\documentclass[11pt,a4paper]{article}

\usepackage{fontspec}
\usepackage[titletoc,title,header]{appendix}

\usepackage{amsmath}
\usepackage{amssymb}
\usepackage{amsthm}
\usepackage{mathtools}
\usepackage{xparse}    

\usepackage[margin=1in]{geometry}

\usepackage{graphicx}
\usepackage{xcolor}

\usepackage{hyperref}
\hypersetup{
  colorlinks=true,
  linkcolor=blue,
  citecolor=blue,
  urlcolor=blue,
  pdfauthor={Yoshitsugu Sekine},
  pdftitle={A Note on the Resolvent Algebra and Functional Integral Approach to the van Hove Model}
}

\usepackage[]{natbib}
\theoremstyle{plain}
\newtheorem{thm}{Theorem}[section]
\newtheorem{lem}[thm]{Lemma}
\newtheorem{prop}[thm]{Proposition}
\newtheorem{cor}[thm]{Corollary}

\newtheorem{fact}[thm]{Fact}

\theoremstyle{definition}
\newtheorem{defn}[thm]{Definition}

\theoremstyle{remark}
\newtheorem{rem}[thm]{Remark}

\makeatletter
\providecommand*{\dashv}{\mathrel{\mathpalette\@Dashv\vDash}}
\newcommand*{\@dashv}[2]{\reflectbox{$\m@th#1#2$}}
\makeatother
\newcommand{\abs}[1]{\left| #1 \right|} 
\NewDocumentCommand{\weaknorm}{O{\dbk} m}{#1{#2}} 
\newcommand{\norm}[1]{\left\Vert #1 \right\Vert} 

\newcommand{\bkt}[2]{\left\langle #1,\,#2 \right\rangle} 
\newcommand{\dualbkt}[2]{\bkt{#1}{#2}_{\txtdual}} 
\newcommand{\rbkt}[2]{\left( #1,\,#2 \right)} 
\newcommand{\slim}{\mathrm{s} \hyphen \lim} 
\newcommand{\wlim}{\mathrm{w} \hyphen \lim} 
\newcommand{\isomto}{\mathrel{\rightarrowtail\kern-1.9ex\twoheadrightarrow}} 
\newcommand{\basebk}[1]{\left\langle #1 \right\rangle} 
\newcommand{\cbk}[1]{\left\{ #1 \right\}} 
\newcommand{\dbk}[1]{\left\langle #1 \right\rangle} 
\newcommand{\pairbk}[1]{\rbk{#1}} 
\newcommand{\rbkleft}[1]{\left( #1 \right.} 
\newcommand{\rbkright}[1]{\left. #1 \right)} 
\newcommand{\rbk}[1]{\left( #1 \right)} 
\newcommand{\sqbk}[1]{\left[ #1 \right]} 
\newcommand{\vecbk}[1]{\rbk{#1}} 
\newcommand{\funrbk}[2]{\fun{\rbk{#1}}{#2}} 
\newcommand{\fun}[2]{#1 \rbk{#2}} 
\newcommand{\sqfuncond}[3]{\sqfun{#1}{#2 \middle| #3}} 
\newcommand{\sqfun}[2]{#1 \sqbk{#2}} 
\newcommand{\closedinterval}[2]{\sqbk{#1,\,#2}} 
\newcommand{\openinterval}[2]{\rbk{#1,\,#2}} 
\newcommand{\rightopeninterval}[2]{\left[#1, \, #2 \right)} 
\newcommand{\commutator}[2]{\sqbk{#1,\,#2}} 
\newcommand{\nin}{n \in \monnat} 

\newcommand{\wick}[1]{\left. :\! \hspace{-0.5pt} #1 \hspace{-0.5pt} \!: \right.} 

\DeclareMathOperator*{\opforall}{\forall}
\DeclareMathOperator{\sgn}{sgn} 
\NewDocumentCommand{\imunit}{O{\mathsf{i}}}{#1} 
\NewDocumentCommand{\placeholder}{O{\bullet}}{#1} 
\NewDocumentCommand{\trace}{O{\operatorname{Tr}}}{#1} 
\newcommand{\Ad}{\operatorname{Ad}} 
\newcommand{\Ker}{\operatorname{Ker}} 
\newcommand{\bigmiddleslash}[2]{\left. #1 \middle/ #2 \right.} 
\newcommand{\cmpconj}[1]{\overline{#1}} 
\newcommand{\diracdelta}{\delta} 
\newcommand{\dom}{\operatorname{dom}} 
\newcommand{\ep}{\varepsilon} 
\newcommand{\eqcsq}[1]{\sqbk{#1}} 
\newcommand{\eqisom}{\cong} 
\newcommand{\fml}[2]{\cbk{#1}_{#2}} 
\newcommand{\hyphen}{\hbox{-}} 
\newcommand{\idone}{1} 
\newcommand{\id}{\mathrm{id}} 
\newcommand{\invrbk}[1]{\rbk{#1}^{-1}} 
\newcommand{\inv}[1]{#1^{-1}} 
\newcommand{\kroneckerdelta}{\delta} 
\newcommand{\napiernum}{\mathsf{e}} 
\newcommand{\od}[2]{\frac{d #1}{d #2}} 
\newcommand{\onehalf}{\frac{1}{2}} 
\newcommand{\oneoverfour}{\frac{1}{4}} 
\newcommand{\opimag}{\operatorname{Im}} 
\newcommand{\opod}[1]{\frac{d}{d #1}} 
\newcommand{\opreal}{\operatorname{Re}} 
\newcommand{\pd}[2]{\frac{\partial #1}{\partial #2}} 
\newcommand{\setSymbolDownLeft}[2]{{\vphantom{#2}}_{#1}{#2}} 
\newcommand{\setSymbolUpLeft}[2]{{\vphantom{#2}}^{#1}{#2}} 
\newcommand{\fnsineintegral}[1]{\fun{\mathrm{Si}}{#1}} 
\newcommand{\intsineintegral}[1]{\int_{0}^{#1} \frac{\sin t}{t} \opdmsr{t}} 
\DeclareMathOperator{\eqalgisom}{\cong} 
\DeclareMathOperator{\eqsimelem}{\sim} 
\DeclareMathOperator{\Rep}{Rep} 
\NewDocumentCommand{\agvariety}{O{\mathcal}}{#1} 
\NewDocumentCommand{\cmdrel}{O{\omega}}{#1} 
\NewDocumentCommand{\dfsp}{O{A}}{#1} 
\NewDocumentCommand{\eqcpointed}{O{\eqcsq} m}{#1{#2}_{\ast}} 
\NewDocumentCommand{\fnheaviside}{O{H}}{#1} 
\NewDocumentCommand{\fthol}{O{\mathcal{O}}}{#1} 
\NewDocumentCommand{\ftmero}{O{\mathcal{M}}}{#1} 
\NewDocumentCommand{\grcentralizer}{O{Z}}{#1} 
\NewDocumentCommand{\grmetform}{O{2} m m}{\grmet[#1] \! \rbkt{#2}{#3}} 
\NewDocumentCommand{\grmet}{O{2}}{\setSymbolDownLeft{#1}{g}} 
\NewDocumentCommand{\grnormalizer}{O{N}}{#1} 
\NewDocumentCommand{\gropasym}{O{A}}{#1} 
\NewDocumentCommand{\gropsym}{O{S}}{#1} 
\NewDocumentCommand{\grpermorderedpair}{O{\mathcal{P}}}{#1} 
\NewDocumentCommand{\grsym}{O{\mathfrak{S}} m}{#1_{#2}} 
\NewDocumentCommand{\gtbase}{O{\mathcal}}{#1} 
\NewDocumentCommand{\gtfilter}{O{\mathcal}}{#1} 
\NewDocumentCommand{\gtfmlclosed}{O{\mathcal}}{#1} 
\NewDocumentCommand{\gtfmlopen}{O{\mathcal}}{#1} 
\NewDocumentCommand{\gtopenball}{O{U}}{#1} 
\NewDocumentCommand{\gtopencover}{O{\mathcal}}{#1} 
\NewDocumentCommand{\gtopennbh}{O{\mathcal}}{#1} 
\NewDocumentCommand{\gtpreopencover}{O{\mathcal}}{#1} 
\NewDocumentCommand{\gtsubbase}{O{\mathcal}}{#1} 
\NewDocumentCommand{\gtvicinity}{O{\mathcal}}{#1} 
\NewDocumentCommand{\lasp}{O{\mathcal}}{#1} 
\NewDocumentCommand{\latpright}{O{\top} m}{#2^{#1}} 
\NewDocumentCommand{\latp}{O{t} m}{\setSymbolUpLeft{#1}{#2}} 
\NewDocumentCommand{\lpdistribution}{O{\mu} m}{#2_{\ast,#1}} 
\NewDocumentCommand{\lpmollifier}{O{\rho}}{#1} 
\NewDocumentCommand{\lpofpositive}{O{\chi}}{#1} 
\NewDocumentCommand{\manliederiv}{O{L}}{#1} 
\NewDocumentCommand{\mansmoothnbh}{O{\mathcal}}{#1} 
\NewDocumentCommand{\mblfmldsysgenerated}{O{d} m}{\fun{#1}{#2}} 
\NewDocumentCommand{\mblfmlgenerated}{O{\sigma} m}{\fun{#1}{#2}} 
\NewDocumentCommand{\oacorrfn}{O{\Gamma}}{#1} 
\NewDocumentCommand{\oagnsvector}{O{\Omega}}{#1} 
\NewDocumentCommand{\oaideal}{O{\mathcal}}{#1} 
\NewDocumentCommand{\oanumberoperator}{O{A}}{#1} 
\NewDocumentCommand{\oaposcone}{O{\mathcal{P}}}{#1} 
\NewDocumentCommand{\oapressure}{O{P}}{#1} 
\NewDocumentCommand{\oarepn}{O{\pi}}{#1} 
\NewDocumentCommand{\oaspnormalstate}{O{N}}{#1} 
\NewDocumentCommand{\oasppurestate}{O{P}}{#1} 
\NewDocumentCommand{\oaspstate}{O{E}}{#1} 
\NewDocumentCommand{\oastatevector}{O{\Omega}}{#1} 
\NewDocumentCommand{\oastate}{O{\omega}}{#1} 
\NewDocumentCommand{\opdilation}{O{\delta}}{#1} 
\NewDocumentCommand{\opdmat}{O{\rho}}{#1} 
\NewDocumentCommand{\opfockan}{O{a}}{#1} 
\NewDocumentCommand{\opfockcran}{O{a}}{#1^{\#}} 
\NewDocumentCommand{\opfockcrdagger}{O{a}}{#1^{\dagger}} 
\NewDocumentCommand{\opfockcr}{O{a}}{#1^{\ast}} 
\NewDocumentCommand{\opfocknumber}{O{N}}{#1} 
\NewDocumentCommand{\opfocksegalconj}{O{\pi}}{#1} 
\NewDocumentCommand{\opfocksegal}{O{\phi}}{#1} 
\NewDocumentCommand{\opspecmeas}{O{E}}{#1} 
\NewDocumentCommand{\opspec}{O{} m}{\fun{\sigma_{#1}}{#2}} 
\NewDocumentCommand{\optransl}{O{\tau}}{#1} 
\NewDocumentCommand{\physaction}{O{\mathcal{A}}}{#1} 
\NewDocumentCommand{\physcharge}{O{e}}{\mathrm{#1}} 
\NewDocumentCommand{\physcplconst}{O{\mathsf{g}}}{#1} 
\NewDocumentCommand{\physelectrostaticcapasity}{O{\mathrm{Cap}}}{#1} 
\NewDocumentCommand{\physenergy}{O{E}}{#1} 
\NewDocumentCommand{\physgse}{O{E}}{#1_{0}} 
\NewDocumentCommand{\physham}{O{H}}{#1} 
\NewDocumentCommand{\physlagdensity}{O{\mathcal{L}}}{#1} 
\NewDocumentCommand{\physlag}{O{L}}{#1} 
\NewDocumentCommand{\physliouvilean}{O{L}}{#1} 
\NewDocumentCommand{\physmass}{O{m}}{#1} 
\NewDocumentCommand{\prbcharfun}{O{\chi}}{#1} 
\NewDocumentCommand{\prbdist}{O{\mathcal{P}}}{#1} 
\NewDocumentCommand{\prbgaussianmeasure}{O{\msrcal{N}}}{#1} 
\NewDocumentCommand{\prbnormaldist}{O{N}}{#1} 
\NewDocumentCommand{\prbprocess}{O{X}}{#1} 
\NewDocumentCommand{\prbqspace}{O{\mathcal{Q}}}{#1} 
\NewDocumentCommand{\prbspsample}{O{\Omega}}{#1} 
\NewDocumentCommand{\psh}{O{\mathfrak}}{#1} 
\NewDocumentCommand{\qtquantumchannel}{O{\mathcal{L}}}{#1} 
\NewDocumentCommand{\repn}{O{\pi}}{#1} 
\NewDocumentCommand{\schattencls}{O{\mathbb{K}}}{#1} 
\NewDocumentCommand{\setfmlcylinder}{O{\mathcal{C}}}{#1} 
\NewDocumentCommand{\setfml}{O{\mathcal}}{#1} 
\NewDocumentCommand{\setindex}{O{\mathcal} m}{#1{#2}} 
\NewDocumentCommand{\setlattice}{O{\Gamma}}{#1} 
\NewDocumentCommand{\setspecial}{O{\mathcal} m}{#1{#2}} 
\NewDocumentCommand{\shdiffform}{O{\sheaf{A}}}{#1} 
\NewDocumentCommand{\sheaf}{O{\mathfrak}}{#1} 
\NewDocumentCommand{\smchemicalpotential}{O{\mu}}{#1} 
\NewDocumentCommand{\smenergydensity}{O{\varrho}}{#1} 
\NewDocumentCommand{\smfluctuationwithdmat}{O{\beta} m}{\smuncertaintywithdmat[#1]{#2}^2} 
\NewDocumentCommand{\sminvtemperature}{O{\beta}}{#1} 
\NewDocumentCommand{\smlocaldensityoperator}{O{\rho}}{#1} 
\NewDocumentCommand{\smmicrocanonicalstate}{O{\beta} m}{\physmean{#2}_{#1}} 
\NewDocumentCommand{\smnumberdensity}{O{\rho}}{#1} 
\NewDocumentCommand{\smooth}{O{\mathcal{E}}}{#1} 
\NewDocumentCommand{\smparticlenumber}{O{N}}{#1} 
\NewDocumentCommand{\smpressure}{O{p}}{#1} 
\NewDocumentCommand{\smspecificfreeenergy}{O{\bar{f}}}{#1} 
\NewDocumentCommand{\smthermalvac}{O{\beta}}{\Omega_{#1}} 
\NewDocumentCommand{\smuncertaintywithdmat}{O{\beta} m}{\rbk{\triangle #2}_{#1}} 
\NewDocumentCommand{\sphilb}{O{\mathcal}}{#1} 
\NewDocumentCommand{\splowerhalf}{O{\mathbb{H}}}{#1_{\txtneg}} 
\NewDocumentCommand{\spupperhalf}{O{\mathbb{H}}}{#1_{\txtnonneg}} 
\NewDocumentCommand{\topmetric}{O{d}}{#1} 
\NewDocumentCommand{\vaoutnormal}{O{\widehat}}{#1} 
\newcommand{\category}[1]{\mathop{\mathsf{#1}}} 

\newcommand{\catpresheaf}[1]{\category{PSh}} 
\newcommand{\conti}{C} 
\newcommand{\dstrapiddec}{\mathcal{S}} 
\newcommand{\dsttempered}{\dstrapiddec^{\ast}} 
\newcommand{\faadjrbk}[1]{\rbk{#1}^{\ast}} 
\newcommand{\faadjpresharp}[1]{#1_{\#}} 
\newcommand{\faadj}[1]{#1^{\ast}} 
\newcommand{\fafouriertr}{\mathcal{F}} 
\newcommand{\faftrrbk}[1]{\widehat{\rbk{#1}}} 
\newcommand{\faftr}[1]{\widehat{#1}} 
\newcommand{\fldcmp}{\fld{C}} 
\newcommand{\fldmultiplicativegroup}[1]{#1^{\times}} 
\newcommand{\fldreal}{\fld{R}} 
\newcommand{\fld}[1]{\mathbb{#1}} 
\newcommand{\fndef}[1]{\boldsymbol{1}_{#1}} 
\newcommand{\fnexp}[1]{\fun{\exp}{#1}} 
\newcommand{\fnrestr}[2]{\left. #1 \right|_{#2}} 
\newcommand{\greuctr}[2]{\fldreal_{#1}^{#2}} 
\newcommand{\gtclos}[1]{\overline{#1}} 
\newcommand{\lacomplexication}[1]{#1_{\fldcmp}} 
\newcommand{\lamat}[2]{\rbk{#1}_{#2}} 
\newcommand{\lpseq}{\ell} 
\newcommand{\lp}{L} 
\newcommand{\mblfmlborel}{\mblfml{B}} 
\newcommand{\mblfmlfrak}[1]{\mathfrak{#1}} 
\newcommand{\mblfml}[1]{\mathcal{#1}} 
\newcommand{\mblfn}{M} 
\newcommand{\monnat}{\mathbb{N}} 
\newcommand{\msrcal}[1]{\mathcal{#1}} 
\newcommand{\msrsf}[1]{\mathsf{#1}} 
\newcommand{\msr}[1]{#1} 

\newcommand{\nfoldvar}[2]{\underline{#1}_{#2}} 
\newcommand{\oacstar}{C^{\ast}} 
\newcommand{\oaderiv}{\delta} 
\newcommand{\oaresolventalgebra}{\oa{R}} 
\newcommand{\oaresolvent}{R} 
\newcommand{\oastaralgebra}{\operatorname{\ast \hyphen \mathrm{alg}}} 
\newcommand{\oaweyl}{\oa{W}} 
\newcommand{\oa}[1]{\mathcal{#1}} 
\newcommand{\opclos}[1]{\overline{#1}} 
\newcommand{\opdmsr}[1]{\mathop{d #1}} 
\newcommand{\opfnresolvent}[1]{\invrbk{#1}} 
\newcommand{\opfockschwingerfunctional}{S} 
\newcommand{\opfocksndqntdiff}{d \Gamma} 
\newcommand{\opfockvac}{\Omega} 
\newcommand{\opfockweyl}{W} 
\newcommand{\opformdomain}{\mathop{Q}} 
\newcommand{\opform}[1]{\mathsf{#1}} 
\newcommand{\oppr}{\mathrm{pr}} 
\newcommand{\opspbddlin}[1]{\fun{\mathbb{B}}{#1}} 
\newcommand{\opspecint}[1]{\mathcal{E}} 
\newcommand{\physmean}[1]{\dbk{#1}} 
\newcommand{\prbcov}{\mathrm{Cov}} 
\newcommand{\prbexp}{\mathbb{E}} 
\newcommand{\prbou}{\xi} 
\newcommand{\prbvar}{\operatorname{Var}} 
\newcommand{\pushoutrbk}[1]{\rbk{#1}_{\ast}} 
\newcommand{\ringratint}{\mathbb{Z}} 
\newcommand{\semigrposint}{\monnat_1} 

\newcommand{\seqn}[1]{\seq{#1_{n}}{\nin}} 
\newcommand{\seq}[2]{\if\relax\detokenize{#1}\relax \rbk{#1} \else \rbk{#1}_{#2} \fi} 
\newcommand{\setcardop}[1]{\operatorname{card} #1} 
\newcommand{\setcpl}[1]{#1^{c}} 
\newcommand{\setisomorphism}[1]{\operatorname{Iso}} 
\newcommand{\setone}[1]{\cbk{#1}}
\newcommand{\setpowfin}[1]{\fun{\mathrm{Pow}_{\txtfin}}{#1}} 
\newcommand{\setquot}[2]{\bigmiddleslash{#1}{#2}} 
\newcommand{\set}[2]{\left\{#1 \, \middle| \, #2\right\}}
\newcommand{\smgrandpartitionfunc}{\Xi} 
\newcommand{\smpartitionfunc}{Z} 
\newcommand{\spfock}{\mathcal{F}} 
\newcommand{\splinspan}{\operatorname{span}} 
\newcommand{\txtbdd}{\mathrm{b}} 
\newcommand{\txtbec}{\mathrm{BEC}} 
\newcommand{\txtborel}{\mathrm{b}} 
\newcommand{\txtbsn}{\mathrm{b}} 
\newcommand{\txtclassical}{\mathrm{cl}} 
\newcommand{\txtcoulomb}{\mathrm{C}} 
\newcommand{\txtcpt}{\mathrm{c}} 
\newcommand{\txtdiscrete}{\mathrm{d}} 
\newcommand{\txtdual}{\mathrm{dual}} 
\newcommand{\txteuclid}{\mathrm{E}} 
\newcommand{\txtexchange}{\mathrm{ex}} 
\newcommand{\txtfin}{\mathrm{fin}} 
\newcommand{\txtfock}{\mathrm{F}} 
\newcommand{\txtfr}{\mathrm{fr}} 
\newcommand{\txtgs}{\mathrm{gs}} 
\newcommand{\txtinflaredred}{\mathrm{IR}} 
\newcommand{\txtirsingular}{\mathrm{irs}} 
\newcommand{\txtneg}{\mathrm{-}} 
\newcommand{\txtnonneg}{\mathrm{+}} 
\newcommand{\txtnonzero}{\mathrm{nz}} 
\newcommand{\txtreal}{\mathrm{real}} 
\newcommand{\txtgrandcanonical}{\mathrm{GC}} 
\newcommand{\txtrenormalization}{\mathrm{ren}} 
\newcommand{\txtself}{\mathrm{self}} 
\newcommand{\txtsym}{\mathrm{s}} 
\newcommand{\txttot}{\mathrm{tot}} 
\newcommand{\txttruncatedfn}{\mathrm{T}} 
\newcommand{\txtultraviolet}{\mathrm{UV}} 
\newcommand{\txtvanhowe}{\mathrm{vH}} 
\newcommand{\txtweak}{\mathrm{w}} 
\newcommand{\vagrad}{\operatorname{grad}} 
\newcommand{\vainnprod}{\cdot} 

\title{A Note on the Resolvent Algebra and Functional Integral Approach to the van Hove Model}

\author{%
Yoshitsugu Sekine\\{\small\texttt{4429sekine@gmail.com}}%
}

\date{\today}

\begin{document}

\maketitle

\begin{abstract}
This paper is a collection of the author's computational notes on the van Hove model and contains no essentially new results. We discuss, from both the operator-algebraic perspective via the Weyl algebra and the resolvent algebra and the functional integral approach, the removal of infrared and ultraviolet cutoffs and the existence of the ground state and \(\beta\)-KMS states in the case of a point source. In the infinite-volume system at finite temperature, Bose--Einstein condensation can arise.

\end{abstract}

\setcounter{tocdepth}{3}
\tableofcontents

\section{Introduction}\label{expedition0012078}

The primary purpose of this paper is to publish computational notes on the van Hove model \cite{LorincziHiroshimaBetz2,LorincziHiroshimaBetz3,AsaoArai26}. We do not intend to present essentially new results; the modest value of this work lies, if any, in the description via the resolvent algebra introduced by Buchholz \cite{DetlevBuchholz001,BuchholzGrundling2} and in translating its general theory into explicit computations for the van Hove model. In constructive quantum field theory and rigorous statistical mechanics, the focus is on concrete ground states and equilibrium states, so the advantages offered by \(\oacstar\)-algebras are often not apparent. We have paid attention to this point throughout our discussion, intending this work to serve more as a review and educational material than as an original research paper.

\subsection{The Van Hove Model}\label{the-van-hove-model}

The van Hove model is an exactly solvable model of quantum field theory defined by the linear coupling \[\physham_{\txtvanhowe}
=
\physham_{\txtbsn,\txtfr}
+\fun{\opfocksegal_{\txtfock}}{\frac{\varrho}{\sqrt{\omega}}}\] of the free scalar field Hamiltonian \(\physham_{\txtbsn,\txtfr}\) via a fixed source \(\varrho\). By the Kato--Rellich theorem, \(\physham_{\txtvanhowe}\) is a self-adjoint operator bounded from below, and it reduces essentially to a free field via a Bogoliubov transformation (a unitary transformation by Weyl operators). In this paper, we choose in particular the point source given by the Dirac delta function at the origin, \(\varrho = \diracdelta_0\), and begin with the cutoff source \(\varrho_{\kappa,\Lambda}\) equipped with an infrared cutoff \(\kappa\) and an ultraviolet cutoff \(\Lambda\), then discuss the removal of these cutoffs.

\subsection{Summary of Main Results}\label{summary-of-main-results}

Without going into details, which are stated as theorems below, the main theorem of this paper reads as follows: even after removing the infrared and ultraviolet cutoffs, the van Hove model admits a ground state and a \(\sminvtemperature\)-KMS state for every inverse temperature \(\sminvtemperature > 0\). Since the analysis reduces essentially to treating a free Bose gas via a Bogoliubov transformation, the model exhibits behavior very close to that of a free Bose field.

The removal of the ultraviolet cutoff requires renormalization of the ground-state energy; in the case of a point source, the self-energy term diverges linearly to \(-\infty\), and upon renormalization a Coulomb potential emerges. The details of the operator-theoretic argument are discussed thoroughly in Arai's textbook \cite{AsaoArai26}.

In the step of removing the infrared cutoff, an infrared divergence appears essentially and manifests as a constraint on observable physical quantities. In the resolvent algebra, this constraint can be expressed in terms of ideal structure.

In equilibrium states at finite temperature, Bose-Einstein condensation (BEC) in the infinite-volume limit can arise, for example below the critical temperature in a three-dimensional system. To make the descriptions of BEC and infrared divergence precise, we sometimes restrict the spatial dimension to \(d = 3\). The appropriate generalization is straightforward. For the relationship to the no-go theorem of BEC for quasiparticles, see \cite{YoshitsuguSekine006}.

\subsection{Organization and References}\label{organization-and-references}

The paper is organized around three main approaches: operator-theoretic arguments in the Weyl algebra, operator-algebraic arguments in the resolvent algebra, and arguments via functional integral theory.

The operator-theoretic treatment of the van Hove model at zero temperature follows Arai's textbook \cite{AsaoArai26} as its foundation, with the particular aim of transplanting the arguments into the operator-algebra setting. The description of Bose-Einstein condensation for the free Bose gas at finite temperature is based on Arai's textbook \cite{AsaoArai28}.

The resolvent algebra is introduced following the formulation of \cite{DetlevBuchholz001,BuchholzGrundling2}, and its application to the van Hove model is discussed. For the resolvent algebra at finite temperature, we also refer to \cite{YoshitsuguSekine004,YoshitsuguSekine005}.

For functional integrals, we discuss the \(Q\)-space representation and the Ornstein--Uhlenbeck representation at zero temperature. The discussion in the \(Q\)-space representation is included in a form that rewrites the treatment in \cite{LorincziHiroshimaBetz2,LorincziHiroshimaBetz3} for a simple special case. The reference \cite{LorincziHiroshimaBetz3} works under infrared regularization. In this paper, we choose the point source and explicitly discuss the removal of the infrared and ultraviolet cutoffs within the functional integral framework; this part may not have an explicit treatment in the existing literature. For the functional integral at finite temperature, we follow the presentation of \cite[Chapter~21]{DerezinskiGerard001}, which improves upon \cite{KleinLandau001}.

\section{Main Results}\label{main-results}

\subsection{Spatial Setup}\label{spatial-setup}

We denote the fundamental complex Hilbert space and its real subspace by \[\sphilb{H}
=
\fun{\lp^{2}}{\fldreal^{d}},
\quad
\sphilb{H}_{\txtreal}
=
\fun{\lp_{\txtreal}^{2}}{\fldreal^{d}}
=
\fun{\lp^{2}}{\fldreal^{d};\fldreal}\] and we may restrict the dimension to \(d = 3\) for simplicity or when discussing BEC, in which case we always state this explicitly. The symplectic form is \(\sigma(f,g) = \opimag \bkt{f}{g}_{\sphilb{H}}\). For the real symplectic space \((X,\sigma)\) we choose \(X\) to be the realification of the complex Hilbert space \(\sphilb{H}\), and the inner product of this real Hilbert space is defined by \(\opreal \bkt{f}{g}\).

For a positive real number \(s > 0\), the single-particle Hamiltonian (dispersion relation) defined on the momentum space \(\greuctr{k}{d}\) is \(\omega(k) = \omega_s(k) = \abs{k}^{s}\). For a non-positive chemical potential \(\smchemicalpotential \leq 0\), the non-negative self-adjoint operator is defined as \(K_{\sminvtemperature,\smchemicalpotential} = \coth \frac{\sminvtemperature \rbk{\omega - \smchemicalpotential}}{2}\). For the associated non-degenerate non-negative symmetric sesquilinear form \(\opform{q}_{\txtnonzero,\sminvtemperature,\smchemicalpotential}\), the associated inner product space and its completion are defined as \[\sphilb{D}_{\sminvtemperature,\smchemicalpotential}
=
\pairbk{\fun{\opformdomain}{\opform{q}_{\txtnonzero,\sminvtemperature,\smchemicalpotential}},\opform{q}_{\txtnonzero,\sminvtemperature,\smchemicalpotential}},
\quad
\sphilb{H}_{\sminvtemperature,\smchemicalpotential}
=
\gtclos{\sphilb{D}_{\sminvtemperature,\smchemicalpotential}}^{\opform{q}_{\txtnonzero,\sminvtemperature,\smchemicalpotential}}.\] Furthermore, when \(\sminvtemperature = \infty\) we regard \(K_{\sminvtemperature,\smchemicalpotential} \to \idone\) and define the sesquilinear form \[\opform{q}_{\txtnonzero,\infty,\smchemicalpotential}(f) = \opform{q}_{\txtnonzero,\infty}(f) = \norm{f}^2.\] Unless otherwise stated and except when discussing BEC, we assume the dispersion relation \(s = 1\), i.e., \(\omega(k) = \abs{k}\).

\begin{rem}
When discussing BEC, we assume phonons in particular.
For the original phonon dispersion $s = 1$, the case $s > 1$ is introduced as a nonlinear term that suppresses excess phonons to prevent infrared and ultraviolet divergences from blowing up, which is to be renormalized into the background field.
As we discuss later, the key point is that the treatment strategy changes significantly between the linear dispersion natural to acoustic phonons and the nonlinear dispersion.
The same arguments apply to other appropriate dispersion relations.
\end{rem}

In general, the bosonic Fock space over a Hilbert space \(\sphilb{H}\) is defined as \(\fun{\spfock_{\txtbsn}}{\sphilb{H}} = \bigoplus_{n=0}^{\infty} \bigotimes_{\txtsym}^{n} \sphilb{H}\), and for any \(f \in \sphilb{H}\) the creation and annihilation operators on the bosonic Fock space are denoted \(\opfockcran_{\txtfock}(f)\). The Segal field operator is defined by \[\opfocksegal_{\txtfock}(f)
=
\frac{1}{\sqrt{2}}
\rbk{\opfockcr_{\txtfock}(f) + \opfockan_{\txtfock}(f)}\] and the Weyl operator is defined as \(\opfockweyl_{\txtfock}(f) = \napiernum^{\imunit \opfocksegal_{\txtfock}(f)}\). The self-adjoint operator \(\physham_{\txtbsn,\txtfr} = \fun{\opfocksndqntdiff_{\txtbsn}}{\omega}\) on Fock space is the second quantization operator determined by the dispersion relation, and its ground state is the Fock vacuum. The Hamiltonian of the van Hove model is defined as a perturbation of the free-field Hamiltonian by the Segal field operator with the source function \(\varrho\) as its argument; specifically, \[\physham_{\txtvanhowe}
=
\physham_{\txtbsn,\txtfr}
+\fun{\opfocksegal_{\txtfock}}{\frac{\varrho}{\sqrt{\omega}}}.\] By the Kato--Rellich theorem, this is a self-adjoint operator bounded from below. The infrared and ultraviolet cutoffs are imposed on the source \(\varrho\), and these will be defined precisely later.

\subsection{Weyl Algebra}\label{weyl-algebra}

Let \(\sphilb{H}\) be a complex Hilbert space, and define the bilinear map \(\sigma\) as the symplectic form by \[\sigma \colon \sphilb{H} \times \sphilb{H} \to \fldreal; \quad \sigma(f,g) = \opimag \bkt{f}{g}_{\sphilb{H}}.\] The \(\oacstar\)-algebra such that for any \(f,g \in \sphilb{H}\) the generators \(\opfockweyl(f)\) satisfy the Weyl relations \begin{equation}
\begin{aligned}
\faadj{\opfockweyl(f)}
&=
\opfockweyl(-f), \\
\opfockweyl(f) \opfockweyl(g)
&=
\napiernum^{-\frac{\imunit}{2} \opimag \bkt{f}{g}_{\sphilb{H}}}
\opfockweyl(f+g)
\end{aligned}
\end{equation} is called the Weyl algebra: \[\oaweyl = \oaweyl(\sphilb{H}, \sigma) = \oaweyl(\sphilb{H}) = \oacstar \set{\opfockweyl(f)}{f \in \sphilb{H}}.\] Where there is no risk of confusion, we use appropriate abbreviations for the Weyl algebra. We may specify an appropriate subspace rather than the full Hilbert space, and the Weyl algebra over the full Hilbert space is also called the full Weyl algebra. When we wish to distinguish clearly from the Fock representation or the Araki--Woods representation of the Weyl algebra, the abstractly defined Weyl algebra above is also called the abstract Weyl algebra.

Let \((\sphilb{H}_{\repn},\repn)\) be a representation of the abstract Weyl algebra \(\oaweyl(\sphilb{H})\). This representation \((\sphilb{H}_{\repn},\repn)\) is called a regular representation if, for any \(f \in \sphilb{H}\), the unitary group \(t \in \fldreal \to \repn(\opfockweyl(tf))\) is strongly continuous. A state \(\omega\) of the Weyl algebra \(\oaweyl(\sphilb{H})\) is called a regular state if its GNS representation is regular. Furthermore, let \(\oaweyl(\sphilb{D})\) denote the Weyl algebra over a pre-Hilbert space \(\sphilb{D}\). A representation \(\pairbk{\sphilb{H},\repn}\) of the Weyl algebra is called a regular representation if, for any \(f \in \sphilb{D}\), the map \(\fldreal \ni t \mapsto \repn(\opfockweyl(tf))\) is strongly continuous. Furthermore, a state \(\oastate\) on the Weyl algebra is called a regular state if its GNS representation is a regular representation.

Denoting by \(\fun{\spfock_{\txtbsn}}{\sphilb{H}}\) the Fock space over the complex Hilbert space \(\sphilb{H}\), and defining the Segal field operator and the symplectic form \(\sigma\) as \[\opfocksegal_{\txtfock}(f) = \frac{1}{\sqrt{2}} \rbk{\opfockcr_{\txtfock}(f) + \opfockan_{\txtfock}(f)}, \quad \sigma_{\txtfock}(f,g) = \opimag \bkt{f}{g}_{\sphilb{H}},\] we define \[\repn_{\txtfock} \colon \oaweyl(\sphilb{H},\sigma) \to \opspbddlin{\fun{\spfock_{\txtbsn}}{\sphilb{H}}}; \quad \repn_{\txtfock}(\opfockweyl(f)) = \opfockweyl_{\txtfock}(f) = \napiernum^{\imunit \opfocksegal_{\txtfock}(f)}.\] This \(\repn_{\txtfock}\) is called the Fock representation of the abstract Weyl algebra, or simply the Fock representation of the Weyl algebra.

\subsection{Resolvent Algebra}\label{expedition0012083}

We introduce the definition and basic properties of the resolvent algebra following \cite{DetlevBuchholz001}. Let \((X,\sigma)\) be a symplectic space. Let \(\oaresolventalgebra_0\) be the universal unital \(\ast\)-algebra generated by the set \(\set{\oaresolvent(\lambda,f)}{\lambda \in \fldmultiplicativegroup{\fldreal}, f \in \sphilb{H}}\), which satisfies in particular the following resolvent relations: \begin{align}
\oaresolvent(\lambda,0)
&=
-\frac{\imunit}{\lambda} \idone, \\ 
\faadj{\oaresolvent(\lambda,f)}
&=
\oaresolvent(-\lambda,f), \\ 
\nu \oaresolvent(\nu \lambda, \nu f)
&=
\oaresolvent(\lambda, f), \\ 
\oaresolvent(\lambda,f) - \oaresolvent(\mu,f)
&=
\imunit
(\mu - \lambda)
\oaresolvent(\lambda,f) \cdot \oaresolvent(\mu,f) \\ 
&=
\imunit
(\mu - \lambda)
\oaresolvent(\mu,f) \cdot \oaresolvent(\lambda,f), \\ 
\commutator{\oaresolvent(\lambda,f)}{\oaresolvent(\mu,g)}
&=
\imunit
\sigma(f,g)
\oaresolvent(\lambda,f)
\oaresolvent(\mu,g)^2
\oaresolvent(\lambda,f), \label{expedition0012052} \\ 
\oaresolvent(\lambda,f)
\oaresolvent(\mu,g)
&=
\oaresolvent(\lambda+\mu, f+g)
\cdot
\rbkleft{\oaresolvent(\lambda,f)} \\
&\quad\rbkright{+
\oaresolvent(\mu,g)
+\imunit \sigma(f,g) \oaresolvent(\lambda,f)^2 \oaresolvent(\mu,g)} 
\end{align} In particular, by condition \eqref{expedition0012052}, \(\oaresolvent(\lambda,f)\) and \(\oaresolvent(\mu,f)\) with the same \(f\) commute.

The \(\ast\)-algebra obtained by introducing an appropriate norm on \(\oaresolventalgebra_0\) and completing it is called the abstract resolvent algebra, or simply the resolvent algebra. For details on the norm, see \cite[P.2730, Definition 3.4]{BuchholzGrundling2}. In particular, by \cite[P.2730, Theorem 3.6 (iii)]{BuchholzGrundling2}, we have \(\norm{\oaresolvent(\lambda,f)} = \frac{1}{\abs{\lambda}}\).

As a dense subalgebra, we choose the \(\ast\)-subalgebra generated by finite products of the generators of \(\oaresolventalgebra(\sphilb{H},\sigma)\); in symbols, \[\oaresolventalgebra_{\txtfin} = \oastaralgebra \set{\prod^{\txtfin} \oaresolvent(z_j,f_j)}{z_j \in \fldcmp \setminus \fldreal, f_j \in X}.\] The \(\ast\)-subalgebra obtained by restricting \(\sphilb{H}\) to an arbitrary subspace \(\sphilb{D}\) is denoted specifically by \(\oaresolventalgebra_{\txtfin}(\sphilb{D},\sigma)\). In discussions of Bose-Einstein condensation for the free Bose gas or the van Hove model, the first argument may become lengthy and its boundary with the second argument may be unclear; in such cases we may write \(\oaresolvent(\lambda;f)\) using a semicolon as a separator between the arguments.

As is well known for ordinary resolvents, the resolvent is analytic in the first variable, and the same property holds for the general resolvent algebra. Using this, we obtain relations that extend \(\lambda \in \fldreal\) of the resolvent algebra to a complex variable \(z \in \fldcmp \setminus \imunit \fldreal\): \begin{align}
\oaresolvent(z,0)
&=
-\frac{\imunit}{z} \idone, \\ 
\faadj{\oaresolvent(z,f)}
&=
\oaresolvent(-\cmpconj{z},f), \\ 
\nu \oaresolvent(\nu z, \nu f)
&=
\oaresolvent(z,f),
\quad
\nu \in \fldmultiplicativegroup{\fldreal}, \\ 
\oaresolvent(z,f) - \oaresolvent(w,f)
&=
\imunit
(w - z)
\oaresolvent(z,f) \cdot \oaresolvent(w,f) \\ 
&=
\imunit
(w - z)
\oaresolvent(w,f) \cdot \oaresolvent(z,f), \\ 
\commutator{\oaresolvent(z,f)}{\oaresolvent(w,g)}
&=
\imunit
\sigma(f,g)
\oaresolvent(z,f)
\oaresolvent(w,g)^2
\oaresolvent(z,f), \\ 
\oaresolvent(z,f)
\oaresolvent(w,g)
&=
\oaresolvent(z+w, f+g)
\cdot
\rbkleft{\oaresolvent(z,f)} \\
&\quad\rbkright{+
\oaresolvent(w,g)
+\imunit \sigma(f,g) \oaresolvent(z,f)^2 \oaresolvent(w,g)} 
\end{align} These are also called the resolvent relations.

Let \(\oaresolventalgebra(X,\sigma)\) denote the resolvent algebra and let \(S\) be a subset of the symplectic space \(X\). A representation \(\oarepn \in \Rep(\oaresolventalgebra(X,\sigma), \sphilb{H}_{\oarepn})\) is called a regular representation on \(S\) if \(\Ker \oarepn(\oaresolvent(1,f)) = \setone{0}\) for all \(f \in S\). A state \(\oastate\) on the resolvent algebra is called a regular state if its GNS representation is a regular representation on \(X\).

\begin{prop}[\cite{BuchholzGrundling2}]\label{expedition0011838}
Let $\pairbk{X,\sigma}$ be a symplectic space of arbitrary dimension and let $S \subset X$ be a non-degenerate finite-dimensional subspace.
\begin{enumerate}
\item
The norms of the full resolvent algebra $\oaresolventalgebra(X,\sigma)$ and the subalgebra $\oaresolventalgebra(S,\sigma)$ coincide on the $\ast$-subalgebra $$\oastaralgebra \set{\oaresolvent(\lambda,f)}{f \in S, \lambda \in \fldreal \setminus \setone{0}}.$$
In particular, $\oaresolventalgebra(S,\sigma) \subset \oaresolventalgebra(X,\sigma)$ holds.

\item
The full resolvent algebra is the inductive limit of the net $\fml{\oaresolventalgebra(S,\sigma)}{X \subset S}$ over non-degenerate finite-dimensional subspaces $S \subset X$.

\item
Any regular representation of the full resolvent algebra $\oaresolventalgebra(X,\sigma)$ is faithful.
\end{enumerate}
In particular, the center of the full resolvent algebra is trivial.
\end{prop}

\subsection{Setup for Finite Temperature}\label{expedition0011278}

We again use the setup adopted in \cite{YoshitsuguSekine004}. Setting the closed interval of side length \(L > 0\) as \(I_L = \closedinterval{-\frac{L}{2}}{\frac{L}{2}}\), the Hilbert space of single-particle states moving in the hypercube \(I_{L}^{d}\) is defined as \[\fun{\lp^{2}}{I_{L}^{d}}
=
\set{f \in \fun{\mblfn_{\txtborel}}{I_{L}^{d};\fldcmp}}
{\int_{I_{L}^{d}} \abs{f(x)}^2 \opdmsr{x} < \infty}\] and the momentum-space lattice for the bounded system is \[\setlattice_L
=
\frac{2 \pi}{L} \ringratint 
=
\set{\frac{2 \pi}{L} n}{n \in \ringratint}.\] We impose periodic boundary conditions on \(I_{L}^{d}\) for the single-particle Hamiltonian or dispersion relation \(\omega\). For a positive real number \(\smchemicalpotential > 0\), we set \(\physham_{\txtbsn,\txtfr}(\smchemicalpotential) = \fun{\opfocksndqntdiff_{\txtbsn}}{\omega - \smchemicalpotential}\), and as an operator on the bosonic Fock space \(\fun{\spfock_{\txtbsn}}{\fun{\lp^{2}}{I_L^d}}\), we define the Hamiltonian of the van Hove model in the bounded system with chemical potential as \[\physham_{\txtvanhowe,I_{L}^{d},\kappa,\Lambda}(\smchemicalpotential)
=
\physham_{\txtbsn,\txtfr}(\smchemicalpotential)
+\fun{\opfocksegal_{\txtfock}}{\omega \mathsf{m}_{\kappa,\Lambda}}.\] By the Kato--Rellich theorem, this is a self-adjoint operator bounded from below. The automorphism group \(\alpha_{\txtvanhowe,I_{L}^{d},\kappa,\Lambda,\smchemicalpotential}\) generated by this Hamiltonian is given by \[\alpha_{\txtvanhowe,I_{L}^{d},\kappa,\Lambda,\smchemicalpotential,t}
=
\Ad U_{I_{L}^{d},\kappa,\Lambda,\smchemicalpotential,t},
\quad
U_{I_{L}^{d},\kappa,\Lambda,\smchemicalpotential,t}
=
\napiernum^{\imunit t \physham_{\txtvanhowe,I_{L}^{d},\kappa,\Lambda}(\smchemicalpotential)}.\] For inverse temperature \(\sminvtemperature > 0\), the density operator \(\opdmat_{\txtvanhowe,I_{L}^{d},\kappa,\Lambda,\sminvtemperature,\smchemicalpotential}\) and the grand partition function \(\smgrandpartitionfunc_{\txtvanhowe,I_{L}^{d},\kappa,\Lambda,\sminvtemperature,\smchemicalpotential}\) are defined by \begin{equation}
\begin{aligned}
\opdmat_{\txtvanhowe,I_{L}^{d},\kappa,\Lambda,\sminvtemperature,\smchemicalpotential}
&=
\frac{1}
{\smgrandpartitionfunc_{\txtvanhowe,I_{L}^{d},\kappa,\Lambda,\sminvtemperature,\smchemicalpotential}}
\napiernum^{-\sminvtemperature \physham_{\txtvanhowe,I_{L}^{d},\kappa,\Lambda}(\smchemicalpotential)},
\\ 
\smgrandpartitionfunc_{\txtvanhowe,I_{L}^{d},\kappa,\Lambda,\sminvtemperature,\smchemicalpotential}
&=
\sqfun{\trace}{\napiernum^{-\sminvtemperature
\physham_{\txtvanhowe,I_{L}^{d},\kappa,\Lambda}(\smchemicalpotential)}}
\end{aligned}
\end{equation} For the infinite system, we use the same notation with \(I_{L}^{d}\) removed, and define analogously. When the chemical potential is \(0\), we simply drop the chemical potential notation.

For the local density operator \[\smlocaldensityoperator_{\sminvtemperature,\smchemicalpotential}
=
\frac{1}{\napiernum^{\sminvtemperature (\omega - \smchemicalpotential)} - 1},\] the operator \(K_{\sminvtemperature,\smchemicalpotential}\) is defined as \[K_{\sminvtemperature,\smchemicalpotential}
=
2 \smlocaldensityoperator_{\sminvtemperature,\smchemicalpotential} + 1
=
\frac{1 + \napiernum^{-\sminvtemperature (\omega - \smchemicalpotential)}}
{1 - \napiernum^{-\sminvtemperature (\omega - \smchemicalpotential)}}.\] When the chemical potential is \(0\), we simply drop the chemical potential notation and write \(\smlocaldensityoperator_{\sminvtemperature}\) and \(K_{\sminvtemperature}\). In the system with chemical potential set to \(0\), the algebra of observables is constrained. In particular, the single-particle subspace presupposing the system with infrared and ultraviolet divergences removed is \[\sphilb{D}_{\txtirsingular,\sminvtemperature}
=
\dom \mathsf{m} \cap \sphilb{D}_{\sminvtemperature}.\]

Under the above setup, the algebra of observables as a Weyl algebra is \[\oaweyl_{\txtirsingular,\sminvtemperature}
=
\oaweyl(\sphilb{D}_{\txtirsingular,\sminvtemperature})
=
\oacstar
\set{\opfockweyl(f)}
{f \in \sphilb{D}_{\txtirsingular,\sminvtemperature}}.\]

\subsection{Quantities Related to Bose-Einstein Condensation}\label{quantities-related-to-bose-einstein-condensation}

Let \(\smnumberdensity_0(\sminvtemperature)\) denote the condensate density at inverse temperature \(\sminvtemperature > 0\). Setting the non-closed non-negative symmetric bilinear form corresponding to the condensate component as \[\opform{q}_{0}(f)
=
2 (2 \pi)^d \smnumberdensity_0(\sminvtemperature)
\abs{\faftr{f}(0)}^{2},
\quad
\opformdomain(\opform{q}_{0})
=
\fun{\lp^{1}}{\fldreal^{d}}
\cap
\fun{\lp^{2}}{\fldreal^{d}},\] we define the subspace \(\sphilb{D}_{0,\sminvtemperature,\smchemicalpotential} = \opformdomain(\opform{q}_0) \cap \sphilb{H}_{\sminvtemperature,\smchemicalpotential}\). When the value of the chemical potential \(\smchemicalpotential\) has no particular meaning, or when \(\smchemicalpotential = 0\), we drop the chemical potential subscript from each of the above objects. Furthermore, for any \(f \in \sphilb{D}_{0,\sminvtemperature}\) we define the sesquilinear form \[\opform{q}_{\txtbec}(f)
=
\opform{q}_{0}(f)
+\opform{q}_{\txtnonzero,\sminvtemperature}(f).\]

\subsection{Theorems}\label{theorems}

The main theorem concerns the existence of the ground state and equilibrium state.

\begin{thm}
The van Hove model admits a ground state and a $\sminvtemperature$-KMS state for every inverse temperature $\sminvtemperature > 0$.
These persist even after removing the infrared and ultraviolet cutoffs.
\end{thm}

The operator-theoretic argument for the ground state is discussed in detail in Arai's textbook \cite{AsaoArai26}. The reference \cite{LorincziHiroshimaBetz3} mentions only briefly the treatment via the functional integral approach. Here we discuss both the \(Q\)-space representation argument and the Ornstein--Uhlenbeck representation. The equilibrium state is discussed in both the Weyl algebra and resolvent algebra frameworks, as well as via the functional integral approach.

The following estimate holds for the ground-state energy. This is also mentioned in \cite{AsaoArai26}.

\begin{thm}
Let the source be $\rho = \diracdelta_0$.
Then the removal of the ultraviolet cutoff gives rise to energy renormalization.
In particular, the self-energy term diverges linearly to $-\infty$, and upon renormalization a Coulomb potential emerges.
\end{thm}

We also briefly address BEC in the equilibrium state. Since the argument reduces via the Bogoliubov transformation to one essentially equivalent to that for the free field, it is largely covered by \cite{AsaoArai28,YoshitsuguSekine004}. In particular, there is an argument directly related to the no-go theorem of BEC for quasiparticles, which is discussed in \cite{YoshitsuguSekine006}.

\section{Setup for the Van Hove Model}\label{expedition0011236}

The source \(\varrho\) representing the interaction term is the Dirac delta function concentrated at the origin, i.e., \(\varrho = \diracdelta_0\), and the source \(\varrho_{\kappa,\Lambda}\) with infrared and ultraviolet cutoffs added is defined by \[\faftr{\varrho_{\kappa,\Lambda}}(k)
=
\faftr{\varrho}(k) \fndef{\kappa \leq \abs{k} \leq \Lambda}(k),
\quad
\faftr{\varrho}(k)
=
\faftr{\diracdelta_0}(k)
=
1.\] In particular, when the infrared cutoff is removed by setting \(\kappa = 0\) we write \(\varrho_{\Lambda}\), and when the ultraviolet cutoff is removed by setting \(\Lambda = \infty\) we write \(\varrho_{\kappa}\).

For any \(\kappa > 0\), the infrared regularization condition \(\varrho_{\kappa,\Lambda} \in \dom \omega^{-\frac{3}{2}}\), which amounts to \[\int_{\greuctr{k}{3}}
\frac{\abs{\faftr{\varrho_{\kappa,\Lambda}}(k)}^2}{\omega(k)^3}
\opdmsr{k}
=
\int_{\kappa \leq \abs{k} \leq \Lambda}
\frac{1}{\omega(k)^3}
\opdmsr{k}
<
\infty,\] and the convergence in the topology of the space of tempered distributions \(\fun{\dsttempered}{\fldreal^{3}}\), \[\lim_{\kappa \to 0, \Lambda \to \infty}
\varrho_{\kappa,\Lambda}
=
\varrho,\] both hold. In particular, the cutoff source satisfies the condition corresponding to \(\varrho \in \dom \omega^{-1}\) around the origin, namely the integrability condition for every \(\kappa > 0\): \[\int_{\abs{k} \leq \kappa}
\frac{\abs{\faftr{\varrho_{\kappa,\Lambda}}(k)}^2}{\omega(k)^2}
\opdmsr{k}
<
\infty.\]

We define the mean functionals \(\mathsf{m}_{\kappa,\Lambda}\), \(\mathsf{m}\) and the domain \(\dom \mathsf{m}\) using the source as follows.

\begin{defn}\label{expedition0011100}
In order to make sense as functions in momentum space via the Fourier transform, we define $$\mathsf{m}_{\kappa,\Lambda}
=
\omega^{-\frac{3}{2}} \varrho_{\kappa,\Lambda},
\quad
\mathsf{m}
=
\omega^{-\frac{3}{2}} \varrho,$$
and we also use the same symbols for the mean functionals they induce.
In particular,
\begin{equation}
\begin{aligned}
\mathsf{m}_{\kappa,\Lambda}(f)
&=
\bkt{\omega^{-\frac{3}{2}} \varrho_{\kappa,\Lambda}}{f}_{\sphilb{H}}, \\
\mathsf{m}(f)
&=
\int_{\greuctr{k}{3}}
\frac{\cmpconj{\faftr{\varrho}(k)} \faftr{f}(k)}{\omega(k)^{\frac{3}{2}}}
\opdmsr{k}
\end{aligned}
\end{equation}
the domain of the former is all of $\sphilb{H}$, and the domain of the latter is $\dom \mathsf{m}$.
\end{defn}

Since the \(\lp^{2}\)-norm of the function \(\mathsf{m}\) without cutoffs is \[\norm{\mathsf{m}}_{\fun{\lp^{2}}{\fldreal^{d}}}^2
=
\int_{\greuctr{k}{3}}
\frac{\abs{\faftr{\varrho}(k)}^2}{\omega(k)^3}
\opdmsr{k},\] the integrability of \(\mathsf{m}\) as a function can also be used to diagnose infrared singularities.

As discussed in the operator-theoretic formulation of \cite{AsaoArai26} (using different notation), the square-integrability of the function \(\mathsf{m}\) around the origin of momentum space is a property that controls the infrared divergence. With an infrared cutoff (and ultraviolet cutoff), this square-integrability is always guaranteed; otherwise, when an infrared divergence occurs, the condition that constrains the observable physical quantities is encoded in the definition of the subspace \(\dom \mathsf{m}\).

\begin{rem}
The better the properties of the source $\varrho$, the larger $\dom \mathsf{m}$ can be taken, increasing the number of observable physical quantities; conversely, the worse the properties of $\varrho$, the smaller $\dom \mathsf{m}$ becomes, reducing the number of physical quantities that can be observed without being affected by the divergence.
As the notation suggests, this is also influenced by $\omega$ and by the spatial dimension.
We consider the domain of the functional $\mathsf{m}$ that appears in the definition of the automorphism group.

If the source $\varrho$ is a rapidly decreasing function and $\omega(k) = \abs{k}$, then in the infrared regime there is no singularity around the origin and it is not necessary to incorporate any constraint into $\dom \mathsf{m}$.
This follows from the $\abs{k}^{-\frac{3}{2}}$ in the numerator and the $k^2$ from the Jacobian of the change of variables.
Even in the ultraviolet regime, the rapid decay of $\varrho$ controls the divergence at infinity, so again no constraint on $\dom \mathsf{m}$ is needed.

If the dispersion relation is $\omega(k) = \abs{k}$ and $\varrho$ is a point source, then, as before, the infrared regime does not affect $\dom \mathsf{m}$.
However, in the ultraviolet regime, to suppress the divergence arising from $\abs{k}^{\onehalf}$, an ultraviolet integrability condition such as $$\int_{\abs{k} \geq 1}
\abs{k}^{\onehalf}
\faftr{f}(k)
\opdmsr{k}
<
\infty$$
is required, which does constrain $\dom \mathsf{m}$.

Next, suppose $\varrho$ is a point source and the dispersion relation is $\omega(k) = k^{2+\delta}$ for some $\delta > 0$.
In this case, in the infrared regime, after the $\abs{k}^{-3}$ in the numerator is partly cancelled by the Jacobian, a residual $\abs{k}^{-1}$ remains, and $\dom \mathsf{m}$ is constrained in order to avoid logarithmic divergence.
In particular, behavior of the form $\faftr{f}(k) = \fun{O}{\abs{k}^{1+\delta}}$ as $k \to 0$ is required, which in particular demands $\faftr{f}(0) = 0$.
Conversely, in the ultraviolet regime, the dispersion contributes as a damping term, thereby relaxing the constraints on elements of $\dom \mathsf{m}$.
\end{rem}

We define auxiliary functionals \(\mathsf{M}_{\kappa,\Lambda,t}\) and \(\mathsf{M}_{t}\) for notational conciseness.

\begin{defn}\label{expedition0011628}
As auxiliary functionals, we define
\begin{equation}
\begin{aligned}
\mathsf{M}_{\kappa,\Lambda,t}(f)
&=
\opreal \fun{\mathsf{m}_{\kappa,\Lambda}}{\rbk{\napiernum^{\imunit t \omega} - \idone} f}, \\
\mathsf{M}_{t}(f)
&=
\opreal \fun{\mathsf{m}}{\rbk{\napiernum^{\imunit t \omega} - \idone} f}.
\end{aligned}
\end{equation}
\end{defn}

\begin{prop}\label{expedition0011238}
The auxiliary functionals $\mathsf{M}_{\kappa,\Lambda,t}$ and $\mathsf{M}_{t}$ satisfy the cocycle condition
\begin{equation}
\begin{aligned}
\mathsf{M}_{\kappa,\Lambda,t+s}(f)
&=
\mathsf{M}_{\kappa,\Lambda,s}(\napiernum^{\imunit t \omega} f) + \mathsf{M}_{\kappa,\Lambda,t}(f), \\
\mathsf{M}_{t+s}(f)
&=
\mathsf{M}_{s}(\napiernum^{\imunit t \omega} f) + \mathsf{M}_{t}(f).
\end{aligned}
\end{equation}
\end{prop}

\begin{proof}
Since the two cases are analogous, we work with $\mathsf{M}_{t}(f)$ for notational simplicity.
A direct computation gives
\begin{equation}
\begin{aligned}
&\mathsf{M}_{t+s}(f)
=
\opreal \bkt{\mathsf{m}}{\rbk{\napiernum^{\imunit (t+s) \omega} - 1} f} \\
&=
\opreal \bkt{\mathsf{m}}{\rbk{\napiernum^{\imunit s \omega} - 1} \napiernum^{\imunit t \omega} f}
+\opreal \bkt{\mathsf{m}}{\rbk{\napiernum^{\imunit t \omega} - 1} f} \\
&=
\mathsf{M}_{s}(\napiernum^{\imunit t \omega} f)
+\mathsf{M}_{t}(f).
\end{aligned}
\end{equation}
\end{proof}

We now define the concrete van Hove model on Fock space with infrared and ultraviolet cutoffs as \[\physham_{\txtvanhowe,\kappa,\Lambda}
=
\physham_{\txtbsn,\txtfr}
+\fun{\opfocksegal_{\txtfock}}{\frac{\varrho_{\kappa,\Lambda}}{\omega^{\onehalf}}}
=
\physham_{\txtbsn,\txtfr}
+\fun{\opfocksegal_{\txtfock}}{\omega \mathsf{m}_{\kappa,\Lambda}}.\] By the Kato--Rellich theorem, this is also a self-adjoint operator bounded from below.

\begin{prop}\label{expedition0011094}
For any $t \in \fldreal$, defining the self-map $\alpha_{\txtvanhowe,\kappa,\Lambda,t}$ of the abstract Weyl algebra by $$\alpha_{\txtvanhowe,\kappa,\Lambda,t}(\opfockweyl(f))
=
\napiernum^{\imunit \mathsf{M}_{\kappa,\Lambda,t}(f)}
\opfockweyl(\napiernum^{\imunit t \omega} f)$$
yields an automorphism group of the abstract Weyl algebra.
\end{prop}

\begin{rem}
As will be shown below, this is the action of the Hamiltonian of the van Hove model on the Weyl operators.
Since we are working at an abstract level, note that the automorphism property holds for all appropriate representations discussed below.
\end{rem}

\begin{proof}
($\ast$-isomorphism: preservation of Weyl relations): First, $\alpha_{\txtvanhowe,\kappa,\Lambda,0} = \idone$ is clear.
For each $t$, the inverse map is the map with parameter $-t$, so it remains to verify homomorphism.

Since the operator $\napiernum^{\imunit t \omega}$ is unitary, it preserves the symplectic form determined by the inner product.
Since $\mathsf{M}_{\kappa,\Lambda,t}$ arises from the inner product, it is additive.
The remaining step is a direct computation:
\begin{equation}
\begin{aligned}
&\alpha_{\txtvanhowe,\kappa,\Lambda,t}(\opfockweyl(f))
\cdot
\alpha_{\txtvanhowe,\kappa,\Lambda,t}(\opfockweyl(g)) \\
&=
\napiernum^{\imunit \rbk{\mathsf{M}_{\kappa,\Lambda,t}(f) + \mathsf{M}_{\kappa,\Lambda,t}(g)}}
\opfockweyl(\napiernum^{\imunit t \omega} f)
\opfockweyl(\napiernum^{\imunit t \omega} g) \\
&=
\napiernum^{\imunit \rbk{\mathsf{M}_{\kappa,\Lambda,t}(f) + \mathsf{M}_{\kappa,\Lambda,t}(g)}}
\napiernum^{-\frac{\imunit}{2} \opimag \bkt{\napiernum^{\imunit t \omega} f}{\napiernum^{\imunit t \omega} g}}
\opfockweyl(\napiernum^{\imunit t \omega}(f+g)), \\
&=
\napiernum^{-\frac{\imunit}{2} \opimag \bkt{f}{g}}
\napiernum^{\imunit \mathsf{M}_{\kappa,\Lambda,t}(f+g)}
\opfockweyl(\napiernum^{\imunit t \omega}(f+g)), \\
&\alpha_{\txtvanhowe,\kappa,\Lambda,t}(\opfockweyl(f) \opfockweyl(g)) \\
&=
\napiernum^{-\frac{\imunit}{2} \opimag \bkt{f}{g}}
\alpha_{\txtvanhowe,\kappa,\Lambda,t}(\opfockweyl(f+g)) \\
&=
\napiernum^{-\frac{\imunit}{2} \opimag \bkt{f}{g}}
\napiernum^{\imunit \mathsf{M}_{\kappa,\Lambda,t}(f+g)}
\opfockweyl(\napiernum^{\imunit t \omega}(f+g))
\end{aligned}
\end{equation}
and the two expressions are equal.

(Adjoint): Since the map $\mathsf{M}_{\kappa,\Lambda,t}$ satisfies $\mathsf{M}_{\kappa,\Lambda,t}(-f) = -\mathsf{M}_{\kappa,\Lambda,t}(f)$, we have
\begin{equation}
\begin{aligned}
&\faadj{\alpha_{\txtvanhowe,\kappa,\Lambda,t}(\opfockweyl(f))}
=
\napiernum^{-\imunit \mathsf{M}_{\kappa,\Lambda,t}(f)}
\opfockweyl(-\napiernum^{\imunit t \omega} f) \\
&=
\napiernum^{\imunit \mathsf{M}_{\kappa,\Lambda,t}(-f)}
\opfockweyl(\napiernum^{\imunit t \omega} (-f)) \\
&=
\fun{\alpha_{\txtvanhowe,\kappa,\Lambda,t}}{\opfockweyl(-f)}
=
\fun{\alpha_{\txtvanhowe,\kappa,\Lambda,t}}{\faadj{\opfockweyl(f)}}.
\end{aligned}
\end{equation}

(Group property): Let
\begin{equation}
\begin{aligned}
\alpha_{\txtvanhowe,\kappa,\Lambda,t}(\opfockweyl(f))
&=
\napiernum^{\imunit \mathsf{M}_{\kappa,\Lambda,t}(f)} \opfockweyl(\napiernum^{\imunit t \omega} f), \\
\alpha_{\txtvanhowe,\kappa,\Lambda,s}(\opfockweyl(f))
&=
\napiernum^{\imunit \mathsf{M}_{\kappa,\Lambda,s}(f)} \opfockweyl(\napiernum^{\imunit s \omega} f).
\end{aligned}
\end{equation}
Then $$\alpha_{\txtvanhowe,\kappa,\Lambda,t}
\circ
\alpha_{\txtvanhowe,\kappa,\Lambda,s}(\opfockweyl(f))
=
\napiernum^{\imunit \mathsf{M}_{\kappa,\Lambda,s}(f)}
\napiernum^{\imunit \mathsf{M}_{\kappa,\Lambda,t}(\napiernum^{\imunit s \omega} f)}
\opfockweyl(\napiernum^{\imunit t \omega} \napiernum^{\imunit s \omega} f).$$
On the other hand, $$\alpha_{\kappa,\Lambda,t+s}(\opfockweyl(f))
=
\napiernum^{\imunit \mathsf{M}_{\kappa,\Lambda,t+s}(f)}
\opfockweyl(\napiernum^{\imunit (t+s) \omega} f).$$
The equality of these two expressions is the cocycle condition, which was established in Proposition~\ref{expedition0011238}.
\end{proof}

\section{Ground State in the Weyl Algebra: Reference Implementation}\label{expedition0012095}

As a reference implementation, we discuss the regularized case in the Fock representation. The full treatment involving the removal of infrared and ultraviolet cutoffs is left to the discussions in the resolvent algebra and the functional integral. Since the ground-state energy will be derived explicitly using the functional integral later, some results are accepted without proof here.

\subsection{Discussion in the Fock Representation}\label{discussion-in-the-fock-representation}

Here we present the operator-theoretic arguments of \cite{AsaoArai26} reformulated in terms of the Weyl algebra. In principle the notation follows \cite{YoshitsuguSekine006}.

For the van Hove Hamiltonian \(\physham_{\txtvanhowe,\kappa,\Lambda}\) with infrared and ultraviolet cutoffs on the bosonic Fock space, the ground-state energy, as established by \cite{AsaoArai26} or by the functional integral argument, is \(\fun{\physgse}{\physham_{\txtvanhowe,\kappa,\Lambda}} = -\onehalf \norm{\inv{\omega} \varrho_{\kappa,\Lambda}}_{\sphilb{H}}^2\). Under the infrared regularization \(\kappa > 0\), we define the unitary operator \(V_{\kappa,\Lambda} = \fnexp{\imunit \fun{\opfocksegal_{\txtfock}}{\imunit \mathsf{m}_{\kappa,\Lambda}}}\). Essentially the same transformation is used in the discussion of Hubbard--phonon interaction systems in \cite{YoshitsuguSekine001,YoshitsuguSekine002}. For various computations, we cite from \cite{AsaoArai26,AsaoArai28} the basic properties of Weyl operators and results on the commutation relations with creation and annihilation operators.

\begin{fact}\label{expedition0010960}
Let $f,g$ be arbitrary elements of the Hilbert space $\sphilb{H}$.
\begin{enumerate}
\item
The Weyl operator leaves the domain of the creation and annihilation operators invariant.
That is, $$\napiernum^{\imunit \opfocksegal(f)} \opfockcran(g) = \opfockcran(g).$$
Furthermore, the operator identities
\begin{align}
\napiernum^{\imunit \opfocksegal(f)}
\opfockan(g)
\napiernum^{-\imunit \opfocksegal(f)}
&=
\opfockan(g)
-\frac{\imunit}{\sqrt{2}}
\bkt{g}{f}, \\ 
\napiernum^{\imunit \opfocksegal(f)}
\opfockcr(g)
\napiernum^{-\imunit \opfocksegal(f)}
&=
\opfockcr(g)
+\frac{\imunit}{\sqrt{2}}
\bkt{f}{g} 
\end{align}
hold.

\item
The Weyl operator leaves the domain of the Segal field operator invariant.
That is, $$\napiernum^{\imunit \opfocksegal(f)} \dom \opfocksegal(g) = \dom \opfocksegal(g).$$
Furthermore, the operator identity
\begin{align}
\napiernum^{\imunit \opfocksegal(f)}
\opfocksegal(g)
\napiernum^{-\imunit \opfocksegal(f)}
=
\opfocksegal(g)
-\opimag \bkt{f}{g} 
\end{align}
holds.
\end{enumerate}
\end{fact}

As confirmed in \cite{YoshitsuguSekine006}, under the infrared regularization \(\kappa > 0\) we obtain \[V_{\kappa,\Lambda}
\physham_{\txtbsn,\txtfr}
\inv{V_{\kappa,\Lambda}}
=
\physham_{\txtvanhowe,\kappa,\Lambda}
-\physenergy_0(\physham_{\txtvanhowe,\kappa,\Lambda}).\] Setting \(\Psi_{\txtvanhowe,\txtgs,\txtfock,\kappa,\Lambda} = V_{\kappa,\Lambda} \opfockvac_{\txtbsn}\) for the Fock vacuum \(\opfockvac_{\txtbsn}\), this is the ground state of the van Hove Hamiltonian with infrared and ultraviolet cutoffs. We take the state \(\oastate[\psi_{\txtvanhowe,\txtgs,\txtfock,\kappa,\Lambda}]\) on the algebra of all bounded linear operators to be the vector state determined by the ground state \(\Psi_{\txtvanhowe,\txtgs,\txtfock,\kappa,\Lambda}\), and define the non-negative self-adjoint operator \[\physham_{\oastate[\psi_{\txtvanhowe,\txtgs,\txtfock,\kappa,\Lambda}]}
=
\physham_{\txtvanhowe,\kappa,\Lambda}
-\fun{\physgse}{\physham_{\txtvanhowe,\kappa,\Lambda}}\] shifted by the ground-state energy. We denote by \(\alpha_{\txtvanhowe,\txtfock,\kappa,\Lambda,t}\) the automorphism group generated by the derivation \[\oaderiv_{\psi_{\txtvanhowe,\txtgs,\txtfock,\kappa,\Lambda}}(\repn_{\txtfock}(A))
=
\commutator{\physham_{\psi_{\txtvanhowe,\txtgs,\txtfock,\kappa,\Lambda}}}{\repn_{\txtfock}(A)},
\quad
A \in \opspbddlin{\fun{\spfock_{\txtbsn}}{\sphilb{H}}}.\] By \cite{BratteliRobinson2}, the Fock representation of the Weyl algebra is a regular representation, and the automorphism group \(\alpha_{\txtvanhowe,\txtfock,\kappa,\Lambda,t}\) of the van Hove model on Fock space is strongly continuous.

\begin{prop}\label{expedition0011107}
Regard the state $\psi_{\txtvanhowe,\txtgs,\txtfock,\kappa,\Lambda}
\in
\oaspstate_{\repn_{\txtfock}(\oaweyl(\sphilb{H}))}$ on the Fock representation of the Weyl algebra as the vector state defined by $$\psi_{\txtvanhowe,\txtgs,\txtfock,\kappa,\Lambda}(\repn_{\txtfock}(A))
=
\bkt{\Psi_{\txtvanhowe,\txtgs,\txtfock,\kappa,\Lambda}}
{\repn_{\txtfock}(A) \Psi_{\txtvanhowe,\txtgs,\txtfock,\kappa,\Lambda}},
\quad
\Psi_{\txtvanhowe,\txtgs,\txtfock,\kappa,\Lambda}
=
V_{\kappa,\Lambda}
\opfockvac_{\txtbsn}.$$
This satisfies $\fun{\psi_{\txtvanhowe,\txtgs,\txtfock,\kappa,\Lambda}}{\physham_{\psi_{\txtvanhowe,\txtgs,\txtfock,\kappa,\Lambda}}} = 0$, and in particular it is a ground state of the automorphism group $\alpha_{\txtvanhowe,\txtfock,\kappa,\Lambda}$ on the Fock representation of the Weyl algebra.

For the unitary operator $U_{\psi_{\txtvanhowe,\txtgs,\txtfock,\kappa,\Lambda},t}
=
\napiernum^{\imunit t \physham_{\psi_{\txtvanhowe,\txtgs,\txtfock,\kappa,\Lambda}}}$, the automorphism group can be written as $$\alpha_{\txtvanhowe,\txtfock,\kappa,\Lambda,t}(\repn_{\txtfock}(A))
=
U_{\psi_{\txtvanhowe,\txtgs,\txtfock,\kappa,\Lambda},t}
\repn_{\txtfock}(A)
\inv{U_{\psi_{\txtvanhowe,\txtgs,\txtfock,\kappa,\Lambda},t}},$$
and the one-parameter unitary group determined by $U_{\psi_{\txtvanhowe,\txtgs,\txtfock,\kappa,\Lambda},t}$ is continuous in the strong operator topology.
\end{prop}

\begin{proof}
(Unitary group generating the automorphism group): Let $U_{0,t} = \napiernum^{\imunit t \physham_{\txtbsn,\txtfr}}$ be the unitary group generated by the free Hamiltonian, and set $$V_t
= U_{0,t} V_{\kappa,\Lambda} \inv{U_{0,t}}
= \fun{\opfockweyl_{\txtfock}}{\imunit \napiernum^{\imunit t \omega} \mathsf{m}_{\kappa,\Lambda}}.$$
Setting $U_{t} = V_{\kappa,\Lambda} U_{0,t} \inv{V_{\kappa,\Lambda}}$, we obtain
\begin{equation}
\begin{aligned}
&U_{t} \opfockweyl_{\txtfock}(f) \inv{U_{t}}
=
V_{\kappa,\Lambda} U_{0,t}
\rbk{\inv{V_{\kappa,\Lambda}} \opfockweyl_{\txtfock}(f) V_{\kappa,\Lambda}}
\inv{U_{0,t}} \inv{V_{\kappa,\Lambda}} \\
&=
\napiernum^{-\imunit \opimag \bkt{f}{\imunit \mathsf{m}_{\kappa,\Lambda}}}
V_{\kappa,\Lambda} \opfockweyl_{\txtfock}(\napiernum^{\imunit t \omega} f) \inv{V_{\kappa,\Lambda}} \\
&=
\napiernum^{-\imunit \opreal \bkt{f}{\mathsf{m}_{\kappa,\Lambda}}}
\napiernum^{-\imunit \opimag \bkt{\napiernum^{\imunit t \omega} f}{-\imunit \mathsf{m}_{\kappa,\Lambda}}}
\opfockweyl_{\txtfock}(\napiernum^{\imunit t \omega} f) \\
&=
\napiernum^{-\imunit \opreal \bkt{f}{\mathsf{m}_{\kappa,\Lambda}}}
\napiernum^{\imunit \opreal \bkt{\napiernum^{\imunit t \omega} f}{\mathsf{m}_{\kappa,\Lambda}}}
\opfockweyl_{\txtfock}(\napiernum^{\imunit t \omega} f) \\
&=
\napiernum^{\imunit \opreal \fun{\mathsf{m}_{\kappa,\Lambda}}{\rbk{\napiernum^{\imunit t \omega} - 1} f}}
\opfockweyl_{\txtfock}(\napiernum^{\imunit t \omega} f)
=
\napiernum^{\imunit \fun{\mathsf{M}_{\kappa,\Lambda}}{f}}
\opfockweyl_{\txtfock}(\napiernum^{\imunit t \omega} f).
\end{aligned}
\end{equation}
In particular, the automorphism group $\alpha_{\txtvanhowe,\txtfock,\kappa,\Lambda}$ is represented as $\Ad_{U_t}$.

(Generator and ground-state energy): Since $U_t$ is clearly strongly continuous, computing the derivative gives the generator: $$\physham_{U}
=
-\imunit
\fnrestr{\od{U_t}{t}}{t=0}
=
V_{\kappa,\Lambda}
\physham_{\txtbsn,\txtfr}
\inv{V_{\kappa,\Lambda}}.$$
By \cite{AsaoArai26}, the right-hand side equals $$\physham_{U}
=
\physham_{\txtbsn,\txtfr}
+\fun{\opfocksegal_{\txtfock}}{\omega^{-\onehalf} \varrho_{\kappa,\Lambda}}
+\onehalf \norm{\inv{\omega} \varrho_{\kappa,\Lambda}}^2
=
\physham_{\txtvanhowe,\kappa,\Lambda}
+\onehalf \norm{\inv{\omega} \varrho_{\kappa,\Lambda}}^2,$$
and by the unitary invariance of the spectrum we get $\physham_{U} = \physham_{\txtvanhowe,\kappa,\Lambda} + \onehalf \norm{\inv{\omega} \varrho_{\kappa,\Lambda}}^2 \geq 0$.

(Verification of the ground-state property): First we show the invariance of the state $\psi_{\txtvanhowe,\txtgs,\txtfock,\kappa,\Lambda}$ under the automorphism group $\alpha_{\txtvanhowe,\txtfock,\kappa,\Lambda,t}$.
The Fock vacuum is an eigenvector of $\physham_{\txtbsn,\txtfr}$ belonging to the eigenvalue $0$.
By the definition of the vector $\Psi_{\psi_{\txtvanhowe,\txtgs,\txtfock,\kappa,\Lambda}}$, we obtain $$\inv{U_t} \Psi_{\psi_{\txtvanhowe,\txtgs,\txtfock,\kappa,\Lambda}}
=
V_{\kappa,\Lambda} U_{0,t} \opfockvac_{\txtbsn}
=
\Psi_{\psi_{\txtvanhowe,\txtgs,\txtfock,\kappa,\Lambda}}.$$
In particular, $\psi_{\txtvanhowe,\txtgs,\txtfock,\kappa,\Lambda} \circ \alpha_{\txtvanhowe,\txtfock,\kappa,\Lambda,t} = \psi_{\txtvanhowe,\txtgs,\txtfock,\kappa,\Lambda}$ holds for all $t \in \fldreal$.

The preceding argument shows that all conditions of \cite[P.98, Proposition 5.3.19]{BratteliRobinson2} except for invariance under the automorphism group are also satisfied.
By that proposition, $\psi_{\txtvanhowe,\txtgs,\txtfock,\kappa,\Lambda}$ is a ground state of $\alpha_{\txtvanhowe,\txtfock,\kappa,\Lambda}$.
\end{proof}

\begin{prop}
The ground state $\psi_{\txtvanhowe,\txtgs,\txtfock,\kappa,\Lambda}$ in the Fock representation is a quasi-free state.
\end{prop}

\begin{proof}
It suffices to combine the quasi-free condition from \cite[P.40]{BratteliRobinson2} with the results on the expectation values of the Segal field operator.
In particular, letting $\psi_{\txttruncatedfn,\txtvanhowe,\txtgs,\txtfock,\kappa,\Lambda}$ denote the two-point truncated correlation function, we have
\begin{equation}
\begin{aligned}
&\psi_{\txttruncatedfn,\txtvanhowe,\txtgs,\txtfock,\kappa,\Lambda}(\opfocksegal_{\txtfock}(f),\opfocksegal_{\txtfock}(f))
+\psi_{\txttruncatedfn,\txtvanhowe,\txtgs,\txtfock,\kappa,\Lambda}(\opfocksegal_{\txtfock}(g),\opfocksegal_{\txtfock}(g))
-\opimag \bkt{f}{g} \\
&=
\onehalf \rbk{\norm{f}^2 + \norm{g}^2 - 2 \opimag \bkt{f}{g}}
=
\onehalf \norm{f + \imunit g}^2
\geq
0,
\end{aligned}
\end{equation}
which verifies the condition for all two-point truncated correlation functions.
\end{proof}

We define the representation \(\repn_{\txtirsingular,\kappa,\Lambda}
\colon
\fun{\oaweyl}{\sphilb{H}}
\to
\opspbddlin{\fun{\spfock_{\txtbsn}}{\sphilb{H}}}\) of the abstract Weyl algebra by \[\fun{\repn_{\txtirsingular,\kappa,\Lambda}}{\opfockweyl(f)}
=
\napiernum^{\imunit \opfocksegal_{\txtirsingular,\txtfock,\kappa,\Lambda}(f)},
\quad
\opfocksegal_{\txtirsingular,\txtfock,\kappa,\Lambda}(f)
=
\opfocksegal_{\txtfock}(f)
-\opreal \mathsf{m}_{\kappa,\Lambda}(f),\] and, for notational convenience, write \[\opfockweyl_{\txtirsingular,\kappa,\Lambda}(f)
=
\fun{\repn_{\txtirsingular,\kappa,\Lambda}}{\opfockweyl(f)}
=
\napiernum^{-\imunit \opreal \mathsf{m}_{\kappa,\Lambda}(f)}
\opfockweyl_{\txtfock}(f).\] This representation is called the formal infrared-singular representation.

\begin{prop}
The formal infrared-singular representation is a representation of the Weyl algebra.
\end{prop}

\begin{proof}
Using the cocycle condition of Proposition~\ref{expedition0011238}, a computation yields the preservation of the Weyl relations.
\end{proof}

We define the automorphism group \(\alpha_{\txtvanhowe,\txtirsingular,\kappa,\Lambda}\) on the formal infrared-singular representation \(\fun{\repn_{\txtirsingular,\kappa,\Lambda}}{\oaweyl(\sphilb{H})}\) of the abstract Weyl algebra by \[\fun{\alpha_{\txtvanhowe,\txtirsingular,\kappa,\Lambda,t}}{\opfockweyl_{\txtirsingular,\kappa,\Lambda}(f)}
=
\napiernum^{\imunit \mathsf{M}_{\kappa,\Lambda,t}(f)}
\opfockweyl_{\txtirsingular,\kappa,\Lambda}(\napiernum^{\imunit t \omega} f).\]

\begin{prop}
The formal infrared-singular representation is a regular representation \cite[P.24]{BratteliRobinson2}.
The automorphism group of the van Hove model via the formal infrared-singular representation is strongly continuous, and its derivation is $$\oaderiv_{\psi_{\txtvanhowe,\txtirsingular,\txtgs,\kappa,\Lambda}}(\repn_{\txtirsingular,\kappa,\Lambda}(A))
=
\commutator{\physham_{\psi_{\txtvanhowe,\txtirsingular,\txtgs,\kappa,\Lambda}}}{\repn_{\txtirsingular,\kappa,\Lambda}(A)},$$
where the operator $\physham_{\psi_{\txtvanhowe,\txtirsingular,\txtgs,\kappa,\Lambda}}$ is the free Hamiltonian $\physham_{\txtbsn,\txtfr}$.
The ground-state energy is $\fun{\physgse}{\physham_{\txtbsn,\txtfr}} = 0$.
Setting $\Psi_{\txtvanhowe,\txtirsingular,\txtgs,\kappa,\Lambda} = \opfockvac_{\txtbsn}$ as the bosonic Fock vacuum, define the state $\psi_{\txtvanhowe,\txtirsingular,\txtgs,\kappa,\Lambda} \in \oaspstate_{\repn_{\txtirsingular,\kappa,\Lambda}(\oaweyl(\sphilb{H}))}$ on the formal infrared-singular representation of the Weyl algebra by $$\psi_{\txtvanhowe,\txtirsingular,\txtgs,\kappa,\Lambda}(\repn_{\txtirsingular,\kappa,\Lambda}(A))
=
\bkt{\Psi_{\txtvanhowe,\txtirsingular,\txtgs,\kappa,\Lambda}}
{\repn_{\txtirsingular,\kappa,\Lambda}(A) \Psi_{\txtvanhowe,\txtirsingular,\txtgs,\kappa,\Lambda}}.$$
This satisfies $\fun{\psi_{\txtvanhowe,\txtirsingular,\txtgs,\kappa,\Lambda}}{\physham_{\psi_{\txtvanhowe,\txtirsingular,\txtgs,\kappa,\Lambda}}} = 0$, and in particular it is a ground state of the automorphism group $\alpha_{\txtvanhowe,\txtirsingular,\kappa,\Lambda}$ on the formal infrared-singular representation of the Weyl algebra.

For the unitary operator $U_{\psi_{\txtvanhowe,\txtirsingular,\txtgs,\kappa,\Lambda},t}
=
\napiernum^{\imunit t \physham_{\psi_{\txtvanhowe,\txtirsingular,\txtgs,\kappa,\Lambda}}}$, the automorphism group can be written as $$\alpha_{\txtvanhowe,\txtirsingular,\kappa,\Lambda,t}(\repn_{\txtirsingular,\kappa,\Lambda}(A))
=
U_{\psi_{\txtvanhowe,\txtirsingular,\txtgs,\kappa,\Lambda},t}
\repn_{\txtirsingular,\kappa,\Lambda}(A)
\inv{U_{\psi_{\txtvanhowe,\txtirsingular,\txtgs,\kappa,\Lambda},t}},$$
and the one-parameter unitary group determined by $U_{\psi_{\txtvanhowe,\txtirsingular,\txtgs,\kappa,\Lambda},t}$ is strongly continuous.
\end{prop}

\begin{proof}
(Unitary group generating the automorphism group): By rewriting the definition, we must have $$\alpha_{\txtvanhowe,\txtirsingular,\kappa,\Lambda,t}
=
\napiernum^{\imunit t \oaderiv},
\quad
\oaderiv(A)
=
\commutator{\physham_{\txtbsn,\txtfr}}{A}.$$
In particular, the strongly continuous unitary group generating the automorphism group is $U_t = \napiernum^{\imunit t \physham_{\txtbsn,\txtfr}}$.

(Generator and ground-state energy): The generator has already been discussed, and the ground-state energy is $0$.

(Construction of the state): It suffices to consider the vector state defined by the Fock vacuum: $$\psi_{\txtvanhowe,\txtirsingular,\txtgs,\kappa,\Lambda}(\repn_{\txtirsingular,\kappa,\Lambda}(A))
=
\bkt{\opfockvac_{\txtbsn}}
{\repn_{\txtirsingular,\kappa,\Lambda}(A)
\opfockvac_{\txtbsn}}.$$
The invariance of the state under the automorphism group is clear.

The ground-state property follows by verifying Proposition~\cite[P.98, Proposition 5.3.19]{BratteliRobinson2}.
Since all conditions other than invariance under the automorphism group are clearly satisfied, the vector state $\psi_{\txtvanhowe,\txtirsingular,\txtgs,\kappa,\Lambda}$ is indeed a ground state.
\end{proof}

\begin{prop}
Under infrared regularization, the Fock representation and the formal infrared-singular representation of the abstract Weyl algebra are unitarily equivalent.
In particular, the corresponding objects in the Fock representation from Proposition~\ref{expedition0011107} are recovered via $$V_{\kappa,\Lambda}
\opfockweyl_{\txtirsingular,\kappa,\Lambda}(f)
\inv{V_{\kappa,\Lambda}}
=
\opfockweyl_{\txtfock}(f),
\quad
V_{\kappa,\Lambda}
U_{\psi_{\txtvanhowe,\txtirsingular,\txtgs,\kappa,\Lambda},t}
\inv{V_{\kappa,\Lambda}}
=
U_{\txtfock,\kappa,\Lambda,t}.$$
\end{prop}

\begin{proof}
Consider the unitary operator $V_{\kappa,\Lambda} = \opfockweyl_{\txtfock}(\imunit \mathsf{m}_{\kappa,\Lambda})$.
For the generators (their representations) of the Weyl algebra, we have
\begin{equation}
\begin{aligned}
&V_{\kappa,\Lambda}
\opfockweyl_{\txtirsingular,\kappa,\Lambda}(f)
\inv{V_{\kappa,\Lambda}}
=
\napiernum^{-\imunit \opreal \mathsf{m}_{\kappa,\Lambda}(f)}
V_{\kappa,\Lambda}
\opfockweyl_{\txtfock}(f)
\inv{V_{\kappa,\Lambda}} \\
&=
\napiernum^{-\imunit \opreal \mathsf{m}_{\kappa,\Lambda}(f)}
\napiernum^{-\imunit \opimag \bkt{f}{\imunit \mathsf{m}_{\kappa,\Lambda}(f)}}
\opfockweyl_{\txtfock}(f)
=
\opfockweyl_{\txtfock}(f).
\end{aligned}
\end{equation}
For the unitary group realizing the automorphism group,
\begin{equation}
\begin{aligned}
&V_{\kappa,\Lambda}
U_{\psi_{\txtvanhowe,\txtirsingular,\txtgs,\kappa,\Lambda},t}
\inv{V_{\kappa,\Lambda}}
=
V_{\kappa,\Lambda}
\napiernum^{\imunit t \physham_{\txtbsn,\txtfr}}
\inv{V_{\kappa,\Lambda}} \\
&=
\napiernum^{\imunit t \rbk{\physham_{\txtvanhowe,\kappa,\Lambda} - \physgse(\physham_{\physham_{\txtvanhowe,\kappa,\Lambda}})}}
=
U_{\txtfock,\kappa,\Lambda,t}
\end{aligned}
\end{equation}
holds.
\end{proof}

In preparation for the discussion in the resolvent algebra, let us confirm the action of the automorphism group in the Fock representation. Although the statement is for the automorphism group with infrared and ultraviolet cutoffs, the same expression holds for the automorphism group with cutoffs removed by the exact same computation, so we will cite the following proposition in that context as well.

\begin{prop}\label{expedition0011254}
For any $\lambda \in \fldreal \setminus \setone{0}$ and $f \in \sphilb{H}$, let the Weyl operator and the resolvent be respectively $$\opfockweyl_{\txtfock}(f)
=
\napiernum^{\imunit \opfocksegal_{\txtfock}(f)},
\quad
\oaresolvent_{\txtfock}(\lambda,f)
=
\opfnresolvent{\imunit \lambda - \opfocksegal_{\txtfock}(f)}.$$
Then for any real $\lambda$ and $f \in \sphilb{H}$, the action on the resolvent is $$\fun{\alpha_{\txtvanhowe,\txtfock,\kappa,\Lambda,t}}{\oaresolvent_{\txtfock}(\lambda, f)}
=
\fun{\oaresolvent_{\txtfock}}
{\lambda + \imunit \mathsf{M}_{\kappa,\Lambda,t}(f), \napiernum^{\imunit t \omega} f}.$$
\end{prop}

\begin{proof}
Noting the definition of the resolvent here, a Laplace transform computation gives
\begin{equation}
\begin{aligned}
&\alpha_{\txtvanhowe,\txtfock,\kappa,\Lambda,t}(\oaresolvent_{\txtfock}(\lambda,f))
=
\int_0^{(\sgn \lambda) \infty}
\fun{\alpha_{\txtvanhowe,\txtfock,\kappa,\Lambda,t}}{\napiernum^{-\lambda t - \imunit t \opfocksegal_{\txtfock}(t)}}
\opdmsr{t} \\
&=
\int_0^{(\sgn \lambda) \infty}
\napiernum^{-\lambda t}
\fun{\alpha_{\txtvanhowe,\txtfock,\kappa,\Lambda,t}}{\opfockweyl_{\txtfock}(-tf)}
\opdmsr{t} \\
&=
\int_0^{(\sgn \lambda) \infty}
\napiernum^{-\lambda t}
\napiernum^{\imunit \mathsf{M}_{\kappa,\Lambda,t}(-tf)}
\opfockweyl_{\txtfock}(\napiernum^{\imunit t \omega} (-tf))
\opdmsr{t} \\
&=
\int_0^{(\sgn \lambda) \infty}
\napiernum^{-\rbk{\lambda + \imunit \mathsf{M}_{\kappa,\Lambda,t}(f)} t}
\opfockweyl_{\txtfock}(-t \napiernum^{\imunit t \omega} f)
\opdmsr{t} \\
&=
\fun{\oaresolvent_{\txtfock}}
{\lambda + \imunit \mathsf{M}_{\kappa,\Lambda,t}(f), \napiernum^{\imunit t \omega} f}.
\end{aligned}
\end{equation}
\end{proof}

\subsection{Ground State in the Abstract Weyl Algebra at Zero Temperature}\label{expedition0011250}

Since the Hamiltonian of the van Hove model cannot be defined directly in the abstract Weyl algebra, we define the dynamics by an automorphism group, guided by the Fock representation. The ground state is also defined based on the results of the Fock-representation computation, but we must verify that it is indeed a ground state of the corresponding automorphism group in the operator-algebraic sense.

\begin{defn}\label{expedition0011093}
For any $f \in \sphilb{H}$ and any $\kappa,\Lambda$, we define a state on the abstract Weyl algebra $\oaweyl(\sphilb{H})$ by $$\oastate[\psi_{\txtvanhowe,\txtgs,\kappa,\Lambda}](\opfockweyl(f))
=
\fnexp{-\imunit \opreal \mathsf{m}_{\kappa,\Lambda}(f)
-\frac{1}{4} \opform{q}_{\txtnonzero,\infty}(f)}.$$
This is called the formal ground state on the abstract Weyl algebra.
For any $t \in \fldreal$, we define the self-map $\alpha_{\txtvanhowe,\kappa,\Lambda,t}$ of the abstract Weyl algebra by $$\alpha_{\txtvanhowe,\kappa,\Lambda,t}(\opfockweyl(f))
=
\napiernum^{\imunit \mathsf{M}_{\kappa,\Lambda,t}(f)}
\opfockweyl(\napiernum^{\imunit t \omega} f),$$
and call it the van Hove automorphism group on the abstract Weyl algebra.
Without the shorthand notation, $$\alpha_{\txtvanhowe,\kappa,\Lambda,t}(\opfockweyl(f))
=
\napiernum^{\imunit \opreal \fun{\mathsf{m}_{\kappa,\Lambda}}{\rbk{\napiernum^{\imunit t \omega} - \idone} f}}
\fun{\opfockweyl}{\napiernum^{\imunit t \omega} f}.$$
By Proposition~\ref{expedition0011094}, this is indeed an automorphism group of the Weyl algebra.
\end{defn}

The word ``formal'' attached to the ground state at this point reflects the fact that the ground-state property has not yet been proved. In fact, it does satisfy the operator-algebraic definition of a ground state \cite[Definition 5.3.18]{BratteliRobinson2}. Furthermore, taking the limit \(\sminvtemperature \to \infty\) of the \(\sminvtemperature\)-KMS states discussed in Section~\ref{expedition0012092}, one indeed recovers this state.

We define the Fock representation \(\alpha_{\txtvanhowe,\txtfock,\kappa,\Lambda}\) of the automorphism group \(\alpha_{\txtvanhowe,\kappa,\Lambda}\) by \[\alpha_{\txtvanhowe,\txtfock,\kappa,\Lambda,t}(\opfockweyl_{\txtfock}(f))
=
\napiernum^{\imunit \mathsf{M}_{\kappa,\Lambda,t}(f)}
\fun{\opfockweyl_{\txtfock}}{\napiernum^{\imunit t \omega} f}.\] This is called the van Hove model in the Fock representation of the abstract Weyl algebra, or the van Hove automorphism map. By Proposition~\ref{expedition0011094} applied to the abstract Weyl algebra, this is indeed an automorphism group. Note in particular that \(\napiernum^{\imunit \mathsf{M}_{\kappa,\Lambda,t}(f)}\) does not depend on the representation, and only the Weyl operator depends on the representation.

The infrared and ultraviolet cutoffs appear only in \(\mathsf{M}_{\kappa,\Lambda,t}(f)\). The method for removing the infrared and ultraviolet cutoffs for the automorphism group and the ground state should then be essentially clear.

\section{Renormalization of the Ground-State Energy}\label{expedition0011496}

Following \cite{AsaoArai26}, we discuss the renormalization of the ground-state energy accompanying the removal of the ultraviolet cutoff. In this section, we take the dispersion relation \(\omega(k) = \abs{k}\), the dimension \(d = 3\), and discuss the van Hove model for the following class of point sources, which generalizes the point source. The self-energy term diverges linearly to \(-\infty\) as \(\Lambda \to \infty\), and the energy renormalization consists in removing this divergence.

\begin{defn}\label{expedition0011497}
Assuming that the points $x_{\txtrenormalization,i} \in \fldreal^{3}$ are all distinct, we define the source $$\varrho^{(p)}(x)
=
\sum_{n=1}^N \lambda_n \delta(x-x_{\txtrenormalization,n}),
\quad
\faftr{\varrho^{(p)}}(k)
=
\frac{1}{(2 \pi)^{3/2}}
\sum_{n=1}^N
\lambda_n \napiernum^{-\imunit k \vainnprod x_{\txtrenormalization,n}},
\quad
x_{\txtrenormalization,1},
\dots,
x_{\txtrenormalization,N}
\in \fldreal^{3},$$
which we call a point source cluster.
Considering the momentum-space cutoff $\faftr{\varrho_{\kappa,\Lambda}^{(p)}}(k) = \fndef{[\kappa,\Lambda]}(\abs{k}) \faftr{\varrho^{(p)}}(k)$, it satisfies $\lim_{\kappa \to 0, \Lambda \to \infty} \varrho_{\kappa,\Lambda}^{(p)} = \varrho^{(p)}$ in the topology of the space of tempered distributions $\fun{\dsttempered}{\fldreal^{3}}$.
Decomposing the ground-state energy $\physgse$ as $$\physgse
=
-\onehalf
\norm{\omega^{\onehalf} \faftr{\mathsf{m}_{\kappa,\Lambda}}}_{\fun{\lp^{2}}{\greuctr{k}{3}}}^2
=
\physenergy_{\txtself,\kappa,\Lambda}^{(3)}
+\fun{\physenergy_{\txtexchange,\kappa,\Lambda}^{(3)}}{\nfoldvar{x_{\txtrenormalization}}{N}},$$
we define each term as follows:
\begin{equation}
\begin{aligned}
\physenergy_{\txtself,\kappa,\Lambda}^{(3)}
&=
-\sum_{n=1}^N
\frac{\lambda_n^2}{2(2\pi)^3}
\int_{\kappa \leq \abs{k} \leq \Lambda}
\frac{dk}{\abs{k}^2},
\\ 
\fun{\physenergy_{\txtexchange,\kappa,\Lambda}^{(3)}}{\nfoldvar{x_{\txtrenormalization}}{N}}
&=
-\sum_{j\neq n}
\frac{\lambda_j\lambda_n}{2(2\pi)^3}
\int_{\kappa \leq \abs{k} \leq \Lambda}
\frac{\napiernum^{\imunit k \vainnprod (x_{\txtrenormalization,j}-x_{\txtrenormalization,n})}}{\abs{k}^2}
\opdmsr{k},
\\ 
\fun{\physenergy_{\txtexchange,0}^{(3)}}{\nfoldvar{x_{\txtrenormalization}}{N}}
&=
\sum_{1 \leq j < n \leq N}
\lambda_j \lambda_n
V_{\txtcoulomb}(x_{\txtrenormalization,j}-x_{\txtrenormalization,n}),
\\ 
V_{\txtcoulomb}(x)
&=
-\frac{1}{4 \pi \abs{x}}.
\end{aligned}
\end{equation}
\end{defn}

\begin{prop}\label{expedition0011517}
For the cutoff-removal limit of ultraviolet and infrared cutoffs and the many-body Coulomb potential, the bound
\begin{equation}\label{eq:Vbound}
\abs{\fun{\physenergy_{\txtexchange,\kappa,\Lambda}^{(3)}}{\nfoldvar{x_{\txtrenormalization}}{N}}
-\fun{\physenergy_{\txtexchange,0}^{(3)}}{\nfoldvar{x}{N}}}
\leq
\rbk{\sum_{1 \leq i < j \leq N}
\abs{\lambda_i \lambda_j}}
\rbk{\frac{2}{\pi^2} \kappa
+\frac{4}{\pi^2}
\frac{1}{\Lambda r_0^{2}}}
\end{equation}
holds.
In particular, $$\lim_{\kappa \to 0, \Lambda \to \infty}
\fun{\physenergy_{\txtexchange,\kappa,\Lambda}^{(3)}}{\nfoldvar{x_{\txtrenormalization}}{N}}
=
\fun{\physenergy_{\txtexchange,0}^{(3)}}{\nfoldvar{x_{\txtrenormalization}}{N}}.$$
\end{prop}

\begin{proof}
(Preparation): The following identity holds:
\begin{equation}\label{eq:kernel-id}
\frac{1}{(2\pi)^3}
\int_{\fldreal^3}
\frac{\napiernum^{\imunit k \vainnprod x}}{\abs{k}^2}
\opdmsr{k}
=
\frac{1}{4 \pi \abs{x}}
\quad
(x \neq 0).
\end{equation}
In particular, the Coulomb kernel is $$V_{\txtcoulomb}^{[1]}(x) = -\dfrac{1}{4 \pi \abs{x}}.$$
The superscript indicates the one-body case.

(Explicit formulas for the cutoff kernel and the error): Let $\fnsineintegral{x} = \intsineintegral{x}$ denote the sine integral.
Fixing the radius $r = \abs{x} > 0$ and using spherical symmetry, we obtain
\begin{align}
V_{\kappa,\Lambda}^{[1]}(r)
&=
\frac{1}{(2\pi)^3}
\int_{\kappa \leq \abs{k} \leq \Lambda}
\frac{\napiernum^{\imunit k \vainnprod x}}{\abs{k}^2}
\opdmsr{k}
=
\frac{1}{(2\pi)^3}
4 \pi
\int_{\kappa}^{\Lambda}
\frac{\sin(kr)}{kr}
\opdmsr{k} \\
&=
\frac{1}{2 \pi^2}
\frac{1}{r}
\rbk{\fnsineintegral{\Lambda r}-\fnsineintegral{\kappa r}}, \label{eq:phi-cut} \\
V_0^{[1]}(r)
&=
\frac{1}{(2\pi)^3}
\int_{\fldreal^3}
\frac{\napiernum^{\imunit k \vainnprod x}}{\abs{k}^2}
\opdmsr{k}
=
\frac{1}{2 \pi^2}
\int_{0}^{\infty}
\frac{\sin kr}{k} \opdmsr{k} \\
&=
\frac{1}{2 \pi^2}
\frac{1}{r}
\fnsineintegral{\infty}
=
\frac{1}{4\pi r}. \label{eq:phi0}
\end{align}
Thus, noting $\fnsineintegral{\infty} = \frac{\pi}{2}$, we have
\begin{equation}
\begin{aligned}
V_{\kappa,\Lambda}^{[1]}(r)
&=
V_0^{[1]}(r)
+\delta_{\kappa,\Lambda}^{[1]}(r), \\
\delta_{\kappa,\Lambda}^{[1]}(r)
&=
\frac{1}{2 \pi^2}
\frac{1}{r}
\rbk{\fnsineintegral{\Lambda r}
-\frac{\pi}{2}
-\fnsineintegral{\kappa r}}.
\end{aligned}
\label{eq:delta-def}
\end{equation}

(Upper bounds on the error: infrared and ultraviolet ends): We use the basic property $0 \leq \fnsineintegral{x} \leq \min \setone{x, \dfrac{\pi}{2}}$ of the sine integral $\fnsineintegral{x} = \intsineintegral{x}$.

(Infrared end ($\kappa \downarrow 0$)): Setting $\delta_{\kappa,\Lambda}^{\txtinflaredred}(r) = -\frac{1}{2 \pi^2} \frac{1}{r} \fnsineintegral{\kappa r}$ and using the estimate $\fnsineintegral{\kappa r} \leq \kappa r$, we obtain
\begin{equation}\label{eq:IR-bound}
\abs{\delta_{\kappa,\Lambda}^{\txtinflaredred}(r)}
=
\frac{1}{2\pi^2}
\frac{\fnsineintegral{\kappa r}}{r}
\leq
\frac{1}{2 \pi^2} \kappa.
\end{equation}

(Ultraviolet end ($\Lambda \uparrow \infty$)): Set $\delta_{\kappa,\Lambda}^{\txtultraviolet}(r) = \frac{1}{2 \pi^2} \frac{1}{r} \rbk{\fnsineintegral{\Lambda r}-\frac{\pi}{2}}$.
Integration by parts gives $$\fnsineintegral{\Lambda r}-\frac{\pi}{2} = \int_{\Lambda r}^{\infty} \frac{\sin t}{t} \opdmsr{t} = \frac{\cos(\Lambda r)}{\Lambda r} - \int_{\Lambda r}^{\infty} \frac{\cos t}{t^2} \opdmsr{t}.$$
Therefore,
\begin{equation}\label{eq:UV-bound}
\abs{\delta_{\kappa,\Lambda}^{\txtultraviolet}(r)}
\leq
\frac{1}{2 \pi^2}
\frac{1}{r} \cdot \frac{2}{\Lambda r}
=
\frac{1}{\pi^2}
\frac{1}{\Lambda r^2}.
\end{equation}

(Uniformization with nearest-neighbor distance): Assume $\min_{i \neq j} \abs{x_{\txtrenormalization,i}-x_{\txtrenormalization,j}} \geq r_0 > 0$.
Define the error term as $\delta_{\kappa,\Lambda}(r) = \delta_{\kappa,\Lambda}^{\txtinflaredred}(r) + \delta_{\kappa,\Lambda}^{\txtultraviolet}(r)$.
For all $r \geq r_0$,
\begin{equation}\label{eq:uniform}
\sup_{r \geq r_0}
\abs{\delta_{\kappa,\Lambda}(r)}
\leq
\frac{1}{2 \pi^2}
\kappa
+\frac{1}{\pi^2}
\frac{1}{\Lambda r_0^{2}}.
\end{equation}

(Application to multi-point interactions with fixed constants): Defining the pair kernel and the interaction as
\begin{equation}\label{eq:pair-kernel}
V_{\kappa,\Lambda}^{[1]}(x)
=
\frac{1}{(2 \pi)^3}
\int_{\kappa \leq \abs{k} \leq \Lambda}
\frac{\napiernum^{\imunit k \vainnprod x}}{\abs{k}^2}
\opdmsr{k},
\quad
V_0^{[1]}(x)
=
\frac{1}{4\pi |x|},
\end{equation}
\begin{equation}
\begin{aligned}
\fun{\physenergy_{\txtexchange,\kappa,\Lambda}^{(3)}}{\nfoldvar{x_{\txtrenormalization}}{N}}
&=
-\sum_{1 \leq i < j\leq N}
\lambda_i \lambda_j
V_{\kappa,\Lambda}^{[1]}(x_{\txtrenormalization,i}-x_{\txtrenormalization,j}),
\\ 
\fun{\physenergy_{\txtexchange,0}^{(3)}}{\nfoldvar{x_{\txtrenormalization}}{N}}
&=
-\sum_{1 \leq i < j \leq N}
\lambda_i \lambda_j
V_0^{[1]}(x_{\txtrenormalization,i}-x_{\txtrenormalization,j}),
\end{aligned}
\label{eq:VkL}
\end{equation}
the difference is
\begin{equation}\label{eq:Vdiff}
\fun{\physenergy_{\txtexchange,\kappa,\Lambda}^{(3)}}{\nfoldvar{x_{\txtrenormalization}}{N}}
-\fun{\physenergy_{\txtexchange,0}^{(3)}}{\nfoldvar{x_{\txtrenormalization}}{N}}
=
-\sum_{i < j}
\lambda_i \lambda_j
\delta_{\kappa,\Lambda}(x_{\txtrenormalization,i}-x_{\txtrenormalization,j}).
\end{equation}
Therefore, in the region with nearest-neighbor distance $r_0 > 0$,
\begin{equation}
\abs{\fun{\physenergy_{\txtexchange,\kappa,\Lambda}^{(3)}}{\nfoldvar{x_{\txtrenormalization}}{N}}
-\fun{\physenergy_{\txtexchange,0}^{(3)}}{\nfoldvar{x_{\txtrenormalization}}{N}}}
\leq
\rbk{\sum_{1 \leq i < j \leq N}
\abs{\lambda_i \lambda_j}}
\rbk{\frac{2}{\pi^2} \kappa
+\frac{4}{\pi^2}
\frac{1}{\Lambda r_0^{2}}}
\end{equation}
holds.
\end{proof}

\section{Ground State in the Resolvent Algebra at Zero Temperature}\label{expedition0011615}

We define the objects based on the discussion of the Weyl algebra in Section~\ref{expedition0012095}.

\subsection{Basic Setup}\label{basic-setup}

Based on Proposition~\ref{expedition0011254}, we define the self-maps of the abstract resolvent algebra for any \(t \in \fldreal\) as \begin{equation}
\begin{aligned}
\alpha_{1,\kappa,\Lambda,t}(\oaresolvent(z,f))
&=
\fun{\oaresolvent}{z + \imunit \mathsf{M}_{\kappa,\Lambda,t}(f), f}, \\
\alpha_{2,t}(\oaresolvent(z,f))
&=
\fun{\oaresolvent}{z, \napiernum^{\imunit t \omega} f}, \\
\alpha_{\txtvanhowe,\kappa,\Lambda,t}(\oaresolvent(z,f))
&=
\funrbk{\alpha_{1,\kappa,\Lambda,t} \circ \alpha_{2,t}}{\oaresolvent(z,f)} \\
&=
\fun{\oaresolvent}{z + \imunit \mathsf{M}_{\kappa,\Lambda,t}(f), \napiernum^{\imunit t \omega} f},
\end{aligned}
\end{equation} and call this the automorphism group of the van Hove model with infrared and ultraviolet cutoffs, or simply the automorphism group of the cutoff van Hove model or the van Hove model.

Furthermore, \(\alpha_{\txtvanhowe,t}\) and \(\alpha_{1,t}\) without cutoffs are defined, in a form arising from the domain of the functional \(\mathsf{m}\), as self-maps of the abstract resolvent algebra \(\oaresolventalgebra(\dom \mathsf{m}, \sigma)\): \begin{equation}
\begin{aligned}
\alpha_{1,t}(\oaresolvent(z,f))
&=
\fun{\oaresolvent}{z + \imunit \mathsf{M}_{t}(f), f}, \\
\alpha_{\txtvanhowe,t}(\oaresolvent(z,f))
&=
\funrbk{\alpha_{1,t} \circ \alpha_{2,t}}{\oaresolvent(z,f)} \\
&=
\fun{\oaresolvent}{z + \imunit \mathsf{M}_{t}(f), \napiernum^{\imunit t \omega} f},
\end{aligned}
\end{equation} and this is called the automorphism group of the van Hove model, or the automorphism group of the van Hove model without cutoffs. In particular, Proposition~\ref{expedition0011109} verifies that \[\alpha_{\txtvanhowe,\kappa,\Lambda}
=
\fml{\alpha_{\txtvanhowe,\kappa,\Lambda,t}}{t \in \fldreal},
\quad
\alpha_{\txtvanhowe}
=
\fml{\alpha_{\txtvanhowe,t}}{t \in \fldreal}\] are automorphism groups.

\begin{prop}\label{expedition0011109}
The above-defined $\alpha_{\txtvanhowe,\kappa,\Lambda}$ is an automorphism group.
In particular, even when the domain is restricted to $\oaresolventalgebra(\dom \mathsf{m},\sigma)$, $\alpha_{\txtvanhowe,\kappa,\Lambda}$ is an automorphism group on it.
With an appropriate restriction of the domain, $\alpha_{\txtvanhowe}$ is an automorphism group on the restricted subalgebra.
\end{prop}

\begin{proof}
Since the domain restriction does not change the argument, we work with the cutoff notation.
By Proposition~\ref{expedition0011238}, the auxiliary functional $\mathsf{M}_{\kappa,\Lambda,t}$ is a cocycle.

(Automorphism property): Under the assumption that the first argument of the resolvent is adjusted to avoid becoming $0$, it suffices to verify directly the preservation of the resolvent relations.

(Restricted automorphism property): Let $f \in \dom \mathsf{m}$ be arbitrary.
Then $$\alpha_{\txtvanhowe,\kappa,\Lambda,t}(\oaresolvent(z,f)) = \oaresolvent(z + \imunit \mathsf{M}_{\kappa,\Lambda}(f), \napiernum^{\imunit t \omega} f)$$ is well-defined.
In particular, $\abs{\napiernum^{\imunit t \omega} f} = \abs{f}$ preserves integrability, so $\dom \mathsf{m}$ is invariant.
Hence $\fnrestr{\alpha_{\txtvanhowe,\kappa,\Lambda}}{\oaresolventalgebra(\dom \mathsf{m},\sigma)}$ is an automorphism group of $\oaresolventalgebra(\dom \mathsf{m},\sigma)$.

(Group property): The argument is the same as in Proposition~\ref{expedition0011094}.
\end{proof}

Here we defined the van Hove model via an automorphism group. Since it is not straightforward to confirm whether this is the van Hove model Hamiltonian without seeing the Hamiltonian directly, we investigate the generator of the automorphism group. Although the derivation \(\oaderiv_2\) corresponding to the automorphism group \(\alpha_{2}\) is clearly generated by the free Hamiltonian, the Hamiltonian cannot be introduced without considering an appropriate representation and its pullback. On the other hand, since the resolvent is analytic in the first variable, the derivation \(\oaderiv_{1,\kappa,\Lambda}\) of the automorphism group \(\alpha_{1,\kappa,\Lambda}\) can be derived independently of the representation, and one can observe that it is generated by the perturbation term \(\opfocksegal(\omega^{-\onehalf} f)\) of the van Hove model.

\begin{prop}
Taking $\oaresolventalgebra_{\txtfin}$ as a dense subalgebra and denoting the derivation as the generator of $\alpha_{1,\kappa,\Lambda}$ on it by $\oaderiv_{1,\kappa,\Lambda}$, this satisfies $$\oaderiv_{1,\kappa,\Lambda}(\oaresolvent(z,f)) = \opimag \bkt{\omega \mathsf{m}_{\kappa,\Lambda}}{f}_{\sphilb{H}} \oaresolvent(z,f)^2.$$
The same property holds for the objects without cutoffs.
\end{prop}

\begin{rem}
In the Fock representation, the commutation relation between the Segal field operator and its resolvent yields, on an appropriate subspace, $$\commutator{\opfocksegal_{\txtfock}(g)}{\opfnresolvent{\opfocksegal_{\txtfock}(f) - z}} = \imunit \opimag \bkt{f}{g} \rbk{\opfocksegal_{\txtfock}(f) - z}^{-2}.$$
Using this, one can rewrite $$\oaderiv_{1,\kappa,\Lambda}(\oaresolvent(z,f)) = \commutator{\fun{\opfocksegal_{\txtfock}}{\omega^{-\onehalf} \varrho_{\kappa,\Lambda}}}{\oaresolvent(z,f)},$$
making it apparent that this is the derivation by the perturbation term $\opfocksegal_{\txtfock}(\omega^{-\onehalf} \varrho_{\kappa,\Lambda})$ of the van Hove model.
\end{rem}

\begin{proof}
Since the domain restriction does not change the argument here either, we work with the cutoff notation.

By \cite[P.2730, Theorem 3.6(iv)]{BuchholzGrundling2}, $\oaresolvent(\lambda,f)$ is analytic in the first variable.
Using the fourth resolvent relation, i.e., $$\oaresolvent(z,f) - \oaresolvent(w,f) = \imunit (w-z) \oaresolvent(z,f) \oaresolvent(w,f),$$
and the chain rule, we obtain
\begin{equation}
\begin{aligned}
&\oaderiv_{1,\kappa,\Lambda}(\oaresolvent(z,f))
=
-\imunit
\fnrestr{\od{\mathsf{M}_{\kappa,\Lambda,t}(f)}{t}}{t=0}
\cdot
\opod{z}
\oaresolvent(z,f) \\
&=
-\imunit
\opreal \mathsf{m}_{\kappa,\Lambda}(\imunit \omega f)
\cdot
\imunit
\oaresolvent(z,f)^2 \\
&=
\opimag \bkt{\omega \mathsf{m}_{\kappa,\Lambda}}{f}_{\sphilb{H}}
\oaresolvent(z,f)^2.
\end{aligned}
\end{equation}
The expression $\mathsf{m}_{\kappa,\Lambda}(\imunit \omega f)$ appearing in the above computation is formal notation; in practice, writing $\mathsf{M}_{\kappa,\Lambda,t}(f)$ as an integral and computing directly yields $\bkt{\omega \mathsf{m}_{\kappa,\Lambda}}{f}$ without difficulty.
\end{proof}

To simplify the discussion related to Proposition~\ref{expedition0011112} on the construction of the ground state, we adopt the approach of defining the functional via the two-point values, i.e., the values on products of two generators. In fact, even starting from only the one-point values, one can derive the two-point values used in the definition by exploiting \(\alpha_{\txtvanhowe,\kappa,\Lambda}\)-invariance, and in particular invariance under the derivation \(\oaderiv_{\txtvanhowe,\kappa,\Lambda}\).

\begin{defn}\label{expedition0011111}
For any $$\lambda,\mu \in \fldmultiplicativegroup{\fldreal}, \quad f,g \in \sphilb{H}, \quad \kappa,\Lambda > 0,$$
we define a functional on the abstract resolvent algebra $\oaresolventalgebra(\sphilb{H},\sigma)$ by
\begin{equation}
\begin{aligned}
&\oastate[\psi_{\txtvanhowe,\txtgs,\kappa,\Lambda}]
(\oaresolvent(\lambda,f) \cdot \oaresolvent(\mu,g)) \\
&=
-\int_0^{(\sgn \lambda) \infty}
\int_0^{(\sgn \mu) \infty}
\napiernum^{-\lambda s - \mu t}
\mathsf{T}_{\kappa,\Lambda}(s,f;t,g)
\opdmsr{s}
\opdmsr{t}, \\
&\mathsf{T}_{\kappa,\Lambda}(s,f;t,g) \\
&=
\fnexp{-\frac{1}{4} \opform{q}_{\txtnonzero,\infty}(sf + tg)
+\imunit \opreal \mathsf{m}_{\kappa,\Lambda}(sf+tg)
-\frac{\imunit}{2} st \opimag \bkt{f}{g}},
\end{aligned}
\end{equation}
and call it the ground state of the van Hove model with infrared and ultraviolet cutoffs, or simply the ground state of the cutoff van Hove model or the van Hove model.

Setting $\mu = 1$ and $g = 0$ and noting the resolvent relation $\oaresolvent(1,0) = \frac{1}{\imunit}$, the one-point value is
\begin{equation}
\begin{aligned}
&\fun{\oastate[\psi]_{\txtgs,\kappa,\Lambda}}{\oaresolvent(\lambda,f)} \\
&=
-\imunit
\int_0^{(\sgn \lambda) \infty}
\napiernum^{-t \rbk{\lambda - \imunit \opreal \mathsf{m}_{\kappa,\Lambda}(f)}}
\napiernum^{-t^2 \frac{1}{4} \opform{q}_{\txtnonzero,\infty}(f)}
\opdmsr{t}.
\end{aligned}
\end{equation}
By the Cauchy transform of Gaussian measures, $\oastate[\psi_{\txtvanhowe,\txtgs,\kappa,\Lambda}](\oaresolvent(\lambda,f))$ coincides with the Cauchy transform of the normal distribution $$\fun{\msrcal{\prbnormaldist}}{-\opreal \mathsf{m}_{\kappa,\Lambda}(f), \frac{1}{2} \opform{q}_{\txtnonzero,\infty}(f)}.$$

When $f,g \in \dom \mathsf{m}$ in the above discussion, symbolically removing $\kappa,\Lambda$ gives a functional $\oastate[\psi_{\txtvanhowe,\txtgs}]$ on $\oaresolventalgebra(\dom \mathsf{m},\sigma)$ defined similarly with respect to the automorphism group of the cutoff van Hove model.
This is called the ground state of the van Hove model, or the ground state of the van Hove model without cutoffs.
\end{defn}

Although we have called it the ground state, it is not even obvious from the definition that it is a linear functional. We prepare several results below before proving the ground-state property in Proposition~\ref{expedition0011112}.

\begin{rem}
Although we do not assume it here, since the state is quasi-free in the Fock representation, one could reasonably assume quasi-freeness from the start.
Here we guarantee quasi-freeness by using Proposition~\ref{expedition0011258} and identifying the state with the classical state described later.
While this is convenient for proving regularity, it is not necessarily the most natural way to organize the argument.

If regularity, quasi-freeness, and $\alpha$-invariance are assumed, the ground state is unique since quasi-free states are determined by their mean and variance.
In particular, under infrared regularization, uniqueness holds in the Fock representation.
\end{rem}

\subsection{Verification of the Ground-State Property}\label{verification-of-the-ground-state-property}

We discuss the Gaussian measure representation of expectation values for the ground state.

\begin{prop}\label{expedition0011244}
\begin{enumerate}
\item
For any $\lambda \in \fldmultiplicativegroup{\fldreal}$ and $f \in \sphilb{H}$, the probability measure $p_{m_{\kappa,\Lambda,f},\sigma_f^2}$ with mean $m_{\kappa,\Lambda,f} = -\opreal \mathsf{m}_{\kappa,\Lambda}(f)$ and variance $\sigma_f^2 = \frac{1}{2} \opform{q}_{\txtnonzero,\infty}(f)$ gives the representation $$G_{f}(\lambda) = \oastate[\psi_{\txtvanhowe,\txtgs,\kappa,\Lambda}](\oaresolvent(\lambda,f)) = \int_{\fldreal} \frac{1}{\imunit \lambda - x} \opdmsr{p_{m_{\kappa,\Lambda,f},\sigma_f^2}(x)}.$$

\item
For any $\lambda,\mu \in \fldmultiplicativegroup{\fldreal}$ and $f,g \in \sphilb{H}$ with $\opimag \bkt{f}{g} = 0$, the two-dimensional Gaussian measure $p_{m_{\kappa,\Lambda,f},m_{\kappa,\Lambda,g},C_{f,g}}$ with mean vector $\vecbk{m_{\kappa,\Lambda,f},m_{\kappa,\Lambda,g}}$ from part (1) and covariance $$C_{f,g}
=
\onehalf
\begin{pmatrix}
\opform{q}_{\txtnonzero,\infty}(f) & \opreal \bkt{f}{g} \\
\opreal \bkt{f}{g} & \opform{q}_{\txtnonzero,\infty}(g)
\end{pmatrix}
=
\begin{pmatrix}
\sigma_f^2 & \onehalf \opreal \bkt{f}{g} \\
\onehalf \opreal \bkt{f}{g} & \sigma_g^2
\end{pmatrix}$$
gives the representation
\begin{equation}
\begin{aligned}
G_{f,g}(\lambda,\mu)
&=
\oastate[\psi_{\txtvanhowe,\txtgs,\kappa,\Lambda}](\oaresolvent(\lambda,f) \oaresolvent(\mu,g)) \\
&=
\int_{\fldreal^{2}}
\frac{1}{\imunit \lambda - x}
\frac{1}{\imunit \mu - y}
\opdmsr{p_{m_{\kappa,\Lambda,f},m_{\kappa,\Lambda,g},C_{f,g}}(x,y)}.
\end{aligned}
\end{equation}
\end{enumerate}
The same properties hold for the objects without cutoffs.
\end{prop}

\begin{rem}
We imposed $\opimag \bkt{f}{g} = 0$ in order to represent the expectation value via a real Gaussian measure.
This ensures that $\frac{\imunit}{2} \opimag \bkt{f}{g} = 0$ in the expression for $\mathsf{T}_{\kappa,\Lambda}$ in part (2) of the proposition, allowing the proof to proceed.

What is even more important is that this representation clarifies what happens under the infrared singularity condition.
When $f \notin \dom \mathsf{m}$ under the infrared singularity condition, the expected value $\mathsf{m}(f)$ of the Segal field operator in the ground state diverges.
A formal computation of the expectation of the resolvent naturally gives $0$, and the resolvent as an operator in the GNS representation of the ground state can be regarded as collapsing to $0$; this can thus also be viewed as an extension of (non-)regularity in the sense of the resolvent algebra.
\end{rem}

\begin{proof}
Since the domain restriction does not change the argument here either, we work with the cutoff notation.

(1): The characteristic function of the normal distribution specified in the statement can be written, using a Gaussian random variable $X_f$ with mean $m_{\kappa,\Lambda,f}$ and variance $\sigma_f^2$, as
\begin{equation}
\begin{aligned}
&\int_{\fldreal}
\napiernum^{\imunit t x}
\opdmsr{p_{m_{\kappa,\Lambda,f},\sigma_f^2}(x)}
=
\sqfun{\prbexp}{\napiernum^{\imunit t X_f}} \\
&=
\napiernum^{\imunit t m_{\kappa,\Lambda,f}}
\napiernum^{-\onehalf \sigma_f^2 t^2}
=
\napiernum^{-\imunit t \opreal \mathsf{m}_{\kappa,\Lambda}(f)}
\napiernum^{-\frac{\opform{q}_{\txtnonzero,\infty}(f)}{4} t^2}.
\end{aligned}
\end{equation}
By Fubini's theorem and the Laplace transform representation of the resolvent, we obtain
\begin{equation}
\begin{aligned}
&\oastate[\psi_{\txtvanhowe,\txtgs,\kappa,\Lambda}](\oaresolvent(\lambda,f)) \\
&=
\int_{\fldreal}
\rbk{-\imunit
\int_{0}^{(\sgn \lambda) \infty}
\napiernum^{-(\lambda + \imunit x) t}
\opdmsr{t}}
\opdmsr{p_{m_{\kappa,\Lambda,f},\sigma_f^2}(x)} \\
&=
\int_{\fldreal}
\frac{1}{\imunit \lambda - x}
\opdmsr{p_{m_{\kappa,\Lambda,f},\sigma_f^2}(x)}.
\end{aligned}
\end{equation}

(2): Let the joint characteristic function of the Gaussian random variables $X_f$ and $X_g$ be $$\sqfun{\prbexp}{\napiernum^{\imunit (sX_f + tX_g)}} = \fnexp{-\frac{1}{4} \opform{q}_{\txtnonzero,\infty}(sf+tg) - \imunit \opreal \mathsf{m}_{\kappa,\Lambda}(sf+tg)}.$$
By assumption, $-\frac{\imunit}{2} \opimag \bkt{f}{g} = 0$.
Under this setup, $$\mathsf{T}_{\kappa,\Lambda}(s,f;t,g) = \sqfun{\prbexp}{\napiernum^{-\imunit \rbk{sX_f + tX_g}}}$$
holds.
Therefore,
\begin{equation}
\begin{aligned}
&\oastate[\psi_{\txtvanhowe,\txtgs,\kappa,\Lambda}](\oaresolvent(\lambda,f) \oaresolvent(\mu,g)) \\
&=
\int_{0}^{(\sgn \lambda) \infty}
\int_{0}^{(\sgn \mu) \infty}
\napiernum^{-\lambda s - \mu t} \\
&\quad\times
\rbk{\int_{\fldreal^{2}}
\napiernum^{-\imunit (sx + ty)}
\opdmsr{p_{m_{\kappa,\Lambda,f},m_{\kappa,\Lambda,g},C_{f,g}}(x,y)}}
\opdmsr{s}
\opdmsr{t} \\
&=
\int_{\fldreal^{2}}
\opdmsr{p_{m_{\kappa,\Lambda,f},m_{\kappa,\Lambda,g},C_{f,g}}(x,y)} \\
&\quad\times
\rbk{-\imunit
\int_{0}^{(\sgn \lambda) \infty}
\napiernum^{-\rbk{\lambda + \imunit x} s}
\opdmsr{s}}
\rbk{-\imunit
\int_{0}^{(\sgn \mu) \infty}
\napiernum^{-\rbk{\mu + \imunit x} t}
\opdmsr{s}} \\
&=
\int_{\fldreal^{2}}
\frac{1}{\imunit \lambda - x}
\frac{1}{\imunit \mu - y}
\opdmsr{p_{m_{\kappa,\Lambda,f},m_{\kappa,\Lambda,g},C_{f,g}}(x,y)}.
\end{aligned}
\end{equation}
\end{proof}

\begin{prop}
The functional $\oastate[\psi_{\txtvanhowe,\txtgs,\kappa,\Lambda}]$ is indeed a state on the resolvent algebra.
The same property holds for the objects without cutoffs.
\end{prop}

\begin{proof}
Since the domain restriction does not change the argument here either, we work with the cutoff notation.

Note the isomorphism of one-dimensional subalgebras $$\oaresolventalgebra_f = \oacstar \set{\oaresolvent(z,f)}{z \in \fldcmp \setminus \imunit \fldreal} \eqalgisom \fun{\conti_0}{\fldreal}, \quad \oaresolvent(\lambda,f) \to \frac{1}{\imunit \lambda - x}.$$
By the Cauchy transform of Gaussian measures, or by Proposition~\ref{expedition0011244}, $\oastate[\psi_{\txtvanhowe,\txtgs,\kappa,\Lambda}](\oaresolvent(\lambda,f))$ coincides with the Cauchy transform of the normal distribution $\fun{\msrcal{N}}{-\opreal \mathsf{m}_{\kappa,\Lambda}(f), \frac{1}{2} \opform{q}_{\txtnonzero,\infty}(f)}$.
In particular, this can be viewed as the integral of $\frac{1}{\imunit \lambda - x} \in \fun{\conti_0}{\fldreal}$ with respect to a Gaussian measure.
Therefore, $\oastate[\psi_{\txtvanhowe,\txtgs,\kappa,\Lambda}]$ is a positive linear functional on $\oaresolventalgebra_f$ arising from a probability measure, and is in particular a state.
By the general theory, it can be extended to a state on the full algebra; choosing a state with the correct restrictions to any one-dimensional subalgebra and denoting its extension by the same symbol gives the desired result.
\end{proof}

\begin{prop}\label{expedition0011232}
The state $\oastate[\psi]_{\txtgs,\kappa,\Lambda}$ is $\alpha_{\txtvanhowe,\kappa,\Lambda}$-invariant, i.e., it satisfies $$\oastate[\psi]_{\txtgs,\kappa,\Lambda} \circ \alpha_{\txtvanhowe,\kappa,\Lambda} = \oastate[\psi]_{\txtgs,\kappa,\Lambda}.$$
The same property holds for the objects without cutoffs.
\end{prop}

\begin{proof}
Since the domain restriction does not change the argument here either, we work with the cutoff notation.

In general, let the one-dimensional Gaussian distribution with mean $m$ and variance $C$ be $$p_{m,C}(x) = \frac{1}{\sqrt{2 \pi C}} \napiernum^{-\frac{(x-m)^2}{2C}}.$$
By definition, $$\opreal \mathsf{m}_{\kappa,\Lambda}(\napiernum^{\imunit t \omega} f) = \opreal \mathsf{m}_{\kappa,\Lambda}(f) + \mathsf{M}_{\kappa,\Lambda,t}(f),$$ hence
\begin{equation}
\begin{aligned}
&\fun{\oastate[\psi]_{\txtgs,\kappa,\Lambda}}
{\alpha_{\txtvanhowe,\kappa,\Lambda,t}(\oaresolvent(\lambda,f))}
=
\fun{\oastate[\psi]_{\txtgs,\kappa,\Lambda}}
{\oaresolvent(\lambda + \imunit \mathsf{M}_{\kappa,\Lambda,t}(f), \napiernum^{\imunit t \omega} f)} \\
&=
\int_{\fldreal}
\frac{1}{\imunit \rbk{\lambda + \imunit \mathsf{M}_{\kappa,\Lambda,t}(f)} - x}
p_{-\opreal \mathsf{m}_{\kappa,\Lambda}(\napiernum^{\imunit t \omega} f), \frac{1}{2} \opform{q}_{\txtnonzero,\infty}(\napiernum^{\imunit t \omega} f)}(x)
\opdmsr{x} \\
&=
\int_{\fldreal}
\frac{1}{\imunit \lambda - (x + \mathsf{M}_{\kappa,\Lambda,t}(f))}
p_{-\opreal \mathsf{m}_{\kappa,\Lambda}(f) - \mathsf{M}_{\kappa,\Lambda,t}(f), \frac{1}{2} \opform{q}_{\txtnonzero,\infty}(f)}(x)
\opdmsr{x}.
\end{aligned}
\end{equation}
Applying the change of variables $y = x + \mathsf{M}_{\kappa,\Lambda,t}(f)$ to eliminate the shift of the mean,
\begin{equation}
\begin{aligned}
&\fun{\oastate[\psi]_{\txtgs,\kappa,\Lambda}}
{\alpha_{\txtvanhowe,\kappa,\Lambda,t}(\oaresolvent(\lambda,f))} \\
&=
\int_{\fldreal}
\frac{1}{\imunit \lambda - (x + \mathsf{M}_{\kappa,\Lambda,t}(f))}
p_{-\opreal \mathsf{m}_{\kappa,\Lambda}(f) - \mathsf{M}_{\kappa,\Lambda,t}(f), \frac{1}{2} \opform{q}_{\txtnonzero,\infty}(f)}(x)
\opdmsr{x} \\
&=
\int_{\fldreal}
\frac{1}{\imunit \lambda - y}
p_{-\opreal \mathsf{m}_{\kappa,\Lambda}(f), \frac{1}{2} \opform{q}_{\txtnonzero,\infty}(f)}(y)
\opdmsr{y} \\
&=
\fun{\oastate[\psi]_{\txtgs,\kappa,\Lambda}}
{\oaresolvent(\lambda,f)},
\end{aligned}
\end{equation}
which gives the desired $\alpha_{\txtvanhowe,\kappa,\Lambda}$-invariance.
\end{proof}

We define a functional \(\chi_{\kappa,\Lambda}\) on the realification \(X = \sphilb{H}_{\txtreal}\) by \[\chi_{\kappa,\Lambda} \colon X \to \fldcmp; \quad \chi_{\kappa,\Lambda}(f) = \napiernum^{\imunit m_{\kappa,\Lambda,f} - \frac{1}{4} \opform{q}_{\txtnonzero,\infty}(f)}, \quad m_{\kappa,\Lambda,f} = -\opreal \mathsf{m}_{\kappa,\Lambda}(f).\] The corresponding objects without cutoffs are defined similarly, taking care of the domain.

\begin{prop}\label{expedition0011256}
The functional $\chi_{\kappa,\Lambda}$ is continuous and positive definite.
In particular, for any natural number $n$, any complex numbers $c_1,\ldots,c_n$, and any $f_1,\ldots,f_n \in X$, $$\sum_{j,k=1}^{n} \cmpconj{c_j} c_k \chi_{\kappa,\Lambda}(f_k - f_j) \geq 0.$$
The same property holds for the objects without cutoffs.
\end{prop}

\begin{proof}
Clear.
\end{proof}

By Minlos's theorem, the functional \(\chi_{\kappa,\Lambda}\) on \(X\) defined in Proposition~\ref{expedition0011256} guarantees the existence of a probability space \(\pairbk{\prbqspace_{\sphilb{H}},P_{\kappa,\Lambda}}\) and a real linear map \(\opfocksegal \colon X \to \fun{\lp_{\txtreal}^{2}}{\prbqspace_{\sphilb{H}},P_{\kappa,\Lambda}}\) satisfying \begin{equation}
\begin{aligned}
\sqfun{\prbexp_{P_{\kappa,\Lambda}}}{\napiernum^{\imunit \opfocksegal(f)}}
&=
\chi_{\kappa,\Lambda}(f), \\
\sqfun{\prbexp_{P_{\kappa,\Lambda}}}{\opfocksegal(f)}
&=
m_{\kappa,\Lambda,f}
=
-\opreal \mathsf{m}_{\kappa,\Lambda}(f), \\
\fun{\prbcov_{P_{\kappa,\Lambda}}}{\opfocksegal(f),\opfocksegal(g)}
&=
\onehalf \opreal \bkt{f}{g}.
\end{aligned}
\end{equation} Using this, we define the classical \(\ast\)-representation \(\repn_{\txtclassical,\kappa,\Lambda}\) of the resolvent algebra by \[\repn_{\txtclassical,\kappa,\Lambda} \colon \oaresolventalgebra(X,\sigma) \to \fun{\lp^{\infty}}{\prbqspace_{\sphilb{H}},P_{\kappa,\Lambda}}; \quad \repn_{\txtclassical,\kappa,\Lambda}(\oaresolvent(\lambda,f)) = \frac{1}{\imunit \lambda - \opfocksegal(f)}.\] The corresponding objects without cutoffs are defined similarly, taking care of the domain.

\begin{rem}[Structure related to classicalization and commutatization]
The classical $\ast$-representation $\repn_{\txtclassical,\kappa,\Lambda}$ is not necessarily faithful, since the information about the symplectic form $\sigma$ governing non-commutativity is lost.
This classicalization is intended for computing expectation values in the state defined below, and is not meant to commutatize the resolvent algebra, which is a non-commutative algebra.
In particular, the fundamental construction remains in the non-commutative algebra; we regard only the computations involving expectation values as being commutatized within the state.
\end{rem}

We define a state \(\oastate[\psi]_{\txtclassical,\kappa,\Lambda}\) of the resolvent algebra \(\oaresolventalgebra(X,\sigma)\) by \[\oastate[\psi]_{\txtclassical,\kappa,\Lambda}(A) = \int_{\prbqspace_{\sphilb{H}}} \repn_{\txtclassical,\kappa,\Lambda}(A) \opdmsr{P_{\kappa,\Lambda}(\phi)}, \quad A \in \oaresolventalgebra(X,\sigma).\]

On the dense subalgebra \(\oaresolventalgebra_{\txtfin}\), the state is extended using \[\fun{\oastate[\psi]_{\txtclassical,\kappa,\Lambda}}{\prod_{j=1}^n \oaresolvent(\lambda_j,f_j)} = \sqfun{\prbexp_{P_{\kappa,\Lambda}}}{\prod_{j=1}^n \frac{1}{\imunit \lambda - \opfocksegal(f_j)}},\] and on the full resolvent algebra by taking the limit. The corresponding objects without cutoffs are defined similarly, taking care of the domain.

\begin{prop}\label{expedition0011258}
The states $\oastate[\psi_{\txtvanhowe,\txtgs,\kappa,\Lambda}]$ and $\oastate[\psi]_{\txtclassical,\kappa,\Lambda}$ can be identified.
In particular, the following discussion will be based on this identification.
The same property holds for the objects without cutoffs.
\end{prop}

\begin{proof}
Since the domain restriction does not change the argument here either, we work with the cutoff notation.
The agreement of the one-point and two-point functions on the generators holds by definition.
The state $\oastate[\psi]_{\txtclassical,\kappa,\Lambda}$ can easily be extended to the dense subalgebra $\oaresolventalgebra_{\txtfin}$.
The identification for general elements follows by exhausting the freedom in extending the state $\oastate[\psi_{\txtvanhowe,\txtgs,\kappa,\Lambda}]$.
\end{proof}

\begin{prop}\label{expedition0011239}
The state $\oastate[\psi]_{\txtclassical,\kappa,\Lambda}$ is quasi-free.
Hence, by the identification discussed in Proposition~\ref{expedition0011258}, the state $\oastate[\psi]_{\txtgs,\kappa,\Lambda}$ is also quasi-free, i.e., the $n$-point expectation values are determined by the mean and the two-point function alone.
The same property holds for the objects without cutoffs.
\end{prop}

\begin{proof}
Since the domain restriction does not change the argument here either, we work with the cutoff notation.
The quasi-free property of $\oastate[\psi]_{\txtclassical,\kappa,\Lambda}$ follows from the construction via Minlos's theorem.
\end{proof}

\begin{prop}
The state $\oastate[\psi]_{\txtgs,\kappa,\Lambda}$ is a regular state in the sense of the resolvent algebra.
The same property holds for the objects without cutoffs.
\end{prop}

\begin{proof}
Since the domain restriction does not change the argument here either, we work with the cutoff notation.
By the identification of Proposition~\ref{expedition0011258}, it suffices to use the classical $\ast$-representation and $\oastate[\psi]_{\txtclassical,\kappa,\Lambda}$.

Choose any $\lambda \in \fldmultiplicativegroup{\fldreal}$, $f \in X$, and $A \in \oaresolventalgebra(X,\sigma)$.
Noting that $$\abs{\frac{\imunit \lambda}{\imunit \lambda - \opfocksegal(f)} - 1} \leq \frac{\abs{\opfocksegal(f)}}{\abs{\lambda}}, \quad \sqfun{\prbexp}{\opfocksegal(f)} < \infty,$$
Lebesgue's dominated convergence theorem gives
\begin{equation}
\begin{aligned}
&\imunit \lambda
\oastate[\psi]_{\txtclassical,\kappa,\Lambda}(\oaresolvent(\lambda,f) A)
=
\int_{\prbqspace_{\sphilb{H}}}
\frac{\imunit \lambda}{\imunit \lambda - \opfocksegal(f)}
\repn_{\txtclassical,\kappa,\Lambda}(A)
\opdmsr{P_{\kappa,\Lambda}(\phi)} \\
&\xrightarrow{\lambda \to \infty}
\int_{\prbqspace_{\sphilb{H}}}
\repn_{\txtclassical,\kappa,\Lambda}(A)
\opdmsr{P_{\kappa,\Lambda}(\phi)}
=
\oastate[\psi]_{\txtclassical,\kappa,\Lambda}(A).
\end{aligned}
\end{equation}
This is precisely the definition of regularity.
\end{proof}

\begin{prop}\label{expedition0011233}
Consider the two-point function $$K_{\kappa,\Lambda}(\tau) = \fun{\oastate[\psi_{\txtvanhowe,\txtgs,\kappa,\Lambda}]}{\oaresolvent(\lambda,f) \cdot \alpha_{\txtvanhowe,\kappa,\Lambda,\tau}(\oaresolvent(\mu,g))}.$$
By regularity, it is clearly continuous.
\begin{enumerate}
\item
The two-point function $K_{\kappa,\Lambda}$ can be analytically continued to an entire function as a holomorphic function bounded on the closed upper half-plane $\gtclos{\fldcmp_{\txtnonneg}}$.

\item
The limits at real times $\tau \to \pm\infty$ satisfy $$\lim_{\tau \to \pm \infty} K_{\kappa,\Lambda}(\tau) = \oastate[\psi_{\txtvanhowe,\txtgs,\kappa,\Lambda}](\oaresolvent(\lambda,f)) \oastate[\psi_{\txtvanhowe,\txtgs,\kappa,\Lambda}](\oaresolvent(\mu,g)).$$
In particular, the ground state has the temporal cluster property.
\end{enumerate}
The same properties hold for the objects without cutoffs.
In particular, for $f,g \in \dom \mathsf{m}$, the ground state without cutoffs also has the temporal cluster property.
\end{prop}

The spatial cluster property also holds by the same argument. In view of its importance, we organize and discuss it separately later.

\begin{proof}
(1): Since the domain restriction does not change the argument here either, we work with the cutoff notation.
For notational simplicity, let $\lambda,\mu > 0$.
Since the action of the automorphism group gives $$\alpha_{\txtvanhowe,\kappa,\Lambda,\tau}(\oaresolvent(\mu,g)) = \oaresolvent(\mu + \imunit \mathsf{M}_{\kappa,\Lambda,\tau}(g), \napiernum^{\imunit \tau \omega} g),$$
we obtain
\begin{equation}
\begin{aligned}
&K_{\kappa,\Lambda}(\tau)
=
-\int_0^{\infty}
\int_0^{\infty}
\napiernum^{-\lambda s}
\napiernum^{-(\mu + \imunit \mathsf{M}_{\kappa,\Lambda,\tau}(g)) t}
\mathsf{T}_{\kappa,\Lambda}(s, f; t, \napiernum^{\imunit \tau \omega} g)
\opdmsr{s}
\opdmsr{t} \\
&=
-\int_0^{\infty}
\int_0^{\infty}
\napiernum^{-\lambda s}
\napiernum^{-\mu t}
\napiernum^{-\imunit t \mathsf{M}_{\kappa,\Lambda,\tau}(g)}
\mathsf{T}_{\kappa,\Lambda}(s, f; t, \napiernum^{\imunit \tau \omega} g)
\opdmsr{s}
\opdmsr{t}.
\end{aligned}
\end{equation}
The variable $\tau$ appears only in the form $\tau \mapsto \napiernum^{\imunit \tau \omega}$, and it is clear that this extends to an entire function on the whole complex plane.

It remains to verify boundedness on the closed upper half-plane $\gtclos{\fldcmp_{\txtnonneg}}$.
Boundedness on the real axis is clear, so set $z = \tau + \imunit y$ for $y > 0$.
Denoting the integrand by $F(z;s,t)$, i.e., $$F(z;s,t) = \napiernum^{-\lambda s} \napiernum^{-\mu t} \napiernum^{-\imunit t \mathsf{M}_{\kappa,\Lambda,\tau}(g)} \mathsf{T}_{\kappa,\Lambda}(s, f; t, \napiernum^{\imunit \tau \omega} g),$$
note that although $\napiernum^{\imunit z \omega}$ also appears in the auxiliary functional $\mathsf{M}_{\kappa,\lambda,z}(g)$, by definition it is always real-valued, and $$\abs{\mathsf{T}_{\kappa,\Lambda}(s,f; t, \napiernum^{\imunit z \omega} g)} = \napiernum^{-\frac{1}{4} \norm{sf + t \napiernum^{\imunit z \omega} g}^2} \leq 1.$$
Therefore $\abs{F(z;s,t)} \leq \napiernum^{-\lambda s} \napiernum^{-\mu t}$, and the original $K_{\kappa,\Lambda}$ satisfies $$\abs{K_{\kappa,\Lambda}(z)} \leq \int_{0}^{\infty} \int_{0}^{\infty} \napiernum^{-\lambda s} \napiernum^{-\mu t} \opdmsr{s} \opdmsr{t} = \frac{1}{\lambda \mu},$$
which is uniformly bounded in $z$.

(2): Since the domain restriction does not change the argument here either, we work with the cutoff notation.
In particular, by the absolute continuity of the multiplication operator $\omega$,
\begin{equation}
\begin{aligned}
\lim_{\tau \to \pm \infty} \opreal \bkt{f}{\napiernum^{\imunit \tau \omega} g} &= 0, \\
\lim_{\tau \to \pm \infty} \opimag \bkt{f}{\napiernum^{\imunit \tau \omega} g} &= 0, \\
\lim_{\tau \to \pm \infty} \opreal \mathsf{m}_{\kappa,\Lambda}(\napiernum^{\imunit \tau \omega} g) &= 0.
\end{aligned}
\end{equation}
By the argument in part (1), the integrand $F(\tau;s,t)$ of the two-point function $K_{\kappa,\Lambda}$ is integrable on $\openinterval{0}{\infty}^2$.
By the first limit estimate,
\begin{equation}
\begin{aligned}
\norm{sf + t \napiernum^{\imunit \tau \omega} g}^2
&=
s^2 \opform{q}_{\txtnonzero,\infty}(f)
+t^2 \opform{q}_{\txtnonzero,\infty}(g)
+2st \opreal \bkt{f}{\napiernum^{\imunit \tau \omega} g} \\
&\xrightarrow{\tau \to \pm \infty}
s^2 \opform{q}_{\txtnonzero,\infty}(f)
+t^2 \opform{q}_{\txtnonzero,\infty}(g), \\
-\frac{\imunit}{2} \opimag \bkt{f}{\napiernum^{\imunit \tau \omega} g}
&\xrightarrow{\tau \to \pm \infty} 0, \\
\opreal \mathsf{m}_{\kappa,\Lambda}(sf + t \napiernum^{\imunit \tau \omega} g)
&\xrightarrow{\tau \to \pm \infty} \opreal \mathsf{m}_{\kappa,\Lambda}(sf).
\end{aligned}
\end{equation}
Combining the last term with the surviving phase factor $\napiernum^{-\imunit t \mathsf{M}_{\kappa,\Lambda,\tau} (g)}$ gives $$\napiernum^{-\imunit t \mathsf{M}_{\kappa,\Lambda,\tau} (g)} \napiernum^{\imunit \opreal \mathsf{m}_{\kappa,\Lambda}(t \napiernum^{\imunit \tau \omega} g)} = \napiernum^{\imunit t \opreal \mathsf{m}_{\kappa,\Lambda} (g)},$$ eliminating the $\tau$-dependence.
Therefore the integrand converges to
\begin{equation}
\begin{aligned}
\lim_{\tau \to \pm \infty} F(z;s,t)
=
\napiernum^{-\lambda s - \mu t}
\napiernum^{-\frac{1}{4} (s^2 \opform{q}_{\txtnonzero,\infty}(f) + t^2 \opform{q}_{\txtnonzero,\infty}(g))}
\napiernum^{\imunit s \opreal \mathsf{m}_{\kappa,\Lambda}(f)}
\napiernum^{\imunit t \opreal \mathsf{m}_{\kappa,\Lambda}(g)}.
\end{aligned}
\end{equation}
By the dominated convergence theorem,
\begin{equation}
\begin{aligned}
&\lim_{\tau \to \pm \infty}
K_{\kappa,\Lambda}(\tau) \\
&=
-\int_{\openinterval{0}{\infty}^2}
\napiernum^{-\lambda s - \mu t}
\napiernum^{-\frac{1}{4} (s^2 \opform{q}_{\txtnonzero,\infty}(f) + t^2 \opform{q}_{\txtnonzero,\infty}(g))}
\napiernum^{\imunit s \opreal \mathsf{m}_{\kappa,\Lambda}(f)}
\napiernum^{\imunit t \opreal \mathsf{m}_{\kappa,\Lambda}(g)}
\opdmsr{s}
\opdmsr{t} \\
&=
\rbk{-\int_{0}^{\infty}
\napiernum^{-\lambda s}
\napiernum^{-\frac{1}{4} s^2 \opform{q}_{\txtnonzero,\infty}(f)}
\napiernum^{\imunit s \opreal \mathsf{m}_{\kappa,\Lambda}(f)}
\opdmsr{s}} \\
&\qquad\times
\rbk{-\int_{0}^{\infty}
\napiernum^{-\mu t}
\napiernum^{-\frac{1}{4} t^2 \opform{q}_{\txtnonzero,\infty}(g)}
\napiernum^{\imunit t \opreal \mathsf{m}_{\kappa,\Lambda}(g)}
\opdmsr{t}} \\
&=
\oastate[\psi_{\txtvanhowe,\txtgs,\kappa,\Lambda}](\oaresolvent(\lambda,f))
\oastate[\psi_{\txtvanhowe,\txtgs,\kappa,\Lambda}](\oaresolvent(\mu,g)).
\end{aligned}
\end{equation}
\end{proof}

\begin{prop}\label{expedition0011112}
The state $\oastate[\psi_{\txtvanhowe,\txtgs,\kappa,\Lambda}]$ is a ground state of the van Hove model with respect to the automorphism group $\alpha_{\txtvanhowe,\kappa,\Lambda}$ and satisfies the temporal cluster property.
The same property holds for the objects without cutoffs.
\end{prop}

\begin{proof}
We verify \cite[P.98, Proposition 5.3.19]{BratteliRobinson2}.
It suffices to check on the generators; for any $\tau \in \fldreal$, let $$F(\tau) = \fun{\oastate[\psi_{\txtvanhowe,\txtgs,\kappa,\Lambda}]}{\oaresolvent(\lambda,f) \cdot \fun{\alpha_{\txtvanhowe,\kappa,\Lambda,\tau}}{\oaresolvent(\mu,g)}}.$$
This is the $K_{\kappa,\Lambda}(\tau)$ of Proposition~\ref{expedition0011233}.
\end{proof}

Fix an arbitrary nonzero real number \(q \in \fldmultiplicativegroup{\fldreal}\). We define the map \(\tau_{q,\theta}\) on generators by \[\tau_{q,\theta}(\oaresolvent(\lambda,f)) = \fun{\oaresolvent}{\lambda, \napiernum^{\imunit q \theta} f}\] and call it the global gauge transformation of the first kind.

\begin{prop}
The ground state of the van Hove model is not invariant under the global gauge transformation of the first kind.
\end{prop}

\begin{proof}
By Definition~\ref{expedition0011111}, the action on the ground state is
\begin{equation}
\begin{aligned}
&\psi_{\txtvanhowe,\txtgs,\kappa,\Lambda} \circ \tau_{q,\theta}(\oaresolvent(\lambda,f)) \\
&=
-\imunit
\int_0^{(\sgn \lambda) \infty}
\napiernum^{-\rbk{\lambda - \imunit \opreal \fun{\mathsf{m}_{\kappa,\Lambda}}{\napiernum^{\imunit q \theta} f}} t}
\napiernum^{-\frac{\fun{\opform{q}_{\txtnonzero,\infty}}{\napiernum^{\imunit q \theta} f}}{4} t^2}
\opdmsr{t} \\
&=
-\imunit
\int_0^{(\sgn \lambda) \infty}
\napiernum^{-\rbk{\lambda - \imunit \opreal \fun{\mathsf{m}_{\kappa,\Lambda}}{\napiernum^{\imunit q \theta} f}} t}
\napiernum^{-\frac{\opform{q}_{\txtnonzero,\infty}(f)}{4} t^2}
\opdmsr{t}.
\end{aligned}
\end{equation}
In general, this does not coincide with $\psi_{\txtvanhowe,\txtgs,\kappa,\Lambda}(\oaresolvent(\lambda,f))$.
In particular, the ground state is not invariant under the global gauge transformation of the first kind.
\end{proof}

\subsection{Removal of Infrared and Ultraviolet Cutoffs}\label{removal-of-infrared-and-ultraviolet-cutoffs}

\begin{prop}\label{expedition0011261}
Fix any positive real number $T > 0$.
Then for any $A \in \oaresolventalgebra(\dom \mathsf{m},\sigma)$, $$\lim_{\kappa \to 0, \Lambda \to \infty} \sup_{\abs{t} \leq T} \norm{\alpha_{\txtvanhowe,\kappa,\Lambda,t}(A) - \alpha_{\txtvanhowe,t}(A)} = 0.$$
\end{prop}

\begin{proof}
For any $f \in \dom \mathsf{m}$ and $z \in \fldcmp \setminus \imunit \fldreal$, it suffices to show the limit $$\lim_{\kappa \to 0, \Lambda \to \infty} \sup_{\abs{t} \leq T} \norm{\alpha_{\txtvanhowe,\kappa,\Lambda,t}(\oaresolvent(z,f)) - \alpha_{\txtvanhowe,t}(\oaresolvent(z,f))} = 0$$ on the generators.

By the definition of the automorphism group and the fourth resolvent relation in its complex form,
\begin{equation}
\begin{aligned}
&\norm{\alpha_{\txtvanhowe,\kappa,\Lambda,t}(\oaresolvent(z,f)) - \alpha_{\txtvanhowe,t}(\oaresolvent(z,f))} \\
&=
\norm{\oaresolvent(z + \imunit \mathsf{M}_{\kappa,\Lambda,t}(f), \napiernum^{\imunit t \omega} f) - \oaresolvent(z + \imunit \mathsf{M}_{t}(f), \napiernum^{\imunit t \omega} f)} \\
&=
\norm{\imunit (\mathsf{M}_{\kappa,\Lambda,t}(f) - \mathsf{M}_{t}(f))
\cdot
\oaresolvent(z + \imunit \mathsf{M}_{\kappa,\Lambda,t}(t), \napiernum^{\imunit t \omega} f)
\cdot
\oaresolvent(z + \imunit \mathsf{M}_{t}(f), \napiernum^{\imunit t \omega} f)} \\
&=
\abs{\mathsf{M}_{\kappa,\Lambda,t}(f) - \mathsf{M}_{t}(f)}
\norm{\oaresolvent(z + \imunit \mathsf{M}_{\kappa,\Lambda,t}(f), \napiernum^{\imunit t \omega} f)
\cdot
\oaresolvent(z + \imunit \mathsf{M}_{t}(f), \napiernum^{\imunit t \omega} f)} \\
&\leq
\abs{\mathsf{M}_{\kappa,\Lambda,t}(f) - \mathsf{M}_{t}(f)}
\norm{\oaresolvent(z + \imunit \mathsf{M}_{\kappa,\Lambda,t}(f), \napiernum^{\imunit t \omega} f)}
\cdot
\norm{\oaresolvent(z + \imunit \mathsf{M}_{t}(f), \napiernum^{\imunit t \omega} f)} \\
&\leq
\abs{\mathsf{M}_{\kappa,\Lambda,t}(f) - \mathsf{M}_{t}(f)}
\frac{1}{\abs{\opreal z}^2}.
\end{aligned}
\end{equation}
By the definition of the linear functionals $\mathsf{M}_{\kappa,\Lambda,t}$ and $\mathsf{M}_{t}$,
\begin{equation}
\begin{aligned}
&{\mathsf{M}_{\kappa,\Lambda,t}(f) - \mathsf{M}_{t}(f)} \\
&\leq
\rbk{\fun{\mathsf{m}_{\kappa,\Lambda}}{\napiernum^{\imunit t \omega} f} - \fun{\mathsf{m}}{\napiernum^{\imunit t \omega} f}}
+\rbk{\mathsf{m}_{\kappa,\Lambda}(f) - \mathsf{m}(f)}.
\end{aligned}
\end{equation}
For any $f \in \dom \mathsf{m}$ and $\abs{t} \leq T$,
\begin{equation}
\begin{aligned}
&\abs{\mathsf{m}_{\kappa,\Lambda}(\napiernum^{\imunit t \omega} f) - \mathsf{m}(\napiernum^{\imunit t \omega} f)}
=
\int_{\greuctr{k}{3}}
\frac{\abs{\faftr{\varrho_{\kappa,\Lambda}}(k) - \faftr{\varrho}(k)} \abs{\faftrrbk{\napiernum^{\imunit t \omega} f}(k)}{}}{\omega(k)^{\frac{3}{2}}}
\opdmsr{k} \\
&=
\int_{\abs{k} \leq \kappa}
\frac{\abs{\faftr{\varrho}(k)} \abs{\faftr{f}(k)}}{\omega(k)^{\frac{3}{2}}}
\opdmsr{k}
+\int_{\abs{k} \geq \Lambda}
\frac{\abs{\faftr{\varrho}(k)} \abs{\faftr{f}(k)}}{\omega(k)^{\frac{3}{2}}}
\opdmsr{k} \\
&\xrightarrow{\kappa \to 0, \Lambda \to \infty}
0
\end{aligned}
\end{equation}
uniformly in $t$ satisfying $\abs{t} \leq T$.
Thus the desired limit is obtained.
\end{proof}

\begin{prop}\label{expedition0011266}
On the subalgebra $\oaresolventalgebra(\dom \mathsf{m},\sigma)$, taking the limit $\kappa \to 0, \Lambda \to \infty$ to remove the cutoffs, the family of ground states $\fml{\oastate[\psi_{\txtvanhowe,\txtgs,\kappa,\Lambda}]}{\kappa>0,\Lambda>0}$ converges weakly to $\oastate[\psi_{\txtvanhowe,\txtgs}]$.
\end{prop}

\begin{proof}
By Proposition~\ref{expedition0011112}, the ground state with cutoffs is a ground state on the full resolvent algebra, and the ground state without cutoffs is a ground state on the subalgebra generated by $\dom \mathsf{m}$.
To make the cutoff-free objects well-defined, we need to work on the subalgebra of the full resolvent algebra restricted to $\sphilb{D}_{\omega,\rho}$; accordingly, all operators appearing below, in particular $f \in \sphilb{H}$, are restricted to $\sphilb{D}_{\omega,\rho}$.

Consider the family of ground states $\fml{\psi_{\txtvanhowe,\txtgs,\kappa,\Lambda}}{\kappa>0,\Lambda>0}$.
Since the resolvent algebra is a unital $\oacstar$-algebra, one can extract a subnet that converges weakly in the limit $\kappa \to 0$, $\Lambda \to \infty$.
The mean functional $\mathsf{m}_{\kappa,\Lambda} \to \mathsf{m}$ converges on $\sphilb{D}_{\omega,\rho}$ in the strong topology of $\sphilb{H}$.
By Lebesgue's dominated convergence theorem, $\oastate[\psi_{\txtvanhowe,\txtgs,\kappa,\Lambda}]$ converges weakly to $\oastate[\psi_{\txtvanhowe,\txtgs}]$ on the subalgebra $\oaresolventalgebra_{\txtfin}$.
Hence the limit of the subnet is unique, and the entire net converges to $\oastate[\psi_{\txtvanhowe,\txtgs}]$.
Since all of these are states on a unital $\oacstar$-algebra, $\oastate[\psi_{\txtvanhowe,\txtgs,\kappa,\Lambda}(\idone)] = 1$ and the limit is never $0$.
Putting the above argument together gives the desired weak convergence.
\end{proof}

\subsection{The Van Hove Model under Infrared Singularity and Ideal Theory}\label{the-van-hove-model-under-infrared-singularity-and-ideal-theory}

In what follows, we discuss the van Hove model proper with the infrared singularity condition imposed, after removing the cutoffs. By the operator-theoretic argument of \cite{AsaoArai26} or the discussion in Section~\ref{expedition0011496}, when a point source is chosen for \(\varrho\), the ultraviolet singularity of the van Hove model manifests only in the renormalization of the self-energy. Therefore, the discussion should focus on the effect of the infrared singularity.

\begin{prop}\label{expedition0011265}
Define the function $\mathsf{f}$ by $$\faftr{\mathsf{f}}(k) = -\frac{1}{\abs{k}^{\frac{3}{2}} \log \abs{k}} \fndef{\abs{k} < 1}(k), \quad k \in \greuctr{k}{3}.$$
This function is square-integrable but does not belong to $\dom \mathsf{m}$.
\end{prop}

\begin{rem}
The resolvent corresponding to this $\mathsf{f}$ is not even defined in the GNS representation of the ground state.
In this sense, an extension to handle non-regularity is needed.
For generators satisfying $f \notin \dom \mathsf{m}$, one possible treatment in the GNS representation is to formally set $\repn_{\psi_{\txtvanhowe,\txtgs}}(\oaresolvent(\lambda,f)) = 0$.
Another approach might use appropriate weights.
\end{rem}

\begin{proof}
($f \in \fun{\lp^{2}}{\greuctr{x}{3}}$): We compute in polar coordinates using the substitution $u = -\log r$.
This substitution gives $r = \napiernum^{-u}$ and $\opdmsr{r} = -r \opdmsr{u}$, so
\begin{equation}
\begin{aligned}
&\norm{\faftr{f}}_{\fun{\lp^{2}}{\greuctr{k}{3}}}^2
=
4 \pi
\int_{0}^1 \frac{r^2}{r^3 \rbk{\log r}^2}
\opdmsr{r}
=
4 \pi
\int_{0}^1 \frac{1}{r \rbk{\log r}^2}
\opdmsr{r} \\
&=
4 \pi
\int_{\infty}^{1}
\frac{1}{r u^2}
(-r) \opdmsr{u}
=
4 \pi
\int_{1}^{\infty}
u^{-2}
\opdmsr{u}
=
4 \pi
<
\infty.
\end{aligned}
\end{equation}
By Plancherel's theorem, $f \in \fun{\lp^{2}}{\greuctr{x}{3}}$.

($\abs{k}^{-\frac{3}{2}} \faftr{f} \notin \fun{\lp^{1}}{\greuctr{x}{3}}$): With the same substitution, $$\int_{\greuctr{k}{3}} \frac{\abs{\faftr{f}(k)}}{\abs{k}^{\frac{3}{2}}} \opdmsr{k} = 4 \pi \int_{0}^{1} \frac{r^2}{r^{\frac{3}{2}}} \cdot \frac{1}{r^{\frac{3}{2}} \log r} \opdmsr{r} = 4 \pi \int_{1}^{\infty} u^{-1} \opdmsr{u} = \infty.$$
\end{proof}

\begin{rem}[Temporal and spatial cluster properties]
The argument of Proposition~\ref{expedition0011233} applies directly, and both the temporal and spatial cluster properties hold regardless of whether cutoffs are present.
Of course, for the ground state without cutoffs, the discussion must be restricted to $\dom \mathsf{m}$.
Even if the resolvent with $f \notin \dom \mathsf{m}$ is formally set to $0$ in the GNS representation, both the temporal and spatial cluster properties hold with both sides equal to $0$.
\end{rem}

In the following, we examine the ideal theory from Buchholz's original paper \cite{DetlevBuchholz001} concretely for this system.

\begin{prop}
The set $\dom \mathsf{m}$ is a dense linear subspace of $\fun{\lp^{2}}{\fldreal^{d}}$ and is invariant under the free-field time evolution $f \mapsto \napiernum^{\imunit t \omega}f$.
\end{prop}

\begin{proof}
Linearity and invariance under the time evolution are clear.
The subspace $\fun{\conti_{\txtcpt}^{\infty}}{\fldreal^{d} \setminus \setone{0}}$ is contained in $\dom \mathsf{m}$ and is dense in $\fun{\lp^{2}}{\fldreal^{d}}$.
\end{proof}

We define the regular subspace from Buchholz's original paper \cite{DetlevBuchholz001} as \[X_R = \dom \mathsf{m} = \dom \omega^{-\frac{3}{2}}.\]

\begin{lem}\label{expedition0011906}
For the regular ideal $X_R$, the trivial regular ideal is $$X_T = X_R \cap X_R^{\perp_\sigma} = \setone{0}.$$
\end{lem}

\begin{proof}
The subspace $\dom \mathsf{m}$ is dense in $\lp^{2}$, and the symplectic form is continuous and non-degenerate.
Therefore $X_R^{\perp_\sigma} = \setone{0}$.
\end{proof}

\begin{defn}
Let $\oarepn_{\txtgs}$ also denote the representation obtained by extending the GNS representation of the ground state without cutoffs to $\oaresolventalgebra(X,\sigma)$.
We define the infrared ideal by $$J_{\txtinflaredred} = \gtclos{\set{\oaresolvent(\lambda,f)}{\lambda \in \fldmultiplicativegroup{\fldreal}, f \notin \dom \mathsf{m}}}.$$
Furthermore, we call a closed two-sided ideal $J$ of the resolvent algebra physically admissible if it satisfies the following two conditions:
\begin{enumerate}
\item
The automorphism group $\alpha_{\txtvanhowe}$ is well-defined on the quotient algebra $\setquot{\oaresolventalgebra(X,\sigma)}{J}$.

\item
The ground state $\oastate[\psi_{\txtvanhowe,\txtgs}]$ descends to a state.
\end{enumerate}
\end{defn}

\begin{prop}\label{expedition0011905}
Let $\oarepn_{\txtgs}$ also denote the representation obtained by extending the GNS representation of the ground state without cutoffs to $0$ on $\sphilb{H} \setminus \dom \mathsf{m}$.
Then the following statements hold.
\begin{enumerate}
\item
The regular subspace is $X_R = \dom \mathsf{m}$.

\item
The singular subspace is $X_S = X \setminus \dom \mathsf{m}$.

\item
The trivial regular part is trivial: $X_T = \setone{0}$.

\item
The kernel of the GNS representation of the ground state is $$\Ker \oarepn_{\txtgs} = J_{\txtinflaredred} = \gtclos{\set{\oaresolvent(\lambda,f)}{\lambda \in \fldmultiplicativegroup{\fldreal}, f \in X_S}}.$$

\item
The infrared ideal $J_{\txtinflaredred}$ is a primitive ideal, and $$\setquot{\oaresolventalgebra(X,\sigma)}{J_{\txtinflaredred}} \eqisom \oaresolventalgebra(\dom \mathsf{m},\sigma).$$
\end{enumerate}
\end{prop}

\begin{rem}
Although Bose-Einstein condensation occurs at finite temperature, the $f$ that causes Bose-Einstein condensation is in general a distribution and is difficult to handle.
From the perspective of the emergence of a non-trivial kernel in the representation, the infrared divergence is more transparent, so it is worth discussing here.
When the infrared divergence is strong at finite temperature, the competition of zero modes causes the infrared divergence to kill Bose-Einstein condensation at the level of the algebra of observables.
That is, at finite temperature, even more complex mathematical phenomena can occur \cite{YoshitsuguSekine006}.
\end{rem}

\begin{proof}
(1)--(2): These are definitions.

(3): By Lemma~\ref{expedition0011906}.

(4): This is also essentially by definition.
In general, by part (3) of this proposition, the index term in Buchholz's kernel identification theorem vanishes.
Hence the kernel coincides with the closed two-sided ideal generated by $\set{\oaresolvent(\lambda,f)}{\lambda \in \fldmultiplicativegroup{\fldreal}, f \in X_S}$.

(5): The isomorphism of the quotient algebra follows immediately from the construction.
\end{proof}

\begin{cor}
The algebra of observables under the infrared singularity condition is $\oaresolventalgebra(\dom \mathsf{m},\sigma)$.
The infrared ideal $J_{\txtinflaredred}$ is a physically admissible ideal, and in particular the minimal $\oacstar$-ideal that removes the infrared divergence modes.
\end{cor}

\begin{proof}
For a physically admissible ideal $J$, generators outside $\dom \mathsf{m}$ must not have meaning in the quotient algebra.
In particular, for the quotient map $q \colon \oaresolventalgebra(X,\sigma) \to \setquot{\oaresolventalgebra(X,\sigma)}{J}$, we must have $q(\oaresolvent(\lambda,f)) = 0$ for $f \in X_S$.
The ideal $J$ must contain the closed two-sided ideal generated by all of these, and by definition $J_{\txtinflaredred} \subset J$.
This is simply a reinterpretation of Proposition~\ref{expedition0011905}(5).
\end{proof}

\section{Weyl Algebra at Finite Temperature: Bounded System}\label{expedition0011587}

As discussed also in \cite{YoshitsuguSekine002,YoshitsuguSekine006}, the model reduces essentially to a free Bose gas. However, since references that discuss the van Hove model at finite temperature are rare, we include the basic arguments for the reader's convenience. The basic approach is the same as the ground-state discussion of \cite{AsaoArai26}: a unitary transformation by Weyl operators under the infrared regularization assumption.

Since the computation results are concise, we first compute the various quantities at finite temperature in the bounded system in the Fock representation of the Weyl algebra. In particular, we assume infrared and ultraviolet cutoffs and work in the Fock representation. In this section, with a nonzero chemical potential and under infrared and ultraviolet regularization, the Hilbert space generating the Weyl algebra is the full space \(\sphilb{H}_L = \fun{\lp^{2}}{I_{L}^{d}}\).

First, the following statement is clear by the Fourier transform.

\begin{prop}\label{expedition0011286}
The single-particle Hamiltonian or dispersion relation $\omega$ is a non-negative self-adjoint operator on $\fun{\lp^{2}}{I_{L}^{d}}$.
Its spectrum is discrete and satisfies $$\opspec{\omega} = \opspec[\txtdiscrete]{\omega} = \set{\omega(k)}{k \in \setlattice_L^{d}}.$$
\end{prop}

\begin{prop}
For any inverse temperature $\sminvtemperature > 0$, $$\sqfun{\trace_{\fun{\lp^{2}}{I_{L}^{d}}}}{\napiernum^{-\sminvtemperature \omega}} < \infty.$$
\end{prop}

\begin{proof}
In the bounded system, the spectrum of the dispersion relation $\omega$ is discrete, so $$S = \sqfun{\trace_{\fun{\lp^{2}}{I_{L}^{d}}}}{\napiernum^{-\sminvtemperature \omega}} = \sum_{k \in \setlattice_L^{d}} \napiernum^{-\sminvtemperature \omega(k)}.$$
It remains to verify that the right-hand side converges.

Set $a = \frac{2 \pi}{L}$ for notational convenience.
Decomposing by radial shells, $$S \leq 1 + \sum_{m=0}^{\infty} N_m \napiernum^{-\sminvtemperature m}, \quad N_m = \setcardop \set{n \in \ringratint^{d}}{m \leq \abs{n} \leq m+1}.$$
By volume comparison, the number of lattice points satisfies $$N_m \leq C \rbk{\rbk{m+1}^3 - m^3} \leq C' \rbk{m^2 + 1}$$ for constants $C, C'$ depending only on the dimension.
Therefore, using an appropriate positive constant $C''$, $$S \leq 1 + C'' \sum_{m=0}^{\infty} (m^2+1) \napiernum^{-\sminvtemperature a m} < \infty.$$
\end{proof}

As in Section~\ref{expedition0012095}, we define the unitary transformation introduced in \cite{AsaoArai26} as \(V_{\kappa,\Lambda} = \opfockweyl_{\txtfock}(\imunit \mathsf{m}_{\kappa,\Lambda})\).

\begin{prop}\label{expedition0012088}
The van Hove model Hamiltonian $\physham_{\txtvanhowe,I_{L}^{d},\kappa,\Lambda}(\smchemicalpotential)$ is transformed under the operator $V_{\kappa,\Lambda} = \opfockweyl_{\txtfock}(\imunit \mathsf{m}_{\kappa,\Lambda})$ as $$V_{\kappa,\Lambda} \physham_{\txtbsn,\txtfr}(\smchemicalpotential) \inv{V_{\kappa,\Lambda}} = \physham_{\txtvanhowe,I_{L}^{d},\kappa,\Lambda}(\smchemicalpotential) - \fun{\physgse}{\physham_{\txtvanhowe,I_{L}^{d},\kappa,\Lambda}(\smchemicalpotential)},$$
where $\fun{\physgse}{\physham_{\txtvanhowe,I_{L}^{d},\kappa,\Lambda}(\smchemicalpotential)}$ is the ground-state energy.
\end{prop}

\begin{proof}
It suffices to use the result of the unitary transformation of the free Hamiltonian by the exponential of the Segal field operator, in addition to Fact~\ref{expedition0010960}.
The computations in \cite{YoshitsuguSekine001,YoshitsuguSekine002} will also be helpful as reference.
\end{proof}

By the following proposition, the thermal operator of the van Hove Hamiltonian with chemical potential is trace class in the bounded system.

\begin{prop}\label{expedition0011287}
The thermal operator $\napiernum^{-\sminvtemperature \physham_{\txtvanhowe,I_{L}^{d},\kappa,\Lambda}(\smchemicalpotential)}$ of the van Hove Hamiltonian with chemical potential is trace class, and satisfies in particular $$\sqfun{\trace}{\napiernum^{-\sminvtemperature \physham_{\txtvanhowe,I_{L}^{d},\kappa,\Lambda}(\smchemicalpotential)}} = \frac{1}{\napiernum^{\sminvtemperature \fun{\physgse}{\physham_{\txtvanhowe,I_{L}^{d},\kappa,\Lambda}(\smchemicalpotential)}}} \sqfun{\trace}{\napiernum^{-\sminvtemperature \physham_{\txtbsn,\txtfr}(\smchemicalpotential)}}.$$
\end{prop}

\begin{proof}
As is well known, the partition function of the free Bose gas is $$\sqfun{\trace}{\napiernum^{-\sminvtemperature \physham_{\txtbsn,\txtfr}(\smchemicalpotential)}} = \frac{1}{\prod_{k \in \setlattice_L^{d}} \rbk{1 - \napiernum^{-\sminvtemperature(\omega(k) - \smchemicalpotential)}}} < \infty.$$
Since unitary transformations preserve the spectrum, the thermal operator of $$\physham_{\txtvanhowe,I_{L}^{d},\kappa,\Lambda}(\smchemicalpotential) - \fun{\physgse}{\physham_{\txtvanhowe,I_{L}^{d},\kappa,\Lambda}(\smchemicalpotential)}$$ is also trace class, and $$\sqfun{\trace}{\napiernum^{-\sminvtemperature \physham_{\txtvanhowe,I_{L}^{d},\kappa,\Lambda}(\smchemicalpotential)}} = \napiernum^{-\sminvtemperature \fun{\physgse}{\physham_{\txtvanhowe,I_{L}^{d},\kappa,\Lambda}(\smchemicalpotential)}} \sqfun{\trace}{\napiernum^{-\sminvtemperature \physham_{\txtbsn,\txtfr}(\smchemicalpotential)}}.$$
\end{proof}

Using the above computational results, we define the grand partition function and the grand canonical state. In particular, \begin{equation}
\begin{aligned}
\smgrandpartitionfunc_{\txtbsn,\txtfr,I_{L}^{d},\sminvtemperature,\smchemicalpotential}
&=
\sqfun{\trace}{\napiernum^{-\sminvtemperature \physham_{\txtbsn,\txtfr}(\smchemicalpotential)}}
=
\frac{1}{\prod_{k \in I_{L}^{d}} \rbk{1 - \napiernum^{-\sminvtemperature(\omega(k) - \smchemicalpotential)}}}, \\
\smgrandpartitionfunc_{\txtvanhowe,I_{L}^{d},\kappa,\Lambda,\sminvtemperature,\smchemicalpotential}
&=
\frac{1}
{\napiernum^{\sminvtemperature
\fun{\physgse}
{\physham_{\txtvanhowe,I_{L}^{d},\kappa,\Lambda}(\smchemicalpotential)}}}
\smgrandpartitionfunc_{\txtbsn,\txtfr,I_{L}^{d},\sminvtemperature,\smchemicalpotential}
\end{aligned}
\end{equation}

and we call the latter \(\smgrandpartitionfunc_{\txtvanhowe,I_{L}^{d},\kappa,\Lambda,\sminvtemperature,\smchemicalpotential}\) the grand partition function for the van Hove Hamiltonian.

Furthermore, the state on the algebra of all bounded operators \(\opspbddlin{\fun{\spfock_{\txtbsn}}{\fun{\lp^{2}}{I_{L}^{d}}}}\) defined by \[\oastate[\psi_{\txtgrandcanonical,\txtvanhowe,I_{L}^{d},\kappa,\Lambda,\sminvtemperature,\smchemicalpotential}](A)
=
\frac{1}{\smgrandpartitionfunc_{\txtbsn,\txtvanhowe,I_{L}^{d},\kappa,\Lambda,\sminvtemperature,\smchemicalpotential}}
\sqfun{\trace}{\napiernum^{-\sminvtemperature \physham_{\txtvanhowe,I_{L}^{d},\kappa,\Lambda}(\smchemicalpotential)} A}\]is defined. This is called the grand canonical state for the van Hove Hamiltonian.

As in the discussion of the free Bose gas in \cite{AsaoArai28}, the domain of the grand canonical state extends to the unbounded Fock canonical commutation relation algebra including creation and annihilation operators. We proceed under this assumption in what follows.

\begin{prop}\label{expedition0011296}
The one-point functions for the creation and annihilation operators and the Segal field operator take the values
\begin{equation}
\begin{aligned}
\oastate[\psi_{\txtgrandcanonical,\txtvanhowe,I_{L}^{d},\kappa,\Lambda,\sminvtemperature,\smchemicalpotential}](\opfockcr(f))
&=
-\frac{1}{\sqrt{2}}
\bkt{\mathsf{m}_{\kappa,\Lambda}}{f}, \\
\oastate[\psi_{\txtgrandcanonical,\txtvanhowe,I_{L}^{d},\kappa,\Lambda,\sminvtemperature,\smchemicalpotential}](\opfockan(f))
&=
-\frac{1}{\sqrt{2}}
\cmpconj{\bkt{\mathsf{m}_{\kappa,\Lambda}}{f}}, \\
\oastate[\psi_{\txtgrandcanonical,\txtvanhowe,I_{L}^{d},\kappa,\Lambda,\sminvtemperature,\smchemicalpotential}](\opfocksegal_{\txtfock}(f))
&=
-\opreal
\bkt{\mathsf{m}_{\kappa,\Lambda}}{f}
\end{aligned}
\end{equation}
\end{prop}

\begin{proof}
We argue somewhat informally and in a slightly formal manner.
To simplify notation, let $\physgse$ denote the ground state energy of the van Hove Hamiltonian and set $V_{\kappa,\Lambda}
= \opfockweyl_{\txtfock}(\imunit \mathsf{m}_{\kappa,\Lambda})$.

(Creation operator): Applying Proposition \ref{expedition0012088} to the Hamiltonian and Fact \ref{expedition0010960} to the creation/annihilation operators,
\begin{equation}
\begin{aligned}
&\sqfun{\trace}
{\napiernum^{-\sminvtemperature \physham_{\txtvanhowe,I_{L}^{d},\kappa,\Lambda}(\smchemicalpotential)}
\cdot
\opfockcr(f)} \\
&=
\sqfun{\trace}
{V_{\kappa,\Lambda}
\inv{V_{\kappa,\Lambda}}
\napiernum^{-\sminvtemperature \physham_{\txtvanhowe,I_{L}^{d},\kappa,\Lambda}(\smchemicalpotential)}
V_{\kappa,\Lambda}
\inv{V_{\kappa,\Lambda}}
\opfockcr(f)} \\
&=
\napiernum^{-\sminvtemperature \physgse}
\sqfun{\trace}
{\napiernum^{-\sminvtemperature \physham_{\txtbsn,\txtfr,I_{L}^{d}}(\smchemicalpotential)}
\inv{V_{\kappa,\Lambda}}
\opfockcr(f)
V_{\kappa,\Lambda}} \\
&=
\napiernum^{-\sminvtemperature \physgse}
\sqfun{\trace}
{\napiernum^{-\sminvtemperature \physham_{\txtbsn,\txtfr,I_{L}^{d}}(\smchemicalpotential)}
\rbk{\opfockcr(f) + \frac{\imunit}{\sqrt{2}} \bkt{-\imunit \mathsf{m}_{\kappa,\Lambda}}{f}}} \\
&=
-\frac{1}{\sqrt{2}}
\napiernum^{-\sminvtemperature \physgse}
\bkt{\mathsf{m}_{\kappa,\Lambda}}{f}
\sqfun{\trace}
{\napiernum^{-\sminvtemperature \physham_{\txtbsn,\txtfr,I_{L}^{d}}(\smchemicalpotential)}} \\
&=
-\frac{1}{\sqrt{2}}
\bkt{\mathsf{m}_{\kappa,\Lambda}}{f}
\smgrandpartitionfunc_{\txtbsn,\txtfr,I_{L}^{d},\sminvtemperature,\smchemicalpotential}
\end{aligned}
\end{equation}
is obtained.
The desired result then follows from the evaluation of the grand partition function.

(Annihilation operator): The only difference is in applying Fact \ref{expedition0010960} to the commutation relation, and in particular
\begin{equation}
\begin{aligned}
&\sqfun{\trace}
{\napiernum^{-\sminvtemperature \physham_{\txtvanhowe,I_{L}^{d},\kappa,\Lambda}(\smchemicalpotential)}
\cdot
\opfockan(f)} \\
&=
\napiernum^{-\sminvtemperature \physgse}
\sqfun{\trace}
{\napiernum^{-\sminvtemperature \physham_{\txtbsn,\txtfr,I_{L}^{d}}(\smchemicalpotential)}
\rbk{\opfockan(f) - \frac{\imunit}{\sqrt{2}} \bkt{f}{-\imunit \mathsf{m}_{\kappa,\Lambda}}}} \\
&=
-\frac{1}{\sqrt{2}}
\bkt{f}{\mathsf{m}_{\kappa,\Lambda}}
\napiernum^{-\sminvtemperature \physgse}
\sqfun{\trace}
{\napiernum^{-\sminvtemperature \physham_{\txtbsn,\txtfr,I_{L}^{d}}(\smchemicalpotential)}}
\end{aligned}
\end{equation}
holds.
Again, the desired result follows from the evaluation of the grand partition function.

(Segal field operator): Take $\frac{1}{\sqrt{2}}$ times the sum of the creation and annihilation operator results.
\end{proof}

Next, we compute the two-point functions for the grand canonical state.

\begin{prop}\label{expedition0011299}
As values of the two-point functions for the creation and annihilation operators and the Segal field operator, using fractional notation to highlight commutativity for readability,
\begin{equation}
\begin{aligned}
&\oastate[\psi_{\txtgrandcanonical,\txtvanhowe,I_{L}^{d},\kappa,\Lambda,\sminvtemperature,\smchemicalpotential}](\opfockcr(f) \opfockan(g)) \\
&=
\bkt{g}
{\frac{\napiernum^{-\sminvtemperature (\omega - \smchemicalpotential)}}
{1 - \napiernum^{-\sminvtemperature (\omega - \smchemicalpotential)}} f}
+\onehalf
\bkt{\mathsf{m}_{\kappa,\Lambda}}{f}
\bkt{g}{\mathsf{m}_{\kappa,\Lambda}}, \\
&=
\bkt{g}
{\frac{1}
{\napiernum^{\sminvtemperature (\omega - \smchemicalpotential)} - 1}
f}
+\onehalf
\bkt{\mathsf{m}_{\kappa,\Lambda}}{f}
\bkt{g}{\mathsf{m}_{\kappa,\Lambda}}, \\
&\oastate[\psi_{\txtgrandcanonical,\txtvanhowe,I_{L}^{d},\kappa,\Lambda,\sminvtemperature,\smchemicalpotential}](\opfocksegal_{\txtfock}(f) \opfocksegal_{\txtfock}(g)) \\
&=
\onehalf
\rbk{\opreal \opform{q}_{\txtnonzero,\sminvtemperature,\smchemicalpotential}(f,g)
+\imunit \opimag \bkt{f}{g}}
+\opreal \bkt{\mathsf{m}_{\kappa,\Lambda}}{f} \cdot \opreal \bkt{\mathsf{m}_{\kappa,\Lambda}}{g}
\end{aligned}
\end{equation}
hold.
\end{prop}

\begin{proof}
(Preparation): We proceed as for the one-point function, again arguing somewhat formally.
Let $\physgse$ denote the ground state energy of the van Hove Hamiltonian and define the unitary operator $V_{\kappa,\Lambda}
= \opfockweyl_{\txtfock}(\imunit \mathsf{m}_{\kappa,\Lambda})$.

($\opfockcr(f) \opfockan(g)$): Set the operator to be evaluated as $$W
=
\napiernum^{-\sminvtemperature \physham_{\txtvanhowe,I_{L}^{d},\kappa,\Lambda}(\smchemicalpotential)}
\cdot
\opfockcr(f)
\cdot
\opfockan(g).$$
Applying Proposition \ref{expedition0012088} to the Hamiltonian and Fact \ref{expedition0010960} to the creation/annihilation operators,
\begin{equation}
\begin{aligned}
&\inv{V_{\kappa,\Lambda}}
W
V_{\kappa,\Lambda}
=
\inv{V_{\kappa,\Lambda}}
\napiernum^{-\sminvtemperature \physham_{\txtvanhowe,I_{L}^{d},\kappa,\Lambda}(\smchemicalpotential)}
V_{\kappa,\Lambda}
\cdot
\inv{V_{\kappa,\Lambda}}
\opfockcr(f)
V_{\kappa,\Lambda}
\cdot
\inv{V_{\kappa,\Lambda}}
\opfockan(g)
V_{\kappa,\Lambda}
\\ 
&=
\napiernum^{-\sminvtemperature \physgse}
\napiernum^{-\sminvtemperature \physham_{\txtbsn,\txtfr,I_{L}^{d}}(\smchemicalpotential)}
\rbk{\opfockcr(f) + \frac{\imunit}{\sqrt{2}} \bkt{-\imunit \mathsf{m}_{\kappa,\Lambda}}{f}}
\rbk{\opfockan(g) - \frac{\imunit}{\sqrt{2}} \bkt{g}{-\imunit \mathsf{m}_{\kappa,\Lambda}}}
\\ 
&=
\napiernum^{-\sminvtemperature \physgse}
\napiernum^{-\sminvtemperature \physham_{\txtbsn,\txtfr,I_{L}^{d}}(\smchemicalpotential)}
\rbkleft{\opfockcr(f) \opfockan(g)
-\frac{\imunit}{\sqrt{2}} \opfockcr(f) \bkt{g}{-\imunit \mathsf{m}_{\kappa,\Lambda}}}
\\ 
&\qquad
+\rbkright{\frac{\imunit}{\sqrt{2}} \bkt{-\imunit \mathsf{m}_{\kappa,\Lambda}}{f} \opfockan(g)
+\onehalf \bkt{-\imunit \mathsf{m}_{\kappa,\Lambda}}{f} \bkt{g}{-\imunit \mathsf{m}_{\kappa,\Lambda}}}
\end{aligned}
\end{equation}
holds.

By the cyclicity of trace we have $\sqfun{\trace}{W}
= \sqfun{\trace}{\inv{V_{\kappa,\Lambda}}
W
V_{\kappa,\Lambda}}$.
The above calculation ultimately allows us to use the evaluation for the ideal Bose gas.
Applying Proposition \ref{expedition0011296} for the one-point functions and the ideal Bose gas result, the one-point functions for the creation and annihilation operators vanish, and
\begin{equation}
\begin{aligned}
&\oastate[\psi_{\txtgrandcanonical,\txtvanhowe,I_{L}^{d},\kappa,\Lambda,\sminvtemperature,\smchemicalpotential}](\opfockcr(f) \opfockan(g))
=
\napiernum^{\sminvtemperature \physgse}
\frac{1}{\smgrandpartitionfunc_{\txtbsn,\txtfr,I_{L}^{d},\sminvtemperature,\smchemicalpotential}}
\sqfun{\trace}{W} \\
&=
\bkt{g}
{\frac{\napiernum^{-\sminvtemperature (\omega - \smchemicalpotential)}}
{1 - \napiernum^{-\sminvtemperature (\omega - \smchemicalpotential)}}
f}
+\onehalf
\bkt{\mathsf{m}_{\kappa,\Lambda}}{f}
\bkt{g}{\mathsf{m}_{\kappa,\Lambda}}
\end{aligned}
\end{equation}
is obtained.

($\opfocksegal_{\txtfock}(f) \opfocksegal_{\txtfock}(g)$): The computation proceeds in the same way as for the creation/annihilation operator pair.
Set the operator to be evaluated as $$W
=
\napiernum^{-\sminvtemperature \physham_{\txtvanhowe,I_{L}^{d},\kappa,\Lambda}(\smchemicalpotential)}
\cdot
\opfocksegal_{\txtfock}(f)
\cdot
\opfocksegal_{\txtfock}(g).$$
By the result on the commutator of the Hamiltonian and the Segal field operator,
\begin{equation}
\begin{aligned}
&\inv{V_{\kappa,\Lambda}}
W
V_{\kappa,\Lambda}
=
\inv{V_{\kappa,\Lambda}}
\napiernum^{-\sminvtemperature \physham_{\txtvanhowe,I_{L}^{d},\kappa,\Lambda}(\smchemicalpotential)}
V_{\kappa,\Lambda}
\cdot
\inv{V_{\kappa,\Lambda}}
\opfocksegal_{\txtfock}(f)
V_{\kappa,\Lambda}
\cdot
\inv{V_{\kappa,\Lambda}}
\opfocksegal_{\txtfock}(g)
V_{\kappa,\Lambda} \\
&=
\napiernum^{-\sminvtemperature \physgse}
\napiernum^{-\sminvtemperature \physham_{\txtbsn,\txtfr,I_{L}^{d}}(\smchemicalpotential)}
\rbk{\opfocksegal_{\txtfock}(f) - \opimag \bkt{-\imunit \mathsf{m}_{\kappa,\Lambda}}{f}}
\rbk{\opfocksegal_{\txtfock}(g) - \opimag \bkt{-\imunit \mathsf{m}_{\kappa,\Lambda}}{g}}
\\ 
&=
\napiernum^{-\sminvtemperature \physgse}
\napiernum^{-\sminvtemperature \physham_{\txtbsn,\txtfr,I_{L}^{d}}(\smchemicalpotential)}
\rbk{\opfocksegal_{\txtfock}(f) - \opreal \bkt{\mathsf{m}_{\kappa,\Lambda}}{f}}
\rbk{\opfocksegal_{\txtfock}(g) - \opreal \bkt{\mathsf{m}_{\kappa,\Lambda}}{g}} \\
\end{aligned}
\end{equation}
holds.

By the cyclicity of trace we have $\sqfun{\trace}{W}
= \sqfun{\trace}{\inv{V_{\kappa,\Lambda}}
W
V_{\kappa,\Lambda}}$.
As expectation values for the ideal Bose gas, the one-point functions vanish.
Using the evaluation of the two-point functions for the ideal Bose gas,
\begin{equation}
\begin{aligned}
&\oastate[\psi_{\txtgrandcanonical,\txtvanhowe,I_{L}^{d},\kappa,\Lambda,\sminvtemperature,\smchemicalpotential}]
(\opfocksegal_{\txtfock}(f) \opfocksegal_{\txtfock}(g))
=
\napiernum^{\sminvtemperature \physgse}
\frac{1}{\smgrandpartitionfunc_{\txtbsn,\txtfr,I_{L}^{d},\sminvtemperature,\smchemicalpotential}}
\sqfun{\trace}{W} \\
&=
\onehalf
\rbk{\opreal \opform{q}_{\txtnonzero,\sminvtemperature,\smchemicalpotential}(f,g)
+\imunit \opimag \bkt{f}{g}}
+\opreal \bkt{\mathsf{m}_{\kappa,\Lambda}}{f}
\cdot
\opreal \bkt{\mathsf{m}_{\kappa,\Lambda}}{g}
\end{aligned}
\end{equation}
is obtained.
\end{proof}

From the above argument, we obtain the grand canonical mean for the Weyl operator.

\begin{prop}\label{expedition0011578}
For any $f
\in \fun{\lp^{2}}{I_{L}^{d}}$,
\begin{equation}
\begin{aligned}
\oastate[\psi_{\txtgrandcanonical,\txtvanhowe,I_{L}^{d},\kappa,\Lambda,\sminvtemperature,\smchemicalpotential}]
(\opfockweyl_{\txtfock}(f))
=
\fnexp{-\imunit
\opreal
\bkt{\mathsf{m}_{\kappa,\Lambda}}{f}
-\frac{1}{4}
\opform{q}_{\txtnonzero,\sminvtemperature,\smchemicalpotential}(f)}
\end{aligned}
\end{equation}
holds.
\end{prop}

\begin{proof}
The evaluation can be carried out by the same method as in Proposition \ref{expedition0011299}.
Let $\physgse$ denote the ground state energy of the van Hove Hamiltonian and use the unitary operator $V_{\kappa,\Lambda}
= \opfockweyl_{\txtfock}(\imunit \mathsf{m}_{\kappa,\Lambda})$.

Let $W
=
\opfockweyl_{\txtfock}(f)
\napiernum^{-\sminvtemperature \physham_{\txtvanhowe,I_{L}^{d},\kappa,\Lambda}(\smchemicalpotential)}$ be the operator to be evaluated from the trace perspective.
By Fact \ref{expedition0010960},
\begin{equation}
\begin{aligned}
&\inv{V_{\kappa}}
W
V_{\kappa}
=
\inv{V_{\kappa}}
\opfockweyl_{\txtfock}(f)
\napiernum^{-\sminvtemperature \physham_{\txtvanhowe,I_{L}^{d},\kappa,\Lambda}(\smchemicalpotential)}
V_{\kappa}
\\ 
&=
\napiernum^{-\imunit
\opimag
\bkt{-\imunit \mathsf{m}_{\kappa,\Lambda}}{f}}
\opfockweyl_{\txtfock}(f)
\cdot
\inv{V_{\kappa}}
\napiernum^{-\sminvtemperature \physham_{\txtvanhowe,I_{L}^{d},\kappa,\Lambda}(\smchemicalpotential)}
V_{\kappa}
\\ 
&=
\napiernum^{-\imunit
\opreal
\bkt{\mathsf{m}_{\kappa,\Lambda}}{f}}
\opfockweyl_{\txtfock}(f)
\cdot
\napiernum^{-\sminvtemperature \physgse}
\napiernum^{-\sminvtemperature \physham_{\txtbsn,\txtfr,I_{L}^{d}}(\smchemicalpotential)}
\\ 
&=
\napiernum^{-\imunit
\opreal
\bkt{\mathsf{m}_{\kappa,\Lambda}}{f}}
\napiernum^{-\sminvtemperature \physgse}
\opfockweyl_{\txtfock}(f)
\napiernum^{-\sminvtemperature \physham_{\txtbsn,\txtfr,I_{L}^{d}}(\smchemicalpotential)}
\end{aligned}
\end{equation}
is obtained.

By the cyclicity of trace,
\begin{equation}
\begin{aligned}
&\sqfun{\trace}
{W}
=
\sqfun{\trace}
{\inv{V_{\kappa,\Lambda}}
W
V_{\kappa,\Lambda}}
\\ 
&=
\napiernum^{-\imunit
\opreal \bkt{\mathsf{m}_{\kappa,\Lambda}}{f}}
\napiernum^{-\sminvtemperature \physgse}
\sqfun{\trace}
{\opfockweyl_{\txtfock}(f)
\napiernum^{-\sminvtemperature \physham_{\txtbsn,\txtfr,I_{L}^{d}}(\smchemicalpotential)}}
\\ 
&=
\napiernum^{-\imunit
\opreal
\bkt{\mathsf{m}_{\kappa,\Lambda}}{f}}
\smgrandpartitionfunc_{\txtvanhowe,I_{L}^{d},\kappa,\Lambda,\sminvtemperature,\smchemicalpotential}
\cdot
\psi_{\txtgrandcanonical,\sminvtemperature,\smchemicalpotential}(\opfockweyl_{\txtfock}(f))
\end{aligned}
\end{equation}
is obtained.
Collecting the above argument,
\begin{equation}
\begin{aligned}
&\oastate[\psi_{\txtgrandcanonical,\txtvanhowe,I_{L}^{d},\kappa,\Lambda,\sminvtemperature,\smchemicalpotential}](\opfockweyl_{\txtfock}(f))
\\ 
&=
\napiernum^{-\imunit
\opreal
\bkt{\mathsf{m}_{\kappa,\Lambda}}{f}}
\fnexp{-\frac{1}{4}
\opform{q}_{\txtnonzero,\sminvtemperature,\smchemicalpotential}(f)}
\end{aligned}
\end{equation}
is obtained.
\end{proof}

We also prepare the two-point function for the Weyl operator.

\begin{prop}\label{expedition0011592}
For any $f,g
\in \sphilb{H}$,
\begin{equation}
\begin{aligned}
&\oastate[\psi_{\txtgrandcanonical,\txtvanhowe,I_{L}^{d},\kappa,\Lambda,\sminvtemperature,\smchemicalpotential}]
(\opfockweyl_{\txtfock}(f)
\opfockweyl_{\txtfock}(g))
\\ 
&=
\fnexp{-\frac{\imunit}{2}
\opimag
\bkt{f}{g}
-\imunit
\opreal
\bkt{\mathsf{m}_{\kappa,\Lambda}}{f+g}
-\frac{1}{4}
\opform{q}_{\txtnonzero,\sminvtemperature,\smchemicalpotential}(f+g)}
\end{aligned}
\end{equation}
holds.
\end{prop}

\begin{proof}
By the Weyl relation, $$\opfockweyl_{\txtfock}(f)
\opfockweyl_{\txtfock}(g)
=
\napiernum^{-\frac{\imunit}{2}
\opimag
\bkt{f}{g}}
\opfockweyl_{\txtfock}(f+g)$$ holds.
Applying Proposition \ref{expedition0011578},
\begin{equation}
\begin{aligned}
&\oastate[\psi_{\txtgrandcanonical,\txtvanhowe,I_{L}^{d},\kappa,\Lambda,\sminvtemperature,\smchemicalpotential}]
(\opfockweyl_{\txtfock}(f)
\opfockweyl_{\txtfock}(g))
\\ 
&=
\napiernum^{-\frac{\imunit}{2}
\opimag
\bkt{f}{g}}
\oastate[\psi_{\txtgrandcanonical,\txtvanhowe,I_{L}^{d},\kappa,\Lambda,\sminvtemperature,\smchemicalpotential}]
(\opfockweyl_{\txtfock}(f+g))
\\ 
&=
\napiernum^{-\frac{\imunit}{2}
\opimag
\bkt{f}{g}}
\napiernum^{-\imunit
\opreal
\bkt{\mathsf{m}_{\kappa,\Lambda}}{f+g}}
\napiernum^{-\frac{1}{4}
\opform{q}_{\txtnonzero,\sminvtemperature,\smchemicalpotential}(f+g)}
\end{aligned}
\end{equation}
is obtained.
\end{proof}

The following proposition shows that the grand canonical state for the van Hove model separates clearly into a term requiring infrared/ultraviolet divergence treatment and the grand canonical state for the free field.

\begin{prop}\label{expedition0011594}
Let $\oastate[\psi_{\txtgrandcanonical,\txtfr,I_{L}^{d},\sminvtemperature,\smchemicalpotential}]$ denote the grand canonical state for the free field.
For any $f
\in \sphilb{H}$, the grand canonical state $\oastate[\psi_{\txtgrandcanonical,\txtvanhowe,I_{L}^{d},\kappa,\Lambda,\sminvtemperature,\smchemicalpotential}]$ for the van Hove model satisfies
\begin{equation}
\begin{aligned}
&\oastate[\psi_{\txtgrandcanonical,\txtvanhowe,I_{L}^{d},\kappa,\Lambda,\sminvtemperature,\smchemicalpotential}]
(\opfockweyl_{\txtfock}(f))
=
\napiernum^{-\imunit
\opreal
\bkt{\mathsf{m}_{\kappa,\Lambda}}{f}}
\fun{\oastate[\psi_{\txtgrandcanonical,\txtfr,I_{L}^{d},\sminvtemperature,\smchemicalpotential}]}
{\opfockweyl_{\txtfock}(f)}
\end{aligned}
\end{equation}
\end{prop}

\begin{proof}
By Proposition \ref{expedition0011578} and the result for the ideal Bose gas.
\end{proof}

By an argument analogous to Proposition \ref{expedition0011094}, \(\alpha_{\txtvanhowe,I_{L}^{d},\kappa,\Lambda}\) is also a one-parameter group of automorphisms.

We also prepare the expectation value of the resolvent for use in the resolvent algebra discussion.

\begin{prop}\label{expedition0011613}
The grand canonical mean for the resolvent $\oaresolvent_{\txtfock}(\lambda,f)
= \opfnresolvent{\imunit \lambda - \opfocksegal_{\txtfock}(f)}$ is
\begin{equation}
\begin{aligned}
&\fun{\oastate[\psi_{\txtgrandcanonical,\txtvanhowe,I_{L}^{d},\sminvtemperature,\kappa,\Lambda,\smchemicalpotential}]}
{\oaresolvent_{\txtfock}(\lambda,f)}
\\ 
&=
\imunit
\int_0^{(\sgn \lambda) \infty}
\fnexp{-\rbk{\lambda - \imunit \opreal \mathsf{m}_{\kappa,\Lambda}(f)} t
-\frac{t^2}{4} \opform{q}_{\txtnonzero,\sminvtemperature,\smchemicalpotential}(f)}
\opdmsr{t}
\end{aligned}
\end{equation}
\end{prop}

\begin{proof}
By the Laplace transform, or by the argument in \cite{BuchholzGrundling2}, the resolvent can be expressed in terms of the Weyl operator as
\begin{equation}
\begin{aligned}
&\opfnresolvent{\imunit \lambda - \opfocksegal_{\txtfock}(f)}
=
\imunit
\int_0^{(\sgn \lambda) \infty}
\napiernum^{\imunit t(\imunit \lambda - \opfocksegal_{\txtfock}(f))}
\opdmsr{t}
\\ 
&=
\imunit
\int_0^{(\sgn \lambda) \infty}
\napiernum^{-\lambda t}
\napiernum^{-\imunit t \opfocksegal_{\txtfock}(f)}
\opdmsr{t}
=
\imunit
\int_0^{(\sgn \lambda) \infty}
\napiernum^{-\lambda t}
\opfockweyl_{\txtfock}(-tf)
\opdmsr{t}
\end{aligned}
\end{equation}
Applying Proposition \ref{expedition0011578},
\begin{equation}
\begin{aligned}
&\fun{\oastate[\psi_{\txtgrandcanonical,\txtvanhowe,I_{L}^{d},\sminvtemperature,\kappa,\Lambda,\smchemicalpotential}]}
{\oaresolvent_{\txtfock}(\lambda,f)}
=
\fun{\oastate[\psi_{\txtgrandcanonical,\txtvanhowe,I_{L}^{d},\sminvtemperature,\kappa,\Lambda,\smchemicalpotential}]}
{\opfnresolvent{\imunit \lambda - \opfocksegal_{\txtfock}(f)}}
\\ 
&=
\imunit
\int_0^{(\sgn \lambda) \infty}
\napiernum^{-\lambda t}
\fun{\oastate[\psi_{\txtgrandcanonical,\txtvanhowe,I_{L}^{d},\sminvtemperature,\kappa,\Lambda,\smchemicalpotential}]}
{\opfockweyl_{\txtfock}(-tf)}
\opdmsr{t}
\\ 
&=
\imunit
\int_0^{(\sgn \lambda) \infty}
\napiernum^{-\lambda t}
\napiernum^{-\imunit
\opreal
\bkt{\mathsf{m}_{\kappa,\Lambda}}{-tf}}
\fnexp{-\frac{t^2}{4}
\bkt{f}
{K_{\sminvtemperature,\smchemicalpotential}
f}}
\opdmsr{t}
\\ 
&=
\imunit
\int_0^{(\sgn \lambda) \infty}
\fnexp{-\rbk{\lambda - \imunit \opreal \mathsf{m}_{\kappa,\Lambda}(f)} t
-\frac{t^2}{4} \opform{q}_{\txtnonzero,\sminvtemperature,\smchemicalpotential}(f)}
\opdmsr{t}
\end{aligned}
\end{equation}
is obtained.
\end{proof}

Let us also prepare the two-point function of the resolvent.

\begin{prop}\label{expedition0011614}
For the resolvent $\oaresolvent_{\txtfock}(\lambda,f)
= \opfnresolvent{\imunit \lambda - \opfocksegal_{\txtfock}(f)}$,
\begin{equation}
\begin{aligned}
&\fun{\oastate[\psi_{\txtgrandcanonical,\txtvanhowe,I_{L}^{d},\sminvtemperature,\kappa,\Lambda,\smchemicalpotential}]}
{\oaresolvent_{\txtfock}(\lambda,f)
\oaresolvent_{\txtfock}(\mu,g)}
\\ 
&=
\imunit
\int_0^{(\sgn \lambda) \infty}
\int_0^{(\sgn \mu) \infty}
\napiernum^{-\lambda s - \mu t}
\mathsf{T}_{\sminvtemperature,\smchemicalpotential,\kappa,\Lambda}(s,f;t,g)
\opdmsr{s}
\opdmsr{t},
\\ 
&\mathsf{T}_{\sminvtemperature,\smchemicalpotential,\kappa,\Lambda}(s,f;t,g)
\\
&=
\fnexp{-\frac{1}{4}
\opform{q}_{\txtnonzero,\sminvtemperature,\smchemicalpotential}(sf+tg)
+\imunit \opreal \mathsf{m}_{\kappa,\Lambda}(sf+tg)
-\frac{\imunit}{2} st \opimag \bkt{f}{g}}
\end{aligned}
\end{equation}
holds.
\end{prop}

\begin{proof}
Combine the product into a single term using the Weyl relation, then apply Proposition \ref{expedition0011613}.
\end{proof}

\section{Weyl Algebra at Finite Temperature: Infinite System}\label{expedition0012092}

By Proposition \ref{expedition0011594}, Bose-Einstein condensation can essentially be handled within the free field discussion, so it suffices to consider only the case \(\smchemicalpotential
= 0\) where BEC occurs.

\subsection{Expectation Value Estimates}\label{expectation-value-estimates}

We first argue with infrared/ultraviolet cutoffs, then discuss the situation below the critical temperature in the limit where the cutoffs are removed.

\begin{thm}\label{expedition0011612}
The expectation value of the Weyl operator in the $\sminvtemperature$-KMS state $\oastate[\psi_{\txtvanhowe,\sminvtemperature,\kappa,\Lambda}]$ for the van Hove model with infrared/ultraviolet cutoffs can be written as $$\oastate[\psi_{\txtvanhowe,\sminvtemperature,\kappa,\Lambda}]
(\opfockweyl_{\txtfock}(f))
=
\fnexp{-\imunit
\opreal
\bkt{\mathsf{m}_{\kappa,\Lambda}}{f}
-\oneoverfour
\opform{q}_{\txtnonzero,\sminvtemperature}(f)
-\oneoverfour
\opform{q}_{0}(f)}$$
\end{thm}

\begin{rem}
The difference from the ideal Bose gas is precisely the term $\napiernum^{-\imunit
\opreal
\bkt{\mathsf{m}_{\kappa,\Lambda}}{f}}$ arising from the addition of the field interaction.
\end{rem}

\begin{proof}
Apply the infinite-volume limit discussed in \cite{YoshitsuguSekine004,AsaoArai28} to Proposition \ref{expedition0011594}.
\end{proof}

Next, we remove the infrared/ultraviolet cutoffs. We note one point regarding how to take the limit. Although we have been working on Fock space up to this point, from the representation-theoretic viewpoint we should regard the discussion as taking place in the Fock representation of the abstract Weyl algebra \(\oaweyl(\sphilb{D}_{\txtirsingular,\sminvtemperature})\). In the discussion of removing the cutoffs, the state \(\oastate[\psi_{\txtvanhowe,\sminvtemperature,\kappa,\Lambda}]\) should be regarded as the state \(\oastate[\tilde{\psi}_{\txtvanhowe,\sminvtemperature,\kappa,\Lambda}]
\circ \oarepn_{\txtfock}\) on \(\oaweyl(\sphilb{D}_{\txtirsingular,\sminvtemperature})\) composed with the Fock representation \(\oarepn_{\txtfock}\), and the limit should be taken for \(\oastate[\tilde{\psi}_{\txtvanhowe,\sminvtemperature,\kappa,\Lambda}]\) in the state space over \(\oaweyl(\sphilb{D}_{\txtirsingular,\sminvtemperature})\). Here, for brevity of notation, we write \(\oastate[\tilde{\psi}_{\txtvanhowe,\sminvtemperature,\kappa,\Lambda}]\) simply as \(\oastate[\psi_{\txtvanhowe,\sminvtemperature,\kappa,\Lambda}]\).

\begin{thm}\label{expedition0011637}
\begin{enumerate}
\item
Convergence of the automorphism group: For any $A
\in \oaweyl_{\txtirsingular,\sminvtemperature}$ and $t
\in \fldreal$, $$\alpha_{\txtvanhowe,t}(A)
=
\lim_{\kappa \to 0, \Lambda \to \infty}
\alpha_{\txtvanhowe,\kappa,\Lambda,t}(A)$$ exists in the norm topology, and the family $\fml{\alpha_{\txtvanhowe,t}}{t \in \fldreal}$ defines a one-parameter group of automorphisms on $\oaweyl_{\txtirsingular,\sminvtemperature}$.

\item
Convergence of KMS states: For any $A
\in \oaweyl_{\txtirsingular,\sminvtemperature}$, there exists a state $\oastate[\psi_{\txtvanhowe,\sminvtemperature}]$ on $\oaweyl_{\txtirsingular,\sminvtemperature}$ satisfying $$\oastate[\psi_{\txtvanhowe,\sminvtemperature}](A)
=
\lim_{\kappa \to 0,\Lambda \to \infty}
\oastate[\psi_{\txtvanhowe,\sminvtemperature,\kappa,\Lambda}](A).$$
In particular, for the Weyl operator, $$\oastate[\psi_{\txtvanhowe,\sminvtemperature}]
(\opfockweyl(f))
=
\lim_{\kappa \to 0,\Lambda \to \infty}
\oastate[\psi_{\txtvanhowe,\sminvtemperature,\kappa,\Lambda}]
(\opfockweyl(f)),
\quad
f \in \sphilb{D}_{\txtirsingular,\sminvtemperature}$$ holds, and the right-hand side is obtained as the natural limit of the expression in Theorem \ref{expedition0011612}.
In particular, the expectation value for the Weyl operator is $$\oastate[\psi_{\txtvanhowe,\sminvtemperature}]
(\opfockweyl(f))
=
\fnexp{-\imunit
\opreal
\bkt{\mathsf{m}}{f}
-\oneoverfour
\opform{q}_{\txtnonzero,\sminvtemperature}(f)
-\oneoverfour
\opform{q}_{0}(f)}$$

\item
The state $\oastate[\psi_{\txtvanhowe,\sminvtemperature}]$ is a $\sminvtemperature$-KMS state with respect to the automorphism group $\alpha_{\txtvanhowe}$.
\end{enumerate}
\end{thm}

\begin{proof}
For the limit removing the infrared/ultraviolet cutoffs, it suffices to consider the limit $\kappa, \Lambda$ of the relevant objects.
The question then reduces to showing that these objects converge to appropriate limits.

(Convergence of the automorphism group): By Proposition \ref{expedition0011094} and the remark immediately following it, the limiting automorphism map is defined for each $t$.
The removal limit only concerns $\mathsf{m}_{\kappa,\Lambda}$, and convergence in norm is clear.

(Convergence of the KMS state): The cutoffs act only on the state functional.

It remains to discuss convergence on $\opfockweyl(f)$ for $f
\in \sphilb{D}_{\txtirsingular,\sminvtemperature}$, which is clear.
The KMS property of the limit follows from the fact that the defining relations from the bounded system are algebraic and are preserved in the limit.
\end{proof}

\subsection{Bose-Einstein Condensation}\label{bose-einstein-condensation}

By Theorem \ref{expedition0011637}, the sesquilinear form \(\opform{q}_{0}\) appears in the expectation value. This is the ingredient describing BEC for the free field. BEC for the free field is discussed in the textbook \cite{AsaoArai28} or in \cite{YoshitsuguSekine004}, and an no-go theorem for BEC of quasiparticles is discussed in \cite{YoshitsuguSekine006}. We do not discuss these further here.

\section{Resolvent Algebra at Finite Temperature: Equilibrium State}\label{expedition0012002}

The resolvent algebra is defined in Section \ref{expedition0012083}. We begin here by defining the van Hove model on the resolvent algebra. We then appropriately transplant the bounded-system finite-temperature setting from Section \ref{expedition0011587}, which discussed the Weyl algebra case, and investigate the ideal structure of the resolvent algebra following the approach of \cite{DetlevBuchholz001,YoshitsuguSekine004,YoshitsuguSekine005}. Each object can be defined based on the Weyl algebra discussion.

For any \(t
\in \fldreal\), define the automorphisms of the abstract resolvent algebra as \begin{equation}
\begin{aligned}
\alpha_{1,\kappa,\Lambda,t}(\oaresolvent(z,f))
&=
\fun{\oaresolvent}{z + \imunit \mathsf{M}_{\kappa,\Lambda,t}(f), f}, \\
\alpha_{2,t}(\oaresolvent(z,f))
&=
\fun{\oaresolvent}{z, \napiernum^{\imunit t \omega} f}, \\
\alpha_{\txtvanhowe,\kappa,\Lambda,t}(\oaresolvent(z,f))
&=
\funrbk{\alpha_{1,\kappa,\Lambda,t} \circ \alpha_{2,t}}{\oaresolvent(z,f)} \\
&=
\fun{\oaresolvent}{z + \imunit \mathsf{M}_{\kappa,\Lambda,t}(f), \napiernum^{\imunit t \omega} f}
\end{aligned}
\end{equation} and call this the automorphism group of the van Hove model with infrared/ultraviolet cutoffs, or simply the automorphism group of the cutoff van Hove model or the van Hove model.

The cutoff-free versions \(\alpha_{\txtvanhowe,t}\) and \(\alpha_{1,t}\), arising from the domain of the functional \(\mathsf{m}\), are defined as automorphisms of the abstract resolvent algebra \(\oaresolventalgebra(\dom \mathsf{m}, \sigma)\) by \begin{equation}
\begin{aligned}
\alpha_{1,t}(\oaresolvent(z,f))
&=
\fun{\oaresolvent}{z + \imunit \mathsf{M}_{t}(f), f}, \\
\alpha_{\txtvanhowe,t}(\oaresolvent(z,f))
&=
\funrbk{\alpha_{1,t} \circ \alpha_{2,t}}{\oaresolvent(z,f)} \\
&=
\fun{\oaresolvent}{z + \imunit \mathsf{M}_{t}(f), \napiernum^{\imunit t \omega} f}
\end{aligned}
\end{equation} and this is called the automorphism group of the van Hove model, or the cutoff-free automorphism group of the van Hove model.

\begin{prop}\label{expedition0012084}
The above $\alpha_{\txtvanhowe,\kappa,\Lambda}$ is a group of automorphisms.
In particular, restricting the domain to $\oaresolventalgebra(\dom \mathsf{m},\sigma)$, $\alpha_{\txtvanhowe,\kappa,\Lambda}$ remains a group of automorphisms on this algebra.
With an appropriate restriction of the domain, $\alpha_{\txtvanhowe}$ is a group of automorphisms on the restricted subalgebra.
\end{prop}

\begin{proof}
Since the domain restriction does not affect the computation itself, we argue using notation with cutoffs.
By Proposition \ref{expedition0011238}, the auxiliary functional $\mathsf{M}_{\kappa,\Lambda,t}$ is a cocycle.

(Automorphism property): Under the assumption that the first variable of the resolvent function does not become $0$, verify the preservation of the resolvent relations by direct computation.

(Restricted automorphism property): Choose any $f
\in \dom \mathsf{m}$.
Then $$\alpha_{\txtvanhowe,\kappa,\Lambda,t}(\oaresolvent(z,f))
=
\oaresolvent(z + \imunit \mathsf{M}_{\kappa,\Lambda}(f), \napiernum^{\imunit t \omega} f)$$ is well-defined.
In particular, $\napiernum^{\imunit t \omega} f$ satisfies $\abs{\napiernum^{\imunit t \omega} f}
= \abs{f}$, so integrability is preserved and $\dom \mathsf{m}$ is invariant.
Therefore $\fnrestr{\alpha_{\txtvanhowe,\kappa,\Lambda}}{\oaresolventalgebra(\dom \mathsf{m},\sigma)}$ is a group of automorphisms of $\oaresolventalgebra(\dom \mathsf{m},\sigma)$.

(Group property): Use the cocycle property of the auxiliary functional $\mathsf{M}_{\kappa,\Lambda,t}$ and compute directly.
\end{proof}

Based on Proposition \ref{expedition0011614}, define the state \(\oastate[\psi_{\txtvanhowe,I_{L}^{d},\sminvtemperature,\smchemicalpotential}]\) as the quasi-free state satisfying, in the bounded system, the two-point function \begin{equation}
\begin{aligned}
&\fun{\oastate[\psi_{\txtvanhowe,I_{L}^{d},\sminvtemperature,\smchemicalpotential}]}
{\oaresolvent(\lambda,f)
\oaresolvent(\mu,g)}
\\ 
&=
\imunit
\int_0^{(\sgn \lambda) \infty}
\int_0^{(\sgn \mu) \infty}
\napiernum^{-\lambda s - \mu t}
\mathsf{T}_{\sminvtemperature,\smchemicalpotential,\kappa,\Lambda}(s,f;t,g)
\opdmsr{s}
\opdmsr{t},
\\ 
&\mathsf{T}_{\sminvtemperature,\smchemicalpotential,\kappa,\Lambda}(s,f;t,g)
\\
&=
\fnexp{-\frac{\opform{q}_{\txtnonzero,\sminvtemperature,\smchemicalpotential}(sf+tg)}{4}
+\imunit \opreal \mathsf{m}_{\kappa,\Lambda}(sf+tg)
-\frac{\imunit}{2} st \opimag \bkt{f}{g}}
\end{aligned}
\end{equation}

\begin{prop}\label{expedition0011634}
The state $\oastate[\psi_{\txtvanhowe,I_{L}^{d},\sminvtemperature,\smchemicalpotential}]$ is a $\sminvtemperature$-KMS state with respect to the automorphism group $\alpha_{\txtvanhowe,I_{L}^{d},\kappa,\Lambda,\smchemicalpotential}$ of Definition \ref{expedition0011278}.
\end{prop}

\begin{proof}
The function $\mathsf{T}_{\sminvtemperature,\smchemicalpotential,\kappa,\Lambda}(s,f;t,g)$ is the $\sminvtemperature$-KMS state in the Weyl algebra, determined in particular by the trace.
The result follows from the KMS property for the trace and the KMS condition on a dense $\ast$-algebra \cite{BratteliRobinson2}.
\end{proof}

\begin{prop}
Set the chemical potential to $\smchemicalpotential
= 0$.
Then the infinite-volume limit $\oastate[\psi_{\txtvanhowe,\sminvtemperature,\smchemicalpotential}]$ as a state on $\oaresolventalgebra_{\txtirsingular,\sminvtemperature}
=
\oaresolventalgebra(\sphilb{D}_{\txtirsingular,\sminvtemperature})$ exists, along with the automorphism group $\alpha_{\kappa,\Lambda,\smchemicalpotential}$ of Definition \ref{expedition0011278}.
The infinite-volume limit is also a $\sminvtemperature$-KMS state with respect to the automorphism group of the infinite system.
\end{prop}

\begin{proof}
By an argument analogous to Proposition \ref{expedition0011634}.
\end{proof}

Let us now summarize the situation in which BEC can occur.

\begin{cor}\label{expedition0012086}
The $\sminvtemperature$-KMS state $\oastate[\psi_{\txtvanhowe,\sminvtemperature}]$ obtained in the infinite-volume limit is a quasi-free state, with two-point function
\begin{equation}
\begin{aligned}
\fun{\oastate[\psi_{\txtvanhowe,\sminvtemperature}]}
{\oaresolvent(\lambda,f)
\oaresolvent(\mu,g)}
&=
\imunit
\int_0^{(\sgn \lambda) \infty}
\int_0^{(\sgn \mu) \infty}
\napiernum^{-\lambda s - \mu t}
\mathsf{T}_{\sminvtemperature,\kappa,\Lambda}(s,f;t,g)
\opdmsr{s}
\opdmsr{t},
\\ 
\mathsf{T}_{\sminvtemperature,\kappa,\Lambda}(s,f;t,g)
&=
\fnexp{\imunit \opreal \mathsf{m}_{\kappa,\Lambda}(sf+tg)
-\frac{\imunit}{2} st \opimag \bkt{f}{g}}
\\
&\qquad\times
\fnexp{-\frac{\opform{q}_{\txtnonzero,\sminvtemperature}(sf+tg)}{4}
-\frac{\opform{q}_0(sf+tg)}{4}}
\end{aligned}
\end{equation}
\end{cor}

\begin{proof}
By the argument for the Weyl algebra, in particular Theorem \ref{expedition0011637}.
\end{proof}

\begin{rem}
The integrand $\mathsf{T}_{\sminvtemperature,\kappa,\Lambda}$ is essentially the expectation value in the Weyl algebra.
\end{rem}

As with the Weyl algebra, we remove the infrared/ultraviolet cutoffs.

\begin{thm}
\begin{enumerate}
\item
Convergence of the automorphism group: For any $A
\in \oaresolventalgebra_{\txtirsingular,\sminvtemperature}$ and $t
\in \fldreal$, $$\alpha_{\txtvanhowe,t}(A)
=
\lim_{\kappa \to 0, \Lambda \to \infty}
\alpha_{\txtvanhowe,\kappa,\Lambda,t}(A)$$ exists in the strong topology of $\oaresolventalgebra_{\txtirsingular,\sminvtemperature}$, and the family $\fml{\alpha_{\txtvanhowe,t}}{t \in \fldreal}$ defines a strongly continuous one-parameter group of automorphisms on $\oaresolventalgebra_{\txtirsingular,\sminvtemperature}$.

\item
Convergence of KMS states: For any $A
\in \oaresolventalgebra_{\txtirsingular,\sminvtemperature}$, there exists a state $\oastate[\psi_{\txtvanhowe,\sminvtemperature,\kappa,\Lambda}]$ on $\oaresolventalgebra_{\txtirsingular,\sminvtemperature}$ satisfying $$\oastate[\psi_{\txtvanhowe,\sminvtemperature}](A)
=
\lim_{\kappa \to 0,\Lambda \to \infty}
\oastate[\psi_{\txtvanhowe,\sminvtemperature,\kappa,\Lambda}](A).$$
In particular, for the generators of the resolvent algebra, $$\oastate[\psi_{\txtvanhowe,\sminvtemperature}]
(\oaresolvent(\lambda,f))
=
\lim_{\kappa \to 0,\Lambda \to \infty}
\oastate[\psi_{\txtvanhowe,\sminvtemperature,\kappa,\Lambda}]
(\oaresolvent(\lambda,f))
\quad
f \in \sphilb{D}_{\txtirsingular,\sminvtemperature}$$ holds, and the right-hand side is obtained as the natural limit of the expression in Theorem \ref{expedition0011612}.

\item
The state $\oastate[\psi_{\txtvanhowe,\sminvtemperature}]$ is a $\sminvtemperature$-KMS state with respect to the automorphism group $\alpha_{\txtvanhowe}$.
\end{enumerate}
\end{thm}

\begin{proof}
Proceed in the same manner as Theorem \ref{expedition0011637}.
\end{proof}

\begin{rem}[Impossibility of BEC and Ideal Structure]
As noted in the Weyl algebra discussion, BEC for the free field is discussed in the textbook \cite{AsaoArai28} or in \cite{YoshitsuguSekine004}, and an no-go theorem for BEC of quasiparticles is discussed in \cite{YoshitsuguSekine006}.
\end{rem}

\section{Preparation for Functional Integral Theory}\label{preparation-for-functional-integral-theory}

In the following, we discuss the ground states and equilibrium states treated in the Weyl and resolvent algebra framework via the functional integral formulation.

\subsection{Basic Setup and Adjustments for Functional Integrals}\label{expedition0011385}

We mainly follow the notation and conventions of \cite{AsaoArai30,LorincziHiroshimaBetz3,DerezinskiGerard001}. This paper adopts a notation slightly different from these references: see Subsection \ref{expedition0012096} for details.

Let \(\sphilb{H}\) be a real Hilbert space and \((M, \Sigma, \mu)\) a probability space.

When, for each \(f
\in \sphilb{H}\), a random variable \(\phi(f)\) on \((M, \Sigma, \mu)\) is given satisfying the following properties, the family \(\set{\phi(f)}
{f \in \sphilb{H}}\) is called a Gaussian stochastic process indexed by \(\sphilb{H}\), or simply a Gaussian hyperprocess.

\begin{enumerate}
\def\labelenumi{\arabic{enumi}.}

\item
  (Real linearity) For any \(a,b
  \in \fldreal\) and \(f,g
  \in \sphilb{H}\), \[\phi(af+bg)
  = a \phi(f) + b \phi(g)\] holds almost surely.
\item
  The family of random variables \(\set{\phi(f)}{f \in \sphilb{H}}\) is total.
\item
  For each \(f
  \in \sphilb{H}\), \(\phi(f)\) is a Gaussian random variable satisfying \[\int_M \napiernum^{i \phi(f)} \opdmsr{\mu}
  =
  \fnexp{- \frac{1}{4} \opform{q}_{\txtnonzero,\infty}(f)}\]
\end{enumerate}

In particular, a Gaussian probability space over a real Hilbert space \(\sphilb{K}\) with the covariance (3) above is called a Gaussian \(\mathbf{L}^2\)-space \cite{DerezinskiGerard001}.

For any \(\alpha
\in \fldreal\) and \(f
\in \sphilb{H}\), define \[\wick{\napiernum^{\alpha \opfocksegal(f)}}_{\mu}
=
\fnexp{-\frac{\alpha^2}{4} \opform{q}_{\txtnonzero,\infty}(f)}
\napiernum^{\alpha \opfocksegal(f)}\] and call \(\wick{\napiernum^{\alpha \opfocksegal(f)}}_{\mu}\) the exponential Wick product, or simply the Wick product. When multiple probability measures appear, a subscript distinguishes which measure is being used. When no confusion arises, we simply write \[\wick{\napiernum^{\alpha \opfocksegal(f)}}\] The exponential Wick product is defined to satisfy the normalization condition \[\sqfun{\prbexp_\mu}{\wick{\napiernum^{\alpha \opfocksegal(f)}}}
=
1.\] The treatment of the exponential Wick product when the mean is nonzero is summarized in Remark \ref{expedition0011410}.

\begin{rem}
Since the Segal field operator is defined with a coefficient $\frac{1}{\sqrt{2}}$ as $$\opfocksegal_{\txtfock}(f)
=
\frac{1}{\sqrt{2}}
\rbk{\opfockcr_{\txtfock}(f) + \opfockan_{\txtfock}(f)},$$
the vacuum expectation value of the Weyl operator is $$\bkt{\opfockvac_{\txtbsn}}
{\opfockweyl_{\txtfock}(f)
\opfockvac_{\txtbsn}}
=
\napiernum^{-\frac{1}{4} \opform{q}_{\txtnonzero,\infty}(f)}.$$
In particular, the mean and covariance are defined as $$\sqfun{\prbexp_{\mu}}{\opfocksegal(f)}
=
0,
\quad
\sqfun{\prbcov_{\msr{\mu}}}{\opfocksegal(f), \opfocksegal(g)}
=
\onehalf \bkt{f}{g}.$$
\end{rem}

\begin{rem}\label{expedition0011410}
If the mean of $\opfocksegal(f)$ is changed to $\mathsf{m}(f)$, so that $$\sqfun{\prbexp_{\mu}}{\opfocksegal(f)}
=
\mathsf{m}(f),
\quad
\sqfun{\prbcov_{\msr{\mu}}}{\opfocksegal(f), \opfocksegal(g)}
=
\onehalf \bkt{f}{g},$$
then the characteristic function changes to $$\prbcharfun_{\sphilb{H}}(f)
=
\sqfun{\prbexp_{\msr{P}_{\sphilb{H}}}}
{\napiernum^{\imunit t \opfocksegal(f)}}
=
\fnexp{\imunit t \mathsf{m}(f)
-\oneoverfour t^2 \norm{f}_{\sphilb{H}}^2}.$$
Even when the mean is nonzero, the normalization condition remains unchanged.
\end{rem}

\begin{prop}\label{expedition0011275}
For any $\alpha,\beta
\in \fldreal$ and $f,g
\in \sphilb{H}$, $$\wick{\napiernum^{\alpha \opfocksegal(f)}} \cdot \wick{\napiernum^{\beta \opfocksegal(g)}}
=
\fnexp{\frac{\alpha \beta}{2} \bkt{f}{g}}
\wick{\napiernum^{\opfocksegal(\alpha f + \beta g)}}$$
\end{prop}

\begin{proof}
Noting the almost-everywhere identity $\napiernum^{\alpha \opfocksegal(f)}
\napiernum^{\beta \opfocksegal(g)}
=
\napiernum^{\alpha \opfocksegal(f) + \beta \opfocksegal(g)}
=
\napiernum^{\opfocksegal(\alpha f + \beta g)}$,
\begin{equation}
\begin{aligned}
&\napiernum^{\alpha \opfocksegal(f)}
\napiernum^{\beta \opfocksegal(g)}
=
\napiernum^{\opfocksegal(\alpha f + \beta g)}
\\ 
&=
\fnexp{\oneoverfour
\norm{\alpha f + \beta g}^2}
\wick{\napiernum^{\opfocksegal(\alpha f + \beta g)}}
\\ 
&=
\fnexp{\frac{\alpha^2}{4} \opform{q}_{\txtnonzero,\infty}(f)
+\frac{\beta^2}{4} \opform{q}_{\txtnonzero,\infty}(g)
+\frac{\alpha \beta}{2} \bkt{f}{g}}
\wick{\napiernum^{\opfocksegal(\alpha f + \beta g)}}
\end{aligned}
\end{equation}
holds, and from this and the definition,
\begin{equation}
\begin{aligned}
&\wick{\napiernum^{\alpha \opfocksegal(f)}} \cdot \wick{\napiernum^{\beta \opfocksegal(g)}}
\\ 
&=
\fnexp{-\frac{\alpha^2}{4} \opform{q}_{\txtnonzero,\infty}(f)
-\frac{\beta^2}{4} \opform{q}_{\txtnonzero,\infty}(g)}
\napiernum^{\alpha \opfocksegal(f) + \beta \opfocksegal(g)}
\\ 
&=
\fnexp{\frac{\alpha \beta}{2} \bkt{f}{g}}
\wick{\napiernum^{\opfocksegal(\alpha f + \beta g)}}
\end{aligned}
\end{equation}
is obtained.
\end{proof}

\subsection{\texorpdfstring{\(Q\)-Space Equivalent to Fock Space}{Q-Space Equivalent to Fock Space}}\label{expedition0012091}

Let \(\pairbk{Q_{\sphilb{H}},
\Sigma_{\sphilb{H}},
\mu_{\sphilb{H}}}\) denote the probability space on which the Gaussian hyperprocess \(\set{\opfocksegal(f)}{f \in \sphilb{H}}\) indexed by the real Hilbert space \(\sphilb{H}\) is realized. There exists a unitary transformation \[U
\colon \fun{\spfock_{\txtbsn}}{\lacomplexication{\sphilb{H}}}
\to \fun{\lp^{2}}{Q_{\sphilb{H}}, \mu_{\sphilb{H}}}.\] Let \(h\) be any non-negative self-adjoint operator on \(\sphilb{H}\), set \(\sphilb{H}_{\alpha}
=
\dom h^{\alpha} \cap \sphilb{H}\) for any real \(\alpha\), and define the random variable \(\fun{\opfocksegal_{\mathrm{F}}}{f}\) on \(Q_{\sphilb{H}}\) for each \(f
\in \sphilb{H}_{-\frac{1}{2}}\) by \begin{align}
\fun{\opfocksegal_{\mathrm{F}}}{f}
=
\fun{\opfocksegal}{h^{-\frac{1}{2}} f} 
\end{align} This satisfies the following statements.

\begin{enumerate}
\def\labelenumi{\arabic{enumi}.}

\item
  For all \(n
  \in \semigrposint\) and \(f_1,\cdots,f_n
  \in \sphilb{H}_{-\frac{1}{2}}\), \begin{align}
  U \fun{\opfockcr}{h^{-\frac{1}{2}} f_1} \cdots \fun{\opfockcr}{h^{-\frac{1}{2}} f_n} \opfockvac_{\txtbsn}
  =
  2^{\frac{n}{2}}
  \wick{\fun{\opfocksegal_{\mathrm{F}}}{f_1} \cdots \fun{\opfocksegal_{\mathrm{F}}}{f_n}} 
  \end{align} holds.
\item
  For any \(f
  \in \sphilb{H}_{-\frac{1}{2}}\), \[U \fun{\opfocksegal}{f} \inv{U}
  = \fun{\opfocksegal_{\mathrm{F}}}{f}\] holds.
\end{enumerate}

By statement (2), \(\fun{\opfocksegal_{\mathrm{F}}}{f}\) is the sharp-time field at time \(0\) in the \(Q\)-space representation.

The subspace \(\sphilb{H}_{-\frac{1}{2}}\) of the real Hilbert space \(\sphilb{H}\) is \begin{align}
\bkt{f}{g}_{-\frac{1}{2}}
=
\bkt{h^{-\frac{1}{2}}f}{h^{-\frac{1}{2}}g}_{\sphilb{H}}, \quad f,g
\in \sphilb{H}_{-\frac{1}{2}} 
\end{align}

regarded as a real inner product space with inner product \begin{align}
\bkt{f}{g}_{-\frac{1}{2}}
=
\bkt{h^{-\frac{1}{2}}f}{h^{-\frac{1}{2}}g}_{\sphilb{H}}, \quad f,g
\in \sphilb{H}_{-\frac{1}{2}} 
\end{align} and its completion is denoted \(\gtclos{\sphilb{H}_{-\frac{1}{2}}}\). By the definition of completion, for any \(f
\in \gtclos{\sphilb{H}_{-\frac{1}{2}}}\) there exists a sequence \(\seqn{f}\) in \(\sphilb{H}_{-\frac{1}{2}}\) converging strongly in the norm derived from the above inner product. On the other hand, \[\norm{\fun{\opfocksegal_{\mathrm{F}}}{f_n} - \fun{\opfocksegal_{\mathrm{F}}}{f_m}}_{\fun{\lp^{2}}{Q_{\sphilb{H}}, \mu_{\sphilb{H}}}}^2
=
\frac{1}{2}
\norm{f_n - f_m}_{-\frac{1}{2}}^2\] shows that \(\seq{\fun{\opfocksegal_{\mathrm{F}}}{f_n}}{\nin}\) is a Cauchy sequence in \(\fun{\lp^{2}}{Q_{\sphilb{H}}, \mu_{\sphilb{H}}}\). Therefore there exists a point \(\fun{\tilde{\opfocksegal}_{\mathrm{F}}}{f}
\in \fun{\lp^{2}}{Q_{\sphilb{H}}, \mu_{\sphilb{H}}}\) satisfying \[\lim_{n \to \infty} \fun{\opfocksegal_{\mathrm{F}}}{f_n}
= \fun{\tilde{\opfocksegal}_{\mathrm{F}}}{f},\] and this limit is independent of the choice of \(\seqn{f}\). If \(f
\in \sphilb{H}_{-\frac{1}{2}}\), setting \(f_n
= f\) for all \(\nin\) gives \(\fun{\tilde{\opfocksegal}_{\mathrm{F}}}{f}
= \fun{\opfocksegal_{\mathrm{F}}}{f}\). From the above argument, the family \(\set{\fun{\tilde{\opfocksegal}_{\mathrm{F}}}{f}}{f \in \gtclos{\sphilb{H}_{-\frac{1}{2}}}}\) containing \(\set{\fun{\opfocksegal_{\mathrm{F}}}{f}}{f \in \sphilb{H}_{-\frac{1}{2}}}\) is a Gaussian hyperprocess indexed by \(\gtclos{\sphilb{H}_{-\frac{1}{2}}}\).

\subsection{Probability Space of Euclidean Quantum Fields}\label{probability-space-of-euclidean-quantum-fields}

Let \(\sphilb{H}_{h,-1}\) be the inner product space that is a subspace of the Hilbert-space-valued tempered distributions \(\fun{\dsttempered}{\fldreal;\sphilb{H}}\), consisting of all maps satisfying the following conditions.

\begin{enumerate}
\def\labelenumi{\arabic{enumi}.}

\item
  The Fourier transform \(\faftr{F}\) of a map \(F
  \in \fun{\dsttempered}{\fldreal;\sphilb{H}}\) is a \(\lacomplexication{\sphilb{H}}\)-valued function on \(\fldreal\).
\item
  For almost all \(p\), \[\faftr{F}(p)
  \in \dom \rbk{h^2+p^2}^{-\frac{1}{2}}\] holds.
\item
  The Hilbert-space-valued function \(p
  \mapsto \rbk{h^2+p^2}^{-\frac{1}{2}} \faftr{F}(p)\) is measurable and satisfies \[\int_{\fldreal}
  \norm{\rbk{h^2+p^2}^{-\frac{1}{2}} \faftr{F}(p)}^2
  \opdmsr{p}
  < \infty.\]
\end{enumerate}

This \(\sphilb{H}_{h,-1}\) is a real vector space that is a real inner product space with inner product \begin{align}
\bkt{F}{G}_{h,-1}
=
2 \int_{\fldreal}
\bkt{\rbk{h^2+p^2}^{-\frac{1}{2}} \faftr{F}(p)}
{\rbk{h^2+p^2}^{-\frac{1}{2}} \faftr{G}(p)}_{\sphilb{H}}
\opdmsr{p} 
\end{align} We complete this with the metric determined by the above inner product, and denote the resulting real Hilbert space by \(\opclos{\sphilb{H}_{h,-1}}\). When no confusion arises, we also write \(\sphilb{H}_{h,-1}\) for its closure.

There exists a Gaussian hyperprocess indexed by the real Hilbert space \(\opclos{\sphilb{H}_{h,-1}}\). Let the probability space on which this Gaussian hyperprocess \(\set{\fun{\opfocksegal_{\txteuclid}}{F}}{F \in \opclos{\sphilb{H}_{h,-1}}}\) is realized be \(\rbk{Q_{\txteuclid},
\mblfml{B}_{\txteuclid},
\mu_{\txteuclid}}\). Since \(\delta_t \times f
\in \opclos{\sphilb{H}_{h,-1}}\) holds for any \(t
\in \fldreal\) and \(f
\in \sphilb{H}_{-\frac{1}{2}}\), we introduce the random variable \begin{align}
\fun{\opfocksegal_{\txteuclid}}{t,f}
=
\fun{\opfocksegal_{\txteuclid}}{\delta_t \times f}
\end{align} For any \(\nin\) and \(t_1,\cdots,f_n
\in \sphilb{H}_{-\frac{1}{2}}\), \begin{align}
&\fun{\opfockschwingerfunctional_n}{t_1,f_1,\cdots,t_n,f_n} \\
&=
\int_{Q_{\txteuclid}}
\fun{\opfocksegal_{\txteuclid}}{t_1,f_1}
\cdots
\fun{\opfocksegal_{\txteuclid}}{t_n,f_n}
\opdmsr{\mu_{\txteuclid}} \\
&=
\physmean{\fun{\opfocksegal_{\txteuclid}}{t_1,f_1}
\cdots
\fun{\opfocksegal_{\txteuclid}}{t_n,f_n}} \label{expedition0010003} 
\end{align} holds. The Gaussian hyperprocess \(\fun{\opfocksegal_{\txteuclid}}{t,f}\) satisfying these properties is called the free Euclidean quantum field associated with the free field \(\opfocksegal(t,f)\).

\subsection{\texorpdfstring{Operators \(j_t\) and \(J_t\)}{Operators j\_t and J\_t}}\label{operators-j_t-and-j_t}

Define the map \(j_t\) between the Hilbert spaces \(\sphilb{H}_{-\frac{1}{2}}
\subset \sphilb{H}\) and \(\sphilb{H}_{h,-1}
\subset \fun{\dsttempered_{\txtreal}}{\fldreal;\sphilb{H}}\) by \begin{align}
j_t
\colon \sphilb{H}_{-\frac{1}{2}}
\to \sphilb{H}_{h,-1}; \quad j_t f
=
\delta_t \times f, \quad f
\in \sphilb{H}_{-\frac{1}{2}} 
\end{align}

\begin{prop}\label{expedition0010021}
The above $j_t$ is an isometry satisfying
\begin{align}
\bkt{j_t f}{j_s g}_{h,-1}
=
\bkt{f}{\napiernum^{-\abs{t-s} h} g}_{-\frac{1}{2}}, \quad f,g
\in \sphilb{H}_{-\frac{1}{2}} \label{expedition0010022} 
\end{align}
In particular, $j_t$ extends uniquely to an isometry from $\gtclos{\sphilb{H}_{-\frac{1}{2}}}$ to $\gtclos{\sphilb{H}_{h,-1}}$, and we also denote this extension by $j_t$.
By the above formula, for any $s,t$,
\begin{align}
\faadj{j_t} j_s
=
\napiernum^{-\abs{t-s} h} \quad s,t \in \fldreal 
\end{align}
holds.
\end{prop}

Using this, define the second-quantized version of the operator as follows. There exists a unique isometry \(J_t
\colon \fun{\lp^{2}}{Q_{\sphilb{H}}, \mu_{\sphilb{H}}}
\to \fun{\lp^{2}}{Q_{\txteuclid}, \mu_{\txteuclid}}\) satisfying \begin{align}
J_t 1
&=
1, \\ 
J_t \wick{\fun{\opfocksegal_{\mathrm{F}}}{f_1} \cdots \fun{\opfocksegal_{\mathrm{F}}}{f_n}}
&=
\wick{\fun{\opfocksegal_{\txteuclid}}{j_t f_1} \cdots \fun{\opfocksegal_{\txteuclid}}{j_t f_n}} 
\end{align} Noting that \(\opfocksegal_{\txtfock}(f)
= \fun{\opfocksegal}{h^{-\onehalf} f}\), we have \[J_t \wick{\fun{\opfocksegal}{f_1} \cdots \fun{\opfocksegal}{f_n}}
=
\wick{\fun{\opfocksegal_{\txteuclid}}{j_t \times h^{\onehalf} f_1}
\cdots
\fun{\opfocksegal_{\txteuclid}}{j_t \times h^{\onehalf} f_n}}.\]

\subsection{Hamiltonian}\label{hamiltonian}

For the unitary transformation \(U
\colon \fun{\spfock_{\txtbsn}}{\lacomplexication{\sphilb{H}}}
\to \fun{\lp^{2}}{Q_{\sphilb{H}}, \mu_{\sphilb{H}}}\) from Section \ref{expedition0012091}, set \begin{align}
\physham_{0,\txtfock}
=
U \physham_{\txtbsn,\txtfr} \inv{U} 
\end{align} Objects defined on the bosonic Fock space are transplanted to the probability space by the same unitary transformation.

The Hamiltonian and other operators on the bosonic Fock space have counterparts in the \(Q\)-space representation \(\fun{\lp^{2}}{\prbqspace_{\sphilb{H}},\msr{P}_{\sphilb{H}}}\). For brevity, we use the same notation for these when no confusion arises.

\subsection{Deviations from Existing Literature Notation}\label{expedition0012096}

In the following discussion, we modify the variance of the Euclidean field \(\opfocksegal_{\txteuclid}(f)\) as follows, and accordingly revise the definitions of the isometries \(j_t\) and \(J_t\).

The covariance and variance of the Euclidean field are defined as \begin{equation}
\begin{aligned}
\fun{\prbcov_{\msr{P}_{\txteuclid}}}
{\opfocksegal_{\txteuclid}(j_s f),
\opfocksegal_{\txteuclid}(j_t g)}
&=
\onehalf
\bkt{f}{\napiernum^{-\abs{t-s} \physham[h]} g}_{\sphilb{H}}, \\
\sqfun{\prbvar_{\msr{P}_{\txteuclid}}}{\opfocksegal_{\txteuclid}(j_t f)}
&=
\onehalf
\norm{f}_{\sphilb{H}}^2
\end{aligned}
\end{equation} The isometry \(j_t\) is revised to \[j_t
\colon \sphilb{H}
\to \sphilb{H}_{\physham[h],-1};
\quad
j_t f
= \diracdelta_t \times f,\] and the isometry \(J_t\) is revised to \[J_t
\colon \fun{\lp^{2}}{\prbqspace_{\sphilb{H}},\msr{P}_{\sphilb{H}}}
\to \fun{\lp^{2}}{\prbqspace_{\txteuclid},\msr{P}_{\txteuclid}}\] with action \begin{equation}
\begin{aligned}
J_t 1
&=
1, \\
J_t \wick{\fun{\opfocksegal}{f_1} \cdots \fun{\opfocksegal}{f_n}}
&=
\wick{\fun{\opfocksegal_{\txteuclid}}{j_t f_1} \cdots \fun{\opfocksegal_{\txteuclid}}{j_t f_n}}
\end{aligned}
\end{equation} For any \(F
\in \sphilb{H}_{\physham[h],-1}\), \[\sqfun{\prbexp_{\msr{P}_{\txteuclid}}}{
\opfocksegal_{\txteuclid}(F)}
=
\fnexp{\onehalf
\sqfun{\prbvar_{\msr{P}_{\txteuclid}}}{\opfocksegal_{\txteuclid}(F)}}
=
\fnexp{\oneoverfour \norm{F}_{\sphilb{H}_{\physham[h],-1}}^2}\] holds.

\subsection{Pair Potential}\label{pair-potential}

From here we adopt the notation for the van Hove model. In particular, we replace the general self-adjoint operator \(\physham[h]\) with the dispersion relation \(\omega\).

With the modifications of the previous subsection in mind, define the pair potential of the van Hove model by \[W_{\kappa,\Lambda}(t)
=
\int_{0}^{t}
\fun{j_s}{\omega \mathsf{m}_{\kappa,\Lambda}}
\opdmsr{s}
\in
\sphilb{H}_{\omega,-1}.\]

Define the notation \(\opfocksegal_{\txteuclid}(t,f)
= \opfocksegal_{\txteuclid}(j_t f)\) and, following the notation simplification convention, let \(\physham_{\txtvanhowe,\kappa,\Lambda}\) denote the van Hove model Hamiltonian on \(\prbqspace_{\sphilb{H}}\). With the modifications of the previous subsection in mind, by \cite{AsaoArai30,LorincziHiroshimaBetz3} the functional integral representation \[\bkt{F}
{\napiernum^{-t \physham_{\txtvanhowe,\kappa,\Lambda}} G}
_{\fun{\lp^{2}}{\prbqspace_{\sphilb{H}},\msr{P}_{\sphilb{H}}}}
=
\sqfun{\prbexp_{\msr{P}_{\txteuclid}}}
{\cmpconj{J_0 F}
\cdot
J_t G
\cdot
\napiernum^{-\wick{\, \fun{\opfocksegal_{\txteuclid}}{W_{\kappa,\Lambda}(t)} \,}}}\] holds. By the definition of the Wick product for random variables, the Wick product of a first-order term satisfies \[\wick{\fun{\opfocksegal_{\txteuclid}}{f}}
=
\fun{\opfocksegal_{\txteuclid}}{f}
-\sqfun{\prbexp_{\msr{P}_{\txteuclid}}}{\fun{\opfocksegal_{\txteuclid}}{f}}
=
\fun{\opfocksegal_{\txteuclid}}{f},\] so for the van Hove Hamiltonian one may simply write \[\bkt{F}
{\napiernum^{-t \physham_{\txtvanhowe,\kappa,\Lambda}} G}
_{\fun{\lp^{2}}{\prbqspace_{\sphilb{H}},\msr{P}_{\sphilb{H}}}}
=
\sqfun{\prbexp_{\msr{P}_{\txteuclid}}}
{\cmpconj{J_0 F}
\cdot
J_t G
\cdot
\napiernum^{-\fun{\opfocksegal_{\txteuclid}}{W_{\kappa,\Lambda}(t)}}}.\]

\begin{lem}\label{expedition0011363}
For a positive number $\omega(k)
> 0$, define two double integrals by
\begin{equation}
\begin{aligned}
I_T^{(1)}(\omega(k))
&=
\int_{0}^{T}
\int_{0}^{T}
\napiernum^{-\abs{t-s} \omega(k)}
\opdmsr{t}
\opdmsr{s}, \\
I_T^{(2)}(\omega(k))
&=
\int_{-T}^T
\int_{-T}^T
\napiernum^{-\abs{t-s} \omega(k)}
\opdmsr{t}
\opdmsr{s}
\end{aligned}
\end{equation}
These can be evaluated as
\begin{equation}
\begin{aligned}
I_T^{(1)}(\omega(k))
&=
\frac{2T}{\omega(k)}
-\frac{2}{\omega(k)^2}
\rbk{1 - \napiernum^{-T \omega(k)}}, \\
I_{T}^{(2)}(\omega(k))
&=
\frac{4T}{\omega(k)}
-\frac{2}{\omega(k)^2}
\rbk{1 - \napiernum^{-2 T \omega(k)}}
\end{aligned}
\end{equation}
\end{lem}

\begin{rem}
The first term on the right-hand side is the linearly divergent self-energy; the second term is finite and leads to the Coulomb exchange energy.
In particular, it is useful to recall the Fourier transform of the Coulomb kernel with $\omega(k)
= \abs{k}$: $$\frac{1}{(2 \pi)^3}
\int_{\fldreal^{d}}
\frac{\napiernum^{\imunit k \vainnprod} x}{\abs{k}^2}
\opdmsr{k}
=
\frac{1}{4 \pi \abs{x}}.$$
\end{rem}

\begin{proof}
By careful case analysis and direct computation, $$I_T^{(1)}(k)
=
2
\int_0^T
(T - s)
\napiernum^{- s \omega(k)}
\opdmsr{s}
=
\frac{2T}{\omega(k)}
-\frac{2 \rbk{1 - \napiernum^{-T \omega(k)}}}{\omega(k)^2}$$ is obtained.

We discuss the second integral $I_T^{(2)}(\omega(k))$.
For notational convenience, set $$I(s)
=
\int_{-T}^T
\napiernum^{-\abs{t-s} \omega}
\opdmsr{t},
\quad
I
=
\int_{-T}^T
I(s)
\opdmsr{s}.$$
For any $s
\in \closedinterval{-T}{T}$, substituting $u
= t-s$ gives $\opdmsr{t}
=
\opdmsr{u}$ and
\begin{equation}
\begin{aligned}
&I(s)
=
\int_{-T-s}^{T-s}
\napiernum^{-\abs{u} \omega}
\opdmsr{u}
=
\int_0^{T-s} \napiernum^{-u \omega} \opdmsr{u}
+\int_0^{T+s} \napiernum^{-u \omega} \opdmsr{u} \\
&=
\frac{1 - \napiernum^{-(T-s) \omega}}{\omega}
+\frac{1 - \napiernum^{-(T+s) \omega}}{\omega} \\
&=
\frac{2}{\omega}
-\frac{\napiernum^{-(T-s) \omega} + \napiernum^{-(T-s)} \omega}{\omega}
\end{aligned}
\end{equation}
holds.

Next, we evaluate the $s$-integral.
First,
\begin{equation}
\begin{aligned}
&I
=
\int_{-T}^{T}
\rbk{\frac{2}{\omega}
-\frac{\napiernum^{-(T-s) \omega} + \napiernum^{-(T-s)} \omega}{\omega}}
\opdmsr{s} \\
&=
\frac{4T}{\omega}
-\frac{1}{\omega}
\rbk{\int_{-T}^{T} \napiernum^{-(T-s) \omega} \opdmsr{s}
+\int_{-T}^{T} \napiernum^{-(T+s) \omega} \opdmsr{s}}
\end{aligned}
\end{equation}
holds.

For the first integral in the second term, the substitution $t
= T-s$ gives $$\int_{-T}^{T} \napiernum^{-(T-s) \omega} \opdmsr{s}
=
\int_{2T}^{0} \napiernum^{-\omega t} (-\opdmsr{t})
=
\int_{0}^{2T} \napiernum^{-\omega t} \opdmsr{t}
=
\frac{1 - \napiernum^{-2 T \omega}}{\omega},$$ and for the second, $u
= T+s$ gives $$\int_{-T}^{T} \napiernum^{-(T+s) \omega} \opdmsr{s}
=
\int_0^{2T} \napiernum^{-\omega u} \opdmsr{u}
=
\frac{1 - \napiernum^{-2 T \omega}}{\omega}.$$
Collecting these completes the proof.
\end{proof}

\begin{prop}\label{expedition0011362}
The norm of the pair potential $W_{\kappa,\Lambda}(T)
=
\int_{0}^{t}
j_s(\omega \mathsf{m}_{\kappa,\Lambda})
\opdmsr{s}$ is, using $\physgse(\physham_{\txtvanhowe,\kappa,\Lambda})
=
-\onehalf
\norm{\omega^{\onehalf} \mathsf{m}_{\kappa,\Lambda}}_{\sphilb{H}}^2$,
\begin{equation}
\begin{aligned}
&\norm{W_{\kappa,\Lambda}(T)}_{\sphilb{H}_{\omega,-1}}^2
\\ 
&=
2
\int_{\fldreal^{d}}
\abs{\omega(k)^{\onehalf} \faftr{\mathsf{m}_{\kappa,\Lambda}}(k)}^2
\rbk{T
-\frac{1 - \napiernum^{-T \omega(k)}}{\omega(k)}}
\opdmsr{k}
\\ 
&=
-4
\physgse(\physham_{\txtvanhowe,\kappa,\Lambda})
-2
\bkt{\mathsf{m}_{\kappa,\Lambda}}
{\rbk{1 - \napiernum^{-T \omega}}\mathsf{m}_{\kappa,\Lambda}}
_{\sphilb{H}}
\end{aligned}
\end{equation}
A useful form for later use is
\begin{equation}
\begin{aligned}
&\oneoverfour
\norm{W_{\kappa,\Lambda}(T)}_{\sphilb{H}_{\omega,-1}}^2
\\ 
&=
-T \physgse(\physham_{\txtvanhowe,\kappa,\Lambda})
-\onehalf
\bkt{\mathsf{m}_{\kappa,\Lambda}}
{\rbk{1 - \napiernum^{-T \omega}} \mathsf{m}_{\kappa,\Lambda}}
_{\sphilb{H}}
\end{aligned}
\end{equation}
In particular, the $\msr{P}_{\txteuclid}$-expectation value of the random variable $-\opfocksegal_{\txteuclid}(W_{\kappa,\Lambda}(T))$ is
\begin{equation}
\begin{aligned}
&\sqfun{\prbexp_{\msr{P}_{\txteuclid}}}
{\napiernum^{-\opfocksegal_{\txteuclid}(W_{\kappa,\Lambda}(T))}} \\
&=
\fnexp{\onehalf
\int_{\fldreal^{d}}
\abs{\omega(k)^{\onehalf}
\faftr{\mathsf{m}_{\kappa,\Lambda}}(k)}^2
\rbk{T
-\frac{1 - \napiernum^{-T \omega(k)}}{\omega(k)}}
\opdmsr{k}}
\\
&=
\fnexp{-T \physgse(\physham_{\txtvanhowe,\kappa,\Lambda})
-\onehalf
\bkt{\mathsf{m}_{\kappa,\Lambda}}
{\rbk{1 - \napiernum^{-T \omega}}
\mathsf{m}_{\kappa,\Lambda}}
_{\sphilb{H}}}
\end{aligned}
\end{equation}
and the Wick product is $$\wick{\napiernum^{-\opfocksegal_{\txteuclid}(W_{\kappa,\Lambda}(T))}}
=
\fnexp{-\oneoverfour \norm{W_{\kappa,\Lambda}(T)}_{\sphilb{H}_{\omega,-1}}^2}
\napiernum^{-\opfocksegal_{\txteuclid}(W_{\kappa,\Lambda}(T))}.$$
\end{prop}

\begin{proof}
Using the double integral from Lemma \ref{expedition0011363}, $$I_T^{(1)}(k)
=
\int_{0}^{T}
\int_{0}^{T}
\napiernum^{-\abs{t-s} \omega(k)}
\opdmsr{t}
\opdmsr{s}
=
\frac{2}{\omega(k)}
\rbk{T
-\frac{1 - \napiernum^{-T \omega(k)}}{\omega(k)}},$$ we obtain
\begin{equation}
\begin{aligned}
&\norm{W_{\kappa,\Lambda}(T)}
_{\sphilb{H}_{\omega,-1}}^2 \\
&=
\int_0^{T}
\int_0^{T}
\bkt{\fun{j_s}{\omega \mathsf{m}_{\kappa,\Lambda}}}
{\fun{j_t}{\omega \mathsf{m}_{\kappa,\Lambda}}}
_{\sphilb{H}_{\omega,-1}}
\opdmsr{t}
\opdmsr{s} \\
&=
\int_0^{T}
\int_0^{T}
\bkt{\omega \mathsf{m}_{\kappa,\Lambda}}
{\napiernum^{-\abs{t-s} \omega} \omega \mathsf{m}_{\kappa,\Lambda}}
_{\sphilb{H}}
\opdmsr{t}
\opdmsr{s} \\
&=
2
\bkt{\mathsf{m}_{\kappa,\Lambda}}
{\omega \rbk{T - \frac{1 - \napiernum^{-T \omega}}{\omega}} \mathsf{m}_{\kappa,\Lambda}}
_{\sphilb{H}}
\end{aligned}
\end{equation}
is obtained.

We compute the expectation value $\sqfun{\prbexp_{\msr{P}_{\txteuclid}}}{\napiernum^{-\opfocksegal_{\txteuclid}(W_{\kappa,\Lambda}(T))}}$.
Since the variance is $\oneoverfour \norm{W_{\kappa,\Lambda}(T)}
_{\sphilb{H}_{\omega,-1}}^2$, using the first-half result,
\begin{equation}
\begin{aligned}
&\sqfun{\prbexp_{\msr{P}_{\txteuclid}}}
{\napiernum^{-\opfocksegal_{\txteuclid}(W_{\kappa,\Lambda}(T))}} \\
&=
\fnexp{\onehalf
\int_{\fldreal^{d}}
\abs{\omega(k)^{\onehalf}
\mathsf{m}_{\kappa,\Lambda}}^2
\rbk{T
-\frac{1 - \napiernum^{-T \omega(k)}}{\omega(k)}}
\opdmsr{k}}
\end{aligned}
\end{equation}
is obtained.
\end{proof}

\subsection{Facts}\label{facts}

The following theorem is useful for estimating the ground state energy.

\begin{thm}\label{expedition0004996}
Let $H$ be a bounded-below self-adjoint operator on a complex Hilbert space $\sphilb{H}$, let $\opspecmeas_H$ be the spectral measure of $H$, let $\physgse(\physham)$ be the ground state energy, and let $S_H(\beta)
= \napiernum^{-\beta \physham}$ be the semigroup generated by $H$.
Define the subspace of $\sphilb{H}$ by
\begin{align}
\sphilb{E}(H)
=
\set{\psi \in \sphilb{H} \setminus \setone{0}}
{\parbox{17em}
{there exists a constant $\ep_0$ such that $$E_H\rbk{\rightopeninterval{\physgse(\physham)}{\physgse(\physham) + \ep}} \psi
\neq 0$$ for all $\ep \in (0, \ep_0)$}}
\label{expedition0004997} 
\end{align}
Then for all $\psi
\in \sphilb{E}(H)$, $$\physgse(\physham)
=
-\lim_{\beta \to \infty}
\frac{\log \bkt{\psi}{\napiernum^{-\beta \physham} \psi}}
{\beta}$$ holds.
\end{thm}

\begin{prop}\label{expedition0011392}
The mean of the exponential Wick product
\begin{equation}
\begin{aligned}
\wick{\napiernum^{\alpha \opfocksegal(f)}}_{\mu}
&=
\fnexp{-\frac{\alpha^2}{4} \opform{q}_{\txtnonzero,\infty}(f)}
\napiernum^{\alpha \opfocksegal(f)},
\\ 
\napiernum^{\alpha \opfocksegal(f)}
&=
\fnexp{+\frac{\alpha^2}{4} \opform{q}_{\txtnonzero,\infty}(f)}
\wick{\napiernum^{\alpha \opfocksegal(f)}}_{\mu}
\end{aligned}
\end{equation}
is $$\sqfun{\prbexp}{\opfocksegal(f)}
=
\fnexp{\oneoverfour \opform{q}_{\txtnonzero,\infty}(f)}.$$
The product of exponential Wick products is $$\wick{\napiernum^{\alpha \opfocksegal(f)}} \cdot \wick{\napiernum^{\beta \opfocksegal(g)}}
=
\fnexp{\frac{\alpha \beta}{2} \bkt{f}{g}}
\wick{\napiernum^{\opfocksegal(\alpha f + \beta g)}}.$$
The action of the free Hamiltonian on the exponential Wick product is $$\physham_{\txtbsn,\txtfr}
\wick{\napiernum^{\fun{\opfocksegal}{f}}}
_{\msr{P}_{\sphilb{H}}}
=
\wick{\opfocksegal(\omega f)
\napiernum^{\fun{\opfocksegal}{f}}}
_{\msr{P}_{\sphilb{H}}}.$$
The action of the Segal field operator on the exponential Wick product is $$\opfocksegal(h)
\wick{\napiernum^{\fun{\opfocksegal}{f}}}
_{\msr{P}_{\sphilb{H}}}
=
\wick{\opfocksegal(h)
\napiernum^{\fun{\opfocksegal}{f}}}
_{\msr{P}_{\sphilb{H}}}
+\onehalf
\bkt{h}{f}_{\sphilb{H}}
\wick{\napiernum^{\fun{\opfocksegal}{f}}}
_{\msr{P}_{\sphilb{H}}}.$$
\end{prop}

Toward the construction of the ground state, we use the following theorem.

\begin{thm}\label{expedition0011361}
Let $\physham$ be a bounded-below self-adjoint operator on a complex Hilbert space $\sphilb{H}$ with ground state energy $\physgse$.
For any $f
\in \sphilb{H} \setminus \setone{0}$ and $\beta
> 0$, set $$\gamma_f(\beta)
=
\frac{\bkt{f}{\napiernum^{-\beta \physham} f}}
{\norm{\napiernum^{-\beta \physham} f}^2}.$$
Then the following two conditions are equivalent.
\begin{enumerate}
\item
The limit $a
=
\lim_{\beta \to \infty}
\gamma_f(\beta)$ exists and is strictly positive, i.e., $a
> 0$.

\item
The ground state energy is an eigenvalue, and for the projection $P_{0}$ onto the eigenspace, $\bkt{f}{P_0 f}
> 0$.
\end{enumerate}
If either condition holds, then $P_0 f$ is nonzero and is a ground state.
\end{thm}

\section{\texorpdfstring{Ground State at Zero Temperature via \(Q\)-Space Functional Integral}{Ground State at Zero Temperature via Q-Space Functional Integral}}\label{expedition0011370}

We discuss the formulation via the \(Q\)-space representation. We first briefly show the existence of a ground state under infrared regularization, then discuss the non-existence of a ground state in the Fock representation under infrared-singular conditions.

In operator theory and operator algebras, the GNS-type argument shows that the influence of infrared/ultraviolet divergences appears as a shrinking of the algebra of observables, and in the probabilistic setting the following structure emerges.

\begin{enumerate}
\def\labelenumi{\arabic{enumi}.}

\item
  There is the \(\lp^{2}\)-space \(\fun{\lp^{2}}{\prbqspace_{\sphilb{H}},\msr{P}_{\sphilb{H}}}\) in \(Q\)-representation, unitarily equivalent to the Fock space.
\item
  From the ground state that exists in the cutoff setting, construct a family of probability measures \(\msr{P_{\txtgs,\kappa,\Lambda}}\).
\item
  The limit \(P_{\txtgs}\) is obtained by removing the cutoffs from the family of probability measures.
\item
  The \(\lp^{2}\)-space \(\fun{\lp^{2}}{\prbqspace_{\sphilb{H}},\msr{P}_{\txtgs}}\) for this \(P_{\txtgs}\) is smaller as a set than the original \(\fun{\lp^{2}}{\prbqspace_{\sphilb{H}},\msr{P}_{\sphilb{H}}}\).
\end{enumerate}

In particular, the limit measure is formulated via a concept like weak convergence of measures (as formulated via \(\lp^{2}\) and bounded continuous functions), and the algebra of observables is identified as the \(\lp^{2}\)-closure with respect to that measure. It suffices to discuss weak convergence of measures, and without it the algebra of observables cannot be obtained. Therefore, a proper convergence argument for the measures becomes essential.

Here, since we also derive the ground state energy by the functional integral method, we assume that the self-adjointness and lower-boundedness of the van Hove model are known, but the specific ground state energy is not.

\subsection{Discussion under Infrared Regularization}\label{discussion-under-infrared-regularization}

\begin{prop}
The heat semigroup generated by the van Hove Hamiltonian $\physham_{\txtvanhowe,\kappa,\Lambda}$ is a positivity-improving operator.
In particular, by a Perron-Frobenius-type argument, the ground state of the van Hove Hamiltonian, if it exists, is unique and strictly positive almost everywhere.
\end{prop}

\begin{proof}
Suppose $\bkt{F}
{\napiernum^{-t \physham_{\txtvanhowe,\kappa,\Lambda}} G}
= 0$ for some $t
> 0$ and non-negative, nonzero $F,G$.
In the functional integral representation, the integrand $\napiernum^{-\fun{\opfocksegal_{\txteuclid}}{W_{\kappa,\Lambda}(t)}}$ is strictly positive almost everywhere.
Therefore $J_0 F
\cdot
J_t G
= 0$ must hold $\msr{P}_{\sphilb{H}}$-almost surely.
In particular, for the free Hamiltonian $\physham_{\txtbsn,\txtfr}$, $$0
=
\bkt{J_0 F}
{J_t G}_{\fun{\lp^{2}}{\prbqspace_{\txteuclid},\msr{P}_{\txteuclid}}}
=
\bkt{F}
{\napiernum^{-t \physham_{\txtbsn,\txtfr}} G}_{\fun{\lp^{2}}{\prbqspace_{\sphilb{H}}, \msr{P}_{\sphilb{H}}}}$$ holds.
The operator $\napiernum^{-t \physham_{\txtbsn,\txtfr}}$ is positivity-improving.
For non-negative and nonzero $F,G$, $$\bkt{F}{\napiernum^{-t \physham_{\txtbsn,\txtfr}} G}_{\fun{\lp^{2}}{\prbqspace_{\sphilb{H}},\msr{P}_{\sphilb{H}}}}
> 0$$ must hold.
This is a contradiction.
\end{proof}

\begin{thm}\label{expedition0011360}
Under infrared regularization, the van Hove Hamiltonian has a ground state, and its ground state energy is $$\fun{\physgse}
{\physham_{\txtvanhowe,\kappa,\Lambda}}
=
-\onehalf
\norm{\omega^{\onehalf} \mathsf{m}_{\kappa,\Lambda}}_{\sphilb{H}}^2.$$
In particular, the ground state normalized in $\fun{\lp^{2}}{\prbqspace_{\sphilb{H}},\msr{P}_{\sphilb{H}}}$ is $$\Psi_{\txtgs,\msr{P}_{\sphilb{H}},\kappa,\Lambda}
=
\fnexp{-\oneoverfour
\norm{\mathsf{m}_{\kappa,\Lambda}}_{\sphilb{H}}^2}
\wick{\napiernum^{-\opfocksegal(\mathsf{m}_{\kappa,\Lambda})}}
_{\msr{P}_{\sphilb{H}}}
\in
\fun{\lp^{2}}{\prbqspace_{\sphilb{H}},\msr{P}_{\sphilb{H}}}.$$
\end{thm}

\begin{rem}
Under infrared regularization, $\mathsf{m}_{\kappa,\Lambda}
\in \sphilb{H}$.
The ground state is also the square root of the Radon-Nikodym derivative, i.e., $$\Psi_{\txtgs,\msr{P}_{\sphilb{H}},\kappa,\Lambda}
=
\sqrt{\od{\msr{P}_{\kappa,\Lambda}}
{\msr{P}_{\sphilb{H}}}},$$ which will be used in the ground-state transformation discussed below.
In the following argument, the computation of $\bkt{1}{\napiernum^{-T \physham} 1}$ is needed for the existence and construction discussion.
The typical operator-theoretic method for computing this is the Trotter-Kato product formula.
This is exactly the computation needed for the functional integral representation, making it a clear illustration of what functional integration is trying to accomplish.
\end{rem}

\begin{proof}
(Notation): For brevity, denote the van Hove Hamiltonian by $\physham$ and the ground state energy by $\physgse$.
Introduce the function $$L(T)
=
\bkt{1}{\napiernum^{-T \physham} 1}
_{\fun{\lp^{2}}{\prbqspace_{\sphilb{H}},\msr{P}_{\sphilb{H}}}}.$$
By the functional integral representation and Proposition \ref{expedition0011362}, this can be written as $$L(T)
=
\fnexp{
\frac{T}{2}
\norm{\omega^{\onehalf} \mathsf{m}_{\kappa,\Lambda}}
_{\sphilb{H}}^{2}
-\onehalf
\bkt{\mathsf{m}_{\kappa,\Lambda}}
{\rbk{1 - \napiernum^{-T \omega}} \mathsf{m}_{\kappa,\Lambda}}
_{\sphilb{H}}}.$$
As will be seen, the first term in the exponent is a constant multiple of the ground state energy.

(Existence of the ground state): To apply Proposition \ref{expedition0011361}, set
\begin{equation}
\begin{aligned}
\Psi_{T}
=
\frac{\napiernum^{-T \physham} 1}{\norm{\napiernum^{-T \physham} 1}_{\fun{\lp^{2}}{\prbqspace_{\sphilb{H}},\msr{P}_{\sphilb{H}}}}}
\end{aligned}
\end{equation}
Then for $\gamma_1(T)
=
\frac{\bkt{1}{\napiernum^{-TH} 1}}
{\norm{\napiernum^{-TH} 1}}$,
\begin{equation}
\begin{aligned}
&\log
\gamma_1(T)
=
\log \frac{L(T)}{L(2T)^{\onehalf}}
\\ 
&=
\frac{T}{2}
\norm{\omega^{\onehalf} \mathsf{m}_{\kappa,\Lambda}}
_{\sphilb{H}}^2
-\onehalf
\bkt{\mathsf{m}_{\kappa,\Lambda}}
{\rbk{1 - \napiernum^{-T \omega}} \mathsf{m}_{\kappa,\Lambda}}
_{\sphilb{H}}
\\
&\quad
-\onehalf
\rbk{T \norm{\omega^{\onehalf} \mathsf{m}_{\kappa,\Lambda}}
_{\sphilb{H}}^2
-\onehalf
\bkt{\mathsf{m}_{\kappa,\Lambda}}
{\rbk{1 - \napiernum^{-2T \omega}} \mathsf{m}_{\kappa,\Lambda}}
_{\sphilb{H}}}
\\ 
&=
-\oneoverfour
\bkt{\mathsf{m}_{\kappa,\Lambda}}
{\rbk{1 - \napiernum^{-T \omega}}^2 \mathsf{m}_{\kappa,\Lambda}}
_{\sphilb{H}}
\end{aligned}
\end{equation}
is obtained.

The limit satisfies $$\lim_{T \to \infty}
\gamma_1(T)
=
\fnexp{-\oneoverfour
\norm{\mathsf{m}_{\kappa,\Lambda}}_{\sphilb{H}}^2}
>
0.$$
By Theorem \ref{expedition0011361}, the van Hove Hamiltonian has a ground state.
In particular, by the above argument, $\Psi_T$ has a nonzero limit $\Psi_{\infty}$.

(Ground state energy evaluation): By the preceding argument, $1
\in \fun{\lp^{2}}{\prbqspace_{\sphilb{H}},\msr{P}_{\sphilb{H}}}$ satisfies the hypotheses of Theorem \ref{expedition0004996}.
The ground state energy is
\begin{equation}
\begin{aligned}
&\physgse
=
-\lim_{T \to \infty}
\frac{\log
\bkt{1}{\napiernum^{-T \physham} 1}_{\fun{\lp^{2}}{\prbqspace_{\sphilb{H}},\msr{P}_{\sphilb{H}}}}}
{T} \\
&=
-\lim_{T \to \infty}
\frac{\log L(T)}
{T} \\
&=
-\lim_{T \to \infty}
\rbk{
\onehalf
\norm{\omega^{\onehalf} \mathsf{m}_{\kappa,\Lambda}}
_{\sphilb{H}}^{2}
-\onehalf
\frac{1}{T}
\bkt{\mathsf{m}_{\kappa,\Lambda}}
{\rbk{1 - \napiernum^{-T \omega}} \mathsf{m}_{\kappa,\Lambda}}
_{\sphilb{H}}}
\\ 
&=
-\onehalf
\norm{\omega^{\onehalf} \mathsf{m}_{\kappa,\Lambda}}
_{\sphilb{H}}^{2}
\end{aligned}
\end{equation}

(Construction of the ground state): The limit $\Psi_{\infty}$ above is in fact a (non-normalized) ground state; we prove this below.
For any $s
> 0$, $$\napiernum^{-s \physham} \Psi_T
=
\frac{\napiernum^{-(T+s) \physham} 1}{\norm{\napiernum^{-T \physham} 1}_{\fun{\lp^{2}}{\prbqspace_{\sphilb{H}},\msr{P}_{\sphilb{H}}}}}
=
\frac{\norm{\napiernum^{-(T+s) \physham} 1}_{\fun{\lp^{2}}{\prbqspace_{\sphilb{H}},\msr{P}_{\sphilb{H}}}}}
{\norm{\napiernum^{-T \physham} 1}_{\fun{\lp^{2}}{\prbqspace_{\sphilb{H}},\msr{P}_{\sphilb{H}}}}}
\Psi_{T+s}$$ holds.
By the ground state energy evaluation, the coefficient converges to $\napiernum^{-s \physgse}$.
Therefore, in the limit $T
\to \infty$, we obtain $\napiernum^{-s \physham} \Psi_{\infty}
=
\napiernum^{-s \physgse} \Psi_{\infty}$, so $\Psi_{\infty}$ is a ground state of $\physham$.

(Expression for the ground state): Use the $\Psi_T$ constructed above and its limit.
By density, it suffices to examine $\bkt{F}{\Psi_{\infty}}$ for $F
= \napiernum^{\opfocksegal(f)}$.

(Denominator evaluation): The denominator of $\Psi_T$ before taking the limit has already been computed:
\begin{equation}
\begin{aligned}
&\norm{\napiernum^{-T \physham} 1}
_{\fun{\lp^{2}}{\prbqspace_{\sphilb{H}},\msr{P}_{\sphilb{H}}}}
=
\sqrt{L(2T)}
\\ 
&=
\bkt{1}
{\napiernum^{-2 T \physham} 1}_{\fun{\lp^{2}}{\prbqspace_{\sphilb{H}},\msr{P}_{\sphilb{H}}}}^{\onehalf}
\\ 
&=
\fnexp{
-T \physgse
-\oneoverfour
\bkt{\mathsf{m}_{\kappa,\Lambda}}
{\rbk{1 - \napiernum^{-2T \omega}} \mathsf{m}_{\kappa,\Lambda}}
_{\sphilb{H}}}
\end{aligned}
\end{equation}

(Numerator evaluation): Denoting the numerator by $N_T(f)$,
\begin{equation}
\begin{aligned}
&N_T(f)
=
\bkt{F}
{\napiernum^{-T \physham} 1}_{\fun{\lp^{2}}{\prbqspace_{\sphilb{H}},\msr{P}_{\sphilb{H}}}} \\
&=
\sqfun{\prbexp_{\msr{P}_{\txteuclid}}}
{\napiernum^{\opfocksegal_{\txteuclid}(j_0 f)
-\opfocksegal_{\txteuclid}(W_{\kappa,\Lambda}(T))}} \\
&=
\fnexp{\oneoverfour
\norm{j_0 f
-W_{\kappa,\Lambda}(T)}
_{\sphilb{H}_{\omega,-1}}^2} \\
&=
\fnexp{\oneoverfour \norm{j_0 f}_{\sphilb{H}_{\omega,-1}}^2
+\oneoverfour \norm{W_{\kappa,\Lambda}(T)}_{\sphilb{H}_{\omega,-1}}^2
-\onehalf \bkt{j_0 f}{W_{\kappa,\Lambda}(T)}_{\sphilb{H}_{\omega,-1}}}
\end{aligned}
\end{equation}
can be written.
The first norm, combined with the isometry property of $j_0$, gives $$\oneoverfour
\norm{j_0 f}_{\sphilb{H}_{\omega,-1}}^2
=
\oneoverfour
\norm{f}_{\sphilb{H}}^2.$$
The second norm, by Proposition \ref{expedition0011362}, is
\begin{equation}
\begin{aligned}
&\oneoverfour
\norm{W_{\kappa,\Lambda}(T)}_{\sphilb{H}_{\omega,-1}}^2
=
-T \physgse
-\onehalf
\bkt{\mathsf{m}_{\kappa,\Lambda}}
{\rbk{1 - \napiernum^{-T \omega}} \mathsf{m}_{\kappa,\Lambda}}
_{\sphilb{H}}
\end{aligned}
\end{equation}
The third inner product is
\begin{equation}
\begin{aligned}
&-\onehalf
\bkt{j_0 f}{W_{\kappa,\Lambda}(T)}
_{\sphilb{H}_{\omega,-1}}
\\ 
&=
-\onehalf
\int_{0}^{T}
\bkt{f}
{\napiernum^{-\abs{t} \omega} \cdot \omega \mathsf{m}_{\kappa,\Lambda}}_{\sphilb{H}}
\opdmsr{t} \\
&=
-\onehalf
\bkt{f}
{\rbk{1 - \napiernum^{-T \omega}}
\mathsf{m}_{\kappa,\Lambda}}_{\sphilb{H}}
\end{aligned}
\end{equation}
Collecting the terms,
\begin{equation}
\begin{aligned}
&N_T(f)
=
\fnexp{\oneoverfour
\norm{f}_{\sphilb{H}}^2
-\onehalf \bkt{f}
{\rbk{1 - \napiernum^{-T \omega}}
\mathsf{m}_{\kappa,\Lambda}}_{\sphilb{H}}} \\
&\qquad\times
\fnexp{-T \physgse
-\onehalf
\bkt{\mathsf{m}_{\kappa,\Lambda}}
{\rbk{1 - \napiernum^{-T \omega}} \mathsf{m}_{\kappa,\Lambda}}
_{\sphilb{H}}}
\end{aligned}
\end{equation}
is obtained.

Substituting the above results into $\bkt{F}
{\Psi_T}_{\fun{\lp^{2}}{\prbqspace_{\sphilb{H}}, \msr{P}_{\sphilb{H}}}}
=
\frac{N_T(f)}{\sqrt{L(2T)}}$,
\begin{equation}
\begin{aligned}
&\bkt{\napiernum^{\opfocksegal(f)}}
{\Psi_T}_{\fun{\lp^{2}}{\prbqspace_{\sphilb{H}}, \msr{P}_{\sphilb{H}}}} \\
&=
\fnexp{\oneoverfour \norm{f}_{\sphilb{H}}^2
-\onehalf \bkt{f}{\rbk{1 - \napiernum^{-T \omega}} \mathsf{m}
_{\kappa,\Lambda}}_{\sphilb{H}}} \\
&\qquad\times
\fnexp{-\oneoverfour
\norm{\rbk{1 - \napiernum^{-T \omega}} \mathsf{m}_{\kappa,\Lambda}}
_{\sphilb{H}}^2}
\end{aligned}
\end{equation}
is obtained.
Noting the Wick product $\fnexp{\oneoverfour
\norm{\mathsf{m}_{\kappa,\Lambda}}_{\sphilb{H}}^2}
_{\msr{P}_{\sphilb{H}}}
\napiernum^{-\opfocksegal(\mathsf{m}_{\kappa,\Lambda})}
=
\wick{\napiernum^{-\opfocksegal(\mathsf{m}_{\kappa,\Lambda})}}$ and taking the limit $T
\to \infty$,
\begin{equation}
\begin{aligned}
&\bkt{\napiernum^{\opfocksegal(f)}}{\Psi_{\infty}}
=
\fnexp{\oneoverfour \norm{f}_{\sphilb{H}}^2
-\onehalf \bkt{f}{\mathsf{m}_{\kappa,\Lambda}}_{\sphilb{H}}
-\oneoverfour \norm{\mathsf{m}_{\kappa,\Lambda}}_{\sphilb{H}}^2} \\
&=
\fnexp{\oneoverfour \norm{f}_{\sphilb{H}}^2
+\onehalf
\bkt{f}
{-\mathsf{m}_{\kappa,\Lambda}}_{\sphilb{H}}}
\fnexp{-\oneoverfour
\norm{\mathsf{m}_{\kappa,\Lambda}}_{\sphilb{H}}^2}
\\ 
&=
\bkt{\napiernum^{\opfocksegal(f)}}
{\wick{\napiernum^{-\opfocksegal(\mathsf{m}_{\kappa,\Lambda})}}
_{\msr{P}_{\sphilb{H}}}}
\fnexp{-\oneoverfour
\norm{\mathsf{m}_{\kappa,\Lambda}}_{\sphilb{H}}^2}
\end{aligned}
\end{equation}
is obtained.
In particular, the function with normalized expectation value as $\Psi_{\infty}$ is $\wick{\napiernum^{-\opfocksegal(\mathsf{m}_{\kappa,\Lambda})}}
_{\msr{P}_{\sphilb{H}}}$.

For a ground state, the norm in $\fun{\lp^{2}}{\prbqspace_{\sphilb{H}},\msr{P}_{\sphilb{H}}}$ must be normalized.
In particular, the ground state is
\begin{equation}
\begin{aligned}
&\log \norm{\wick{\napiernum^{-\opfocksegal(\mathsf{m}_{\kappa,\Lambda})}}
_{\msr{P}_{\sphilb{H}}}}
_{\fun{\lp^{2}}{\prbqspace_{\sphilb{H}},\msr{P}_{\sphilb{H}}}}
\\ 
&=
\onehalf
\log
\sqfun{\prbexp_{\msr{P}_{\sphilb{H}}}}
{\wick{\napiernum^{-\opfocksegal(\mathsf{m}_{\kappa,\Lambda})}}
_{\msr{P}_{\sphilb{H}}}
\cdot
\wick{\napiernum^{-\opfocksegal(\mathsf{m}_{\kappa,\Lambda})}}
_{\msr{P}_{\sphilb{H}}}}
\\ 
&=
\oneoverfour
\bkt{\mathsf{m}_{\kappa,\Lambda}}
{\mathsf{m}_{\kappa,\Lambda}}
_{\sphilb{H}}
=
\oneoverfour
\norm{\mathsf{m}_{\kappa,\Lambda}}_{\sphilb{H}}^2
\end{aligned}
\end{equation}

it suffices to normalize using these, giving $$\Psi_{\txtgs,\msr{P}_{\sphilb{H}},\kappa,\Lambda}
=
\fnexp{-\oneoverfour
\norm{\mathsf{m}_{\kappa,\Lambda}}_{\sphilb{H}}^2}
\wick{\napiernum^{-\opfocksegal(\mathsf{m}_{\kappa,\Lambda})}}
_{\msr{P}_{\sphilb{H}}}.$$
\end{proof}

\begin{thm}\label{expedition0011395}
Under infrared-singular conditions, there is no ground state.
\end{thm}

\begin{proof}
Let $\varrho_{\Lambda}$ be the source with the infrared cutoff removed.
Imposing the infrared-singular condition, the preceding argument gives $$\lim_{T \to \infty}
\gamma_1(T)
=
\fnexp{-\oneoverfour
\norm{\mathsf{m}_{\Lambda}}_{\sphilb{H}}^2}
=
0.$$
This is equivalent to the non-existence of a ground state.
\end{proof}

\subsection{Preparation of Measures toward Removal of Infrared/Ultraviolet Cutoffs}\label{preparation-of-measures-toward-removal-of-infraredultraviolet-cutoffs}

Using the ground state \(\Psi_{\txtgs,\msr{P}_{\sphilb{H}},\kappa,\Lambda}\) from Theorem \ref{expedition0011395}, define the measure \(\msr{P_{\txtgs,\kappa,\Lambda}}\) by \[\msr{P_{\txtgs,\kappa,\Lambda}}
=
\Psi_{\txtgs,\msr{P}_{\sphilb{H}},\kappa,\Lambda}^2
\msr{P}_{\sphilb{H}}.\] This is equivalent to the ground state transformation of \cite[Chapter 5, P.450]{LorincziHiroshimaBetz2}.

\begin{prop}\label{expedition0011492}
\begin{enumerate}
\item
The measures $\msr{P_{\txtgs,\kappa,\Lambda}}$ and $\msr{P}_{\sphilb{H}}$ are equivalent, and the function $\Psi_{\txtgs,\msr{P}_{\sphilb{H}},\kappa,\Lambda}^2$ is the Radon-Nikodym derivative $\od{\msr{P_{\txtgs,\kappa,\Lambda}}}
{\msr{P}_{\sphilb{H}}}$.

\item
The measure $\msr{P_{\txtgs,\kappa,\Lambda}}$ is a Gaussian measure with mean and covariance
\begin{equation}
\begin{aligned}
\sqfun{\prbexp_{\msr{P_{\txtgs,\kappa,\Lambda}}}}
{\opfocksegal(f)}
&=
-\bkt{f}{\mathsf{m}_{\kappa,\Lambda}}_{\sphilb{H}}, \\
\sqfun{\prbcov_{\msr{P_{\txtgs,\kappa,\Lambda}}}}{\opfocksegal(f) \opfocksegal(g)}
&=
\onehalf
\bkt{f}{g}
_{\sphilb{H}}
\end{aligned}
\end{equation}
respectively.

\item
The Wick products with respect to these measures coincide.
In particular, $$\wick{\napiernum^{\fun{\opfocksegal}{f}}}
_{\msr{P}_{\sphilb{H}}}
=
\wick{\napiernum^{\fun{\opfocksegal}{f}}}
_{\msr{P_{\txtgs,\kappa,\Lambda}}}$$ holds.

\item
Regarded as a multiplication operator, the ground state gives a unitary transformation $$\Psi_{\txtgs,\msr{P}_{\sphilb{H}},\kappa,\Lambda}
\colon \fun{\lp^{2}}{\prbqspace_{\sphilb{H}},\msr{P}_{\sphilb{H}}}
\to \fun{\lp^{2}}{\prbqspace_{\sphilb{H}},\msr{P_{\txtgs,\kappa,\Lambda}}}.$$
\end{enumerate}
\end{prop}

\begin{proof}
(1): The function $\Psi_{\txtgs,\msr{P}_{\sphilb{H}},\kappa,\Lambda}^2$ is strictly positive.
In particular, this function and the multiplication operator it defines are invertible.
Therefore $\msr{P_{\txtgs,\kappa,\Lambda}}$ and $\msr{P}_{\sphilb{H}}$ are equivalent measures, and $\Psi_{\txtgs,\msr{P}_{\sphilb{H}},\kappa,\Lambda}^2$ is the Radon-Nikodym derivative $\od{\msr{P_{\txtgs,\kappa,\Lambda}}}
{\msr{P}_{\sphilb{H}}}$.

(2): Compute the expectation value and derivatives directly for the Wick products $F
= \wick{\napiernum^{\fun{\opfocksegal}{sf}}}
_{\msr{P}_{\sphilb{H}}}$ and $G
= \wick{\napiernum^{\fun{\opfocksegal}{tg}}}
_{\msr{P}_{\sphilb{H}}}$.

(3): Follows from the fact that the Wick products for Gaussian measures with the same variance but different means coincide.

(4): For any $F,G
\in \fun{\lp^{2}}{\prbqspace_{\sphilb{H}},\msr{P}_{\sphilb{H}}}$,
\begin{equation}
\begin{aligned}
&\bkt{\Psi_{\txtgs,\msr{P}_{\sphilb{H}},\kappa,\Lambda} F}
{\Psi_{\txtgs,\msr{P}_{\sphilb{H}},\kappa,\Lambda} G}_{\fun{\lp^{2}}{\prbqspace_{\sphilb{H}},\msr{P}_{\sphilb{H}}}}
\\ 
&=
\sqfun{\prbexp_{\msr{P}_{\sphilb{H}}}}
{\cmpconj{F} G \Psi_{\txtgs,\msr{P}_{\sphilb{H}},\kappa,\Lambda}^2}
=
\sqfun{\prbexp_{\msr{P}_{\sphilb{H}}}}
{\cmpconj{F} G
\od{\msr{P_{\txtgs,\kappa,\Lambda}}}
{\msr{P}_{\sphilb{H}}}}
\\ 
&=
\sqfun{\prbexp_{\msr{P_{\txtgs,\kappa,\Lambda}}}}
{\cmpconj{F} G}
=
\bkt{F}{G}_{\fun{\lp^{2}}{\prbqspace_{\sphilb{H}},\msr{P_{\txtgs,\kappa,\Lambda}}}}
\end{aligned}
\end{equation}
holds.
\end{proof}

As in the operator algebra discussion, define the ground state \(\psi_{\txtvanhowe,\txtgs,\kappa,\Lambda}\) as a functional. Using the ground state found in Theorem \ref{expedition0011360}, \[\Psi_{\txtgs,\msr{P}_{\sphilb{H}},\kappa,\Lambda}
=
\fnexp{-\oneoverfour
\norm{\mathsf{m}_{\kappa,\Lambda}}_{\sphilb{H}}^2}
\wick{\napiernum^{-\opfocksegal(\mathsf{m}_{\kappa,\Lambda})}}
_{\msr{P}_{\sphilb{H}}}
\in
\fun{\lp^{2}}{\prbqspace_{\sphilb{H}},\msr{P}_{\sphilb{H}}}\] define the linear functional \(\psi_{\txtvanhowe,\txtgs,\kappa,\Lambda}\) for linear operators constructed from the Segal field operator by \[\psi_{\txtvanhowe,\txtgs,\kappa,\Lambda}(A)
=
\bkt{\Psi_{\txtgs,\msr{P}_{\sphilb{H}},\kappa,\Lambda}}
{A
\Psi_{\txtgs,\msr{P}_{\sphilb{H}},\kappa,\Lambda}}
_{\fun{\lp^{2}}{\prbqspace_{\sphilb{H}},\msr{P}_{\sphilb{H}}}}.\] This linear functional \(\psi_{\txtvanhowe,\txtgs,\kappa,\Lambda}\) is also called a ground state of the van Hove model.

\begin{prop}\label{expedition0011494}
The ground state as a functional can also be written as $$\psi_{\txtvanhowe,\txtgs,\kappa,\Lambda}(A)
=
\bkt{1}{A1}
_{\fun{\lp^{2}}{\prbqspace_{\sphilb{H}},\msr{P_{\txtgs,\kappa,\Lambda}}}}.$$
\end{prop}

\begin{proof}
Verify using Wick products; this is in fact the definition of the measure $\msr{P_{\txtgs,\kappa,\Lambda}}$.
\end{proof}

We now remove the cutoffs, first the infrared cutoff and then the ultraviolet cutoff. Since the ground state from Theorem \ref{expedition0011360} loses its meaning after removing the infrared cutoff, we proceed by making good use of the ground state as a functional, or the function \(1\) as ground state in the \(\msr{P_{\txtgs,\kappa,\Lambda}}\)-space of Proposition \ref{expedition0011494}. First, as the measures corresponding to the cutoff-removal limits, define: the limit \(\kappa
\to 0\) of \(\msr{P_{\txtgs,\kappa,\Lambda}}\) as \(\msr{P_{\txtgs,\Lambda}}\), the limit \(\Lambda
\to \infty\) as \(\msr{P_{\txtgs,\kappa}}\), and the double limit as \(\msr{P_{\txtirsingular}}\).

\begin{defn}\label{expedition0011500}
Let $\tilde{\physham}_{\txtvanhowe,\kappa,\Lambda}$ be the non-negative self-adjoint operator shifted by the ground state energy: $$\tilde{\physham}_{\txtvanhowe,\kappa,\Lambda}
=
\physham_{\txtvanhowe,\kappa,\Lambda}
-\physgse(\physham_{\txtvanhowe,\kappa,\Lambda}).$$
Using the ground state $\Psi_{\txtgs,\msr{P}_{\sphilb{H}},\kappa,\Lambda}$ from Theorem \ref{expedition0011360}, define the unitary operator $\mathsf{U}_{\kappa,\Lambda}$ by $$\mathsf{U}_{\kappa,\Lambda}
\colon \fun{\lp^{2}}{\prbqspace_{\sphilb{H}},\msr{P_{\txtgs,\kappa,\Lambda}}}
\to \fun{\lp^{2}}{\prbqspace_{\sphilb{H}},\msr{P}_{\sphilb{H}}};
\quad
\mathsf{U}_{\kappa,\Lambda} F
=
\Psi_{\txtgs,\msr{P}_{\sphilb{H}},\kappa,\Lambda}
F,$$ and use it to define $$L_{\kappa,\Lambda}
=
\inv{\mathsf{U}_{\kappa,\Lambda}}
\tilde{\physham}_{\txtvanhowe,\kappa,\Lambda}
\mathsf{U}_{\kappa,\Lambda}
=
\inv{\Psi_{\txtgs,\msr{P}_{\sphilb{H}},\kappa,\Lambda}}
\tilde{\physham}_{\txtvanhowe,\kappa,\Lambda}
\Psi_{\txtgs,\msr{P}_{\sphilb{H}},\kappa,\Lambda}.$$
This unitary transformation is called the ground state transformation.
\end{defn}

Let us examine the properties of the self-adjoint operator \(\physham[L]_{\kappa,\Lambda}\).

\begin{prop}
\begin{enumerate}
\item
The operator $\physham[L]_{\kappa,\Lambda}$ is a self-adjoint operator on $\fun{\lp^{2}}{\prbqspace_{\sphilb{H}},\msr{P_{\txtgs,\kappa,\Lambda}}}$.
In particular, it can be written as $$\physham[L]_{\kappa,\Lambda}
=
\physham_{\txtbsn,\txtfr}
+\bkt{\omega \mathsf{m}_{\kappa,\Lambda}}{f}
_{\sphilb{H}}.$$

\item
The ground state of $\physham[L]_{\kappa,\Lambda}$ is $1
\in \fun{\lp^{2}}{\prbqspace_{\sphilb{H}},\msr{P_{\txtgs,\kappa,\Lambda}}}$, with ground state energy $0$.

\item
The operator $\physham[L]_{\kappa,\Lambda}$ generates a heat semigroup.

In particular, $$\inv{\mathsf{U}_{\kappa,\Lambda}}
\napiernum^{-t \physham_{\txtvanhowe,\kappa,\Lambda}}
\mathsf{U}_{\kappa,\Lambda}
=
\napiernum^{-t \physgse(\physham_{\txtvanhowe,\kappa,\Lambda})}
\napiernum^{-t \physham[L]_{\kappa,\Lambda}}$$ holds.

\item

The action on the constant function $1$ is $\napiernum^{-t \physham[L]_{\kappa,\Lambda}} 1
= 1$.
\end{enumerate}
\end{prop}

\begin{proof}
(1), Preparation: Self-adjointness is preserved by the unitary property of the operator $\mathsf{U}$.
The explicit form of the operator can be found by examining the action on $F
= \wick{\napiernum^{\fun{\opfocksegal}{f}}}
_{\msr{P}_{\sphilb{H}}}$ expressed as a Wick product.
In particular, it suffices to show $\mathsf{U}_{\kappa,\Lambda}
\physham[K]
=
\tilde{\physham}_{\txtvanhowe,\kappa,\Lambda}
\mathsf{U}_{\kappa,\Lambda}$ for $\physham[K]
=
\physham_{\txtbsn,\txtfr}
+\bkt{\omega \mathsf{m}_{\kappa,\Lambda}}{f}
_{\sphilb{H}}$.
Define the coefficient $$C(f)
=
\fnexp{-\oneoverfour
\norm{\mathsf{m}_{\kappa,\Lambda}}_{\sphilb{H}}^2
+\onehalf
\bkt{f}{-\mathsf{m}_{\kappa,\Lambda}}_{\sphilb{H}}}$$ which will appear later.

(1), Action of $\tilde{\physham}_{\txtvanhowe,\kappa,\Lambda} \mathsf{U}_{\kappa,\Lambda}$:
By Proposition \ref{expedition0011392}, computing $\physham_{\txtvanhowe,\kappa,\Lambda}
\rbk{\Psi_{\txtgs,\msr{P}_{\sphilb{H}},\kappa,\Lambda} F}$,
\begin{equation}
\begin{aligned}
&\physham_{\txtvanhowe,\kappa,\Lambda}
\rbk{\Psi_{\txtgs,\msr{P}_{\sphilb{H}},\kappa,\Lambda} F}
\\ 
&=
C(f)
\funrbk{\physham_{\txtbsn,\txtfr} + \opfocksegal(\omega \mathsf{m}_{\kappa,\Lambda})}
{\wick{\napiernum^{\fun{\opfocksegal}{f - \mathsf{m}_{\kappa,\Lambda}}}}
_{\msr{P}_{\sphilb{H}}}}
\\ 
&=
C(f)
\rbkleft{\wick{\opfocksegal(\omega (f - \mathsf{m}_{\kappa,\Lambda}))
\napiernum^{\fun{\opfocksegal}{f - \mathsf{m}_{\kappa,\Lambda}}}}
_{\msr{P}_{\sphilb{H}}}}
\\
&\qquad
+\wick{\opfocksegal(\omega \mathsf{m}_{\kappa,\Lambda})
\napiernum^{\fun{\opfocksegal}{f - \mathsf{m}_{\kappa,\Lambda}}}}
_{\msr{P}_{\sphilb{H}}}
\\
&\qquad
+\rbkright{\onehalf
\bkt{\omega \mathsf{m}_{\kappa,\Lambda}}
{f - \mathsf{m}_{\kappa,\Lambda}}
_{\sphilb{H}}
\wick{\napiernum^{\fun{\opfocksegal}{f - \mathsf{m}_{\kappa,\Lambda}}}}
_{\msr{P}_{\sphilb{H}}}}
\\ 
&=
C(f)
\rbkleft{\wick{\opfocksegal(\omega f)
\napiernum^{\fun{\opfocksegal}{f - \mathsf{m}_{\kappa,\Lambda}}}}
_{\msr{P}_{\sphilb{H}}}}
\\
&\qquad
+\rbkright{\onehalf
\bkt{\omega \mathsf{m}_{\kappa,\Lambda}}
{f - \mathsf{m}_{\kappa,\Lambda}}
_{\sphilb{H}}
\wick{\napiernum^{\fun{\opfocksegal}{f - \mathsf{m}_{\kappa,\Lambda}}}}
_{\msr{P}_{\sphilb{H}}}}
\end{aligned}
\end{equation}
is obtained.
Therefore,
\begin{equation}
\begin{aligned}
&\tilde{\physham}_{\txtvanhowe,\kappa,\Lambda}
\rbk{\Psi_{\txtgs,\msr{P}_{\sphilb{H}},\kappa,\Lambda} F}
\\ 
&=
C(f)
\rbk{\wick{\opfocksegal(\omega f)
\napiernum^{\fun{\opfocksegal}{f - \mathsf{m}_{\kappa,\Lambda}}}}
_{\msr{P}_{\sphilb{H}}}
+\onehalf
\bkt{\omega \mathsf{m}_{\kappa,\Lambda}}
{f}
_{\sphilb{H}}
\wick{\napiernum^{\fun{\opfocksegal}{f - \mathsf{m}_{\kappa,\Lambda}}}}
_{\msr{P}_{\sphilb{H}}}}
\end{aligned}
\end{equation}
holds.

(1), Action of $\mathsf{U}_{\kappa,\Lambda} \physham[K]$: Likewise, by Proposition \ref{expedition0011392}, $$\physham[K] F
=
\wick{\opfocksegal(\omega f)
\napiernum^{\fun{\opfocksegal}{f}}}
_{\msr{P}_{\sphilb{H}}}
+\bkt{f}
{\omega \mathsf{m}_{\kappa,\Lambda}}
_{\sphilb{H}}
\wick{\napiernum^{\fun{\opfocksegal}{f}}}
_{\msr{P}_{\sphilb{H}}}$$ holds.
Applying $\mathsf{U}_{\kappa,\Lambda}$ to the first term.
It suffices to compute $\fnrestr{\opod{t}
\mathsf{U}_{\kappa,\Lambda}
\wick{\napiernum^{\fun{\opfocksegal}{f + t \omega f}}}
_{\msr{P}_{\sphilb{H}}}}
{t=0}$, and the result is
\begin{equation}
\begin{aligned}
&\mathsf{U}_{\kappa,\Lambda}
\wick{\opfocksegal(\omega f)
\napiernum^{\fun{\opfocksegal}{f}}}
_{\msr{P}_{\sphilb{H}}}
\\ 
&=
C(f)
\rbk{\wick{\opfocksegal(\omega f)
\napiernum^{\fun{\opfocksegal}{f - \mathsf{m}_{\kappa,\Lambda}}}}
_{\msr{P}_{\sphilb{H}}}
-\onehalf
\bkt{\omega \mathsf{m}_{\kappa,\Lambda}}{f}_{\sphilb{H}}
\wick{\napiernum^{\fun{\opfocksegal}{f}}}
_{\msr{P}_{\sphilb{H}}}}
\end{aligned}
\end{equation}
is obtained.

The action of $\mathsf{U}_{\kappa,\Lambda}$ on the second term is $$C(f)
\bkt{f}
{\omega \mathsf{m}_{\kappa,\Lambda}}
_{\sphilb{H}}
\wick{\napiernum^{\fun{\opfocksegal}{f}}}
_{\msr{P}_{\sphilb{H}}}.$$

Taking their sum,
\begin{equation}
\begin{aligned}
&\mathsf{U}_{\kappa,\Lambda}
\physham[K] F
\\ 
&=
C(f)
\rbk{\wick{\opfocksegal(\omega f)
\napiernum^{\fun{\opfocksegal}{f - \mathsf{m}_{\kappa,\Lambda}}}}
_{\msr{P}_{\sphilb{H}}}
+\onehalf
\bkt{\omega \mathsf{m}_{\kappa,\Lambda}}{f}_{\sphilb{H}}
\wick{\napiernum^{\fun{\opfocksegal}{f}}}
_{\msr{P}_{\sphilb{H}}}}
\end{aligned}
\end{equation}
is obtained.
This agrees with the result from the previous step.

(2): By direct computation from the definition: $$\physham[L]_{\kappa,\Lambda} 1
=
\inv{\Psi_{\txtgs,\msr{P}_{\sphilb{H}},\kappa,\Lambda}}
\tilde{\physham}_{\txtvanhowe,\kappa,\Lambda}
\Psi_{\txtgs,\msr{P}_{\sphilb{H}},\kappa,\Lambda} 1
=
\inv{\Psi_{\txtgs,\msr{P}_{\sphilb{H}},\kappa,\Lambda}}
\cdot
0
=
0.$$

(3): It suffices to show the semigroup property, which follows from $$\napiernum^{-t \physham[L]_{\kappa,\Lambda}}
=
\inv{\mathsf{U}}
\napiernum^{-t \tilde{\physham}_{\txtvanhowe,\kappa,\Lambda}}
\mathsf{U}
=
\napiernum^{t \physgse(\physham_{\txtvanhowe,\kappa,\Lambda})}
\inv{\mathsf{U}}
\napiernum^{-t \physham_{\txtvanhowe,\kappa,\Lambda}}
\mathsf{U}.$$

(4): By part (2) of this proposition.
\end{proof}

Consider the ground state \(\psi_{\txtvanhowe,\txtgs,\kappa,\Lambda}\) as a functional. Evaluating the van Hove Hamiltonian with the ground state \(\Psi_{\txtgs,\msr{P}_{\sphilb{H}},\kappa,\Lambda}\) from Theorem \ref{expedition0011360}, \[\psi_{\txtvanhowe,\txtgs,\kappa,\Lambda}
(\physham_{\txtvanhowe,\kappa,\Lambda})
=
\physgse(\physham_{\txtvanhowe,\kappa,\Lambda})\] holds. When computing with the ground state \(1
\in \fun{\lp^{2}}{\prbqspace_{\sphilb{H}},\msr{P_{\txtgs,\kappa,\Lambda}}}\) of Definition \ref{expedition0011494}, formally apply the above evaluation as well. When computing with this ground state \(1\), of course \(\physham[L]_{\kappa,\Lambda}\) and the ground state energy appear via the appropriate unitary transformation, and the evaluation can be carried out properly.

\subsection{Projective System of Probability Measures}\label{expedition0011536}

We develop the theory of projective systems of probability measures as a fundamental concept toward removing the infrared/ultraviolet cutoffs.

\begin{defn}\label{expedition0011527}
Choose a directed set $\pairbk{I,\leq}$ as the index set and set up objects as follows.
\begin{enumerate}
\item
For each $i
\in I$, assign a Polish space $X_i$ and a probability measure $\msr{P_i}$.

\item
Under the condition $i
\leq j$, assign continuous maps $\pi_{ij}
\colon X_j
\to X_i$ satisfying the projection conditions $$\pi_{ii}
=
\id_{X_i},
\quad
\pi_{ik}
=
\pi_{ij} \circ \pi_{jk}
\quad
\rbk{i \leq j \leq k}.$$

\item

The family of probability measures is consistent, i.e., under the condition $i
\leq j$, $$\msr{P_i}
= \msr{P_j} \circ \inv{\pi_{ij}}$$ holds.
\end{enumerate}

In this case, the family $\pairbk{(X_i,\msr{P_i}), \pi_{ij}}$ is called a projective system of probability measures.
For this projective system, define the projective limit of sets $$X
=
\projlim_{i \in I} X_i
=
\set{\fml{x_i}{i \in I} \in \prod_{i \in I} X_i}
{\pi_{ij}(x_j) = x_i (i \leq j)}.$$
Define the natural sigma-algebra on the product space $\prod_{i \in I} X_i$, and take its intersection with $X$ as the sigma-algebra on $X$.
\end{defn}

\begin{thm}\label{expedition0011529}
There exists a family of projections $\fml{\pi_i}{i \in I}$ and a unique probability measure $\msr{P}$ on $X$ satisfying, for all $i
\in I$, $$\msr{P_i}
=
\msr{P} \circ \inv{\pi_i},
\quad
\pi_i
\colon X
\to X_i;
\quad
\fml{x_j}{j \in I}
\mapsto
x_i.$$
\end{thm}

\begin{proof}
This is essentially Kolmogorov's extension theorem; the proof is omitted.
\end{proof}

\begin{defn}\label{expedition0011530}
When the projective system $\pairbk{(X_i,\msr{P_i}), \pi_{ij}}$ satisfies the following condition, it is called tight as a projective system, or projectively tight.

For any $\ep
> 0$, there exists a family of compact sets $\fml{K_i}{i \in I}$ satisfying:
\begin{enumerate}
\item
For each $i
\in I$, $K_i
\subset X_i$ is a compact set.

\item
If $i
\leq j$, then $\pi_{ij}(K_j)
\subset K_i$, i.e., projective consistency holds.
In particular, $\fml{K_i}{i \in I}$ itself forms a projective system of compact sets.

\item
Uniform concentration of mass: $$\inf_{i \in I}
\msr{P_i}(K_i)
> 1 - \ep.$$
\end{enumerate}
\end{defn}

\begin{prop}\label{expedition0011531}
Assume the following conditions for each $i
\in I$:
\begin{enumerate}
\item
For each $i
\in I$, there exists a subset $C_i
\subset X_i$ satisfying $\msr{P_i}(C_i)
= 1$.

\item
The family $\fml{C_i}{i \in I}$ is projectively consistent, i.e., $\pi_{ij}(C_j)
\subset C_i$ for $i
\leq j$.
\end{enumerate}
Under these conditions, define the subset of the projective limit $X$: $$C
=
\set{\fml{x_i}{i \in I} \in X}
{\fun{\opforall_{i \in I}}{x_i \in C_i}}.$$
Then the projective limit measure $\msr{P}$ whose existence is guaranteed by Kolmogorov's extension theorem \ref{expedition0011529} satisfies $\msr{P}(C)
= 1$.
This property is called preservation of path continuity.
\end{prop}

\begin{proof}
The complement of $C$ is $$\setcpl{C}
=
\set{\fml{x_i}{i \in I} \in X}
{\text{there exists $i \in I$ with $x_i \notin C_i$}}.$$
In particular, $\setcpl{C}
\subset
\bigcup_{i \in I}
\fun{\inv{\pi_i}}{\setcpl{C_i}}$ holds.
By monotonicity of measure, $$\msr{P}(\setcpl{C})
\leq
\fun{\msr{P}}
{\bigcup_{i \in I}
\fun{\inv{\pi_i}}{\setcpl{C_i}}}
\leq
\sup_{i \in I}
\fun{\msr{P}}{\fun{\inv{\pi_i}}{\setcpl{C_i}}}
=
\sup_{i \in I}
\fun{\msr{P_i}}{\setcpl{C_i}}
=
0$$ is obtained.
\end{proof}

We use the result of Kolmogorov's extension theorem \ref{expedition0011529} for projective systems of probability measures. A map \(F
\colon X
\to \fldreal\) satisfying the following condition is called a simply projective function.

\begin{itemize}

\item
  There exists \(i
  \in I\) and a function \(f_i
  \colon X_i
  \to \fldreal\) such that \(F
  = f_i \circ \pi_i\).
\end{itemize}

Depending on the properties of \(f_i\), one also has simply projective bounded functions, simply projective continuous functions, and simply projective bounded continuous functions.

This property also affects weak convergence of probability measures. By the definition of the product topology, a continuous function on the product space \(\prod_{i \in I} X_i\) depends on only finitely many \(X_i\). In particular, a continuous function \(f
\colon X
\to \fldreal\) on \(X\) can always be expressed as \(f(x)
= \fun{g}{\fml{x_i}{i \in F}}\) for some continuous function \(g\) on \(\prod_{i \in F} X_i\) for a suitable finite set \(F
\subset I\).

A projective system of probability measures \(\fml{\msr{P_i}}{i \in I}\) is said to converge weakly to the projective limit measure \(\msr{P}\) when:

\begin{itemize}

\item
  For any simply projective bounded continuous function \(F
  \colon X
  \to \fldreal\), \[\lim_{i \in I}
  \int_{X_i}
  f_i(x)
  \opdmsr{P_i(x)}
  =
  \int_X
  F(x)
  \opdmsr{P(x)}\] holds.
\end{itemize}

Since continuous functions always depend on only finitely many \(X_i\), it suffices to consider simply projective continuous functions depending on a single \(i\).

\begin{thm}\label{expedition0011535}
Suppose the projective system of probability measures further satisfies the following two conditions.
\begin{enumerate}
\item
Tightness as a projective system, i.e., the condition of Definition \ref{expedition0011530}.

\item
Preservation of path continuity, i.e., the condition of Proposition \ref{expedition0011531}.
\end{enumerate}
Then the projective limit of spaces $X
=
\projlim_{i \in I}
X_i$, there exist a family of projections $\fml{\pi_i}{i \in I}$ and a unique probability measure $\msr{P}$ on $X$ satisfying, for all $i
\in I$, $$\msr{P_i}
=
\msr{P} \circ \inv{\pi_i},
\quad
\pi_i
\colon X
\to X_i;
\quad
\fml{x_j}{j \in I}
\mapsto
x_i,$$ and the family of probability measures $\msr{P_i}$ converges weakly to $\msr{P}$ in the sense defined above.
\end{thm}

\begin{proof}
The existence of the family of projections and the unique probability measure on $X$ follows from Kolmogorov's extension theorem \ref{expedition0011529} for projective systems.
It remains to show weak convergence.
Let $F
= f_{i_0} \circ \pi_{i_0}$ be a simply projective bounded continuous function.

(Integral with respect to the projective limit measure): By Kolmogorov's extension theorem \ref{expedition0011529}, the integral of this $F$ can be written as
\begin{equation}
\begin{aligned}
&\int_X F \opdmsr{P}
=
\int_X f_{i_0} \circ \pi_{i_0} \opdmsr{\msr{P}}
=
\int_{X_{i_0}} f_{i_0} \opdmsr{\rbk{\msr{P} \circ \inv{\pi_{i_0}}}}
\\ 
&=
\int_{X_{i_0}} f_{i_0} \opdmsr{\msr{P}_{i_0}}
\end{aligned}
\end{equation}
can be written.

(Integral of pullback for $i \geq i_0$): By consistency of the projective system, for $i
\geq i_0$, $\msr{P_{i_0}}
=
\msr{P_i} \circ \inv{\pi_{i_0 i}}$ holds.
The integral can be written as $$\int_{X_i}
f_{i_0} \circ \pi_{i_0 i}
\opdmsr{\msr{P_i}}
=
\int_{X_{i_0}}
f_{i_0}
\opdmsr{\msr{P_{i_0}}}.$$
In particular, the integral of the pullback $f_{i_0} \circ \pi_{i_0 i}$ is independent of $i$.

(Definition of the $i$-th test function corresponding to $F$): Lift the simply projective bounded continuous function $F
= f_{i_0} \circ \pi_{i_0}$ to $F_i
= f_{i_0} \circ \pi_{i_0 i}
\in \fun{\conti_{\txtbdd}}{X_i}$ for $i
\geq i_0$.
Then the preceding pullback integral can be written as $$\int_{X_i} F_i \opdmsr{\msr{P_i}}
=
\int_{X_{i_0}}
f_{i_0}
\opdmsr{\msr{P_{i_0}}}.$$

(Weak convergence): Collecting the above two steps, for $i
\geq i_0$, $$\int_{X_i} F_i \opdmsr{\msr{P_i}}
=
\int_{X_{i_0}}
f_{i_0}
\opdmsr{\msr{P_{i_0}}}
=
\int_X F \opdmsr{\msr{P}}$$ holds.

The left-hand side is constant for $i
\geq i_0$, and in particular $$\projlim_{i \in I}
\int_{X_i} F_i \opdmsr{\msr{P_i}}
=
\int_X F \opdmsr{\msr{P}}$$ is obtained.
This is precisely weak convergence.
\end{proof}

Use \(\dom \mathsf{m}\) as the space of test functions, and define the index set \[\setindex{I}
=
\setpowfin{\dom \mathsf{m}}.\] For any \(I
=
\setone{f_1,\cdots,f_n}
\in \setindex{I}\), abbreviate \(\fldreal^{I}
= \fldreal^{n}\).

Define the configuration space \[\prbqspace
= \fldreal^{\dom \mathsf{m}}\] and for each \(f
\in \dom \mathsf{m}\) define the evaluation map \[\pi_f
\colon \prbqspace
\to \fldreal;
\quad
\pi_f(\omega)
=
\omega(f).\] Using this, define the sigma-algebra generated by cylinder sets: \[\mblfml{C}
=
\mblfmlgenerated{\set{\inv{\pi_f(B)}}
{f \in \dom \mathsf{m}, B \in \mblfmlborel(\fldreal)}}.\]

For any \(I
= \setone{f_1,\cdots,f_n}
\in \setindex{I}\), set \[\pi_I
\colon \prbqspace
\to \fldreal^{I};
\quad
\pi_{I}(\omega)
=
\vecbk{\omega(f_1),\cdots,\omega(f_n)}.\] For a probability measure \(\msr{P}\) on the measurable space \(\pairbk{\prbqspace,\mblfml{C}}\), define its finite-dimensional marginal for \(I\) by \[\msr{P_I}
=
P \circ \inv{\pi_I}.\]

Define \(\setindex{K}\) as the set of cutoff parameters \((\kappa,\Lambda)\). For each \((\kappa,\Lambda)
\in \setindex{K}\), consider \(\msr{P_{\txtgs,\kappa,\Lambda}}\) as a probability measure on \(\pairbk{\prbqspace,\mblfml{C}}\), and for each \(I
\in \setindex{I}\) define the finite-dimensional distribution: \[\msr{P_{\txtgs,\kappa,\Lambda,I}}
=
\msr{P_{\txtgs,\kappa,\Lambda}} \circ \inv{\pi_{I}}.\]

For any \(I
= \setone{f_1,\cdots,f_n}
\in \setindex{I}\), define the finite-dimensional characteristic function: \[\prbcharfun_{\txtgs,\kappa,\Lambda,I}(\xi)
=
\int_{\fldreal^{I}}
\napiernum^{\imunit k \vainnprod x}
\opdmsr{\msr{P_{\txtgs,\kappa,\Lambda,I}}(x)},
\quad
\xi \in \fldreal^{I}.\]

We now define a concept restricting ordinary weak convergence to \(\dom \mathsf{m}\)-cylinder sets on \(\prbqspace\).

When, for any \(I
= \setone{f_1,\cdots,f_n}
\in \setpowfin{\dom \mathsf{m}}\) and any \(F
\in \fun{\conti_{\txtbdd}}{\fldreal^{I}}\), \[\sqfun{\prbexp_{\msr{P_{\txtgs,\kappa,\Lambda}}}}
{F \circ \pi_I}
\to
\sqfun{\prbexp_{\msr{P}}}
{F \circ \pi_I}\] holds, we say the family \(\fml{\msr{P_{\txtgs,\kappa,\Lambda}}}
{\kappa,\Lambda > 0}\) of probability measures converges \(\dom \mathsf{m}\)-weakly to \(\msr{P}\).

\begin{prop}\label{expedition0011554}
Let $\msr{P}$ be a probability measure on the sample space $\prbqspace$.
Then the following statements are equivalent.
\begin{enumerate}
\item
The family $\fml{\msr{P_{\txtgs,\kappa,\Lambda}}}{\kappa,\Lambda > 0}$ of cutoff measures converges $\dom \mathsf{m}$-weakly to $\msr{P}$.

\item
For any $I
\in \setpowfin{\dom \mathsf{m}}$, ordinary weak convergence on $\fldreal^{I}$ holds.
In particular, $$\wlim_{\kappa \to 0, \Lambda \to \infty}
\msr{P_{\txtgs,\kappa,\Lambda,I}}
=
\msr{P}_I$$ holds.
\end{enumerate}
\end{prop}

\begin{proof}
((1)$\Rightarrow$(2)): Fix any $I
= \setone{f_1,\cdots,f_n}
\in \setpowfin{\dom \mathsf{m}}$ and any $F
\in \fun{\conti_{\txtbdd}}{\fldreal^{I}}$. Then
\begin{equation}
\begin{aligned}
\sqfun{\prbexp_{\msr{P_{\txtgs,\kappa,\Lambda,I}}}}
{F}
=
\sqfun{\prbexp_{\msr{P_{\txtgs,\kappa,\Lambda}}}}
{F \circ \pi_I},
\quad
\sqfun{\prbexp_{\msr{P_I}}}
{F}
=
\sqfun{\prbexp_{\msr{P}}}
{F \circ \pi_I}
\end{aligned}
\end{equation}
holds.
By hypothesis, the right-hand side of the first equation converges as $\kappa
\to 0$ and $\Lambda
\to \infty$.
In particular, for all $F
\in \fun{\conti_{\txtbdd}}{\fldreal^{I}}$, $$\sqfun{\prbexp_{\msr{P_{\txtgs,\kappa,\Lambda,I}}}}
{F}
\to
\sqfun{\prbexp_{\msr{P_I}}}
{F}$$ holds.
This is precisely ordinary weak convergence on $\fldreal^{I}$.

((2)$\Rightarrow$(1)): Again fix any $I
= \setone{f_1,\cdots,f_n}
\in \setpowfin{\dom \mathsf{m}}$ and $F
\in \fun{\conti_{\txtbdd}}{\fldreal^{I}}$.
The same equations as before
\begin{equation}
\begin{aligned}
\sqfun{\prbexp_{\msr{P_{\txtgs,\kappa,\Lambda,I}}}}
{F}
=
\sqfun{\prbexp_{\msr{P_{\txtgs,\kappa,\Lambda}}}}
{F \circ \pi_I},
\quad
\sqfun{\prbexp_{\msr{P_I}}}
{F}
=
\sqfun{\prbexp_{\msr{P}}}
{F \circ \pi_I}
\end{aligned}
\end{equation}
hold.
By hypothesis, weak convergence $\msr{P_{\txtgs,\kappa,\Lambda,I}}
\Rightarrow \msr{P_I}$ on $\fldreal^{I}$ holds.
In particular, $\sqfun{\prbexp_{\msr{P_{\txtgs,\kappa,\Lambda,I}}}}
{F}
\to
\sqfun{\prbexp_{\msr{P_I}}}
{F}$.
Applying this to the equations above, for any $I$ and $F$, $$\sqfun{\prbexp_{\msr{P_{\txtgs,\kappa,\Lambda}}}}
{F \circ \pi_I}
\to
\sqfun{\prbexp_{\msr{P}}}
{F \circ \pi_I}$$ holds.
This is $\dom \mathsf{m}$-convergence.
\end{proof}

We discuss \(\dom \mathsf{m}\)-weak convergence via finite-dimensional characteristic functions.

\begin{prop}
Suppose each $\msr{P_{\txtgs,\kappa,\Lambda,I}}$ is a probability measure on $\fldreal^I$ and assume the following conditions.
\begin{enumerate}
\item
For any $I
\in \setpowfin{\dom \mathsf{m}}$ and any $\xi
\in \fldreal^I$, the limit of characteristic functions $$\prbcharfun_{\txtgs,I}(\xi)
=
\lim_{\kappa \to 0, \Lambda \to \infty}
\prbcharfun_{\txtgs,\kappa,\Lambda,I}(\xi)$$ exists.

\item
Each $\prbcharfun_{\txtgs,I}$ is a characteristic function on $\fldreal^I$.

\item
For the corresponding probability measures $\msr{\mu_I}$, projective consistency $$\msr{\mu_I} \circ \invrbk{{\pi_{J}^I}}
=
\msr{\mu_J}$$ holds.
\end{enumerate}
Then there exists a probability measure $\msr{P}$ to which the family $\fml{\msr{P_{\txtgs,\kappa,\Lambda}}}{\kappa,\Lambda > 0}$ converges $\dom \mathsf{m}$-weakly.
\end{prop}

\begin{proof}
This is a direct consequence of Kolmogorov's extension theorem \ref{expedition0011529} for projective systems of probability measures.
\end{proof}

So far we have constructed the argument on \(\dom \mathsf{m}\). We now need to take the appropriate closure.

For a symmetric positive semi-definite form \(C
\colon \dom \mathsf{m} \times \dom \mathsf{m}
\to \fldreal\) on the family of test functions \(\dom \mathsf{m}\), define the positive semi-definite inner product \[\bkt{f}{g}_{C}
=
C(f,g),
\quad
\norm{f}_C
=
\sqrt{C(f,f)}.\] The completion of the quotient space \(\setquot{\dom \mathsf{m}}
{\sphilb{N}_C}\) by \[\sphilb{N}_{C}
=
\set{f \in \dom \mathsf{m}}
{\norm{f}_C = 0}\] is denoted \(\sphilb{H}_{\txtirsingular}\) and called the \(C\)-closure of \(\dom \mathsf{m}\).

For each \(\pairbk{\kappa,\Lambda}\), define the linear operator on \(\pairbk{\prbqspace,\mblfml{C},\msr{P_{\txtgs,\kappa,\Lambda}}}\) by \[X_{\kappa,\Lambda}
\colon \dom \mathsf{m}
\to \fun{\lp^{2}}{\prbqspace,\msr{P_{\txtgs,\kappa,\Lambda}}};
\quad
f \mapsto \pi_{\kappa,\Lambda,f};
\quad
\pi_{\kappa,\Lambda,f}(\omega)
=
\omega(f).\] Assume in particular that the covariance \[C(f,g)
=
\sqfun{\prbcov_{\msr{P_{\txtgs,\kappa,\Lambda}}}}
{\pi_{\kappa,\Lambda,f},
\pi_{\kappa,\Lambda,g}}\] is constant, independent of \(\pairbk{\kappa,\Lambda}\).

\begin{prop}\label{expedition0011559}
When no confusion arises, denote elements of the quotient space $\setquot{\dom \mathsf{m}}{N_C}$ by the same symbol $f$ for $f
\in
\dom \mathsf{m}$.
For each $\pairbk{\kappa,\Lambda}$, define the centering map
\begin{equation}
\begin{aligned}
\iota_0
\colon \setquot{\dom \mathsf{m}}{N_C}
\to
\fun{\lp^{2}}{\prbqspace,\msr{P_{\txtgs,\kappa,\Lambda}}};
\quad
\iota_0(f)
=
\pi_{\kappa,\Lambda,f}
-\sqfun{\prbexp_{\msr{P_{\txtgs,\kappa,\Lambda}}}}
{\pi_{\kappa,\Lambda,f}}
\end{aligned}
\end{equation}
This is an isometry with a unique extension $$\iota_{\kappa,\Lambda}
\colon \sphilb{H}_{\txtirsingular}
\to \fun{\lp^{2}}{\prbqspace,\msr{P_{\txtgs,\kappa,\Lambda}}}.$$

In particular, if convergence of finite-dimensional distributions or characteristic functions holds on $\dom \mathsf{m}$, it extends naturally to $\sphilb{H}_{\txtirsingular}$.
\end{prop}

\begin{proof}
(Well-definedness): By definition, $f
\eqsimelem g$ implies $C(f-g,f-g)
= 0$.
For the centered evaluation map,
\begin{equation}
\begin{aligned}
&\norm{\iota_0(f)
-\iota_0(g)}
_{\fun{\lp^{2}}{\prbqspace},\msr{P_{\txtgs,\kappa,\Lambda}}}^2
\\ 
&=
\sqfun{\prbexp_{\msr{P_{\txtgs,\kappa,\Lambda}}}}
{\rbk{\iota_0(f)
-\iota_0(g)}^2}
\\ 
&=
\sqfun{\prbexp_{\msr{P_{\txtgs,\kappa,\Lambda}}}}
{\iota_0(f-g)^2}
=
C(f-g,f-g)
=
0
\end{aligned}
\end{equation}
is obtained.
Therefore $\iota_0$ is well-defined, independent of the choice of representative.

(Linearity): By the linearity of the map $\dom \mathsf{m}
\ni f
\mapsto \pi_{\kappa,\Lambda,f}$ and the linearity of expectation.

(Isometry): Clear by direct computation.

(Extension to closure): By definition, $\dom \mathsf{m}$ is dense in $\sphilb{H}_{\txtirsingular}$.
Since $\iota_0$ is continuous and isometric with respect to $\norm{\cdot}_C$, the continuous extension follows by general theory.
\end{proof}

\begin{thm}\label{expedition0011558}
Suppose the family $\fml{\msr{P_{\txtgs,\kappa,\Lambda}}}{\kappa,\Lambda>0}$ of probability measures satisfies the following conditions.
\begin{enumerate}
\item[(a)]
Gaussian structure and fixed covariance: For each $\pairbk{\kappa,\Lambda}$, $\fml{\pi_{\kappa,\Lambda,f}}{f \in \dom \mathsf{m}}$ is a Gaussian linear functional, and the covariance form $C(f,g)
= \sqfun{\prbcov_{\msr{P_{\txtgs,\kappa,\Lambda}}}}
{\pi_{\kappa,\Lambda,f},\pi_{\kappa,\Lambda,g}}$ is common to all $\pairbk{\kappa,\Lambda}$.

\item[(b)]
Convergence of the mean functional: For each $\pairbk{\kappa,\Lambda}$, the mean functional $$\mathsf{m}_{\kappa,\Lambda}(f)
=
\sqfun{\prbexp_{\msr{P_{\txtgs,\kappa,\Lambda}}}}
{\pi_{\kappa,\Lambda,f}}$$ is given, and there exists a linear functional $\mathsf{m}
\colon \dom \mathsf{m}
\to \fldreal$ satisfying $\mathsf{m}_{\kappa,\Lambda}(f)
\to \mathsf{m}(f)$ for any $f
\in \dom \mathsf{m}$.

\item[(c)]
Convergence of finite-dimensional characteristic functions and preservation of Gaussianity: For any $I
= \setone{f_1,\cdots,f_n}
\in \setpowfin{\dom \mathsf{m}}$ and $\xi
\in \fldreal^{n}$, $$\prbcharfun_{\txtgs,\kappa,\Lambda,I}(\xi)
=
\int_{\prbqspace}
\fnexp{\imunit
\sum_{j=1}^n
\xi_j \pi_{f_j}(\omega)}
\opdmsr{\msr{P_{\kappa,\Lambda(\omega)}}}$$ exists, and the limit $\prbcharfun_{\txtgs,I}
=
\lim_{\kappa \to 0, \Lambda \to \infty}
\prbcharfun_{\txtgs,\kappa,\Lambda,I}$ is the Gaussian characteristic function $$\prbcharfun_{\txtgs,I}(\xi)
=
\fnexp{\imunit
\sum_{j=1}^n \xi_j \mathsf{m}(f_j)
-\onehalf
\sum_{j,k=1}^n
\xi_j \xi_k
C(f_j,f_k)}.$$
\end{enumerate}
Then the following statements hold.
\begin{enumerate}
\item
For each $I
\subset \setpowfin{\dom \mathsf{m}}$, there exists a Gaussian probability measure $\msr{\mu_I}$ corresponding to the characteristic function $\prbcharfun_{\txtgs,I}$.
Moreover, the family $\fml{\msr{\mu_I}}{I \in \setpowfin{\dom \mathsf{m}}}$ is a projective system of probability measures in the sense of Definition \ref{expedition0011527}.

\item
On the measurable space $\pairbk{\prbqspace,\mblfml{C}}$, there exists a unique probability measure satisfying $$P \circ \inv{\pi_I}
=
\msr{\mu_I},
\quad I
\in \setpowfin{\dom \mathsf{m}}.$$

\item
The family $\fml{\msr{P_{\txtgs,\kappa,\Lambda}}}{\kappa,\Lambda > 0}$ of Gaussian probability measures converges $\dom \mathsf{m}$-weakly to the probability measure $\msr{P}$ of statement (2).
\end{enumerate}
\end{thm}

\begin{proof}
((1), Existence of Gaussian measure $\msr{\mu_I}$): Fix a finite set $I
= \setone{f_1,\cdots,f_n}
\in \setpowfin{\dom \mathsf{m}}$.
By hypothesis (c), for $\xi
\in \fldreal^{n}$, we have $$\prbcharfun_{\txtgs,I}(\xi)
=
\fnexp{\imunit
\sum_{j=1}^n \xi_j \mathsf{m}(f_j)
-\onehalf
\sum_{j,k=1}^n
\xi_j \xi_k
C(f_j,f_k)}.$$
Setting the mean vector and covariance matrix $$\mathsf{m}_I
=
\vecbk{\mathsf{m}(f_1),\cdots,\mathsf{m}(f_n)}
\in \fldreal^{n},
\quad
\Sigma_I
=
\lamat{C(f_j,f_k)}{1 \leq j,k \leq n},$$ the characteristic function can be rewritten as $$\prbcharfun_{\txtgs,I}(\xi)
=
\fnexp{\imunit \xi \vainnprod \mathsf{m}_I
-\onehalf
\bkt{\xi}{\Sigma_I \xi}_{\fldreal^{n}}}.$$
By hypothesis, $C$ is a symmetric positive semidefinite form, so $\Sigma_I$ is a symmetric positive semidefinite matrix.
In particular, $\prbcharfun_{\txtgs,I}$ is a continuous positive definite function on $\fldreal^{n}$ satisfying $\prbcharfun_{\txtgs,I}(0)
= 1$.
By Bochner's theorem, there exists a unique probability measure $\msr{\mu_I}$ whose characteristic function coincides with $\prbcharfun_{\txtgs,I}$.
This $\msr{\mu_I}$ is the Gaussian measure with mean $\mathsf{m}_I$
and covariance $\Sigma_I$.

((1), Projective consistency): For any pair of finite subsets $J
=
\setone{f_1,\cdots,f_n}
\subset I
\in \setpowfin{\dom \mathsf{m}}$, let $\pi_J^I
\colon \fldreal^{I}
\to \fldreal^{J}$ be the coordinate projection.
Denote the characteristic function of measure $\msr{\mu}$ by $\sqfun{\fafouriertr}{\msr{\mu}}$.
The characteristic function of the projected measure $\msr{\mu_I} \circ \invrbk{\pi_{J}^{I}}$ satisfies, for any $\xi
\in \fldreal^{J}$, $$\sqfun{\fafouriertr}{\msr{\mu_I} \circ \invrbk{\pi_J^{I}}}(\xi)
=
\sqfun{\prbexp_{\msr{\mu_I} \circ \invrbk{\pi_J^{I}}}}
{\napiernum^{i \xi \vainnprod y}}
=
\sqfun{\prbexp_{\msr{\mu_I}}}
{\napiernum^{i \xi \vainnprod \pi_J^I(x)}}.$$
Let $\tilde{\xi}
\in \fldreal^{I}$ be $\xi$ naturally embedded into $\fldreal^{I}$.
Using $\xi \vainnprod \pi_J^I(x)
= \tilde{\xi} \vainnprod x$, $$\sqfun{\fafouriertr}{\msr{\mu_I}
\circ
\invrbk{\pi_J^{I}}}(\xi)
=
\sqfun{\prbexp_{\msr{\mu_I}}}
{\napiernum^{i \xi \vainnprod \pi_J^I(x)}}
=
\fun{\prbcharfun_{\txtgs,I}}
{\tilde{\xi}}$$ holds.
On the other hand, the characteristic function for $J$ is $$\prbcharfun_J(\xi)
=
\fnexp{\imunit
\sum_{l=1}^k
\xi_l \mathsf{m}(f_{j_l})
-\onehalf
\sum_{l,l'=1}^k
\xi_l \xi_{l'}
C(f_{j_l}, f_{j_{l'}})},$$ which equals $\prbcharfun_{\txtgs,I}(\tilde{\xi})$.
Therefore $\sqfun{\fafouriertr}{\msr{\mu_I} \circ \invrbk{\pi_J^{I}}}
= \prbcharfun_J$ holds.
Since characteristic functions uniquely determine probability measures, $\msr{\mu_I}
\circ
\invrbk{\pi_J^{I}}
=
\msr{\mu_J}$, which is a projective system of probability measures in the sense of Definition \ref{expedition0011527}.

((2), Existence of limit measure $\msr{P}$): The family $\fml{\msr{\mu_I}}{I \in \setpowfin{\dom \mathsf{m}}}$ from (1) is a projective system of probability measures.
The conclusion follows from Kolmogorov's extension theorem \ref{expedition0011529} for projective systems.

((3), Weak convergence of finite-dimensional distributions):
For each $\kappa,\Lambda,I$, let $\prbcharfun_{\txtgs,\kappa,\Lambda,I}$ be the characteristic function of $\msr{P_{\txtgs,\kappa,\Lambda,I}}$ and $\prbcharfun_{\txtgs,I}$ that of $\msr{\mu_I}$.
By hypothesis (c), for any $I
\in \setpowfin{\dom \mathsf{m}}$ and $\xi
\in \fldreal^{I}$, $$\prbcharfun_{\txtgs,\kappa,\Lambda,I}(\xi)
\to \prbcharfun_{\txtgs,I}(\xi)$$ holds.
By Lévy's continuity theorem in finite dimensions, $$\msr{P_{\txtgs,\kappa,\Lambda,I}}
\xRightarrow[\kappa \to 0, \Lambda \to \infty]{\txtweak}
\msr{\mu_I}$$ holds.
On the other hand, by (2), $\msr{P_I}
=
\msr{P} \circ \inv{\pi_I}
=
\msr{\mu_I}$.
In particular, for any $I
\in \setpowfin{\dom \mathsf{m}}$, $\msr{P_{\txtgs,\kappa,\Lambda,I}}
\xRightarrow{\txtweak}
\msr{P_I}$ is obtained.

((3), $\dom \mathsf{m}$-weak convergence): By Proposition \ref{expedition0011554} and the above argument, the desired result follows.
\end{proof}

\begin{cor}\label{expedition0011563}
\begin{enumerate}
\item
Existence and uniqueness of the limit Gaussian measure: On the measurable space $\pairbk{\prbqspace,C}$, the measure $\msr{P_{\txtirsingular}}$ obtained as the weak limit of the family $\fml{\msr{P_{\txtgs,\kappa,\Lambda}}}{\kappa,\Lambda > 0}$ with cutoffs removed exists uniquely.

\item
Preservation of Gaussian structure of the limit measure $\msr{P_{\txtirsingular}}$: For the probability measure $\msr{P_{\txtirsingular}}$ from (1), $$\pi_{\txtirsingular,f}(\omega)
=
\omega(f)$$ for each $f
\in \dom \mathsf{m}$ is a real Gaussian random variable with mean $\mathsf{m}(f)$ and covariance $C(f,g)$, and the finite-dimensional distribution for any $I
\in \setpowfin{\dom \mathsf{m}}$ is the Gaussian measure characterized by the characteristic function $$\prbcharfun_{\txtgs,I}(\xi)
=
\fnexp{\imunit
\sum_{j=1}^{n}
\xi_{j}
\mathsf{m}(f_j)
-\onehalf
\sum_{j,k=1}^{n}
\xi_{j} \xi_{k}}
C(f_j,f_k).$$

\item
Limit of the isometric extension to the closure $\sphilb{H}_{\txtirsingular}$: Consider the closure $\sphilb{H}_{\txtirsingular}$ for the symmetric positive semi-definite form $C
\colon \dom \mathsf{m} \times \dom \mathsf{m}
\to \fldreal$.
As the limit of the isometry $\iota_{\kappa,\Lambda}
\colon \sphilb{H}_{\txtirsingular}
\to \fun{\lp^{2}}{\prbqspace,\msr{P_{\txtgs,\kappa,\Lambda}}}$ from Proposition \ref{expedition0011559}, there exists a unique isometry $$\iota_{\txtirsingular}
\colon \sphilb{H}_{\txtirsingular}
\to \fun{\lp^{2}}{\prbqspace,\msr{P_{\txtirsingular}}}.$$
This satisfies the following properties.
\begin{enumerate}
\item[(1)]
For any $f
\in \dom \mathsf{m}$, the finite-dimensional distributions of $\iota_{\kappa,\Lambda}(f)$ converge to those of $\iota_{\txtirsingular}(f)$.

\item[(2)]
The inner product satisfies $$\bkt{\iota_{\txtirsingular}(f)}
{\iota_{\txtirsingular}(f)}
_{\fun{\lp^{2}}{\prbqspace,\msr{P_{\txtirsingular}}}}
=
C(f,g).$$
\end{enumerate}
In particular, the $\fun{\lp^{2}}{\prbqspace,\msr{P_{\txtirsingular}}}$-closure of $\iota_{\txtirsingular}(\sphilb{H}_{\txtirsingular})$ can be regarded as the $Q$-representation after removal of infrared and ultraviolet cutoffs, naturally obtained as the limit of the cutoff $Q$-representation $\fun{\lp^{2}}{\prbqspace_{\sphilb{H}},\msr{P_{\txtgs,\kappa,\Lambda}}}$.
\end{enumerate}
The limit measure $\msr{P_{\txtirsingular}}$ constructed in this corollary is called the ground state measure on the $Q$-representation after removal of infrared and ultraviolet cutoffs,
and its $\lp^{2}$-space $\fun{\lp^{2}}{\prbqspace_{\txtirsingular},\msr{P_{\txtirsingular}}}$ is called the renormalized physical Hilbert space.
\end{cor}

\begin{proof}

(1)-(2): Direct corollaries of Theorem \ref{expedition0011558}.

(3-1): Direct corollary of Theorem \ref{expedition0011558}.

(3-2): For any $I
= \setone{f_1,\cdots,f_n}
\in \setpowfin{\dom \mathsf{m}}$, compute the characteristic function of $$X_{\kappa,\Lambda}
=
\vecbk{\iota_{\kappa,\Lambda}(f_1),
\cdots,
\iota_{\kappa,\Lambda}(f_n)}.$$

For the finite-dimensional characteristic function $\prbcharfun_{\txtgs,\kappa,\Lambda,I}$ of Theorem \ref{expedition0011558}(c), by definition,
\begin{equation}
\begin{aligned}
&\sqfun{\prbexp_{\msr{P_{\txtgs,\kappa,\Lambda}}}}
{\fnexp{\imunit
\sum_{j=1}^n
\xi_{j}
\iota_{\kappa,\Lambda}(f_j)}}
\\ 
&=
\sqfun{\prbexp_{\msr{P_{\txtgs,\kappa,\Lambda}}}}
{\fnexp{\imunit
\sum_{j=1}^n
\xi_{j}
\rbk{\pi_{\kappa,\Lambda,f_j} - \mathsf{m}_{\kappa,\Lambda}(f_j)}}}
\\ 
&=
\fnexp{-\imunit
\sum_{j=1}^n
\xi_j
\mathsf{m}_{\kappa,\Lambda}(f_j)}
\prbcharfun_{\txtgs,\kappa,\Lambda,I}(\xi)
\end{aligned}
\end{equation}
is obtained.
Also by Theorem \ref{expedition0011558}(c), $$\prbcharfun_{\txtgs,\kappa,\Lambda,I}(\xi)
=
\fnexp{\imunit
\sum_{j=1}^{n}
\xi_j
\mathsf{m}_{\kappa,\Lambda}(f_j)
-\onehalf
\sum_{j,k=1}^n
\xi_j \xi_k
C(f_j,f_k)}$$ holds.
The preceding expectation value becomes $$\sqfun{\prbexp_{\msr{P_{\txtgs,\kappa,\Lambda}}}}
{\fnexp{\imunit
\sum_{j=1}^n
\xi_{j}
\iota_{\kappa,\Lambda}(f_j)}}
=
\fnexp{-\onehalf
\sum_{j,k=1}^n
\xi_j \xi_k
C(f_j,f_k)},$$ and the right-hand side is independent of the cutoff parameters $\pairbk{\kappa,\Lambda}$.
On the other hand, computing the centering $\iota_{\txtirsingular}(f_j)
=
\pi_{\txtirsingular,f_j} - \mathsf{m}(f_j)$ with respect to $\msr{P_{\txtirsingular}}$, $$\sqfun{\prbexp_{\msr{P_{\txtirsingular}}}}
{\fnexp{\imunit
\sum_{j=1}^n
\xi_j \iota_{\txtirsingular}(f_j)}}
=
\fnexp{-\onehalf
\sum_{j,k=1}^n
\xi_j \xi_k
C(f_j,f_k)}$$ is obtained.
For any $I
= \vecbk{f_1,\cdots,f_n}
\in \dom \mathsf{m}$ and $\xi
\in \fldreal^{I}$, $$\sqfun{\prbexp_{\msr{P_{\txtgs,\kappa,\Lambda}}}}
{\fnexp{\imunit
\sum_{j=1}^n
\xi_j \iota_{\txtirsingular}(f_j)}}
=
\sqfun{\prbexp_{\msr{P_{\txtirsingular}}}}
{\fnexp{\imunit
\sum_{j=1}^n
\xi_j \iota_{\txtirsingular}(f_j)}}$$ is obtained.
In particular, the finite-dimensional distributions of $\vecbk{\iota_{\kappa,\Lambda}(f_1),
\cdots,
\iota_{\kappa,\Lambda}(f_n)}$ agree with those of $\vecbk{\iota_{\txtirsingular}(f_1),
\cdots,
\iota_{\txtirsingular}(f_n)}$.
The removal limit of infrared/ultraviolet cutoffs follows trivially.
\end{proof}

\begin{prop}\label{expedition0011562}
The set-theoretic inclusion $$\fun{\lp^{2}}{\prbqspace_{\txtirsingular},\msr{P_{\txtirsingular}}}
\subset
\fun{\lp^{2}}{\prbqspace_{\sphilb{H}},\msr{P_{\txtgs,\kappa,\Lambda}}}$$ holds.
In particular, for any $F
\in \fun{\lp^{2}}{\prbqspace_{\txtirsingular},\msr{P_{\txtirsingular}}}$, $F
\in \fun{\lp^{2}}{\prbqspace_{\sphilb{H}},\msr{P_{\txtgs,\kappa,\Lambda}}}$.
\end{prop}

\begin{proof}
Fix non-negative $f
\in \sphilb{H}_{\txtirsingular}$.
The measures $\msr{P_{\txtirsingular}}$ and $\msr{P_{\txtgs,\kappa,\Lambda}}$ have equal covariance, and for the means, $$\abs{\sqfun{\prbexp_{\msr{P_{\txtgs,\kappa,\Lambda}}}}
{\opfocksegal(f)}}
=
\bkt{f}{\mathsf{m}_{\kappa,\Lambda}}
\leq
\bkt{f}{\mathsf{m}}
=
\abs{\sqfun{\prbexp_{\msr{P_{\txtirsingular}}}}
{\opfocksegal(f)}}$$ holds.
By definition, $\fun{\lp^{2}}{\prbqspace,\msr{P_{\txtirsingular}}}$ is constructed from $\opfocksegal(f)$, and since $\prbqspace_{\txtirsingular}
\subset \prbqspace_{\sphilb{H}}$ holds for the sample spaces, the desired result follows.
\end{proof}

\subsection{Renormalization of Hamiltonian and Ground State}\label{renormalization-of-hamiltonian-and-ground-state}

Define the Hamiltonian with the self-energy term subtracted by \[\physham_{\txtvanhowe,\txtrenormalization,\kappa,\Lambda}
=
\physham_{\txtvanhowe,\kappa,\Lambda} - \physenergy_{\txtself,\kappa,\Lambda}^{(3)}\] and apply the ground state transformation of Definition \ref{expedition0011500} to obtain \[\physham[K]_{\txtrenormalization,\kappa,\Lambda}
=
\inv{\mathsf{U}_{\kappa,\Lambda}}
\physham_{\txtvanhowe,\txtrenormalization,\kappa,\Lambda}
\mathsf{U}_{\kappa,\Lambda}
=
\physham_{\txtbsn,\txtfr}
-\physenergy_{\txtexchange,\kappa,\Lambda}^{(3)}.\] Furthermore, formally define the Hamiltonian with infrared/ultraviolet cutoffs removed by \[\physham[K]_{\txtrenormalization}
=
\physham_{\txtbsn,\txtfr}
-\physenergy_{\txtexchange}^{(3)}.\]

Toward the renormalization of the Hamiltonian and ground state, we formulate the functional integral representation of the free field in the \(\lp^{2}\)-space with respect to the ground state measure.

\begin{thm}\label{expedition0011561}
Consider $\fun{\lp^{2}}{\prbqspace_{\sphilb{H}},\msr{P_{\txtgs,\kappa,\Lambda}}}$ with respect to the ground state measure for each cutoff parameter $\kappa,\Lambda$, and also denote the free field Hamiltonian directly defined on this space by $\physham_{\txtbsn,\txtfr}$.
Then the functional integral representation for the free field gives $$\bkt{F}
{\napiernum^{-t \physham_{\txtbsn,\txtfr}}
G}
_{\fun{\lp^{2}}{\prbqspace_{\sphilb{H}},\msr{P_{\txtgs,\kappa,\Lambda}}}}
=
\sqfun{\prbexp_{\msr{P}_{\txteuclid}}}
{J_0 F
\cdot
J_t G}.$$
Here $J_t$ is the isometry defined via Wick products, analogously to the standard isometry $$J_t
\colon \fun{\lp^{2}}{\prbqspace_{\sphilb{H}},\msr{P}_{\sphilb{H}}}
\to \fun{\lp^{2}}{\prbqspace_{\txteuclid},\msr{P}_{\txteuclid}}.$$
In particular, for $\physham[K]_{\txtrenormalization,\kappa,\Lambda}$, $$\bkt{F}
{\napiernum^{-t \physham[K]_{\txtrenormalization,\kappa,\Lambda}}
G}
_{\fun{\lp^{2}}{\prbqspace_{\sphilb{H}},\msr{P_{\txtgs,\kappa,\Lambda}}}}
=
\fnexp{t \physenergy_{\txtexchange,\kappa,\Lambda}^{(3)}}
\sqfun{\prbexp_{\msr{P}_{\txteuclid}}}
{J_0 F
\cdot
J_t G}$$ holds.
Similarly, for any $F,G
\in \fun{\lp^{2}}{\prbqspace,\msr{P_{\txtirsingular}}}$, $$\bkt{F}
{\napiernum^{-t \physham[K]_{\txtrenormalization}}
G}
_{\fun{\lp^{2}}{\prbqspace_{\txtirsingular},\msr{P_{\txtirsingular}}}}
=
\fnexp{t \physenergy_{\txtexchange}^{(3)}}
\sqfun{\prbexp_{\msr{P}_{\txteuclid}}}
{J_0 F
\cdot
J_t G}$$ holds.
\end{thm}

\begin{proof}
Since both cases are the same, we discuss only the cutoff case.
The ground state measure $\msr{P_{\txtgs,\kappa,\Lambda}}$ differs from $\msr{P}_{\sphilb{H}}$ only in the mean; it is a Gaussian measure with equal variance.
In particular, the (exponential) Wick products they define coincide.
Therefore, the same $J_t$ defines an isometry $J_t
\colon \fun{\lp^{2}}{\prbqspace_{\sphilb{H}},\msr{P_{\txtgs,\kappa,\Lambda}}}
\to \fun{\lp^{2}}{\prbqspace_{\txteuclid},\msr{P}_{\txteuclid}}$.
In particular, the functional integral representation for the free Hamiltonian holds in the same way.
\end{proof}

\begin{thm}
\begin{enumerate}
\item
The family $\fml{\physham[K]_{\txtrenormalization,\kappa,\Lambda}}
{\kappa,\Lambda > 0}$ of bounded-below self-adjoint operators is uniformly bounded below by the respective operator norms on their spaces, and for each cutoff parameter the heat semigroup $\napiernum^{-t \physham[K]_{\txtrenormalization,\kappa,\Lambda}}$ is well-defined.

\item
For any $t
> 0$, the strong limit $$\slim
_{\kappa \to 0, \Lambda \to \infty}
\napiernum^{-t \physham[K]_{\txtrenormalization,\kappa,\Lambda}}
=
\napiernum^{-t \physham[K]_{\txtrenormalization}}$$ holds on $\fun{\lp^{2}}{\prbqspace_{\txtirsingular},\msr{P_{\txtirsingular}}}$.
The ground state energy can be expressed as $\physgse(\physham[K]_{\txtrenormalization})
=
\physenergy_{\txtexchange}^{(3)}$ in terms of the exchange interaction energy from Definition \ref{expedition0011497}.
The corresponding ground state is $1
\in \fun{\lp^{2}}{\prbqspace_{\txtirsingular},\msr{P_{\txtirsingular}}}$, and is unique.
Furthermore, the Feynman-Kac-Nelson formula $$\bkt{F}
{\napiernum^{-t \physham[K]_{\txtrenormalization}}
G}
_{\fun{\lp^{2}}{\prbqspace_{\txtirsingular},\msr{P}_{\txtirsingular}}}
=
\napiernum^{t \physenergy_{\txtexchange}^{(3)}}
\sqfun{\prbexp_{\msr{P}_{\txteuclid}}}
{J_0 F
\cdot
J_t G}$$ holds.
\end{enumerate}
\end{thm}

\begin{proof}
(1): By Section \ref{expedition0011496} and Proposition \ref{expedition0011517}, the exchange interaction term $\physenergy_{\txtexchange,\kappa,\Lambda}^{(3)}$ converges.

(2): Fix any $F,G
\in \fun{\lp^{2}}{\prbqspace_{\txtirsingular},\msr{P_{\txtirsingular}}}$.
By definition, $\prbqspace_{\txtirsingular}
\subset \prbqspace_{\sphilb{H}}$, so in particular $F,G
\in \fun{\lp^{2}}{\prbqspace_{\sphilb{H}},\msr{P_{\txtgs,\kappa,\Lambda}}}$ holds.
By the Feynman-Kac-Nelson formula from Theorem \ref{expedition0011561}, for any cutoff parameter, $$\bkt{F}
{\napiernum^{-t \physham[K]_{\txtrenormalization,\kappa,\Lambda}}
G}
_{\fun{\lp^{2}}{\prbqspace_{\txtirsingular},\msr{P}_{\txtirsingular}}}
=
\fnexp{t \physenergy_{\txtexchange,\kappa,\Lambda}^{(3)}}
\sqfun{\prbexp_{\msr{P}_{\txteuclid}}}
{J_0 F
\cdot
J_t G}$$ holds.
By Proposition \ref{expedition0011517}, the exchange interaction term $\physenergy_{\txtexchange,\kappa,\Lambda}^{(3)}$ converges to a finite definite value, so the cutoff-removal limit can be taken on the right-hand side.
By Theorem \ref{expedition0011561}, the right-hand side reduces to $\bkt{F}
{\napiernum^{-t \physham[K]_{\txtrenormalization}}
G}
_{\fun{\lp^{2}}{\prbqspace_{\txtirsingular},\msr{P_{\txtirsingular}}}}$.
Since the right-hand side is continuous in $t$, so is the left-hand side.
Viewing the left-hand side as the limit of a quadratic form, the KLMN theorem yields a bounded operator $T_t$ satisfying $\bkt{F}
{T_t G}
_{\fun{\lp^{2}}{\prbqspace_{\txtirsingular},\msr{P}_{\txtirsingular}}}$.
Comparing with the right-hand side, $T_t$ must be a strongly continuous semigroup coinciding with $\napiernum^{-t \physham[K]_{\txtrenormalization}}$.
The statements about the ground state energy, ground state, and its uniqueness are clear.
\end{proof}

\section{Ground State at Zero Temperature via Infinite-Dimensional Ornstein-Uhlenbeck Representation}\label{expedition0011528}

We discuss the argument of Section \ref{expedition0011370} in the Ornstein-Uhlenbeck representation. The notation follows primarily from \cite{LorincziHiroshimaBetz3}. Writing the entire argument in terms of Ornstein-Uhlenbeck processes has the advantage of unifying the following structures into a single probabilistic framework.

\begin{itemize}

\item
  Free field construction: generator of the Ornstein-Uhlenbeck process
\item
  Introduction of interaction: exponential weight of the time integral
\item
  Existence of the ground state: long-time limit evaluation of the exponential integral
\item
  Coordinate process (\(P(\phi)_1\)-process): Ornstein-Uhlenbeck process after ground state transformation
\item
  Construction of the Gibbs measure: infinite-volume limit of the Ornstein-Uhlenbeck process with weight
\item
  Integrability and regularity: Gaussian properties of the Ornstein-Uhlenbeck process and Gaussian probability theory for linear functionals
\end{itemize}

The correspondence with the \(Q\)-representation is as follows.

\begin{itemize}

\item
  Static Gaussian measure: time-\(0\) section of the Ornstein-Uhlenbeck process
\item
  Dynamic probabilistic structure: the entire strongly Markovian Ornstein-Uhlenbeck process
\item
  Ground state measure \(\msr{P_{\txtgs,\kappa,\Lambda}}\): Gibbs measure of the Ornstein-Uhlenbeck process weighted by the ground state
\item
  Hamiltonian: deformation of the generator of the Ornstein-Uhlenbeck process by weighting
\end{itemize}

The parallel shifts of Gaussian measures, Wick ordering, and changes in field means that were needed in the \(Q\)-representation are all unified as drift terms of the generator of the Ornstein-Uhlenbeck process, which eliminates the branching in the argument.

\subsection{Basic Setup}\label{expedition0011372}

Specifically, let the Hilbert spaces be \[\sphilb{H}_M
= \sphilb{H}_{-\onehalf}(\fldreal^{3}),
\quad
\sphilb{M}
=
\set{f \in \sphilb{H}_M}
{\text{$f$ is real-valued}}.\] The inner product for \(f,g
\in \sphilb{H}_M\) is \[\bkt{f}{g}
_{\sphilb{H}_M}
=
\bkt{\abs{k}^{-\onehalf} \faftr{f}}
{\abs{k}^{-\onehalf} \faftr{g}}
_{\fun{\lp^{2}}{\fldreal^{3}}}.\] A given positive self-adjoint operator \(\mathcal{D}\) with Hilbert-Schmidt inverse on \(\sphilb{M}\) and such that \(\sqrt{\omega} \inv{\mathcal{D}}\) is bounded on \(\sphilb{M}\).

Construct a probability measure \(\msrcal{G}\) on the space \[\sphilb[\mathfrak]{Y}
= \fun{\conti}{\fldreal; \sphilb{M}_{-2}} 
\] of continuous paths with values in the Hilbert space \(\sphilb{M}_{-2}\), and on this space define the \(\sphilb{M}_{-2}\)-valued continuous process \(\fml{\prbou_s}{s \in \fldreal}\). This process has mean \(0\) and covariance \[\sqfun{\prbexp_{\msrcal{G}}}{\bkt{\prbou_s}{\prbou_t}_{\sphilb{M}_{-2}}}
=
\frac{1}{4}
\sum_{i,j}
\frac{1}{\lambda_i \lambda_j}
\bkt{e_i}{\napiernum^{-\abs{t-s} \omega} e_j}_{\sphilb{H}},\] and for each \(f,g
\in \sphilb{M}\), satisfies \[\sqfun{\prbexp_{\msrcal{G}}}
{\bkt{\prbou_s(f)}
{\prbou_t(g)}
_{\sphilb{M}_{-2}}}
=
\onehalf
\bkt{f}
{\napiernum^{-\abs{s-t} \omega} g}
_{\sphilb{H}}.\]

This process is called the infinite-dimensional Ornstein-Uhlenbeck process. The inner product on the right-hand side of the covariance uses \(\sphilb{H}\) instead of \(\sphilb{H}_M\) for consistency with \cite[P.81 (1.5.36), (1.5.37)]{LorincziHiroshimaBetz3}. For more detailed properties and discussion, see \cite[P.101, Theorem1.129]{LorincziHiroshimaBetz3} and the surrounding arguments.

Let \(L_0\) be the generator of the Markov property of the infinite-dimensional Ornstein-Uhlenbeck process \(\fml{\prbou_t}{t \in \fldreal}\), regarded as a self-adjoint operator on the Hilbert space \(\fun{\lp^{2}}{\sphilb{M}_{-2},\msrsf{G}}\). This is in fact the free field Hamiltonian, and the corresponding semigroup \[T_{\txtfr}
=
\fml{T_{\txtfr,t}}{t \geq 0},
\quad
\funrbk{T_{\txtfr,t} F}{\opfocksegal}
=
\sqfuncond{\prbexp_{\msrcal{G}}}
{F(\prbou_t)}
{\prbou_0 = \opfocksegal}\] satisfies \(T_{\txtfr,t}
= \napiernum^{-t \physham[L]_0}\), and the invariant measure is \(\msrsf{G}\).

We cite \cite[p.50, Corollary 1.54]{LorincziHiroshimaBetz3}.

\begin{fact}
The following statements hold.
\begin{enumerate}
\item
The operator $\physham[L]_0$ is the free field Hamiltonian.

\item
Conservation: $\napiernum^{-s \physham[L]_0} 1
= 1$ holds.

\item
Stationarity (invariance with respect to the static distribution $\msrcal{G}$): In particular, $$\sqfun{\prbexp_{\msrcal{G}}}
{F}
=
\sqfun{\prbexp_{\msrcal{G}}}
{\napiernum^{-s \physham[L]_0} F}$$ holds.

\item
Symmetry with respect to the inner product: $$\bkt{F}{\napiernum^{-s \physham[L]_0} G}
_{\fun{\lp^{2}}{\sphilb{M}_{-2},\msrsf{G}}}
=
\bkt{\napiernum^{-s \physham[L]_0} F}{G}
_{\fun{\lp^{2}}{\sphilb{M}_{-2},\msrsf{G}}}$$ holds.

\item
Positivity improvement: For $F
\geq 0$, $$\napiernum^{-s \physham[L]_0} F
> 0$$ holds.
\end{enumerate}
\end{fact}

Define the van Hove Hamiltonian by \[\physham_{\txtvanhowe,\kappa,\Lambda}
=
\physham_{\txtbsn,\txtfr} + \prbou(\omega \mathsf{m}_{\kappa,\Lambda})\] and the pair potential by \[W_{\kappa,\Lambda}(T)
=
\int_0^{T}
\prbou_s(\omega \mathsf{m}_{\kappa,\Lambda})
\opdmsr{s}.\] Then for any \(F,G
\in \fun{\lp^{2}}{\sphilb{M}_{-2},\msrsf{G}}\), the functional integral representation \[\bkt{F}
{\napiernum^{-t \physham_{\txtvanhowe,\kappa,\Lambda}} G}
_{\fun{\lp^{2}}{\sphilb{M}_{-2},\msrsf{G}}}
=
\sqfun{\prbexp_{\msrcal{G}}}{\cmpconj{F(\prbou_0)}
G(\prbou_t)
\napiernum^{-W_{\kappa,\Lambda}(t)}}\] holds.

\begin{defn}\label{expedition0011537}
Let $\sphilb{H}$ be a Hilbert space.
For an appropriate positive integer $n$ and $f
\in \fun{\dstrapiddec}{\fldreal^{n}}$, a map $F
\colon \sphilb{H}
\to \fldcmp$ expressed as $$F(\opfocksegal)
=
\fun{f}
{\bkt{\opfocksegal}{h_1}_{\sphilb{H}},
\cdots,
\bkt{\opfocksegal}{h_n}_{\sphilb{H}}}
\quad
h_1,\cdots,h_n
\in \sphilb{H}$$ is called a cylindrical function.
\end{defn}

For Banach spaces \(X,Y\) consider a map \(F
\colon X
\to Y\), and fix any \(\opfocksegal,k
\in X\). When the limit \[dF(\opfocksegal;k)
=
\lim_{t \to 0}
\frac{F(\opfocksegal + tk) - F(\opfocksegal)}{t}
=
\fnrestr{\opod{t} F(\opfocksegal + tk)}
{t = 0}\] exists, it is called the Gâteaux derivative of \(F\) at \(\opfocksegal
\in X\) in the direction \(k
\in X\).

\begin{defn}\label{expedition0011538}
For the cylindrical function $F$ of Definition \ref{expedition0011537}, define the gradient via Gâteaux derivative by $$\vagrad F(\opfocksegal)
=
\sum_{j=1}^n
\fun{\pd{f}{x_j}}
{\bkt{\opfocksegal}{h_1}_{\sphilb{H}},
\cdots,
\bkt{\opfocksegal}{h_n}_{\sphilb{H}}}
\cdot
h_j
\in \sphilb{H}.$$
\end{defn}

\begin{prop}
\begin{enumerate}
\item
For any $\opfocksegal
\in \sphilb{H}$, the map $\sphilb{H}
\ni k
\mapsto dF(\opfocksegal;k) \in \fldcmp$ is a continuous linear functional on $\sphilb{H}$.
By the Riesz representation theorem, there exists a unique $G(\opfocksegal)
\in \sphilb{H}$ satisfying $$dF(\opfocksegal;k)
=
\bkt{G(\opfocksegal)}{k}_{\sphilb{H}}$$ for all $k
\in \sphilb{H}$.
In particular, the gradient is uniquely determined.

\item
The gradient defined by Definition \ref{expedition0011538} is well-defined, and the value of $G(\opfocksegal)$ from (1) agrees regardless of how the cylindrical function $F$ is expressed.
In particular, computing with another representation of $F$, $$F(\opfocksegal)
=
\fun{g}
{\bkt{\opfocksegal}{h'_1}_{\sphilb{H}},
\cdots,
\bkt{\opfocksegal}{h'_m}_{\sphilb{H}}},$$ the gradient $$\sum_{l=1}^m
\fun{\pd{g}{x_l}}
{\fun{g}
{\bkt{\opfocksegal}{h'_1}_{\sphilb{H}},
\cdots,
\bkt{\opfocksegal}{h'_m}_{\sphilb{H}}}}
\cdot
h'_l$$ agrees with $G(\opfocksegal)$ from (1), establishing the well-definedness of the gradient definition in coordinate representation.
\end{enumerate}
\end{prop}

\begin{proof}
(1): Fix any $\opfocksegal
\in \sphilb{H}$.
By definition, $$F(\opfocksegal + tk)
=
\fun{f}
{\bkt{\opfocksegal + tk}{h_1}_{\sphilb{H}},
\cdots,
\bkt{\opfocksegal + tk}{h_n}_{\sphilb{H}}}$$ holds.
Since each argument satisfies $$\bkt{\opfocksegal + tk}{h_j}_{\sphilb{H}}
=
\bkt{\opfocksegal}{h_j}_{\sphilb{H}}
+t \bkt{k}{h_j}_{\sphilb{H}},$$ the chain rule for multivariable derivatives gives $$dF(\opfocksegal;k)
=
\sum_{j=1}^n
\fun{\pd{f}{x_j}}
{\bkt{\opfocksegal}{h_1}_{\sphilb{H}},
\cdots,
\bkt{\opfocksegal}{h_n}_{\sphilb{H}}}
\cdot
\bkt{k}{h_j}_{\sphilb{H}}$$
This shows that $k \mapsto dF(\opfocksegal;k)$ can be written as $$dF(\opfocksegal;k)
=
\sum_{j=1}^n
a_j(\opfocksegal)
\bkt{k}{h_j}_{\sphilb{H}},$$ which is in particular a continuous linear functional on $\sphilb{H}$.
By the Riesz representation theorem, there exists a unique $G(\opfocksegal)
\in \sphilb{H}$ satisfying $$dF(\opfocksegal;k)
=
\bkt{G(\opfocksegal)}{k}_{\sphilb{H}}$$ for any $k
\in \sphilb{H}$.

(2): The notation from (1) shows that the right-hand side comes from the inner product on the finite-dimensional subspace $$\sphilb{V}
=
\splinspan \setone{h_1,\cdots,h_n}
\subset
\sphilb{H}.$$
Setting $u(\opfocksegal)
=
\sum_{j=1}^n
a_j(\opfocksegal) h_j
\in \sphilb{V}$, we get $dF(\opfocksegal;k)
=
\bkt{u(\opfocksegal)}{k}_{\sphilb{H}}$ for any $k
\in \sphilb{H}$.
By uniqueness in the Riesz representation theorem, $G(\opfocksegal)
=
u(\opfocksegal)$, which is the coordinate representation of the gradient.
By the uniqueness of $G(\opfocksegal)$, any other representation yields the same vector.
In particular, $\vagrad F(\opfocksegal)$ is independent of the representation of $F$ and the choice of $\setone{h_j}_{j=1}^n$.
\end{proof}

\subsection{Discussion under Infrared Regularization}\label{discussion-under-infrared-regularization-1}

\begin{prop}\label{expedition0011507}
The heat semigroup $\napiernum^{-t \physham_{\txtvanhowe,\kappa,\Lambda}}$ is a positivity-improving operator.
In particular, if the van Hove Hamiltonian has a ground state, it is unique.
\end{prop}

\begin{proof}
By the Feynman-Kac-Nelson formula, for any $F,G
\geq 0$, $$\bkt{F}
{\napiernum^{-t \physham_{\txtvanhowe,\kappa,\Lambda}}
G}
_{\fun{\lp^{2}}{\sphilb{M}_{-2},\msrsf{G}}}
=
\sqfun{\prbexp_{\msrcal{G}}}
{F(\prbou_0)
G(\prbou_t)
\napiernum^{-W_{\kappa,\Lambda}(t)}}$$ holds.

If the semigroup were not positivity-improving, for some $t,F,G$, the inner product would be $0$.
Since $\napiernum^{-W_{\kappa,\Lambda}(t)}
> 0$ in the Feynman-Kac-Nelson representation, we would need $F(\prbou_0) G(\prbou_t)
= 0$.
By the Feynman-Kac-Nelson formula for the free field, $$\bkt{F}
{\napiernum^{-t \physham_{\txtbsn,\txtfr}}
G}
_{\fun{\lp^{2}}{\sphilb{M}_{-2},\msrsf{G}}}
=
\sqfun{\prbexp_{\msrcal{G}}}
{F(\prbou_0)
G(\prbou_t)}
=
0$$ would hold.
Since the heat semigroup of the free field is positivity-improving, this is a contradiction.
Therefore, the heat semigroup of the van Hove Hamiltonian is positivity-improving.
\end{proof}

\begin{prop}\label{expedition0011501}
\begin{enumerate}
\item
The pair potential $W_{\kappa,\Lambda}(T)$ is a Gaussian random variable with mean $0$ and variance$$\sqfun{\prbvar_{\msrcal{G}}}
{W_{\kappa,\Lambda}(T)}
=
T
\norm{\omega^{\onehalf} \mathsf{m}_{\kappa,\Lambda}}
_{\sphilb{H}}^2
-\bkt{\mathsf{m}_{\kappa,\Lambda}}
{\rbk{1 - \napiernum^{-T \omega}}
\mathsf{m}_{\kappa,\Lambda}}
_{\sphilb{H}}.$$
In particular, the mean of $\napiernum^{-W_{\kappa,\Lambda(T)}}$ is
\begin{equation}
\begin{aligned}
&\sqfun{\prbexp_{\msrcal{G}}}
{\napiernum^{-W_{\kappa,\Lambda(T)}}}
=
\fnexp{\onehalf
\sqfun{\prbvar_{\msrcal{G}}}
{W_{\kappa,\Lambda(T)}}}
\\ 
&=
\fnexp{\frac{T}{2}
\norm{\omega^{\onehalf} \mathsf{m}_{\kappa,\Lambda}}
_{\sphilb{H}}^2
-\onehalf
\bkt{\mathsf{m}_{\kappa,\Lambda}}
{\rbk{1 - \napiernum^{-T \omega}}
\mathsf{m}_{\kappa,\Lambda}}
_{\sphilb{H}}}
\end{aligned}
\end{equation}

\item
In the notation of Definition \ref{expedition0011372},
\begin{equation}
\begin{aligned}
&\bkt{1}
{\napiernum^{-T \physham_{\txtvanhowe,\kappa,\Lambda}} 1}
_{\fun{\lp^{2}}{\sphilb{M}_{-2},\msrsf{G}}}
=
\sqfun{\prbexp_{\msrcal{G}}}
{\napiernum^{-W_{\kappa,\Lambda}(T)}} \\
&=
\fnexp{
\frac{T}{2}
\norm{\omega^{\onehalf} \mathsf{m}_{\kappa,\Lambda}}_{\sphilb{H}}^2
-\onehalf
\bkt{\mathsf{m}_{\kappa,\Lambda}}
{\rbk{1 - \napiernum^{-T \omega}} \mathsf{m}_{\kappa,\Lambda}}
_{\sphilb{H}}}
\end{aligned}
\end{equation}
is obtained.

\item
The norm of the function $\napiernum^{-T \physham_{\txtvanhowe,\kappa,\Lambda}} 1$ can be estimated as
\begin{equation}
\begin{aligned}
&\norm{\napiernum^{-T \physham_{\txtvanhowe,\kappa,\Lambda}} 1}
_{\fun{\lp^{2}}{\sphilb{M}_{-2},\msrsf{G}}}
\\ 
&=
\fnexp{
\frac{T}{2}
\norm{\omega^{\onehalf} \mathsf{m}_{\kappa,\Lambda}}_{\sphilb{H}}^2
-\oneoverfour
\bkt{\mathsf{m}_{\kappa,\Lambda}}
{\rbk{1 - \napiernum^{-2T \omega}} \mathsf{m}_{\kappa,\Lambda}}
_{\sphilb{H}}}
\end{aligned}
\end{equation}

\end{enumerate}
\end{prop}

As will be shown, the first term in the variance of the pair potential is the term corresponding to the ground state energy: \[\physgse(\physham_{\txtvanhowe,\kappa,\Lambda})
=
-\onehalf
\norm{\omega^{\onehalf} \mathsf{m}_{\kappa,\Lambda}}_{\sphilb{H}}^2.\]

\begin{proof}
(1): The mean follows by linearity: $$\sqfun{\prbexp_{\msrcal{G}}}
{W_{\kappa,\Lambda}(T)}
=
\int_{0}^{T}
\sqfun{\prbexp_{\msrcal{G}}}
{\prbou_s(\omega \mathsf{m}_{\kappa,\Lambda})}
\opdmsr{s}
=
0.$$
The variance is computed by definition as
\begin{equation}
\begin{aligned}
&\sqfun{\prbvar_{\msrcal{G}}}
{W_{\kappa,\Lambda}(T)}
=
\sqfun{\prbexp_{\msrcal{G}}}
{W_{\kappa,\Lambda}(T)^2}
\\ 
&=
\int_0^T
\int_0^T
\sqfun{\prbexp_{\msrcal{G}}}
{\prbou_s(\omega \mathsf{m}_{\kappa,\Lambda})
\prbou_t(\omega \mathsf{m}_{\kappa,\Lambda})}
\opdmsr{s}
\opdmsr{t}
\\ 
&=
\onehalf
\int_0^T
\int_0^T
\bkt{\omega \mathsf{m}_{\kappa,\Lambda}}
{\napiernum^{-\abs{t-s} \omega} \omega \mathsf{m}_{\kappa,\Lambda}}
_{\sphilb{H}}
\opdmsr{s}
\opdmsr{t}
\end{aligned}
\end{equation}
can be computed.
Applying the double integral from Lemma \ref{expedition0011363}, $$\sqfun{\prbvar_{\msrcal{G}}}
{W_{\kappa,\Lambda}(T)}
=
T
\norm{\omega^{\onehalf} \mathsf{m}_{\kappa,\Lambda}}
_{\sphilb{H}}^2
-\bkt{\mathsf{m}_{\kappa,\Lambda}}
{\rbk{1 - \napiernum^{-T \omega}}
\mathsf{m}_{\kappa,\Lambda}}
_{\sphilb{H}}$$ is obtained.
The expectation of the exponential $\napiernum^{-W_{\kappa,\Lambda}(T)}$ follows easily from the variance.

(2): Set $F
= G
= 1$ in the Feynman-Kac-Nelson formula and compute the resulting object. The conclusion then follows from (1).

(3): Note that $$\norm{\napiernum^{-T \physham_{\txtvanhowe,\kappa,\Lambda}} 1}
_{\fun{\lp^{2}}{\sphilb{M}_{-2},\msrsf{G}}}
=
\bkt{1}{\napiernum^{-2T \physham_{\txtvanhowe,\kappa,\Lambda}} 1}
_{\fun{\lp^{2}}{\sphilb{M}_{-2},\msrsf{G}}}^{\onehalf}$$ and apply (2).
\end{proof}

Here also, introduce the approximate ground state \(\Psi_T\) and the criterion \(\gamma_{1}(T)\).

To apply Theorem \ref{expedition0011361}, for a positive real number \(T\), introduce \[\Psi_T
=
\frac{\napiernum^{-T \physham_{\txtvanhowe,\kappa,\Lambda}} 1}
{\norm{\napiernum^{-T \physham_{\txtvanhowe,\kappa,\Lambda}} 1}
_{\fun{\lp^{2}}{\sphilb{M}_{-2},\msrsf{G}}}}\] and the criterion \[\gamma_1(T)
=
\bkt{1}{\Psi_T}
_{\fun{\lp^{2}}{\sphilb{M}_{-2},\msrsf{G}}}
=
\frac{\bkt{1}
{\napiernum^{-T \physham_{\txtvanhowe,\kappa,\Lambda}} 1}
_{\fun{\lp^{2}}{\sphilb{M}_{-2},\msrsf{G}}}}
{\norm{\napiernum^{-T \physham_{\txtvanhowe,\kappa,\Lambda}} 1}
_{\fun{\lp^{2}}{\sphilb{M}_{-2},\msrsf{G}}}}.\]

\begin{prop}
The criterion $\gamma_1(T)$ can be estimated as $$\gamma_1(T)
=
\fnexp{-\oneoverfour
\bkt{\mathsf{m}_{\kappa,\Lambda}}
{\rbk{1 - \napiernum^{-T \omega}}^2 \mathsf{m}_{\kappa,\Lambda}}
_{\sphilb{H}}}.$$
\end{prop}

\begin{proof}
The numerator is computed by Proposition \ref{expedition0011501}(2) and the denominator by the same Proposition (3).
Collecting these,
\begin{equation}
\begin{aligned}
&\log \gamma_1(T)
\\ 
&=
\frac{T}{2} \norm{\mathsf{m}_{\kappa,\Lambda}}_{\sphilb{H}}^2
-\onehalf
\bkt{\mathsf{m}_{\kappa,\Lambda}}
{\rbk{1 - \napiernum^{-T \omega}} \mathsf{m}_{\kappa,\Lambda}}
_{\sphilb{H}}
\\
&\quad
-\frac{T}{2} \norm{\mathsf{m}_{\kappa,\Lambda}}_{\sphilb{H}}^2
+\oneoverfour
\bkt{\mathsf{m}_{\kappa,\Lambda}}
{\rbk{1 - \napiernum^{-2T \omega}} \mathsf{m}_{\kappa,\Lambda}}
_{\sphilb{H}}
\\ 
&=
-\oneoverfour
\bkt{\mathsf{m}_{\kappa,\Lambda}}
{\rbk{1 - \napiernum^{-T \omega}}^2 \mathsf{m}_{\kappa,\Lambda}}
_{\sphilb{H}}
\end{aligned}
\end{equation}
is obtained.
\end{proof}

We construct the ground state of the van Hove Hamiltonian.

\begin{thm}\label{expedition0011502}
Under infrared regularization, the van Hove Hamiltonian has a ground state.
In particular, the ground state energy is $$\physgse(\physham_{\txtvanhowe,\kappa,\Lambda})
=
-\onehalf
\norm{\omega^{\onehalf}
\mathsf{m}_{\kappa,\Lambda}}_{\sphilb{H}}^2,$$ and the ground state normalized in $\fun{\lp^{2}}{\sphilb{M}_{-2},\msrsf{G}}$ is $$\Psi_{\txtgs,\kappa,\Lambda}
=
\fnexp{-\onehalf
\norm{\mathsf{m}_{\kappa,\Lambda}}
_{\sphilb{H}}^2}
\napiernum^{-\prbou_0(\mathsf{m}_{\kappa,\Lambda})}.$$
In particular, the ground state is unique and strictly positive almost everywhere.
\end{thm}

\begin{proof}
(Existence and uniqueness of ground state): Argue as in Theorem \ref{expedition0011360} in the $Q$-representation.
In particular, $$\lim_{T \to \infty}
\gamma_1(T)
=
\fnexp{-\oneoverfour
\norm{\mathsf{m}_{\kappa,\Lambda}}_{\sphilb{H}}^2}
>
0$$ shows that a ground state exists.
Uniqueness follows from Proposition \ref{expedition0011507}.

(Ground state energy evaluation): By the preceding argument, $1
\in
\fun{\lp^{2}}{\sphilb{M}_{-2},\msrsf{G}}$ satisfies the hypothesis of Theorem \ref{expedition0004996}.
By Proposition \ref{expedition0011501}, the ground state energy is
\begin{equation}
\begin{aligned}
&\physgse(\physham_{\txtvanhowe,\kappa,\Lambda})
=
-\lim_{T \to \infty}
\frac{\log
\bkt{1}{\napiernum^{-T \physham_{\txtvanhowe,\kappa,\Lambda}} 1}
_{\fun{\lp^{2}}{\sphilb{M}_{-2},\msrsf{G}}}}
{T}
\\ 
&=
-\lim_{T \to \infty}
\rbk{
\onehalf
\norm{\omega^{\onehalf} \mathsf{m}_{\kappa,\Lambda}}
_{\sphilb{H}}^{2}
-\onehalf
\frac{1}{T}
\bkt{\mathsf{m}_{\kappa,\Lambda}}
{\rbk{1 - \napiernum^{-T \omega}} \mathsf{m}_{\kappa,\Lambda}}
_{\sphilb{H}}}
\\ 
&=
-\onehalf
\norm{\omega^{\onehalf} \mathsf{m}_{\kappa,\Lambda}}
_{\sphilb{H}}^{2}
\end{aligned}
\end{equation}
is obtained.

(Construction of ground state): As an application of the ground state energy evaluation argument, $\lim_{T \to \infty}
\Psi_T
=
\Psi_{\infty}$ exists.
In particular, $\Psi_{\infty}$ is a (non-normalized) ground state.

We prove this below.
For any $s
> 0$, $$\napiernum^{-s \physham} \Psi_T
=
\frac{\napiernum^{-(T+s) \physham} 1}
{\norm{\napiernum^{-T \physham} 1}
_{\fun{\lp^{2}}{\sphilb{M}_{-2},\msrsf{G}}}}
=
\frac{\norm{\napiernum^{-(T+s) \physham} 1}
_{\fun{\lp^{2}}{\sphilb{M}_{-2},\msrsf{G}}}}
{\norm{\napiernum^{-T \physham} 1}
_{\fun{\lp^{2}}{\sphilb{M}_{-2},\msrsf{G}}}}
\Psi_{T+s}$$ holds.
By the ground state energy evaluation, the coefficient converges to $\napiernum^{-s
\physgse(\physham_{\txtvanhowe,\kappa,\Lambda})}$.
Therefore, in the limit $T
\to \infty$, $$\napiernum^{-s \physham_{\txtvanhowe,\kappa,\Lambda}}
\Psi_{\infty}
=
\napiernum^{-s \physgse(\physham_{\txtvanhowe,\kappa,\Lambda})}
\Psi_{\infty}$$ is obtained.
This shows $\Psi_{\infty}$ is a ground state of $\physham$.

(Expression of the ground state): Use $\Psi_T$ and its limit.
By density, it suffices to examine $\bkt{F}{\Psi_{\infty}}$ for the time-$0$ projection $F
= \napiernum^{\prbou_0(f)}$.
The denominator of $\Psi_T$ before taking the limit has already been computed by Proposition \ref{expedition0011501}(2).

(Numerator evaluation): Denoting the numerator by $N_T$,
\begin{equation}
\begin{aligned}
&N_T(f)
=
\bkt{F}
{\napiernum^{-T \physham} 1}_{\fun{\lp^{2}}{\sphilb{M}_{-2},\msrsf{G}}}
\\ 
&=
\sqfun{\prbexp_{\msrcal{G}}}
{\napiernum^{\prbou_0(f)
-W_{\kappa,\Lambda}(T)}} \\
&=
\fnexp{\onehalf
\norm{\prbou_0(f)
-W_{\kappa,\Lambda}(T)}^2}
\\ 
&=
\fnexp{\onehalf \norm{\prbou_0(f)}^2
+\onehalf \norm{W_{\kappa,\Lambda}(T)}^2
-\bkt{\prbou_0(f)}{W_{\kappa,\Lambda}(T)}}
\end{aligned}
\end{equation}
can be written.
The first norm gives $$\onehalf
\norm{\prbou_0(f)}^2
=
\oneoverfour
\norm{f}_{\sphilb{H}}^2.$$
The second norm, by Proposition \ref{expedition0011501}, is
\begin{equation}
\begin{aligned}
&\onehalf
\norm{W_{\kappa,\Lambda}(T)}^2
=
-T \physgse
-\onehalf
\bkt{\mathsf{m}_{\kappa,\Lambda}}
{\rbk{1 - \napiernum^{-T \omega}} \mathsf{m}_{\kappa,\Lambda}}
_{\sphilb{H}}.
\end{aligned}
\end{equation}
The third inner product is
\begin{equation}
\begin{aligned}
&-\onehalf
\bkt{\prbou_0(f)}{W_{\kappa,\Lambda}(T)}
\\ 
&=
-\onehalf
\int_{0}^{T}
\bkt{f}
{\napiernum^{-\abs{t} \omega} \cdot \omega \mathsf{m}_{\kappa,\Lambda}}_{\sphilb{H}}
\opdmsr{t} \\
&=
-\onehalf
\bkt{f}
{\rbk{1 - \napiernum^{-T \omega}}
\mathsf{m}_{\kappa,\Lambda}}_{\sphilb{H}}.
\end{aligned}
\end{equation}

Collecting the terms,
\begin{equation}
\begin{aligned}
&N_T(f)
=
\fnexp{\oneoverfour
\norm{f}_{\sphilb{H}}^2
-\onehalf \bkt{f}
{\rbk{1 - \napiernum^{-T \omega}}
\mathsf{m}_{\kappa,\Lambda}}_{\sphilb{H}}} \\
&\qquad\times
\fnexp{-T \physgse
-\onehalf
\bkt{\mathsf{m}_{\kappa,\Lambda}}
{\rbk{1 - \napiernum^{-T \omega}} \mathsf{m}_{\kappa,\Lambda}}
_{\sphilb{H}}}
\end{aligned}
\end{equation}
is obtained.

(Normalized ground state): Substituting the above results into $\bkt{F}
{\Psi_T}_{\fun{\lp^{2}}{\sphilb{M}_{-2},\msrsf{G}}}$,
\begin{equation}
\begin{aligned}
&\bkt{\napiernum^{\prbou_0(f)}}
{\Psi_T}
_{\fun{\lp^{2}}{\sphilb{M}_{-2},\msrsf{G}}} \\
&=
\fnexp{\oneoverfour \norm{f}_{\sphilb{H}}^2
-\onehalf \bkt{f}{\rbk{1 - \napiernum^{-T \omega}} \mathsf{m}
_{\kappa,\Lambda}}_{\sphilb{H}}} \\
&\qquad\times
\fnexp{-\oneoverfour
\norm{\sqrt{1 - \napiernum^{-T \omega}} \mathsf{m}_{\kappa,\Lambda}}
_{\sphilb{H}}^2}
\end{aligned}
\end{equation}
is obtained.
Taking the limit $T
\to \infty$, $$\bkt{\napiernum^{\prbou_0(f)}}
{\Psi_{\infty}}
_{\fun{\lp^{2}}{\sphilb{M}_{-2},\msrsf{G}}}
=
\fnexp{\oneoverfour \norm{f}_{\sphilb{H}}^2
-\onehalf \bkt{f}{\mathsf{m}_{\kappa,\Lambda}}_{\sphilb{H}}
-\oneoverfour \norm{\mathsf{m}_{\kappa,\Lambda}}_{\sphilb{H}}^2}$$ is obtained.
The right-hand side coincides with $\bkt{\napiernum^{\prbou_0(f)}}
{\Phi}
_{\fun{\lp^{2}}{\sphilb{M}_{-2},\msrsf{G}}}$ for the function $$\Phi(\prbou)
=
\fnexp{-\oneoverfour
\norm{\mathsf{m}_{\kappa,\Lambda}}_{\sphilb{H}}^2}
\napiernum^{-\prbou_0(\mathsf{m}_{\kappa,\Lambda})}$$ in $\fun{\lp^{2}}{\sphilb{M}_{-2},\msrsf{G}}$.
Therefore $\Phi$ is the expectation-normalized version of $\Psi_{\infty}$.
For a ground state, the norm in $\fun{\lp^{2}}{\sphilb{M}_{-2},\msrsf{G}}$ must be normalized.
Consider the function $\tilde{\Psi}(\prbou)
= \napiernum^{-\prbou_0(\mathsf{m}_{\kappa,\Lambda})}$ with the constant factor removed.
Its $\fun{\lp^{2}}{\sphilb{M}_{-2},\msrsf{G}}$-norm is
\begin{equation}
\begin{aligned}
&\norm{\tilde{\Psi}}
_{\fun{\lp^{2}}{\sphilb{M}_{-2},\msrsf{G}}}
=
\sqfun{\prbexp_{\msrcal{G}}}
{\napiernum^{-2 \prbou_0(\mathsf{m}_{\kappa,\Lambda})}}^{\onehalf}
\\ 
&=
\fnexp{\onehalf
\cdot
4
\prbvar
\prbou_0(\mathsf{m}_{\kappa,\Lambda})}^{\onehalf}
=
\fnexp{\onehalf
\norm{\mathsf{m}_{\kappa,\Lambda}}_{\sphilb{H}}^2}
\end{aligned}
\end{equation}

In particular, the ground state normalized in $\fun{\lp^{2}}{\sphilb{M}_{-2},\msrsf{G}}$ is $$\Psi_{\txtgs,\kappa,\Lambda}(\prbou)
=
\fnexp{-\onehalf
\norm{\mathsf{m}_{\kappa,\Lambda}}
_{\sphilb{H}}^2}
\napiernum^{-\prbou_0(\mathsf{m}_{\kappa,\Lambda})}.$$
\end{proof}

Here too, we show the non-existence of a ground state of the van Hove Hamiltonian under infrared-singular conditions.

\begin{thm}
Under infrared-singular conditions, there is no ground state in $\fun{\lp^{2}}{\sphilb{M}_{-2},\msrsf{G}}$.
\end{thm}

\begin{proof}
Imposing the infrared-singular condition, the criterion satisfies $$\lim_{T \to \infty}
\gamma_1(T)
=
\fnexp{-\oneoverfour
\norm{\mathsf{m}_{\Lambda}}_{\sphilb{H}}^2}
=
0$$ by the preceding argument.
This is equivalent to the non-existence of a ground state.
\end{proof}

\subsection{Preparation toward Removal of Infrared/Ultraviolet Cutoffs}\label{preparation-toward-removal-of-infraredultraviolet-cutoffs}

Based on the results from operator theory and operator algebras, the existence of a ground state can be proved even under infrared-singular conditions with appropriate treatment, so we investigate what the appropriate probabilistic treatment should be. More specifically, we proceed as follows.

\begin{itemize}

\item
  We consider the ground state transformation with cutoffs and discuss how the Hamiltonian and ground state change.
\item
  We examine the probabilistic Markov property of the objects transformed by the ground state transformation.
\item
  We examine the existence of the objects with infrared and ultraviolet cutoffs removed and their properties, in particular the Markov property of the semigroup formed by the generator obtained by transforming the Hamiltonian.
\item
  We discuss whether the objects with cutoffs, in particular the generator, semigroup, measure, and ground state, converge to the objects with cutoffs removed in an appropriate topology.
\end{itemize}

Denote the non-negative operator combining the van Hove Hamiltonian and ground state energy by \[\tilde{\physham}_{\txtvanhowe,\txtrenormalization,\kappa,\Lambda}
=
\physham_{\txtvanhowe,\txtrenormalization,\kappa,\Lambda}
-\physgse(\physham_{\txtvanhowe,\txtrenormalization,\kappa,\Lambda}).\]

Using the ground state \(\Psi_{\txtgs,\kappa,\Lambda}\) from Theorem \ref{expedition0011502}, define the measure on \(\sphilb{M}_{-2}\) by \[\opdmsr{\msrsf{G}_{\txtgs,\kappa,\Lambda}(\opfocksegal)}
=
\Psi_{\txtgs,\kappa,\Lambda}(\opfocksegal)^2
\opdmsr{\msrsf{G}(\opfocksegal)}.\] Equivalently, \(\Psi_{\txtgs,\kappa,\Lambda}(\opfocksegal)^2\) is the Radon-Nikodym derivative. Furthermore, define the unitary transformation \(\mathsf{U}_{\kappa,\Lambda}\) by \[\mathsf{U}_{\kappa,\Lambda}
\colon \fun{\lp^{2}}{\sphilb{M}_{-2},\msrsf{G}_{\txtgs,\kappa,\Lambda}}
\to \fun{\lp^{2}}{\sphilb{M}_{-2},\msrsf{G}};
\quad
\mathsf{U}_{\kappa,\Lambda} F
=
\Psi_{\txtgs,\kappa,\Lambda} F.\] This is the static ground state distribution and the stationary distribution of the path measure that will appear later.

\begin{lem}\label{expedition0011548}
The mean of the Gaussian measure $\msrsf{G}_{\txtgs,\kappa,\Lambda}$ is $$\sqfun{\prbexp_{\msrsf{G}_{\txtgs,\kappa,\Lambda}}}
{\prbou_0(f)}
=
-\bkt{f}{\mathsf{m}_{\kappa,\Lambda}}_{\sphilb{H}},$$ and its covariance is $$\sqfun{\prbcov_{\msrsf{G}_{\txtgs,\kappa,\Lambda}}}
{\prbou_0(f),\prbou_0(g)}
=
\onehalf
\bkt{f}{g}
_{\sphilb{H}}.$$
\end{lem}

\begin{proof}
(Mean): Compute using the Radon-Nikodym derivative.
For the original measure $\msrsf{G}$, $$\sqfun{\prbcov_{\msrsf{G}}}
{\prbou_0(f),\prbou_0(\mathsf{m}_{\kappa,\Lambda})}
=
\onehalf
\bkt{f}{\mathsf{m}_{\kappa,\Lambda}}_{\sphilb{H}}.$$
Considering the moment generating function of the two-dimensional Gaussian random variable, $$\sqfun{\prbexp_{\msrsf{G}}}
{\prbou_0(f) \napiernum^{t \prbou_0(\mathsf{m}_{\kappa,\Lambda})}}
=
t
\cdot
\sqfun{\prbcov_{\msrsf{G}}}
{\prbou_0(f), \prbou_0(\mathsf{m}_{\kappa,\Lambda})}
\napiernum^{\onehalf
t^2
\sqfun{\prbvar_{\msrsf{G}}}
{\prbou_0(\mathsf{m}_{\kappa,\Lambda})}}$$ holds.
Setting $t
= -2$,
\begin{equation}
\begin{aligned}
\sqfun{\prbexp_{\msrsf{G}}}
{\napiernum^{-2 \prbou_0(\mathsf{m}_{\kappa,\Lambda})}}
&=
\fnexp{\norm{\mathsf{m}_{\kappa,\Lambda}}_{\sphilb{H}}^2},
\\ 
\sqfun{\prbexp_{\msrsf{G}}}
{\prbou_0(f) \napiernum^{-2 \prbou_0(\mathsf{m}_{\kappa,\Lambda})}}
&=
-2
\sqfun{\prbcov_{\msrsf{G}}}
{\prbou_0(f), \prbou_0(\mathsf{m}_{\kappa,\Lambda})}
\napiernum^{2
\sqfun{\prbvar_{\msrsf{G}}}
{\prbou_0(\mathsf{m}_{\kappa,\Lambda})}}
\end{aligned}
\end{equation}
is obtained.
Taking into account the $\lp^{2}$-normalization of the ground state,
\begin{equation}
\begin{aligned}
&\sqfun{\prbexp_{\msrsf{G}_{\kappa,\Lambda}}}
{\prbou_0(f)}
=
\sqfun{\prbexp_{\msrsf{G}}}
{\prbou_0(f)
\Psi_{\txtgs,\kappa,\Lambda}^2}
=
\sqfun{\prbcov_{\msrsf{G}}}
{\prbou_0(f),\Psi_{\txtgs,\kappa,\Lambda}^2}
\\ 
&=
-2 \sqfun{\prbcov_{\msrsf{G}}}
{\prbou_0(f), \prbou_0(\mathsf{m}_{\kappa,\Lambda})}
=
-\bkt{f}{\mathsf{m}_{\kappa,\Lambda}}_{\sphilb{H}}
\end{aligned}
\end{equation}
is obtained.

(Covariance): This Radon-Nikodym derivative is a translation of the mean of the Gaussian measure and does not change the variance.
Therefore the variance agrees with that of $\msrsf{G}$.
\end{proof}

Using the operators \(\physham[L]_0\) and \(\mathsf{U}_{\kappa,\Lambda}\), \begin{equation}
\begin{aligned}
&\physham[L]_{\kappa,\Lambda}
=
\inv{\mathsf{U}_{\kappa,\Lambda}}
\tilde{\physham}_{\txtvanhowe,\kappa,\Lambda}
\mathsf{U}_{\kappa,\Lambda}
\\ 
&=
\inv{\Psi_{\txtgs,\kappa,\Lambda}}
\rbk{\physham[L]_0 + \prbou(\omega \mathsf{m}_{\kappa,\Lambda}) - \physgse(\physham_{\txtvanhowe,\kappa,\Lambda})}
\Psi_{\txtgs,\kappa,\Lambda}
\end{aligned}
\end{equation} is defined. In particular, for all \(\kappa,\Lambda
> 0\), the domains of \(\physham[L]_{\kappa,\Lambda}\) agree in the sense of unitary equivalence. As will be discussed in detail later, \(L_{\kappa,\Lambda}\) is the generator of the Markov process obtained by adding a drift from the ground state to the generator of the free Ornstein-Uhlenbeck process.

\begin{prop}
Under the unitary transformation by the ground state transformation, the van Hove Hamiltonian maps to $\physham[L]_{\kappa,\Lambda}$, and the ground state maps from $\Psi_{\txtgs,\kappa,\Lambda}
\in \fun{\lp^{2}}{\sphilb{M}_{-2},\msrsf{G}}$ to $1
\in \fun{\lp^{2}}{\sphilb{M}_{-2},\msrsf{G}_{\txtgs,\kappa,\Lambda}}$.
\end{prop}

The ground state \(1\) on \(\fun{\lp^{2}}{\sphilb{M}_{-2},\msrsf{G}_{\txtgs,\kappa,\Lambda}}\) obtained by the ground state transformation is independent of \(\kappa,\Lambda\). This suggests the possibility that a ground state exists in the limit where the infrared/ultraviolet cutoffs are removed. In particular, it can also be interpreted as the ground state of the generator \(\physham[L]_{\kappa,\Lambda}\). It remains to discuss whether the limits of the associated generator, measure, and semigroup exist and become non-trivial objects with the appropriate properties. In particular, since Markovianity in the probabilistic sense is the key at various stages, Markovianity must be ensured at each step.

\begin{proof}
Both follow directly from the definitions.
\end{proof}

Let us investigate the drift structure of the generator \(\physham[L]_{\kappa,\Lambda}\).

\begin{prop}\label{expedition0011574}
For any cylindrical function $F$, $$\physham[L]_{\kappa,\Lambda} F
=
\physham[L]_0 F
+2
\bkt{\vagrad F}
{\omega \mathsf{m}_{\kappa,\Lambda}}
_{\sphilb{H}}$$ holds.
\end{prop}

The gradient on cylindrical functions was formulated in Definition \ref{expedition0011538}. What is added by the ground state transformation is a drift term, corresponding in particular to a translation of the mean.

\begin{proof}
By definition, the gradient of the ground state is $$\vagrad
\Psi_{\txtgs,\kappa,\Lambda}
=
-\Psi_{\txtgs,\kappa,\Lambda}
\mathsf{m}_{\kappa,\Lambda}.$$
For a cylindrical function in the sense of Definition \ref{expedition0011537}, $$\vagrad(\Psi_{\txtgs,\kappa,\Lambda} F)
=
F \vagrad \Psi_{\txtgs,\kappa,\Lambda}
+\Psi_{\txtgs,\kappa,\Lambda} \vagrad F$$ holds.
Furthermore, the generator $\physham[L]_0$ satisfies $$\physham[L]_0 (\Psi_{\txtgs,\kappa,\Lambda} F)
=
\Psi_{\txtgs,\kappa,\Lambda} \physham[L]_0 F
+F \physham[L]_0 \Psi_{\txtgs,\kappa,\Lambda}
+2
\bkt{\vagrad \Psi_{\txtgs,\kappa,\Lambda}}
{\vagrad F}
_{\sphilb{H}}.$$
Since the generator $\physham[L]_0$ is actually the free Hamiltonian, $$0
=
\tilde{\physham}_{\txtvanhowe,\kappa,\Lambda}
\Psi_{\txtgs,\kappa,\Lambda}
=
\rbk{L_0 + \prbou(\omega \mathsf{m}_{\kappa,\Lambda}) - \physgse(\physham_{\txtvanhowe,\kappa,\Lambda})}
\Psi_{\txtgs,\kappa,\Lambda},$$ i.e., $$\physham[L]_0
\Psi_{\txtgs,\kappa,\Lambda}
=
\rbk{\physgse(\physham_{\txtvanhowe,\kappa,\Lambda})
-\prbou(\omega \mathsf{m}_{\kappa,\Lambda})}
\Psi_{\txtgs,\kappa,\Lambda}$$ is obtained.
Using these,
\begin{equation}
\begin{aligned}
&\physham[L]_{\kappa,\Lambda} F
\\ 
&=
\inv{\Psi_{\txtgs,\kappa,\Lambda}}
\rbk{\physham[L]_0 (\Psi_{\txtgs,\kappa,\Lambda} F)
+\prbou(\omega \mathsf{m}_{\kappa,\Lambda})
\Psi_{\txtgs,\kappa,\Lambda} F
-\physgse(\physham_{\txtvanhowe,\kappa,\Lambda}) \Psi_{\txtgs,\kappa,\Lambda} F}
\\ 
&=
\inv{\Psi_{\txtgs,\kappa,\Lambda}}
\rbkleft{\Psi_{\txtgs,\kappa,\Lambda} \physham[L]_0 F
+ F \physham[L]_0 \Psi_{\txtgs,\kappa,\Lambda}
+2
\bkt{\vagrad \Psi_{\txtgs,\kappa,\Lambda}}
{\vagrad F}}
\\
&\quad
\rbkright{
+\prbou(\omega \mathsf{m}_{\kappa,\Lambda}) \Psi_{\txtgs,\kappa,\Lambda} F
-\physgse(\physham_{\txtvanhowe,\kappa,\Lambda})
\Psi_{\txtgs,\kappa,\Lambda} F}
\\ 
&=
\physham[L]_0 F
+2
\inv{\Psi_{\txtgs,\kappa,\Lambda}}
\bkt{\vagrad \Psi_{\txtgs,\kappa,\Lambda}}
{\vagrad F}
\\ 
&=
\physham[L]_0 F
+2\bkt{\vagrad F}
{\omega \mathsf{m}_{\kappa,\Lambda}}
_{\sphilb{H}}
\end{aligned}
\end{equation}
is obtained.
\end{proof}

Using the operator \(\physham[L]_{\kappa,\Lambda}\), define the semigroup \[T_{\kappa,\Lambda,t}
=
\napiernum^{-t \physham[L]_{\kappa,\Lambda}}
=
\inv{\Psi_{\txtgs,\kappa,\Lambda}}
\napiernum^{-t \tilde{\physham}_{\txtvanhowe,\kappa,\Lambda}}
\Psi_{\txtgs,\kappa,\Lambda}.\]

This is the Markov semigroup after the ground state transformation. Although the semigroup property is clear, the Markov property is not obvious at this stage of the definition. Furthermore, while the Markov property as a semigroup follows from Proposition \ref{expedition0011511}, the Markov property as a stochastic process is only established in Proposition \ref{expedition0011514}.

\begin{prop}\label{expedition0011511}
The following statements hold for the operator semigroup $T_{\kappa,\Lambda,t}$.
\begin{enumerate}
\item
Conservation: $T_{\kappa,\Lambda,t} 1
= 1$ holds.

\item
It is a positivity-improving operator.

\item
It is a symmetric semigroup with $\msrsf{G}_{\txtgs,\kappa,\Lambda}$ as its invariant measure.
\end{enumerate}
In particular, the semigroup $T_{\kappa,\Lambda}
= \fml{T_{\kappa,\Lambda,t}}{t > 0}$ is a canonical Markov semigroup.
\end{prop}

\begin{proof}
(1): This follows from the fact that $\Psi_{\txtgs,\kappa,\Lambda}$ is the ground state of the van Hove Hamiltonian.

(2): Fix any $F
\geq 0$.
By Theorem \ref{expedition0011502}, $\Psi_{\txtgs,\kappa,\Lambda} F
\geq 0$.
By Proposition \ref{expedition0011507}, $\napiernum^{-t \tilde{\physham}_{\txtvanhowe,\kappa,\Lambda}}
(\Psi_{\txtgs,\kappa,\Lambda} F)
> 0$.
Since the ground state is strictly positive almost everywhere,$$T_{\kappa,\Lambda,t} F
=
\inv{\Psi_{\txtgs,\kappa,\Lambda}}
\napiernum^{-t \tilde{\physham}_{\txtvanhowe,\kappa,\Lambda}}
(\Psi_{\txtgs,\kappa,\Lambda} F)
> 0$$holds.

((3), Symmetry):
Choose arbitrary $F,G
\in
\fun{\lp^{2}}{\sphilb{M}_{-2},\msrsf{G}_{\txtgs,\kappa,\Lambda}}$.
By the Radon-Nikod\'{y}m derivative and the self-adjointness of the van Hove Hamiltonian,
\begin{equation}
\begin{aligned}
&\bkt{F}{T_{\kappa,\Lambda,t} G}
_{\fun{\lp^{2}}{\sphilb{M}_{-2},\msrsf{G}_{\txtgs,\kappa,\Lambda}}}
\\ 
&=
\sqfun{\prbexp_{\msrsf{G}_{\txtgs,\kappa,\Lambda}}}
{\cmpconj{F}
\inv{\Psi_{\txtgs,\kappa,\Lambda}}
\napiernum^{-t \tilde{\physham}_{\txtvanhowe,\kappa,\Lambda}}
(\Psi_{\txtgs,\kappa,\Lambda} G)}
\\ 
&=
\sqfun{\prbexp_{\msrsf{G}}}
{\cmpconj{F}
\inv{\Psi_{\txtgs,\kappa,\Lambda}}
\napiernum^{-t \tilde{\physham}_{\txtvanhowe,\kappa,\Lambda}}
(\Psi_{\txtgs,\kappa,\Lambda} G)
\cdot
\Psi_{\txtgs,\kappa,\Lambda}^2}
\\ 
&=
\sqfun{\prbexp_{\msrsf{G}}}
{\cmpconj{\Psi_{\txtgs,\kappa,\Lambda} F}
\napiernum^{-t \tilde{\physham}_{\txtvanhowe,\kappa,\Lambda}}
(\Psi_{\txtgs,\kappa,\Lambda} G)}
\\ 
&=
\bkt{\Psi_{\txtgs,\kappa,\Lambda} F}
{\napiernum^{-t \tilde{\physham}_{\txtvanhowe,\kappa,\Lambda}}
(\Psi_{\txtgs,\kappa,\Lambda} G)}
_{\fun{\lp^{2}}{\sphilb{M}_{-2},\msrsf{G}}}
\\ 
&=
\bkt{\napiernum^{-t \tilde{\physham}_{\txtvanhowe,\kappa,\Lambda}}
(\Psi_{\txtgs,\kappa,\Lambda} F)}
{\Psi_{\txtgs,\kappa,\Lambda} G}
_{\fun{\lp^{2}}{\sphilb{M}_{-2},\msrsf{G}}}
\end{aligned}
\end{equation}
holds, giving symmetry.

((3), Invariance):
For any $F$ and $t
> 0$, we need $\sqfun{\prbexp_{\msrsf{G}_{\txtgs,\kappa,\Lambda}}}
{T_{\kappa,\Lambda,t} F}
=
\sqfun{\prbexp_{\msrsf{G}_{\txtgs,\kappa,\Lambda}}}
{F}$.
Using the inner product identity $\sqfun{\prbexp_{\msrsf{G}_{\txtgs,\kappa,\Lambda}}}
{T_{\kappa,\Lambda,t} F}
=
\bkt{1}
{T_{\kappa,\Lambda,t} F}
_{\fun{\lp^{2}}{\sphilb{M}_{-2},\msrsf{G}}}$
and applying (1) and (2) of this proposition,$$\bkt{1}
{T_{\kappa,\Lambda,t} F}
_{\fun{\lp^{2}}{\sphilb{M}_{-2},\msrsf{G}}}
=
\bkt{T_{\kappa,\Lambda,t} 1}
{F}
_{\fun{\lp^{2}}{\sphilb{M}_{-2},\msrsf{G}}}
=
\bkt{1}
{F}
_{\fun{\lp^{2}}{\sphilb{M}_{-2},\msrsf{G}}}$$holds, which is the desired result.
\end{proof}

In what follows, we discuss the Markov property of the semigroup \(T_{\kappa,\Lambda}\) as a stochastic process. The abstract framework involves the discussion in Section \ref{expedition0011536} on projective systems of probability measures and weak convergence, and we proceed with the preparations needed to apply it.

First, we define the path measure on a finite interval. Let the path space over the finite interval \(\closedinterval{-T}{T}\) be\[\Xi_T
=
\fun{\conti}{\closedinterval{-T}{T}; \sphilb{M}_{-2}}\]and define the normalization constant\[\smpartitionfunc_{\kappa,\Lambda,T}
=
\sqfun{\prbexp_{\msrcal{G}}}
{\Psi_{\txtgs,\kappa,\Lambda}(\prbou_{-T})
\Psi_{\txtgs,\kappa,\Lambda}(\prbou_{T})
\napiernum^{-\int_{-T}^{T}
\rbk{\prbou_{s}(\omega \mathsf{m}_{\kappa,\Lambda})
-\physgse(\physham_{\txtvanhowe,\kappa,\Lambda})}
\opdmsr{s}}}.\]This normalization constant is also called the partition function.

We then define the measure \(\msrcal{G}_{\txtgs,\kappa,\Lambda,T}\) by \begin{equation}
\begin{aligned}
&\fnrestr{\od{\msrcal{G}_{\txtgs,\kappa,\Lambda,T}}
{\msrcal{G}}}
{\Xi_T}
(\prbou)
\\ 
&=
\frac{\Psi_{\txtgs,\kappa,\Lambda}(\prbou_{-T})
\Psi_{\txtgs,\kappa,\Lambda}(\prbou_{T})
\napiernum^{-\int_{-T}^{T}
\rbk{\prbou_{s}(\omega \mathsf{m}_{\kappa,\Lambda})
-\physgse(\physham_{\txtvanhowe,\kappa,\Lambda})}
\opdmsr{s}}}
{\smpartitionfunc_{\kappa,\Lambda,T}}
\end{aligned}
\end{equation}

\begin{prop}
The partition function satisfies $\smpartitionfunc_{\kappa,\Lambda,T}
=
1$.
Henceforth, the path measure on the finite interval is described under this convention.
\end{prop}

\begin{proof}
The Ornstein-Uhlenbeck process is a stationary process.
In particular, we may replace the interval $\closedinterval{-T}{T}$ by $\closedinterval{0}{2T}$, giving$$\smpartitionfunc_{\kappa,\Lambda,T}
=
\sqfun{\prbexp_{\msrcal{G}}}
{\Psi_{\txtgs,\kappa,\Lambda}(\prbou_{0})
\Psi_{\txtgs,\kappa,\Lambda}(\prbou_{2T})
\napiernum^{-\int_{0}^{2T}
\rbk{\prbou_{s}(\omega \mathsf{m}_{\kappa,\Lambda})
-\physgse(\physham_{\txtvanhowe,\kappa,\Lambda})}
\opdmsr{s}}}.$$Applying the Feynman-Kac-Nelson formula with $F
= G
= \Psi_{\txtgs,\kappa,\Lambda}$, and noting the action of the van Hove Hamiltonian on the ground state,
\begin{equation}
\begin{aligned}
&\smpartitionfunc_{\kappa,\Lambda,T}
=
\bkt{\Psi_{\txtgs,\kappa,\Lambda}}
{\napiernum^{-2T \tilde{\physham}_{\txtvanhowe,\kappa,\Lambda}}
\Psi_{\txtgs,\kappa,\Lambda}}
_{\fun{\lp^{2}}{\sphilb{M}_{-2},\msrsf{G}}}
\\ 
&=
\bkt{\Psi_{\txtgs,\kappa,\Lambda}}
{\Psi_{\txtgs,\kappa,\Lambda}}
_{\fun{\lp^{2}}{\sphilb{M}_{-2},\msrsf{G}}}
=
1.
\end{aligned}
\end{equation}
This gives the desired result.
\end{proof}

Let \(\pi_{\closedinterval{-S}{S}}\) denote the restriction map on paths. When \(\msrcal{G}_{\txtgs,\kappa,\Lambda,T}
\circ
\inv{\pi_{\closedinterval{-S}{S}}}\) is independent of \(T\) under the condition \(T
\geq S\), we say that the finite-interval measures are consistent.

\begin{prop}\label{expedition0011508}
The family of finite-interval measures $\fml{\msrcal{G}_{\txtgs,\kappa,\Lambda,T}}
{T > 0}$ is consistent as a projective system of probability measures, and is projectively tight in the sense of Definition \ref{expedition0011530}.
In particular, as the projective limit of probability measures on the projective limit space $\fun{\conti}{\fldreal;\sphilb{M}_{-2}}$, the all-time path measure $\msrcal{G}_{\txtgs,\kappa,\Lambda}$ exists.
Moreover, the projective limit measure $\msrcal{G}_{\txtgs,\kappa,\Lambda}$ is also the weak limit of the finite-interval measures:$$\msrcal{G}_{\txtgs,\kappa,\Lambda}
=
\wlim_{T \to \infty}
\msrcal{G}_{\txtgs,\kappa,\Lambda,T}.$$
\end{prop}

\begin{proof}
(Preparation): Let $F$ be a cylinder function supported on $\closedinterval{-S}{S}$ in the sense of Definition \ref{expedition0011537}.
Under the condition $T
\geq S$,$$\sqfun{\prbexp_{\msrcal{G}_{\txtgs,\kappa,\Lambda,T}}}
{F}
=
\sqfun{\prbexp_{\msrcal{G}}}
{F(\prbou)
\Psi_{\txtgs,\kappa,\Lambda}(\prbou_{-T})
\Psi_{\txtgs,\kappa,\Lambda}(\prbou_{T})
\napiernum^{-\int_{-T}^{T}
\rbk{\prbou_{s}(\omega \mathsf{m}_{\kappa,\Lambda})
-\physgse}
\opdmsr{s}}}$$holds.
We decompose this into three intervals $\closedinterval{-T}{-S},
\closedinterval{-S}{S},
\closedinterval{S}{T}$.
The middle interval $\closedinterval{-S}{S}$ is independent of $T$, while the outer intervals $\closedinterval{-T}{-S}$ and $\closedinterval{S}{T}$ contribute $1$ using the ground state property.
We now prove this fact.
For conciseness, we denote the van Hove Hamiltonian and ground state energy simply by $\physham$ and $\physgse$.

(Simplification of outer intervals): Using the fact that splitting the integration interval results in a product,
for the left outer interval
\begin{equation}
\begin{aligned}
&\Psi_{\txtgs,\kappa,\Lambda}(\prbou_{-T})
\napiernum^{-\int_{-T}^{-S}
\rbk{\prbou_{s}(\omega \mathsf{m}_{\kappa,\Lambda})
-\physgse}
\opdmsr{s}}
\\ 
&=
\funrbk{\napiernum^{-(T-S) \tilde{\physham}}
\Psi_{\txtgs,\kappa,\Lambda}}
{\prbou_{-S}}
=
\Psi_{\txtgs,\kappa,\Lambda}(\prbou_{-S}),
\end{aligned}
\end{equation}
it suffices to show this.
Since the Ornstein-Uhlenbeck process is stationary, we may translate time.
Setting $\eta_u
= \prbou_{u-T}$, the process $\fml{\eta_u}{u \in \fldreal}$ is an Ornstein-Uhlenbeck process with the same law $\msrcal{G}$.
In particular,
\begin{equation}
\begin{aligned}
\prbou_{-T}
&=
\eta_0,
\\ 
\prbou_{-S}
&=
\eta_{T-S},
\\ 
\int_{-T}^{-S}
\rbk{\prbou_{s}(\omega \mathsf{m}_{\kappa,\Lambda}) - \physgse}
\opdmsr{s}
&=
\int_{0}^{T-S}
\rbk{\eta_{u}(\omega \mathsf{m}_{\kappa,\Lambda}) - \physgse}
\opdmsr{u}
\end{aligned}
\end{equation}
holds.
Renaming the Ornstein-Uhlenbeck process $\eta$ back to $\prbou$,
the left outer interval becomes
\begin{equation}
\begin{aligned}
&\Psi_{\txtgs,\kappa,\Lambda}(\prbou_{-T})
\napiernum^{-\int_{-T}^{-S}
\rbk{\prbou_{s}(\omega \mathsf{m}_{\kappa,\Lambda})
-\physgse}
\opdmsr{s}}
\\ 
&=
\Psi_{\txtgs,\kappa,\Lambda}(\prbou_{0})
\napiernum^{-\int_{0}^{T-S}
\rbk{\prbou_{s}(\omega \mathsf{m}_{\kappa,\Lambda})
-\physgse}
\opdmsr{s}}
\end{aligned}
\end{equation}
can be rewritten in this form.
Applying the Feynman-Kac-Nelson formula with $F
= \Psi_{\txtgs,\kappa,\Lambda}$ and $t
= T-S$,
\begin{equation}
\begin{aligned}
&\bkt{\Psi_{\txtgs,\kappa,\Lambda}}
{\napiernum^{-(T-S) \tilde{\physham}} G}
_{\fun{\lp^{2}}{\sphilb{M}_{-2},\msrsf{G}}}
\\ 
&=
\sqfun{\prbexp_{\msrcal{G}}}
{\Psi_{\txtgs,\kappa,\Lambda}(\prbou_0)
G(\prbou_{T-S})
\napiernum^{-\int_{0}^{T-S}
\rbk{\prbou_u(\omega \mathsf{m}_{\kappa,\Lambda}) - \physgse}
\opdmsr{u}}}
\end{aligned}
\end{equation}
is obtained.
Choosing $\Phi$ as a test function corresponding to evaluation at time $-S$,
in particular taking $\Phi$ to be a bounded measurable function,
\begin{equation}
\begin{aligned}
&\sqfun{\prbexp_{\msrcal{G}}}
{\Psi_{\txtgs,\kappa,\Lambda}(\prbou_0)
\Phi(\prbou_{T-S})
\napiernum^{-\int_{0}^{T-S}
\rbk{\prbou_u(\omega \mathsf{m}_{\kappa,\Lambda}) - \physgse}
\opdmsr{u}}}
\\ 
&=
\bkt{\Psi_{\txtgs,\kappa,\Lambda}}
{\napiernum^{-(T-S) \tilde{\physham}}
\Phi}
_{\fun{\lp^{2}}{\sphilb{M}_{-2},\msrsf{G}}}
\\ 
&=
\bkt{\Psi_{\txtgs,\kappa,\Lambda}}
{\Phi}
_{\fun{\lp^{2}}{\sphilb{M}_{-2},\msrsf{G}}}
\end{aligned}
\end{equation}
is obtained.
Viewing this as an identity for the distribution of $\prbou_{T-S}$, it can be interpreted as follows: for Ornstein-Uhlenbeck paths on $\closedinterval{0}{T-S}$, the expectation with weight$$\Psi_{\txtgs,\kappa,\Lambda}(\prbou_{0})
\napiernum^{-\int_0^{T-S}
\rbk{\prbou_{u}(\omega \mathsf{m}_{\kappa,\Lambda}) - \physgse}
\opdmsr{u}}$$applied to the initial value $\prbou_{0}$ coincides with the expected value of $\Psi_{\txtgs,\kappa,\Lambda}(\prbou_{T-S})$ at the endpoint.
This establishes the desired behavior for the left outer interval.
The right outer interval is handled analogously.

(Consistency of finite-interval measures and construction of the all-time path measure): From the above discussion,$$\sqfun{\prbexp_{\msrcal{G}_{\txtgs,\kappa,\Lambda,T}}}
{F}
=
\sqfun{\prbexp_{\msrcal{G}}}
{F(\prbou)
\Psi_{\txtgs,\kappa,\Lambda}(\prbou_{-S})
\Psi_{\txtgs,\kappa,\Lambda}(\prbou_{S})
\napiernum^{-\int_{-S}^{S}
\rbk{\prbou_{s}(\omega \mathsf{m}_{\kappa,\Lambda})
-\physgse}
\opdmsr{s}}}$$is obtained.
Using the notation for finite-interval measures, the right-hand side is independent of $T$.
Under the condition $T
\geq S$, the quantity $\msrcal{G}_{\txtgs,\kappa,\Lambda,T}
\circ
\inv{\pi_{\closedinterval{-S}{S}}}$ is constant in $T$.
That is, the family of probability measures formed by the finite-interval measures is consistent in the sense of Definition \ref{expedition0011527}, and in particular forms a projective system of probability measures.
By the Kolmogorov extension theorem \ref{expedition0011529} for projective systems of probability measures, the all-time path measure $\msrcal{G}_{\txtgs,\kappa,\Lambda}$ can be constructed as the projective limit measure.

(Path-continuity preservation): Path-continuity preservation here is understood in the sense of Proposition \ref{expedition0011531}.
To align with the setting of Proposition \ref{expedition0011531}, take the index set $I
= \openinterval{0}{\infty}$ and set $C_T
=
\fun{\conti}{\closedinterval{-T}{T};\sphilb{M}_{-2}}$ for each $T
> 0$.
In particular, $\msrcal{G}_{\txtgs,\kappa,\Lambda,T}(C_T)
= 1$ holds for each $T
> 0$.
For the projective limit space $C
= \fun{\conti}{\fldreal;\sphilb{M}_{-2}}$, by Proposition \ref{expedition0011531},$$\fun{\msrcal{G}_{\txtgs,\kappa,\Lambda}}
{\fun{\conti}{\fldreal;\sphilb{M}_{-2}}}
=
1$$holds.

(Tightness and weak convergence): Again fix $S$ with $T
\geq S$.
Since the space of continuous functions taking values in a separable Banach space is a Polish space, $\fun{\conti}{\closedinterval{-S}{S};\sphilb{M}_{-2}}$ is a Polish space.
Every probability measure on a Polish space is tight, so all probability measures on this space are tight.
By tightness, for any $\ep
> 0$, there exists a compact set $K_S
\subset
\fun{\conti}{\closedinterval{-S}{S};\sphilb{M}_{-2}}$ such that $\msrcal{G}_{\txtgs,\kappa,\Lambda,T}
\circ
\inv{\pi_{\closedinterval{-S}{S}}}(\setcpl{K_S})
< \ep$.
Since the right-hand side is independent of $T$, under the condition $T
\geq S$,$$\sup_{T \geq S}
\msrcal{G}_{\txtgs,\kappa,\Lambda,T}
\circ
\inv{\pi_{\closedinterval{-S}{S}}}(\setcpl{K_S})
\leq
\ep$$holds.
For any fixed $S
> 0$, one can always choose a compact set $K_S$ that bounds the measure of the restriction to $\closedinterval{-S}{S}$ to arbitrary precision.
In particular, $\fml{\msrcal{G}_{\txtgs,\kappa,\Lambda,T}}{T > 0}$ is tight as a projective system in the sense of Definition \ref{expedition0011530}.
By Theorem \ref{expedition0011535}, the projective system of probability measures $\fml{\msrcal{G}_{\txtgs,\kappa,\Lambda,T}}{T > 0}$ converges weakly to the projective limit measure $\msrcal{G}_{\txtgs,\kappa,\Lambda}$.
\end{proof}

\begin{defn}
Suppose there is a probability measure $\msrsf{G}$ on a Hilbert space $\sphilb{K}$, and a self-adjoint semigroup $T
=
\fml{T_t}
{t > 0}$ on $\fun{\lp^{2}}{\sphilb{K},\msrsf{G}}$ is a standard Markov semigroup.
Let $\fml{\prbou_t}
{t \in \fldreal}$ be the coordinate process defined by a probability measure $\msrcal{G}$ on the path space $\mathfrak{Y}
= \fun{\conti}{\fldreal;\sphilb{K}}$.
When the coordinate process $\fml{\prbou_t}
{t \in \fldreal}$ satisfies the following conditions, this stochastic process is called a $P(\phi)_1$-process.
\begin{enumerate}
\item
Stationarity: for any bounded measurable function $F
\colon \sphilb{K}
\to \fldreal$ and any $t
\in \fldreal$, $\sqfun{\prbexp_{\msrcal{G}}}{F(\prbou_t)}
= \sqfun{\prbexp_{\msrsf{G}}}{F}$ holds.
That is, the distribution of $\prbou_t$ is $\msrsf{G}$ at every time $t$, and the coordinate process is stationary under time translation.

\item
The Markov representation of finite-dimensional distributions holds:
for any bounded measurable functions $F_0,\cdots,F_n
\colon \sphilb{K}
\to \fldreal$ and a partition of the interval$$0
\leq t_0
< t_1
< \cdots
< t_n,$$$$\sqfun{\prbexp_{\msrcal{G}}}
{\prod_{j=0}^n F_j(\prbou_{t_j})}
=
\bkt{F_0}
{T_{t_1-t_0} F_1
\cdots
T_{t_{n}-t_{n-1}} F_n}
_{\fun{\lp^{2}}{\sphilb{K},\msrsf{G}}}$$ holds.
That is, the coordinate process $\fml{\prbou_t}{t \in \fldreal}$ is a stationary Markov process corresponding to the semigroup $T
= \fml{T_t}{t > 0}$.
\end{enumerate}
\end{defn}

\begin{prop}\label{expedition0011514}
The infinite-dimensional Ornstein-Uhlenbeck process $\fml{\prbou_t}{t \in \fldreal}$, as the coordinate process defined by the projective and weak limit probability measure $\msrcal{G}_{\txtgs,\kappa,\Lambda}$ of Proposition \ref{expedition0011508}, is a $P(\phi)_1$-process with stationary distribution $\msrsf{G}_{\txtgs,\kappa,\Lambda}$.
In particular, under the measure $\msrcal{G}_{\txtgs,\kappa,\Lambda}$, the coordinate process $\fml{\prbou_t}{t \in \fldreal}$ is a Markov process with generator $\physham[L]_{\kappa,\Lambda}$, and for bounded measurable $F_0,\cdots,F_n$ and $0
\leq t_0
< \cdots
< t_n$,
\begin{equation}
\begin{aligned}
&\sqfun{\prbexp_{\msrcal{G}_{\txtgs,\kappa,\Lambda}}}
{\prod_{j=0}^n F_j(\prbou_{t_j})}
\\ 
&=
\bkt{F_0}
{T_{\kappa,\Lambda,t_1-t_0} F_1
\cdots
T_{\kappa,\Lambda,t_n-t_{n-1}} F_n}
_{\fun{\lp^{2}}{\sphilb{M}_{-2},\msrsf{G}_{\txtgs,\kappa,\Lambda}}}
\end{aligned}
\end{equation}
holds.
\end{prop}

\begin{proof}
(The ground state measure $\msrsf{G}_{\txtgs,\kappa,\Lambda}$ is the stationary distribution): For bounded measurable $F$, compute the expectation at time $0$ from the finite-interval measure.
Using the finite-interval path measure and $Z_{\kappa,\Lambda,T}
= 1$,
\begin{equation}
\begin{aligned}
&\sqfun{\prbexp_{\msrcal{G}_{\txtgs,\kappa,\Lambda,T}}}
{F(\prbou_0)}
\\ 
&=
\sqfun{\prbexp_{\msrcal{G}}}
{F(\prbou_0)
\Psi_{\txtgs,\kappa,\Lambda}(\prbou_{-T})
\Psi_{\txtgs,\kappa,\Lambda}(\prbou_{T})
\napiernum^{-\int_{-T}^{T}
\rbk{\prbou_(\omega \mathsf{m}_{\kappa,\Lambda}) - \physgse}
\opdmsr{s}}}
\end{aligned}
\end{equation}
holds.
Using the same decomposition into $\closedinterval{-T}{-S},
\closedinterval{-S}{S},
\closedinterval{S}{T}$ as in the proof of Proposition \ref{expedition0011508}(2), setting $S
= 0$, the contributions from both endpoints collapse to $\Psi_{\txtgs,\kappa,\Lambda}(\prbou_0)$ via the outer-interval cancellation.
In particular,
\begin{equation}
\begin{aligned}
&\sqfun{\prbexp_{\msrcal{G}_{\txtgs,\kappa,\Lambda,T}}}
{F(\prbou_0)}
=
\sqfun{\prbexp_{\msrcal{G}}}
{F(\prbou_0)
\Psi_{\txtgs,\kappa,\Lambda}(\prbou_{0})^2}
\\ 
&=
\sqfun{\prbexp_{\msrsf{G}}}
{F
\Psi_{\txtgs,\kappa,\Lambda}^2}
=
\sqfun{\prbexp_{\msrsf{G}_{\txtgs,\kappa,\Lambda}}}
{F}
\end{aligned}
\end{equation}
is obtained.
Since the right-hand side is independent of $T$, taking the weak limit gives$$\sqfun{\prbexp_{\msrcal{G}_{\txtgs,\kappa,\Lambda}}}
{F(\prbou_0)}
=
\sqfun{\prbexp_{\msrsf{G}_{\txtgs,\kappa,\Lambda}}}
{F}.$$In particular, the distribution at time $0$ is $\msrsf{G}_{\txtgs,\kappa,\Lambda}$.
By the stationarity (translation invariance) of the Ornstein-Uhlenbeck process and the construction of the finite-interval measures, the distribution of $\prbou_t$ is also $\msrsf{G}_{\txtgs,\kappa,\Lambda}$ for any $t
\in \fldreal$, and the finite-dimensional distributions are invariant under time translation.
Therefore, the infinite-dimensional Ornstein-Uhlenbeck process $\fml{\xi_t}{t \in \fldreal}$ under $\msrcal{G}_{\txtgs,\kappa,\Lambda}$ is stationary with static distribution equal to the ground state measure $\msrsf{G}_{\txtgs,\kappa,\Lambda}$.

(Markov property): By Proposition \ref{expedition0011511}, the semigroup $T_{\kappa,\Lambda,t}
=
\napiernum^{-t \physham[L]_{\kappa,\Lambda}}$ is a Markov semigroup with conservation, positivity preservation, and symmetry, and its generator is $\physham[L]_{\kappa,\Lambda}
= \physham[L]_0 + (2 \omega \mathsf{m}_{\kappa,\Lambda}) \vainnprod \vagrad$.
Once the finite-dimensional distributions of $T_{\kappa,\Lambda,t}$ and $\msrcal{G}_{\txtgs,\kappa,\Lambda}$ agree, one can confirm that they form a Markov process with generator $L_{\kappa,\Lambda}$.
Combining the construction so far with the Feynman-Kac-Nelson formula and the ground state transformation, for bounded measurable $F_0,\cdots,F_n$ and $0
\leq t_0
< \cdots
< t_n$,
\begin{equation}
\begin{aligned}
&\sqfun{\prbexp_{\msrcal{G}_{\txtgs,\kappa,\Lambda}}}
{\prod_{j=0}^n F_j(\prbou_{t_j})}
\\ 
&=
\bkt{F_0}
{T_{\kappa,\Lambda,t_1-t_0} F_1
\cdots
T_{\kappa,\Lambda,t_n-t_{n-1}} F_n}
_{\fun{\lp^{2}}{\sphilb{M}_{-2},\msrsf{G}_{\txtgs,\kappa,\Lambda}}}
\end{aligned}
\end{equation}
is obtained.
This is the Markov representation of finite-dimensional distributions for the semigroup $T_{\kappa,\Lambda,t}
= \napiernum^{-t \physham[L]_{\kappa,\Lambda}}$.
Therefore, the coordinate process under $\msrcal{G}_{\txtgs,\kappa,\Lambda}$ is a Markov process with generator $\physham[L]_{\kappa,\Lambda}$.
\end{proof}

As a corollary, the functional integral representation for the Markov semigroup \(\napiernum^{-t \physham[L]_{\kappa,\Lambda}}\) is obtained.

\begin{cor}\label{expedition0011569}
The Feynman-Kac-Nelson formula for any $F,G
\in \fun{\lp^{2}}{\sphilb{M}_{-2},\msrsf{G}}$ reads:$$\bkt{F}
{\napiernum^{-t \physham[L]_{\kappa,\Lambda}} G}
_{\fun{\lp^{2}}{\sphilb{M}_{-2},\msrsf{G}_{\txtgs,\kappa,\Lambda}}}
=
\sqfun{\prbexp_{\msrcal{G}_{\txtgs,\kappa,\Lambda}}}
{F(\prbou_0)
G(\prbou_t)}.$$
\end{cor}

\begin{proof}
This is simply the Markov process representation of Proposition \ref{expedition0011514} with $n
= 1$ rewritten.
\end{proof}

In Proposition \ref{expedition0011514}, only properties are listed, and we now verify that the explicit form of the finite-dimensional distributions and their consistency with marginalization (projection) are clear. The next proposition focuses on establishing the projective structure and Markov property of the stochastic process.

\begin{prop}
For any finite collection of times $t_0
< t_1
\cdots
< t_n$, define the projection$$\pi_{t_0,\cdots,t_n}
\colon \sphilb[\mathfrak]{Y}
\to \sphilb{M}_{-2}^{n+1};
\quad
\pi_{t_0,\cdots,t_n}(\prbou)
=
\vecbk{\prbou_{t_0},\cdots,\prbou_{t_n}}.$$Then the following statements hold for the all-time path measure $\msrcal{G}_{\txtgs,\kappa,\Lambda}$ defined in Proposition \ref{expedition0011508}.
\begin{enumerate}
\item
As the Markov representation of finite-dimensional distributions, on the algebra generated by indicator functions $\fndef{A_j}$,$$\msr{\mu}_{t_0,\cdots,t_n}
=
\msrcal{G}_{\txtgs,\kappa,\Lambda}
\circ
\inv{\pi_{t_0,\cdots,t_n}}$$is given by
\begin{equation}
\begin{aligned}
&\fun{\msr{\mu_{t_0,\cdots,t_n}}}{\prod_{j=0}^{n} A_j}
\\ 
&=
\bkt{\fndef{A_0}}
{T_{\kappa,\Lambda,t_1-t_0} \fndef{A_1}
\cdots
T_{\kappa,\Lambda,t_n-t_{n-1}} \fndef{A_n}}
_{\fun{\lp^{2}}{\sphilb{M}_{-2},\msrsf{G}_{\txtgs,\kappa,\Lambda}}}.
\end{aligned}
\end{equation}

\item
For any pair of finite sets $I
\subset J
\subset \fldreal$, let $p_{I,J}
\colon \sphilb{M}_{-2}^{J}
\to \sphilb{M}_{-2}^{I}$ denote the natural projection.
Then$$\msr{\mu_I}
= \msr{\mu_J}
\circ
\inv{p_{I,J}}$$holds.
In particular, the finite-dimensional distributions are consistent with marginalization, and the family $\fml{\msr{\mu_I}}{I \in \setpowfin{\fldreal}}$ forms a projective system.

\item
For a symmetric interval $\closedinterval{-S}{S}$,$$\msrcal{G}_{\txtgs,\kappa,\Lambda}
\circ
\inv{\pi_{\closedinterval{-S}{S}}}
=
\lim_{T \to \infty}
\msrcal{G}_{\txtgs,\kappa,\Lambda,T}
\circ
\inv{\pi_{\closedinterval{-S}{S}}}$$holds.
In particular, the right-hand side coincides with the projective system.
\end{enumerate}
\end{prop}

\begin{proof}
(1): Take $F_j
= \fndef{A_j}$ in Proposition \ref{expedition0011514}(2).

(2): The semigroup property of the Markov semigroup gives the marginalization obtained by collapsing one time point in the finite-dimensional distribution.
Specifically, let $J
= \setone{t_0,\cdots,t_n}$, and for some $0 < k < n$ set $I
= J \setminus \setone{t_k}$.
Integration over the $k$-th component in $\msr{\mu_J}$ via $T_{t_k - t_{k-1}}
\circ
T_{t_{k+1} - t_{k}}
=
T_{t_{k+1} - t_{k-1}}$ gives
\begin{equation}
\begin{aligned}
&\fun{\msr{\mu_J}}{\prod_{j=0}^n A_j}
\\ 
&=
\bkt{\fndef{A_0}}
{T_{t_1 - t_0} \fndef{A_1}
\cdots
T_{t_{k} - t_{k-1}} \fndef{A_k}
\cdots
T_{t_{n} - t_{n-1}} \fndef{A_n}}
\\ 
&=
\bkt{\fndef{A_0}}
{T_{t_1 - t_0} \fndef{A_1}
\cdots
T_{t_{k+1} - t_{k-1}} \fndef{A_{k+1}}
\cdots
T_{t_{n} - t_{n-1}} \fndef{A_n}}
\end{aligned}
\end{equation}
is obtained.
This agrees with the finite-dimensional distribution for $I$.
The same argument applies for general $I
\subset J$.
Therefore, $\fml{\msr{\mu_I}}{I \in \setpowfin{\fldreal}}$ is a projective system, and the measure $\msrcal{G}_{\txtgs,\kappa,\Lambda}$ is the path measure realizing this projective system.

(3): Proposition \ref{expedition0011508} establishes the consistency of the finite-interval measures $\msrcal{G}_{\txtgs,\kappa,\Lambda,T}$.
By the tightness and weak limit argument of that proposition, these converge weakly to $\msrcal{G}_{\txtgs,\kappa,\Lambda}$.
By continuity of the projection,$$\msrcal{G}_{\txtgs,\kappa,\Lambda}
\circ
\inv{\pi_{\closedinterval{-S}{S}}}
=
\lim_{T \to \infty}
\msrcal{G}_{\txtgs,\kappa,\Lambda}
\circ
\inv{\pi_{\closedinterval{-S}{S}}}$$is obtained.
Since the left-hand side is independent of $T$, the restriction of the all-time measure to $\closedinterval{-S}{S}$ agrees with the family of finite-interval measures.
\end{proof}

\subsection{Removal of Infrared and Ultraviolet Cutoffs}\label{removal-of-infrared-and-ultraviolet-cutoffs-1}

Since the divergence of the ground state energy cannot be canceled when the interaction term remains in the original Hamiltonian, the entire discussion must be carried out in the \(\physham[L]\) system where the interaction term has been formally removed. Due to the special structure of the van Hove model, if we work on \(\fun{\lp^{2}}{\sphilb{M}_{-2},\msrsf{G}_{\txtgs}}\) or \(\fun{\lp^{2}}{\sphilb{M}_{-2},\msrsf{G}_{\kappa,\Lambda}}\) with ground state measures \(\msrsf{G}_{\txtgs}\) or \(\msrsf{G}_{\kappa,\Lambda}\), the ground state is \(1\), so there is no need to take a limit of ground states in this sense. Again due to the special nature of the van Hove model, all that needs to be handled is the trivial limit of \(\mathsf{m}_{\kappa,\Lambda}\) and the renormalization of the energy, and the important properties of the objects after cutoff removal are essentially already proved. What remains is the transition to \(\physham[L]\), \(\msrsf{G}_{\txtgs}\), \(\msrcal{G}_{\txtgs}\), \(T_t = \napiernum^{-t \physham[L]}\), and the establishment of the Feynman-Kac-Nelson formula in the new framework. First, preparation equivalent to the discussion in Section \ref{expedition0011370} is needed. We leave the basic arguments to Section \ref{expedition0011370} and focus here on the modifications.

The limit of the mean functional and the cutoff independence of the covariance are clear by definition.

\begin{prop}\label{expedition0011566}
\begin{enumerate}
\item
For any element $f$ of $\dom \mathsf{m}$ defined in Definition \ref{expedition0011100}, the IR/UV cutoff removal limit $\mathsf{m}(f)$ is well-defined.
The domain can be uniquely continuously extended to $\sphilb{H}_{\txtirsingular}
=
\gtclos{\dom \mathsf{m}}^{\norm{\cdot}_{C}}$.

\item
At each stage, for the ground state measures with IR/UV cutoffs associated with $\msrcal{G}$ and $\msrsf{G}$, the covariance form is independent of the IR/UV cutoffs in all cases.
\end{enumerate}
\end{prop}

For the domain \(\dom \mathsf{m}\) of the functional \(\mathsf{m}\), define the index set \(\setindex{I}\) by\[\setindex{I}
=
\setpowfin{\fldreal \times \dom \mathsf{m}}.\]Furthermore, for any\[I
=
\setone{\pairbk{t_1,f_1},
\cdots,
\pairbk{t_n,f_n}}
\in
\setindex{I},\]define the finite-dimensional projection \(\pi_I\) on \(\fun{\conti}{\fldreal;\sphilb{M}_{-2}}\) by\[\pi_I
\colon \fun{\conti}{\fldreal;\sphilb{M}_{-2}}
\to \fldreal^{n};
\quad
\pi_I(\prbou)
=
\rbk{\prbou_{t_1}(f_1),
\cdots,
\prbou_{t_n}(f_n)}.\]The image measure\[\msrcal{G}_{\txtgs,\kappa,\Lambda,I}
=
\msrcal{G}_{\txtgs,\kappa,\Lambda}
\circ
\inv{\pi_I}\]is called the finite-dimensional distribution, and the covariance form for the above setting is defined as\[C_{t_j,t_k}(f_j,f_k)
=
\sqfun{\prbcov_{\msrcal{G}_{\txtgs,\kappa,\Lambda}}}
{\prbou_{t_j}(f_j),
\prbou_{t_k}(f_k)}.\]

First, we discuss the IR/UV cutoff removal limit of the finite-dimensional characteristic functions.

\begin{prop}\label{expedition0011564}
For any $I
=
\setone{\pairbk{t_1,f_1},
\cdots,
\pairbk{t_n,f_n}}
\in
\setindex{I}$, define the finite-dimensional characteristic function$$\prbcharfun_{\txtgs,\kappa,\Lambda,I}(\theta)
=
\sqfun{\prbexp_{\msrcal{G}_{\txtgs,\kappa,\Lambda}}}
{\fnexp{\imunit
\sum_{j=1}^n
\theta_j
\prbou_{t_j}(f_j)}},
\quad
\theta \in \fldreal^{n}.$$
\begin{enumerate}
\item
Existence of the limit: for any $I
\in \setindex{I}$ and $\theta
\in \fldreal^{n}$, the limit$$\prbcharfun_{\txtgs,I}(\theta)
=
\lim_{\kappa \to 0, \Lambda \to \infty}
\prbcharfun_{\txtgs,\kappa,\Lambda,I}(\theta)$$exists.

\item
The limit $\prbcharfun_{\txtgs,I}$ is a Gaussian characteristic function; in particular,$$\prbcharfun_{\txtgs,I}(\theta)
=
\fnexp{\imunit
\sum_{j=1}^n
\theta_j \mathsf{m}(f_j)
-\onehalf
\sum_{j,k=1}^n
\theta_j \theta_k
C_{t_j,t_k}(f_j,f_k)}.$$

\item
There exists an IR/UV cutoff removal limit in the sense of weak convergence of finite-dimensional distributions.
In particular, for each $I
\in \setindex{I}$, there exists a Gaussian measure $\msr{\mu_I}$ corresponding to $\prbcharfun_{\txtgs,I}$ such that, in the ordinary weak convergence on $\fldreal^{n}$,$$\msrcal{G}_{\txtgs,\kappa,\Lambda,I}
\xRightarrow[\kappa \to 0, \Lambda \to \infty]{\txtweak}
\msr{\mu_I}.$$
\end{enumerate}
\end{prop}

\begin{proof}
((1), Gaussian structure at each cutoff parameter $\pairbk{\kappa,\Lambda}$): By Proposition \ref{expedition0011514}, the coordinate process $\fml{\prbou_{t}}
{t \in \fldreal}$ is an infinite-dimensional Ornstein-Uhlenbeck-type Markov process with generator $\physham[L]_{\kappa,\Lambda}$, and is a stationary Gaussian process with static distribution $\msrsf{G}_{\txtgs,\kappa,\Lambda}$.
For any finite set $I = \setone{(t_1,f_1), \cdots, (t_n,f_n)} \in \setindex{I}$, the vector defined by
$$X_{\kappa,\Lambda,I}
=
\vecbk{X_{\kappa,\Lambda,j}},
\quad
X_{\kappa,\Lambda,j}
=
\prbou_{t_j}(f_j)$$is a Gaussian vector in $\fldreal^{n}$, and the corresponding measure $\msrcal{G}_{\txtgs,\kappa,\Lambda,I}$ is a Gaussian measure.
In particular, the finite-dimensional characteristic function is
\begin{equation}
\begin{aligned}
\vecbk{\mathsf{m}_{\kappa,\Lambda,I}}_j
&=
\sqfun{\prbexp_{\msrcal{G}_{\txtgs,\kappa,\Lambda}}}
{\prbou_{t_j}(f_j)},
\\ 
\lamat{C_I}{jk}
&=
\sqfun{\prbcov_{\msrcal{G}_{\txtgs,\kappa,\Lambda}}}
{\prbou_{t_j}(f_j),
\prbou_{t_k}(f_k)}
\end{aligned}
\end{equation}
represented as$$\prbcharfun_{\txtgs,\kappa,\Lambda,I}(\theta)
=
\fnexp{\imunit
\bkt{\theta}{\mathsf{m}_{\kappa,\Lambda,I}}_{\fldreal^{n}}
-\onehalf
\bkt{\theta}{C_I \theta}_{\fldreal^{n}}}.$$

((1), Limit of the mean): By Proposition \ref{expedition0011566}, the limit $\mathsf{m}_{\kappa,\Lambda}(f)
\to \mathsf{m}(f)$ holds for $f
\in \dom \mathsf{m}$.
By the stationarity of the ground state process,$$\sqfun{\prbexp_{\msrcal{G}_{\txtgs,\kappa,\Lambda}}}
{\prbou_t(f)}
=
\mathsf{m}_{\kappa,\Lambda}(f)$$holds for any $t
\in \fldreal$.
In particular, $\vecbk{\mathsf{m}_{\kappa,\Lambda,I}}_j
=
\mathsf{m}_{\kappa,\Lambda}(f_j)
\to
\mathsf{m}(f_j)$ holds for each $j$.
Therefore,$$\lim_{\kappa\to 0,\Lambda\to\infty}
\mathsf{m}_{\kappa,\Lambda,I}
=
\mathsf{m}_I
=
\vecbk{\mathsf{m}(f_1),
\cdots,
\mathsf{m}(f_n)}
\in
\fldreal^{n}$$is obtained.

((1), Limit of the covariance matrix): By Proposition \ref{expedition0011566}, the covariance form is independent of the cutoff parameters $\pairbk{\kappa,\Lambda}$.
The IR/UV cutoff removal limit is trivial.

((1)-(2), Pointwise convergence of characteristic functions and Gaussian nature of the limit): From the discussion so far, setting$$\prbcharfun_{\txtgs,I}(\theta)
=
\fnexp{\imunit
\bkt{\theta}{\mathsf{m}_I}_{\fldreal^{n}}
-\onehalf
\bkt{\theta}{C_I \theta}_{\fldreal^{n}}}$$gives$$\lim_{\kappa\to 0, \Lambda\to\infty}
\prbcharfun_{\txtgs,\kappa,\Lambda,I}(\theta)
=
\prbcharfun_{\txtgs,I}(\theta).$$By assumption and construction, $\prbcharfun_{\txtgs,I}$ is a Gaussian characteristic function.
Let $\msr{\mu_I}$ denote the corresponding unique Gaussian measure.

(3): This follows from (1)--(2) of this proposition and Lévy's continuity theorem on $\fldreal^{n}$.
\end{proof}

We establish the projective consistency of the limiting finite-dimensional distribution family and the limiting measure.

\begin{prop}\label{expedition0011567}
\begin{enumerate}
\item
The limiting finite-dimensional distribution family is projectively consistent.
In particular, for any pair of finite sets $J
\subset I
\in \setindex{I}$, let $\pi_J^I
\colon \fldreal^{I}
\to \fldreal^{J}$ denote the natural coordinate projection.
Then $\msr{\mu_J}
=
\msr{\mu_I} \circ \invrbk{\pi_J^{I}}$ holds.
Therefore the family $\fml{\msr{\mu_I}}{I \in \setindex{I}}$ is a projective system of probability measures in the sense of Definition \ref{expedition0011527}.

\item
The limiting path measure exists.
In particular, letting $\mblfml{C}$ denote the cylinder set algebra on $\fun{\conti}{\fldreal;\sphilb{M}_{-2}}$, for the projective system $\fml{\msr{\mu_I}}{I \in \setindex{I}}$ of (1), there exists a probability measure $\msrcal{G}_{\txtgs}$ on $\pairbk{\fun{\conti}{\fldreal;\sphilb{M}_{-2}},
\mblfml{C}}$ satisfying$$\msrcal{G}_{\txtgs}
\circ
\inv{\pi_I}
=
\msr{\mu_I},
\quad
I
\in \setindex{I}.$$This $\msrcal{G}_{\txtgs}$ is called the ground state path measure after cutoff removal for the van Hove model.
\end{enumerate}
\end{prop}

\begin{proof}
(1): Consider a pair of finite sets$$J
=
\setone{(s_1,g_1),
\cdots,
(s_m,g_m)}
\subset
I
=
\setone{(t_1,f_1),
\cdots,
(t_n,f_n)}
\in \setindex{I}.$$By Proposition \ref{expedition0011564}, the characteristic function of the measure $\msr{\mu_I}$ can be written as, for $\theta
= \vecbk{\theta_{\alpha}}_{\alpha \in I}
\in \fldreal^{I}$,$$\prbcharfun_{\txtgs,I}(\theta)
=
\fnexp{\imunit
\sum_{j=1}^{n}
t_j
\mathsf{m}(f_j)
-\onehalf
\sum_{j,k=1}^{n}
t_{j} t_{k}
C_{t_j t_k}(f_j,f_k)}.$$For any $\eta
= \vecbk{\eta_{j}}_{j=1}^m
\in \fldreal^{J}$, define the natural embedding$$\tilde{\eta}_{k}
=
\begin{dcases}
\eta_k, & k \in J, \\
0, & k \in I \setminus J
\end{dcases}$$where $k
\in J$ corresponds to the index $i$ of each pair $(t_i,f_i)$ in $J$.
Then $\fun{\prbcharfun_{\txtgs,I}}{\tilde{\eta}}$ can be written as$$\fun{\prbcharfun_{\txtgs,I}}{\tilde{\eta}}
=
\fnexp{\imunit
\sum_{j \in J}
\eta_j
\mathsf{m}(f_j)
-\onehalf
\sum_{j,k \in J}
\eta_j \eta_k
C_{t_j t_k}(f_j,f_k)},$$and we simply write this as $\prbcharfun_{\txtgs,J}(\eta)$.
On the other hand, the characteristic function of the projection $\msr{\mu_I} \circ \invrbk{\pi_J^I}$ of the measure $\msr{\mu_I}$ is
\begin{equation}
\begin{aligned}
&\sqfun{\fafouriertr}
{\msr{\mu_I} \circ \invrbk{\pi_J^I}}(\eta)
=
\sqfun{\prbexp_{\msr{\mu_I}}}
{\fnexp{\imunit
\eta \vainnprod \pi_J^I(X)}}
\\ 
&=
\sqfun{\prbexp_{\msr{\mu_I}}}
{\fnexp{\imunit
\tilde{\eta} \vainnprod X}}
=
\fun{\prbcharfun_{\txtgs,I}}{\tilde{\eta}}
=
\prbcharfun_{\txtgs,J}(\eta)
\end{aligned}
\end{equation}
is obtained.
Since a characteristic function uniquely determines a probability measure, $\msr{\mu_I}
\circ
\invrbk{\pi_J^I}
=
\msr{\mu_J}$ follows.

(2): By (1) of this proposition, the family $\fml{\msr{\mu_I}}{I \in \setindex{I}}$ forms a projective system.
By the Kolmogorov extension theorem \ref{expedition0011529} for projective systems of probability measures, there exists a probability space $\pairbk{\fun{\conti}{\fldreal;\sphilb{M}_{-2}},
\mblfml{C},
\msrcal{G}_{\txtgs}}$ satisfying$$\msrcal{G}_{\txtgs}
\circ
\inv{\pi_I}
=
\msr{\mu_I}$$for each $I
\in \setindex{I}$.
\end{proof}

Define the ground state measure \(\mathsf{G}_{\txtgs}\) on the Hilbert space \(\sphilb{M}_{-2}\) as the projection of the ground state measure \(\msrcal{G}_{\txtgs}\) on the path space \(\fun{\conti}{\fldreal;\sphilb{M}_{-2}}\) to time \(0\):\[\mathsf{G}_{\txtgs}
=
\msrcal{G}_{\txtgs}
\circ
\inv{\prbou_0}.\]

\begin{thm}\label{expedition0011572}
The ground state measure $\msrsf{G}_{\txtgs}$ is the Gaussian measure with mean functional $\mathsf{m}$ and covariance form $C$ from Proposition \ref{expedition0011566}, uniquely determined as the IR/UV cutoff removal limit of the cutoff measures $\msrsf{G}_{\txtgs,\kappa,\Lambda}$.
For any $f
\in \dom \mathsf{m}$, the finite-dimensional characteristic function is$$\sqfun{\prbexp_{\msrsf{G}_{\txtgs}}}
{\napiernum^{\imunit \prbou_0(f)}}
=
\fnexp{\imunit \mathsf{m}(f)
-\onehalf
C(f,f)},$$which coincides with the characteristic functions and ground state measures of Propositions \ref{expedition0011564} and \ref{expedition0011567}.
\end{thm}

\begin{proof}
The path measure $\msrcal{G}_{\txtgs}$ is the projective limit of the finite-dimensional distributions $\msrcal{G}_{\txtgs,I}$ for $I
\in \setpowfin{\fldreal}$.
In particular, for $I
=
\setone{0}$, set $\msrsf{G}_{\txtgs}
=
\msrcal{G}_{\txtgs}
\circ
\inv{\prbou_0}$.
This coincides with the IR/UV cutoff removal limit measure of the finite-dimensional distribution $\msrcal{G}_{\txtgs,\kappa,\Lambda,\setone{0}}$ of the cutoff measure $\msrcal{G}_{\txtgs,\kappa,\Lambda}$.
By Proposition \ref{expedition0011564}, the cutoff removal limits of the mean functional and covariance form are obtained.
Therefore, the finite-dimensional characteristic function at time $0$ is
\begin{equation}
\begin{aligned}
&\sqfun{\prbexp_{\msrsf{G}_{\txtgs}}}
{\napiernum^{\imunit \prbou_0(f)}}
=
\lim_{\kappa \to 0, \Lambda \to \infty}
\sqfun{\prbexp_{\msrsf{G}_{\txtgs,\kappa,\Lambda}}}
{\napiernum^{\imunit \prbou_0(f)}}
\\ 
&=
\lim_{\kappa \to 0, \Lambda \to \infty}
\fnexp{\imunit
\mathsf{m}_{\kappa,\Lambda}(f)
-\onehalf
C(f,f)}
\\ 
&=
\fnexp{\imunit
\mathsf{m}(f)
-\onehalf
C(f,f)}
\end{aligned}
\end{equation}
is obtained.
Since the right-hand side is the characteristic function of a Gaussian measure with mean functional $\mathsf{m}$ and covariance form $C$, $\msrsf{G}_{\txtgs}$ is the corresponding Gaussian measure.
\end{proof}

Let us define the Hilbert space incorporating the infrared singularity condition.

\begin{defn}\label{expedition0011573}
As shorthand notation, we adopt$$\fun{F}{\nfoldvar{\prbou_t(f)}{n}}
=
\fun{F}{\prbou_{t_1}(f_1),
\cdots,
\prbou_{t_n}(f_n)}.$$Consider the linear span of cylinder functions using $\dom \mathsf{m}$ from Definition \ref{expedition0011100}:$$\sphilb{S}
=
\splinspan
\set{\fun{F}{\nfoldvar{\prbou_t(f)}{n}}}
{\parbox{15em}{$n$ is a natural number; when $n
= 0$ the element is $1$.
For positive integers, $t_1,\cdots,t_n
\in \fldreal$, $f_1,\cdots,f_n
\in \dom \mathsf{m}$,
$F
\in \fun{\dstrapiddec}{\fldreal^{n}}$.}}$$The closed subspace of $\fun{\lp^{2}}
{\sphilb{M}_{-2},\msrsf{G}_{\txtgs}}$ obtained by completing with respect to the norm induced by the ground state measure $\mathsf{G}_{\txtgs}$ is denoted$$\fun{\lp^{2}}
{\sphilb{M}_{-2,\txtirsingular},\msrsf{G}_{\txtgs}}$$and called the Hilbert space incorporating the infrared singularity condition.
\end{defn}

For the operator defined by the analogue of the cylinder function identity of Proposition \ref{expedition0011574},\[\physham[L]_{\txtirsingular,0} F
=
\physham[L]_0 F
+2
\bkt{\vagrad F}{\omega \mathsf{m}},\]let \(\physham[L]_{\txtirsingular}\) denote its closure on \(\fun{\lp^{2}}{\sphilb{M}_{-2,\txtirsingular},\msrsf{G}_{\txtgs}}\). In particular, this is a bounded-below self-adjoint operator.

\begin{lem}\label{expedition0011576}
\begin{enumerate}
\item
The set-theoretic inclusion$$\fun{\lp^{2}}
{\sphilb{M}_{-2,\txtirsingular},\msrsf{G}_{\txtgs}}
\subset
\fun{\lp^{2}}{\sphilb{M}_{-2},\msrsf{G}_{\kappa,\Lambda}}$$holds.

\item
Consider the linear span $\sphilb{S}$ defined in Definition \ref{expedition0011573}.
For the exponents $a
= \onehalf,1$, the set-theoretic inclusion$$\sphilb{S}
\subset
\dom \physham_{\txtbsn,\txtfr}^{a}$$holds.
\end{enumerate}
\end{lem}

\begin{proof}
(1): This follows from the inclusion $\dom \mathsf{m}
\subset
\dom \mathsf{m}_{\kappa,\Lambda}
=
\sphilb{H}$ arising from the infrared singularity condition.

(2): By definition.
\end{proof}

In view of Lemma \ref{expedition0011576}(2), the free Hamiltonian \(\physham_{\txtbsn,\txtfr}\) is defined on \(\fun{\lp^{2}}
{\sphilb{M}_{-2,\txtirsingular},\msrsf{G}_{\txtgs}}\). When distinction is needed, the measure is appended as \(\dom_{\msrsf{G}_{\txtgs}} \physham_{\txtbsn,\txtfr}\); when there is no risk of confusion, the domain of this Hamiltonian is simply written as \(\dom \physham_{\txtbsn,\txtfr}\).

\begin{prop}\label{expedition0011540}
The family of self-adjoint operators $\fml{\physham[L]_{\kappa,\Lambda}}
{\kappa > 0, \Lambda > 0}$ converges in the strong resolvent sense to a bounded-below self-adjoint operator $\physham[L]_{\txtirsingular}$.
In particular, the corresponding family of heat semigroups also converges uniformly on compact sets in the strong operator topology to the heat semigroup $T_{\txtirsingular}
= \fml{\napiernum^{-t \physham[L]_{\txtirsingular}}}
{t \geq 0}$.
This $T_{\txtirsingular}$ is a symmetric, non-expansive, strongly continuous semigroup on $\fun{\lp^{2}}{\sphilb{M}_{-2,\txtirsingular},\msrsf{G}_{\txtgs}}$.
It is a standard Markov semigroup satisfying: conservation $T_t 1
= 1$; positivity improvement $T_t F
> 0$ for $F
\geq 0$; and symmetry $\bkt{F}{T_t G}
=
\bkt{T_t F}{G}$.
\end{prop}

\begin{proof}
(Strong resolvent convergence): By the set-theoretic inclusion in Lemma \ref{expedition0011576}, consider the quadratic form on $\dom_{\msrsf{G}_{\txtgs}}
\physham_{\txtbsn,\txtfr}^{\onehalf}$.
In particular, for $\Psi
\in \dom_{\msrsf{G}_{\txtgs}}
\physham_{\txtbsn,\txtfr}^{\onehalf}$, set
\begin{equation}
\begin{aligned}
\opform{q}_{\kappa,\Lambda}(\Psi)
&=
\bkt{\Psi}{\physham[L]_{\kappa,\Lambda} \Psi},
\quad
\opform{q}_{\txtirsingular}(\Psi)
&=
\bkt{\Psi}{\physham[L]_{\txtirsingular} \Psi}.
\end{aligned}
\end{equation}
Considering the action on cylinder functions, $\opform{q}_{\kappa,\Lambda}$ converges to $\opform{q}_{\txtirsingular}$ on the set $\dom_{\msrsf{G}_{\txtgs}}
\physham_{\txtbsn,\txtfr}^{\onehalf}$.
By the general theory, convergence of bounded-below sesquilinear forms implies strong resolvent convergence of the associated self-adjoint operators.

(Heat semigroup): The convergence of heat semigroups follows from strong resolvent convergence and the sufficient conditions for strong resolvent convergence of self-adjoint operators.
The uniform-on-compacts strong convergence in the variable $t$ follows from the heat semigroup convergence from strong resolvent convergence.
The canonical Markov property is a direct consequence of the properties of $\physham[L]_{\kappa,\Lambda}$.
\end{proof}

We can now take the IR/UV cutoff removal limit of the Markov semigroup \(\napiernum^{-t \physham[L]_{\kappa,\Lambda}}\).

\begin{prop}\label{expedition0011570}
For any finite time sequence $0
\leq t_0
< \cdots
< t_n$ and bounded Borel measurable functions $F_0,
\cdots,
F_n$,
\begin{equation}
\begin{aligned}
&\sqfun{\prbexp_{\msrcal{G}_{\txtgs}}}
{\prod_{j=0}^n F_j(\prbou_{t_j})}
\\ 
&=
\bkt{F_0}
{T_{t_1 - t_0} F_1
\cdots
T_{t_{n} - t_{n-1}} F_n}
_{\fun{\lp^{2}}{\sphilb{M}_{-2},\msrsf{G}_{\txtgs}}}
\end{aligned}
\end{equation}
holds.
This is the Markov representation of finite-dimensional distributions.
In particular, the coordinate process $\fml{\prbou_t}{t \in \fldreal}$ is a stationary Gaussian process with static distribution $\msrsf{G}_{\txtgs}$ and can be regarded as an Ornstein-Uhlenbeck-type $P(\phi)_1$-process by the semigroup representation of finite-dimensional distributions.
\end{prop}

\begin{proof}
By Proposition \ref{expedition0011514}, under $\msrcal{G}_{\txtgs,\kappa,\Lambda}$, the coordinate process $\fml{\prbou_t}
{t \in \fldreal}$ is an Ornstein-Uhlenbeck-type $P(\phi)_1$-process with generator $\physham[L]_{\kappa,\Lambda}$.
The finite-dimensional distributions are expressed via the semigroup $T_{\kappa,\Lambda,t}
= \napiernum^{-t \physham[L]_{\kappa,\Lambda}}$ as
\begin{equation}
\begin{aligned}
&\sqfun{\prbexp_{\msrcal{G}_{\txtgs,\kappa,\Lambda}}}
{\prod_{j=0}^n F_j(\prbou_{t_j})}
\\ 
&=
\bkt{F_0}
{T_{\kappa,\Lambda,t_1-t_0} F_1
\cdots
T_{\kappa,\Lambda,t_n-t_{n-1}} F_n}
_{\fun{\lp^{2}}{\sphilb{M}_{-2},\msrsf{G}_{\txtgs}}}.
\end{aligned}
\end{equation}
By the hypothesis of Proposition \ref{expedition0011564}, approximating the bounded Borel measurable functions $F_j$ by appropriate test functions, the IR/UV cutoff removal limit of each finite-dimensional characteristic function can be taken.
In particular, they coincide with the finite-dimensional distributions of $\msrcal{G}_{\txtgs}$ constructed in Proposition \ref{expedition0011567}.
Therefore, for any finite time sequence $t_0
< \cdots
< t_n$ and bounded measurable $F_j$,$$\sqfun{\prbexp_{\msrcal{G}_{\txtgs}}}
{\prod_{j=0}^n F_j(\prbou_{t_j})}
=
\lim_{\kappa \to 0, \Lambda \to \infty}
\sqfun{\prbexp_{\msrcal{G}_{\txtgs,\kappa,\Lambda}}}
{\prod_{j=0}^n F_j(\prbou_{t_j})}$$holds.

From the above discussion,
\begin{equation}
\begin{aligned}
&\sqfun{\prbexp_{\msrcal{G}_{\txtgs}}}
{\prod_{j=0}^n F_j(\prbou_{t_j})}
\\ 
&=
\lim_{\kappa \to 0, \Lambda \to \infty}
\bkt{F_0}
{T_{\kappa,\Lambda,t_1-t_0} F_1
\cdots
T_{\kappa,\Lambda,t_n-t_{n-1}} F_n}
_{\fun{\lp^{2}}{\sphilb{M}_{-2},\msrsf{G}_{\txtgs,\kappa,\Lambda}}}
\end{aligned}
\end{equation}
holds.
By Proposition \ref{expedition0011540}, the cutoff removal limit of the heat semigroup on the right-hand side can be taken, and
\begin{equation}
\begin{aligned}
&\sqfun{\prbexp_{\msrcal{G}_{\txtgs}}}
{\prod_{j=0}^n F_j(\prbou_{t_j})}
\\ 
&=
\bkt{F_0}
{T_{t_1-t_0} F_1
\cdots
T_{t_n-t_{n-1}} F_n}
_{\fun{\lp^{2}}{\sphilb{M}_{-2},\msrsf{G}_{\txtgs}}}
\end{aligned}
\end{equation}
holds.
This is the Markov representation by the semigroup $T
= \fml{T_t}{t > 0}$ after cutoff removal.
\end{proof}

As a corollary, the functional integral representation for the Markov semigroup \(\napiernum^{-t \physham[L]_{\txtirsingular}}\) is obtained.

\begin{cor}\label{expedition0011571}
The Feynman-Kac-Nelson formula for any $F,G
\in \fun{\lp^{2}}
{\sphilb{M}_{-2,\txtirsingular},\msrsf{G}_{\txtgs}}$ reads:$$\bkt{F}
{\napiernum^{-t \physham[L]_{\txtirsingular}} G}
_{\fun{\lp^{2}}{\sphilb{M}_{-2,\txtirsingular},\msrsf{G}_{\txtgs}}}
=
\sqfun{\prbexp_{\msrcal{G}_{\txtgs}}}
{F(\prbou_0)
G(\prbou_t)}.$$
\end{cor}

\begin{proof}
Take $n
= 1$, $t_0
= 0$, and $t_1
= t$ in Proposition \ref{expedition0011570}.
In the proof of Proposition \ref{expedition0011570}, use Proposition \ref{expedition0011514} or Corollary \ref{expedition0011569} in the argument before taking the limit.
Since the weak convergence of measures is also established by Proposition \ref{expedition0011567}, the desired identity is obtained.
\end{proof}

\begin{thm}
An Ornstein-Uhlenbeck-type $P(\phi)_1$-process after cutoff removal exists.
In particular, the coordinate process $\fml{\prbou_t}{t \in \fldreal}$ on $\msrcal{G}_{\txtgs}$ has the following properties.
\begin{enumerate}
\item
The finite-dimensional distributions are given by the IR/UV cutoff removal limits of the finite-dimensional distributions for the Ornstein-Uhlenbeck-type $P(\phi)_1$-process with cutoffs.
In particular, for any finite time sequence $-\infty
< t_0
< \cdots
< t_n
< \infty$ and bounded Borel measurable functions $F_0,
\cdots,
F_n$,$$\sqfun{\prbexp_{\msrcal{G}_{\txtgs}}}
{\prod_{j=0}^n F_j(\prbou_{t_j})}
=
\lim_{\kappa\to 0,\Lambda\to\infty}
\sqfun{\prbexp_{\msrcal{G}_{\txtgs,\kappa,\Lambda}}}
{\prod_{j=0}^n F_j(\prbou_{t_j})}$$holds.

\item
The coordinate process $\fml{\prbou_t}{t \in \fldreal}$ is an Ornstein-Uhlenbeck-type $P(\phi)_1$-process after IR/UV cutoff removal.
In particular, for any $0
\leq t_0
< \cdots
< t_n$, as the cutoff removal limit of the Markov representation in Proposition \ref{expedition0011514},
\begin{equation}
\begin{aligned}
&\sqfun{\prbexp_{\msrcal{G}_{\txtgs}}}
{\prod_{j=0}^n F_j(\prbou_{t_j})}
\\ 
&=
\bkt{F_0}
{T_{t_1-t_0} F_1
\cdots
T_{t_n-t_{n-1}} F_n}
_{\fun{\lp^{2}}{\sphilb{M}_{-2,\txtirsingular},\msrsf{G}_{\txtgs}}}
\end{aligned}
\end{equation}

there exists a semigroup $T_{\txtirsingular}
=
\fml{T_{\txtirsingular,t}}{t \geq 0}$ on $\fun{\lp^{2}}{\sphilb{M}_{-2,\txtirsingular},\msrsf{G}_{\txtgs}}$ satisfying this.

\item
The semigroup $T_{\txtirsingular}$ is obtained as the cutoff removal limit of the Markov semigroup $T_{\kappa,\Lambda}$.
In particular, the generator $\physham[L]_{\txtirsingular}$ of $T_{\txtirsingular}$ is obtained as the strong resolvent convergence limit of $\physham[L]_{\kappa,\Lambda}$ under cutoff removal.
\end{enumerate}
\end{thm}

\begin{proof}
(1): By Proposition \ref{expedition0011514}, under $\msrcal{G}_{\txtgs,\kappa,\Lambda}$, the coordinate process $\fml{\prbou_t}
{t \in \fldreal}$ is an Ornstein-Uhlenbeck-type $P(\phi)_1$-process with generator $\physham[L]_{\kappa,\Lambda}$.
The finite-dimensional distributions are expressed by the semigroup $T_{\kappa,\Lambda,t}
= \napiernum^{-t \physham[L]_{\kappa,\Lambda}}$ as
\begin{equation}
\begin{aligned}
&\sqfun{\prbexp_{\msrcal{G}_{\txtgs,\kappa,\Lambda}}}
{\prod_{j=0}^n F_j(\prbou_{t_j})}
\\ 
&=
\bkt{F_0}
{T_{\kappa,\Lambda,t_1-t_0} F_1
\cdots
T_{\kappa,\Lambda,t_n-t_{n-1}} F_n}
_{\fun{\lp^{2}}{\sphilb{M}_{-2,\txtirsingular},\msrsf{G}_{\txtgs}}}
\end{aligned}
\end{equation}
as above.
By the hypothesis of Proposition \ref{expedition0011564}, approximating the bounded Borel measurable functions $F_j$ by appropriate test functions, the IR/UV cutoff removal limit of each finite-dimensional characteristic function can be taken.
In particular, they coincide with the finite-dimensional distributions of $\msrcal{G}_{\txtgs}$ constructed in Proposition \ref{expedition0011567}.
Therefore, for any finite time sequence $t_0
< \cdots
< t_n$ and bounded measurable $F_j$,$$\sqfun{\prbexp_{\msrcal{G}_{\txtgs}}}
{\prod_{j=0}^n F_j(\prbou_{t_j})}
=
\lim_{\kappa \to 0, \Lambda \to \infty}
\sqfun{\prbexp_{\msrcal{G}_{\txtgs,\kappa,\Lambda}}}
{\prod_{j=0}^n F_j(\prbou_{t_j})}$$holds.

(2): From the preceding argument,
\begin{equation}
\begin{aligned}
&\sqfun{\prbexp_{\msrcal{G}_{\txtgs}}}
{\prod_{j=0}^n F_j(\prbou_{t_j})}
\\ 
&=
\lim_{\kappa \to 0, \Lambda \to \infty}
\bkt{F_0}
{T_{\kappa,\Lambda,t_1-t_0} F_1
\cdots
T_{\kappa,\Lambda,t_n-t_{n-1}} F_n}
_{\fun{\lp^{2}}{\sphilb{M}_{-2},\msrsf{G}_{\txtgs,\kappa,\Lambda}}}
\end{aligned}
\end{equation}
holds.
By Propositions \ref{expedition0011570} and \ref{expedition0011571},
\begin{equation}
\begin{aligned}
&\sqfun{\prbexp_{\msrcal{G}_{\txtgs}}}
{\prod_{j=0}^n F_j(\prbou_{t_j})}
\\ 
&=
\bkt{F_0}
{T_{t_1-t_0} F_1
\cdots
T_{t_n-t_{n-1}} F_n}
_{\fun{\lp^{2}}{\sphilb{M}_{-2,\txtirsingular},\msrsf{G}_{\txtgs}}}
\end{aligned}
\end{equation}
holds.
This is the Markov representation by the semigroup $T_{\txtirsingular}
= \fml{T_{\txtirsingular,t}}{t > 0}$ after cutoff removal.
From the discussion so far, the coordinate process $\fml{\prbou_t}
{t \in \fldreal}$ is a stationary Gaussian process with static distribution $\msrsf{G}_{\txtgs}$, and since its finite-dimensional distributions satisfy the semigroup representation as above, it can be regarded as an Ornstein-Uhlenbeck-type $P(\phi)_1$-process.
\end{proof}

With this, the existence of the final ground state is established.

\begin{thm}
The following statements hold under the infrared singularity condition.
\begin{enumerate}
\item
The operator $\physham[L]_{\txtirsingular}$ is a non-negative self-adjoint operator with ground state energy $0$.

\item
The constant function $1
\in \fun{\lp^{2}}{\sphilb{M}_{-2,\txtirsingular},\msrsf{G}_{\txtgs}}$ belongs to the domain $\dom \physham[L]_{\txtirsingular}$ and satisfies$$\physham[L]_{\txtirsingular} 1
=
0.$$Hence $0$ is an eigenvalue of $\physham[L]_{\txtirsingular}$, and the corresponding eigenfunction $1$ is the ground state.

\item
The ground state is unique.
\end{enumerate}
\end{thm}

\begin{proof}
((1)--(2)): The generator $\physham[L]_{\kappa,\Lambda}$ with IR/UV cutoffs is non-negative, and the ground state corresponding to ground state energy $0$ is $1$.
By Proposition \ref{expedition0011540}, this family of generators converges in the strong resolvent sense to $\physham[L]_{\txtirsingular}$.
The strong resolvent limit of non-negative self-adjoint operators is again a non-negative self-adjoint operator.
By the Markov property, $\napiernum^{-t \physham[L]_{\kappa,\Lambda}} 1
=
1$ holds for each $\pairbk{\kappa,\Lambda}$ and $t
> 0$.
Since strong resolvent convergence implies uniform-on-compacts strong convergence of heat semigroups,$$\napiernum^{-t \physham[L]_{\txtirsingular}} 1
=
\lim_{\kappa \to 0, \Lambda \to \infty}
\napiernum^{-t \physham[L]_{\kappa,\Lambda}} 1
=
1$$is obtained.
Therefore $1 \in \dom \physham[L]_{\txtirsingular}$ and $\physham[L]_{\txtirsingular} 1
= 0$.

(3): By the Feynman-Kac-Nelson formula of Corollary \ref{expedition0011571}, $\napiernum^{-t \physham[L]_{\txtirsingular}}$ is a positivity-improving operator.
Therefore the ground state is unique.
\end{proof}

\section{Functional Integrals at Finite Temperature}\label{functional-integrals-at-finite-temperature}

\subsection{Introduction}\label{introduction}

Here we summarize the functional integral theory for bounded systems at finite temperature, following the discussion in \cite{YoshitsuguSekine004,DerezinskiGerard001}. For bounded systems, we assume the chemical potential satisfies \(\smchemicalpotential
< 0\) so that \(K_{\sminvtemperature,\smchemicalpotential}\) is a bounded operator.

\subsection{Reference: Quasi-local Construction of the Free Field}\label{expedition0012093}

We recall the setup from \cite{YoshitsuguSekine004} for the reader's convenience; readers already familiar may proceed to the next section. Throughout, the chemical potential satisfies \(\smchemicalpotential
< 0\).

For the circle \(S_{\sminvtemperature}
= \closedinterval{-\frac{\sminvtemperature}{2}}{\frac{\sminvtemperature}{2}}\), define the real Hilbert space \(\sphilb{K}
=
\fun{\lp^{2}}{S_{\sminvtemperature};\sphilb{H}_{\txtreal}}\) and the regularization operator \(B\) on \(\sphilb{K}\) as the multiplication operator whose Fourier transform is\[B(\omega,k)
=
\frac{\rbk{1 + \omega^2}^r
\rbk{1 + \abs{k}^2}^u}
{\abs{k}^{2a} \land 1}.\]In the three-dimensional case, using the notation of \cite{YoshitsuguSekine004},\[r = 1,
\quad
u = 2,
\quad
a =
\begin{dcases}
0, & s = 1, \\
1, & s = 2.
\end{dcases}\]Using \(B\), define\[\faadjpresharp{{\prbqspace_{\txttot}}}
=
\dom B^{\onehalf}
\subset
\sphilb{K},
\quad
\norm{f}_{\faadjpresharp{{\prbqspace_{\txttot}}}}
=
\norm{B^{\onehalf} f}_{\sphilb{K}}.\]On the other hand, \(\prbqspace_{\txttot}\) is the completion of \(\dom B^{-\onehalf}\) under the norm \(\norm{u}_{\prbqspace_{\txttot}}
=
\norm{B^{-\onehalf} u}_{\sphilb{K}}\), i.e.,\[\prbqspace_{\txttot}
=
\gtclos{\dom B^{-\onehalf}}^{\norm{\cdot}_{\prbqspace_{\txttot}}},\]and the generic variable in this space is denoted \(\opfocksegal\). Until further specified, the generic probability measure on this space is denoted \(\msr{\mu_{\txttot}}\).

Since \(\prbqspace_{\txttot}\) is a separable Hilbert space, its Borel algebra is countably generated. Let \(\msr{\mu_{\txtnonzero}}\) denote the Gaussian measure on \(\prbqspace_{\txttot}\): this Gaussian measure corresponds to the characteristic functional \(\prbcharfun_{C_{\smchemicalpotential}}(f)
= \napiernum^{-\oneoverfour \opform{q}_{C_{\smchemicalpotential}}(f)}\) defined from the sesquilinear form \(\opform{q}_{C_{\smchemicalpotential}}(f)
=
\norm{C_{\smchemicalpotential}^{\onehalf} f}_{\sphilb{K}}^{2}\) associated with the covariance operator \(C
= C_{\smchemicalpotential}\) to be specified below. This \(\opform{q}_{C_{\smchemicalpotential}}\) essentially corresponds to \(\opform{q}_{\txtnonzero}\).

For any \(t
\in S_{\sminvtemperature}\), define the isometric operator \(j_{t}\) by \begin{equation}
\begin{aligned}
j_t
\colon \rbk{2 \omega \tanh \frac{\sminvtemperature \omega}{2}}^{\onehalf} \sphilb{H}_{\txtreal}
\to \faadjpresharp{{\prbqspace_{\txttot}}};
\quad
j_t g
=
\diracdelta_t \otimes g.
\end{aligned}
\end{equation} Noting the agreement with the Araki-Woods vacuum expectation values of Weyl operators in the Araki-Woods representation as in \cite{YoshitsuguSekine004,DerezinskiGerard001}, the covariance operator \(C_{\smchemicalpotential}\) and its associated sesquilinear form for any \(g_{1},g_{2}
\in
\rbk{2 \omega \tanh \frac{\sminvtemperature \omega}{2}}^{\onehalf}
\sphilb{H}_{\txtreal}\) are \begin{equation}
\begin{aligned}
\opform{q}_{C_{\smchemicalpotential}}(j_{t_1} g_1, j_{t_2} g_2)
=
\bkt{j_{t_1} g_1}
{j_{t_2} g_2}
_{\faadjpresharp{{\prbqspace_{\txttot}}}}
=
\bkt{g_1}{\napiernum^{-\abs{t_1-t_2}}
K_{\sminvtemperature,\smchemicalpotential}
g_2}_{\sphilb{H}_{\txtreal}}
\end{aligned}
\end{equation}

as written above. Using this, we define the sharp-time field. In particular, for any\[s
\in S_{\sminvtemperature},
\quad
g
\in \rbk{2 \omega \tanh \frac{\sminvtemperature (\omega - \smchemicalpotential)}{2}}^{\onehalf}
\sphilb{H}_{\txtreal},\]set\[\opfocksegal_s(g)
=
\opfocksegal(j_s g)
=
\opfocksegal(\diracdelta_s \otimes g)
\in
\bigcap_{1 \leq p < \infty}
\fun{\lp^{p}}{\prbqspace_{\txttot}},
\quad
\opfocksegal
\in
\prbqspace_{\txttot}.\]

The reflection and time evolution are defined as follows. First, using \(r\) and \(u_{t}\), define \(\tilde{r} \opfocksegal
\in \prbqspace_{\txttot}\) and \(\tilde{u}_t \opfocksegal
\in \prbqspace_{\txttot}\) via the duality pairing:\[\dualbkt{\tilde{r} \opfocksegal}{f}
=
\dualbkt{\opfocksegal}{rf},
\quad
\dualbkt{\tilde{u}_{t} \opfocksegal}{f}
=
\dualbkt{\opfocksegal}{u_{t} f}.\]Furthermore, for \(F
\in \fun{\lp^{\infty}}{Q_{\txttot},\msr{\mu_{\txttot}}}\), define the reflection \(R\) and time evolution \(U_{t}\) by\[\funrbk{RF}{\opfocksegal}
=
\fun{F}{\tilde{r} \opfocksegal},
\quad
\funrbk{U_t F}{\opfocksegal}
=
\fun{F}{\tilde{u}_{t} \opfocksegal},\]and extend to \(\fun{\lp^{2}}{\prbqspace_{\txttot},
\mu_{\txttot}}\) by linearity and density.

The quasi-local construction is defined as follows. We consider a construction analogous to the quasi-local construction of the resolvent algebra. Specifically, for \(L
> 0\) set \(I_{L}
= \closedinterval{-\frac{L}{2}}{\frac{L}{2}}\), define all objects by the same construction as above on the complex Hilbert space \(\sphilb{H}_{L}
= \fun{\lp^{2}}{I_{L}^{d};\fldcmp}\), and attach the subscript \(L\) to the notation as in \(\sphilb{H}_{L}\).

Rather than working with a net of general bounded domains, we fix a reference length \(L_{0}
> 0\), take a strictly increasing sequence of positive integers \(\seqn{m}\) satisfying \(m_{n}
\mid m_{{n+1}}\), set \(L_{n}
= m_{n} L_{0}\), and consider the strictly monotone sequence given by \(\seq{I_{L_n}^{d}}{n \in \monnat}\). Similarly, denote the wavenumber space corresponding to the Fourier transform on \(I_{L}\) by \(\setlattice_{L}
=
\frac{2 \pi}{L} \ringratint\), and define the complete orthonormal system \(\basebk{e_{k}}_{k \in \setlattice_{L}^{d}}\) of \(\fun{\lpseq^{2}}{\setlattice_{L}^{d}}\) by \(e_{k}(m)
= \kroneckerdelta_{km}\). We also introduce \(V
= L^{d}\) as a notational shorthand. By the construction of the sequence, \(\setlattice_{L_{n}}^{d}
\subset \setlattice_{L_{n+1}}^{d}\) holds in wavenumber space.

In particular, denote the singular Gaussian \(\sminvtemperature\)-Markov path space for each \(L\) by\[\pairbk{\prbqspace_{\txttot,L},
\mblfmlfrak{S}_{\txttot,L},
\mblfmlfrak{S}_{\txttot,0,L},
R,
U_{t},
\msr{\mu_{\txttot,L}}}.\]Since the formal actions of the reflection and time translation do not change, we suppress the subscript \(L\) unless distinction is needed. The sequence formed by this quasi-local construction is called the local net in the general framework.

By the definitions in the quasi-local construction, the spatial inclusion \(\setlattice_{L_{n}}^{d}
\hookrightarrow \setlattice_{L_{n+1}}^{d}\) in wavenumber space induces the natural isometric embedding \begin{gather}
\tilde{\iota}_{L_{n}}^{L_{n+1}}
\colon \fun{\lpseq^{2}}{\setlattice_{L_{n}}^{d}}
\hookrightarrow \fun{\lpseq^{2}}{\setlattice_{L_{n+1}}^{d}};
\\ 
\funrbk{\tilde{\iota}_{L_{n}}^{L_{n+1}} \faftr{f}}{k}
=
\begin{dcases}
\faftr{f}(k), & k \in \setlattice_{L_{n}}^{d}, \\
0, & k \in \setlattice_{L_{n+1}}^{d} \setminus \setlattice_{L_{n}}^{d}.
\end{dcases}
\end{gather} For brevity, we denote the corresponding map on real space by the same symbol \(\tilde{\iota}_{L_{n}}^{L_{n+1}}\) when no confusion arises. Taking the tensor product with the time direction,\[\iota_{L_{n}}^{L_{n+1}}
= \idone \otimes \tilde{\iota}_{L_{n}}^{L_{n+1}}
\colon \sphilb{K}_{L_{n}}
\to \sphilb{K}_{L_{n+1}}\]is defined.

Furthermore, the adjoint of \(\iota_{L_{n}}^{L_{n+1}}\),\[P
= \faadjrbk{\iota_{L_{n}}^{L_{n+1}}}
\colon \sphilb{K}_{L_{n+1}}
\to \sphilb{K}_{L_{n}},\]is linear, continuous, and satisfies \(P \iota_{L_{n}}^{L_{n+1}}
= \id_{K_{L_{n}}}\). This \(P\) extends to \(\prbqspace_{\txttot}\) by regularization; denoting the extension also by \(P\), it becomes \(P
\colon \prbqspace_{\txttot,L_{n+1}}
\to \prbqspace_{\txttot,L_{n}}\). In particular, \(P\) is a projection, and we define the map \(\pi_{L_{n}}^{L_{n+1}}
= P\): this is simply the map that restricts wavenumbers to \(\setlattice_{L_{n}}^{d}\).

Next we show the commutativity of the isometric embedding and the regularization operator.

\begin{lem}\label{expedition0012071}
The isometric embedding $\iota$ and the regularization operator $B$ commute.
In particular, they also commute with the covariance operator $C_{\smchemicalpotential}$, and as an identity for the duality pairing $\dualbkt{\opfocksegal}{f}$, for any $f
\in \sphilb{K}_{L_{n}}$,$$\dualbkt{\pi_{L_{n}}^{L_{n+1}} \opfocksegal_{L_{n+1}}}
{f}
=
\dualbkt{\opfocksegal_{L_{n+1}}}
{\iota_{L_{n}}^{L_{n+1}} f}$$holds.
For readability, using the local version of the covariance operator denoted $\opform{q}_{C_{\smchemicalpotential},L_{n}}$:$$\fun{\opform{q}_{C_{\smchemicalpotential},L_{n+1}}}
{\iota_{L_{n}}^{L_{n+1}} f}
=
\fun{\opform{q}_{C_{\smchemicalpotential},L_{n}}}
{f},
\quad
f \in \opformdomain(\opform{q}_{C,L_{n}})$$holds.
\end{lem}
\begin{proof}
The original embedding $\iota$ is merely restriction in wavenumber space and is consistent with the restriction of $\omega(k)
= \abs{k}^{s}$ by the construction of the monotone sequence of spaces.
Thus $\iota$ is an intertwiner for $\omega$ and its operator calculus.
Since $\iota$ including the time direction acts as the identity in the time direction, by definition it acts as an intertwiner for the operators $B$ and $C$.

Since the sesquilinear form $\opform{q}_{C_{\smchemicalpotential},L_n}$ is associated with the covariance $C_{\smchemicalpotential,L_{n}}$, using the intertwining and isometry, for any $f
\in \sphilb{K}_{L_{n+1}}$,
\begin{equation}
\begin{aligned}
&\fun{\opform{q}_{C_{\smchemicalpotential},L_{n}}}
{\iota_{L_{n}}^{L_{n+1}} f}
=
\norm{C_{\smchemicalpotential,L_{n}}^{\onehalf} \circ \iota_{L_{n}}^{L_{n+1}} f}_{\sphilb{K}_{L_{n}}}^{2}
\\ 
&=
\norm{\iota_{L_{n}}^{L_{n+1}} \circ C_{\smchemicalpotential,L_{n+1}}^{\onehalf} f}_{\sphilb{K}_{L_{n+1}}}^{2}
=
\fun{\opform{q}_{C_{\smchemicalpotential},L_{n+1}}}
{f}
\end{aligned}
\end{equation}
holds.
Regarding the duality pairing, for any $f
\in \dom B^{\onehalf}$,
\begin{equation}
\begin{aligned}
&\dualbkt{\pi_{L_{n}}^{L_{n+1}} \opfocksegal_{L_{n+1}}}{f}
=
\bkt{B_{L_{n}}^{-\onehalf} \pi_{L_{n}}^{L_{n+1}} \opfocksegal_{L_{n+1}}}
{B_{L_{n}}^{\onehalf} f}
_{\sphilb{K}_{L_{n}}}
\\ 
&=
\bkt{\pi_{L_{n}}^{L_{n+1}} B_{L_{n+1}}^{-\onehalf} \opfocksegal_{L_{n+1}}}
{B_{L_{n}}^{\onehalf} f}
_{\sphilb{K}_{L_{n}}}
\\ 
&=
\bkt{B_{L_{n+1}}^{-\onehalf} \opfocksegal_{L_{n+1}}}
{\iota_{L_{n}}^{L_{n+1}} B_{L_{n}}^{\onehalf} f}
_{\sphilb{K}_{L_{n+1}}}
\\ 
&=
\bkt{B_{L_{n+1}}^{-\onehalf} \opfocksegal_{L_{n+1}}}
{B_{L_{n+1}}^{\onehalf} \iota_{L_{n}}^{L_{n+1}} f}
_{\sphilb{K}_{L_{n+1}}}
\\ 
&=
\dualbkt{\opfocksegal_{L_{n+1}}}{\iota_{L_{n}}^{L_{n+1}} f}
\end{aligned}
\end{equation}
holds.
\end{proof}

This quasi-local construction is consistent in the following sense.

\begin{prop}\label{expedition0012055}
For the isometric embeddings and projections associated with the monotone sequence $\seqn{L}$, the objects such as time evolution and probability measures are consistent.
In particular, for any natural number $n$, let the measurable map given by projection be$$\pi_{L_{n}}^{L_{n+1}}
\colon \pairbk{\prbqspace_{\txttot,L_{n+1}},\mblfmlfrak{S}_{\txttot,L_2}}
\to \pairbk{\prbqspace_{\txttot,L_{n}},\mblfmlfrak{S}_{\txttot,L_1}}.$$This satisfies the following properties.

\begin{enumerate}
\item
Consistency as a projective system: for any natural number $n$, almost surely with respect to $\msr{\mu_{\txttot,L_{n+2}}}$, $\pi_{L_{n}}^{L_{n+2}}
=
\pi_{L_{n}}^{L_{n+1}} \circ \pi_{L_{n+1}}^{L_{n+2}}$ holds.
The same consistency holds for any $n
< m$.

\item
Consistency of the field functional and the algebra: for any $f
\in \sphilb{K}_{L_{n}}$,$$\dualbkt{\pi_{L_{n}}^{L_{n+1}} \opfocksegal_{L_{n+1}}}{f}
=
\dualbkt{\opfocksegal_{L_{n+1}}}{\iota_{L_{n}}^{L_{n+1}} f}$$holds.
In particular, for the algebras,
\begin{equation}
\begin{aligned}
\fun{\invrbk{\pi_{L_{n}}^{L_{n+1}}}}{\mblfmlfrak{S}_{\txttot,0,L_{n}}}
&\subset \mblfmlfrak{S}_{\txttot,0,L_{n+1}},
\\ 
\fun{\invrbk{\pi_{L_{n}}^{L_{n+1}}}}{\mblfmlfrak{S}_{\txttot,L_{n}}}
&\subset
\mblfmlfrak{S}_{\txttot,L_{n+1}}
\end{aligned}
\end{equation}
holds.

\item
Consistency of the covariance operator or sesquilinear form: for the covariance operator or sesquilinear form,$$C_{\smchemicalpotential,L_{n+1}} \iota_{L_{n}}^{L_{n+1}}
=
\iota_{{L_{n}}}^{L_{n+1}} C_{\smchemicalpotential,L_{n}},
\quad
\fun{\opform{q}_{C_{\smchemicalpotential},L_{n+1}}}
{\iota_{L_{n}}^{L_{n+1}} f}
=
\fun{\opform{q}_{C_{\smchemicalpotential},L_{n}}}
{f}$$holds.

\item
Consistency of measures: as equality of marginal distributions, $\pushoutrbk{\pi_{L_{n}}^{L_{n+1}}}
\msr{\mu_{\txttot,L_{n+1}}}
=
\msr{\mu_{\txttot,L_{n}}}$ holds.

\item
Consistency of reflection and time translation: for any $p
\in \closedinterval{1}{\infty}$,
\begin{equation}
\begin{aligned}
\fun{U_{L_{n+1},t}}
{F \circ \pi_{L_{n}}^{L_{n+1}}}
&=
\rbk{U_{L_{n},t} F} \circ \pi_{L_{n}}^{L_{n+1}},
\quad
F \in \fun{\lp^{p}}{\prbqspace_{\txttot,L_{n}}},
\\ 
\fun{R_{L_{n+1}}}{F \circ \pi_{L_{n}}^{L_{n+1}}}
&=
\rbk{R_{L_{n}} F} \circ \pi_{L_{n}}^{L_{n+1}},
\quad
F \in \fun{\lp^{p}}{\prbqspace_{\txttot,L_{n}}}
\end{aligned}
\end{equation}
holds.
\end{enumerate}
For brevity, under this proposition the quantities that should be considered for $\seqn{L}$ may be written simply as $L$.
\end{prop}
\begin{proof}
(1): The isometric embeddings clearly satisfy the composition rule$$\iota_{L_{n}}^{L_{n+2}}
=
\iota_{L_{n+1}}^{L_{n+2}}
\circ
\iota_{L_{n}}^{L_{n+1}}$$by definition.
In particular, $\iota$ and its adjoint $\pi$ also satisfy the same consistency.

(2): The identity for the field operators was shown in Lemma \ref{expedition0012071}.
Since the algebra $\mblfmlfrak{S}_{\txttot,0,L}$ is generated by the sharp-time field at time $0$, the identity for the field operators carries over directly.
Since $\mblfmlfrak{S}_{\txttot,L}$ is generated by time translations of $\mblfmlfrak{S}_{\txttot,0,L}$, this is also clear.

(3): Shown in Lemma \ref{expedition0012071}.

(4): Since the Gaussian measure is uniquely determined by its characteristic functional, it suffices to show the consistency of the characteristic functionals, which was shown in (3) of this proposition.
More explicitly, first
\begin{equation}
\begin{aligned}
&\sqfun{\prbexp_{\msr{\pushoutrbk{\pi_{L_{n}}^{L_{n+1}}} \mu_{\txttot,L_{n+1}}}}}
{\napiernum^{\imunit \opfocksegal_{L_{n+1}}(f)}}
=
\sqfun{\prbexp_{\msr{\mu_{\txttot,L_{n+1}}}}}
{\napiernum^{\imunit \dualbkt{\pi_{L_{n}}^{L_{n+1}} \opfocksegal_{L_{n+1}}}{f}}}
\\ 
&=
\sqfun{\prbexp_{\msr{\mu_{\txttot,L_{n+1}}}}}
{\napiernum^{\imunit \dualbkt{\opfocksegal_{L_{n+1}}}{\iota_{L_{n}}^{L_{n+1}} f}}}
=
\fnexp{-\oneoverfour
\fun{\opform{q}_{C_{\smchemicalpotential},L_{n+1}}}
{\iota_{L_{n}}^{L_{n+1}} f}}
\\ 
&=
\fnexp{-\oneoverfour
\fun{\opform{q}_{C_{\smchemicalpotential},L_{n}}}
{f}}
=
\sqfun{\prbexp_{\msr{\mu_{\txttot,L_{n}}}}}
{\napiernum^{\imunit \opfocksegal_{L_{n}}(f)}}
\end{aligned}
\end{equation}
holds.
Extending this to $L^{2}$ and then applying indicator functions gives the consistency for the measures.

(5): By density of the linear span, it suffices to show for $f
\in \sphilb{K}_{L_{n}}$ and $F
= \napiernum^{\imunit \opfocksegal_{L_{n}}(f)}$.
Since the embedding $\iota_{L_{n}}^{L_{n+1}}$ does not affect the time variable,$$\iota_{L_{n}}^{L_{n+1}} (u_{-t} f)
=
\fun{u_{-t}}{\iota_{L_{n}}^{L_{n+1}} f},
\quad
\iota_{L_{n}}^{L_{n+1}} (r f)
=
\fun{r}{\iota_{L_{n}}^{L_{n+1}} f}$$hold.
For brevity, we occasionally suppress the variable $L$.
Then for time translation,
\begin{equation}
\begin{aligned}
&\fun{U_{L_{n+1},t}}
{\napiernum^{\imunit \opfocksegal_{L_{n}}(f)} \circ \pi}
=
\fun{U_{L_{n+1},t}}
{\napiernum^{\imunit \dualbkt{\pi \opfocksegal_{L_{n}}}{f}}}
\\ 
&=
\fun{U_{L_{n+1},t}}
{\napiernum^{\imunit \dualbkt{\opfocksegal_{L_{n}}}{\iota f}}}
\\ 
&=
\fnexp{\imunit \dualbkt{\opfocksegal_{L_{n+1}}}{u_{-t} \circ \iota f}}
=
\fnexp{\imunit \dualbkt{\opfocksegal_{L_{n+1}}}{\iota \circ u_{-t} f}}
\\ 
&=
\fnexp{\imunit \dualbkt{\pi \opfocksegal_{L_{n+1}}}{u_{-t} f}}
=
U_{L_{n},t}
\fnexp{\imunit \dualbkt{\pi \opfocksegal_{L_{n+1}}}{f}}
\\ 
&=
\rbk{U_{L_{n},t} \fnexp{\imunit \dualbkt{\opfocksegal_{L_{n}}}{f}}} \circ \pi
\end{aligned}
\end{equation}
holds.

The reflection case is computed in the same way.
\end{proof}

We now show the existence of the projective limit of the local measures.

\begin{prop}
A probability measure can be constructed as the projective limit of the sequence of measures associated with the sequence $\seqn{L}$.
Moreover, this probability measure is defined on a countably generated algebra.
\end{prop}
\begin{proof}
Define the projective limit set $\prbqspace_{\txttot,\infty}$ by$$\prbqspace_{\txttot,\infty}
=
\set{\seqn{q} \in \prod_{n \in \monnat} \prbqspace_{\txttot,L_{n}}}
{\pi_{L_{m}}^{L_{n}}(q_{n}) = q_{m} (m \leq n)},$$and let the coordinate projection be $\oppr_{n}
\colon \prbqspace_{\txttot,\infty}
\to \prbqspace_{\txttot,L_{n}}$.
The cylinder set family$$\mblfmlfrak{S}_{\txttot,\infty}
=
\mblfmlgenerated{\set{\inv{\oppr_{n}}(A)}
{n \in \monnat, A \in \mblfmlfrak{S}_{\txttot,L_{n}}}}$$is countably generated.
By the consistency in Proposition \ref{expedition0012055}, the finite-dimensional distributions of the family of probability measures on finite intersections of cylinder sets are well-defined.
By the Kolmogorov extension theorem, there exists a probability measure $\msr{\mu_{\txttot,\infty}}$ satisfying the pushforward condition $\pushoutrbk{\oppr_{n}} \msr{\mu_{\txttot,\infty}}
= \msr{\mu_{\txttot,L_{n}}}$.
\end{proof}

For self-containedness, we formulate a lemma for the sesquilinear form \(\opform{q}_{0}\) that does not use the zeroth-order Bessel function.

\begin{lem}\label{expedition0012066}
For any $f
\in \fun{\lp^{1}}{\fldreal^{d}}$, there exists a probability measure $\msr{\chi}$ realizing the sesquilinear form $\opform{q}_{0}$ as
\begin{equation}
\begin{aligned}
\fnexp{-\oneoverfour \opform{q}_{0}(f)}
&=
\int_{\fldreal^{2}}
\napiernum^{\imunit \ell_{\sminvtemperature,r,\theta}(f)}
\opdmsr{\chi(r,\theta)},
\\ 
\ell_{\sminvtemperature,r,\theta}(f)
&=
\sqrt{2 (2 \pi)^{d} \smnumberdensity_{0}(\sminvtemperature) r}
\fun{\opreal}{\napiernum^{\imunit \theta} \hat{f}(0)}.
\end{aligned}
\end{equation}
\end{lem}
\begin{proof}
For brevity, fix $f$ and set $z
= \hat{f}(0)
= a + \imunit b$, and let the coefficient of the functional $\ell_{\sminvtemperature,r,\theta}$ be $c
= \sqrt{2 (2 \pi)^{d} \smnumberdensity_{0}(\sminvtemperature)}$.

First, we examine the characteristic function of Gaussian random variables.
Let $\pairbk{X,Y}$ be a pair of centered Gaussian random variables in two dimensions that are mutually independent with covariance $\onehalf$; then$$\sqfun{\prbexp}
{\napiernum^{\imunit (ax + by)}}
= \napiernum^{-\oneoverfour (a^{2} + b^{2})}.$$Setting $G_{f}
=
c
\fun{\opreal}
{(X + \imunit Y) z}
=
c (aX - bY)$, this is a one-dimensional centered Gaussian random variable.
The variance is
\begin{equation}
\begin{aligned}
&\sqfun{\prbvar}{G_{f}}
=
c^{2}
\sqfun{\prbvar}
{aX - bY}
=
c^{2}
\rbk{a^{2} \prbvar X + b^{2} \prbvar Y}
\\ 
&=
c^{2}
\frac{a^{2} + b^{2}}{2}
=
\frac{c^{2} \abs{z}^{2}}{2}
\end{aligned}
\end{equation}
can be computed, giving$$\sqfun{\prbexp_{\pairbk{X,Y}}}
{\fnexp{\imunit c \fun{\opreal}{(X + \imunit Y) z}}}
=
\napiernum^{-\oneoverfour \opform{q}_{0}(f)}.$$

Represent $\fldreal^{2}$ in two-dimensional polar coordinates, and let the radial and angular measures be$$\opdmsr{\nu(r)}
=
\napiernum^{-r} \opdmsr{r},
\quad
\opdmsr{\lambda}(\theta)
=
\frac{1}{2 \pi} \opdmsr{\theta},$$and take their product measure to be $\msr{\chi}$.
As random variables, for $\omega
= \vecbk{r,\theta}
\in \fldreal^{2}$, define $R(\omega)
= r$ and $\Theta(\omega)
= \theta$.
These satisfy$$R
= X^{2} + Y^{2},
\quad
\Theta
=
\arg (X + \imunit Y),$$and are meaningful as measurable maps.
They satisfy$$\sqrt{R}
\fun{\opreal}
{\napiernum^{\imunit \Theta} z}
=
\fun{\opreal}
{(X + \imunit Y) z},
\quad
\ell_{\sminvtemperature,R,\Theta}(f)
=
c
\fun{\opreal}
{(X + \imunit Y) z}.$$Rewriting the integral representation,
\begin{equation}
\begin{aligned}
&\int_{\fldreal^{2}}
\napiernum^{\imunit \ell_{\sminvtemperature,r,\theta}(f)}
\opdmsr{\chi(r,\theta)}
=
\sqfun{\prbexp_{\pairbk{R,\Theta}}}
{\napiernum^{\imunit \ell_{\sminvtemperature,R,\Theta}(f)}}
\\ 
&=
\sqfun{\prbexp_{\pairbk{X,Y}}}
{\fnexp{\imunit
c
\fun{\opreal}
{(X + \imunit Y) z}}}
\end{aligned}
\end{equation}
is obtained.
By the argument for the characteristic function of the one-dimensional Gaussian random variable above,
\begin{equation}
\begin{aligned}
&\int_{\fldreal^{2}}
\napiernum^{\imunit \ell_{\sminvtemperature,r,\theta}(f)}
\opdmsr{\chi(r,\theta)}
=
\sqfun{\prbexp_{\pairbk{X,Y}}}
{\fnexp{\imunit
c
\fun{\opreal}
{(X + \imunit Y) z}}}
\\ 
&=
\fnexp{-\oneoverfour
c^{2} z^{2}}
=
\fnexp{-\oneoverfour \opform{q}_{0}(f)}
\end{aligned}
\end{equation}
is obtained.
\end{proof}

Let \(\msr{\mu_{\txtbec}}\) denote the probability measure on \(\prbqspace_{\txttot}\) associated with the sesquilinear form \(\opform{q}_{\txtbec}\) via\[\sqfun{\prbexp_{\msr{\mu_{\txtbec}}}}
{\napiernum^{\imunit \opfocksegal(f)}}
=
\fnexp{-\oneoverfour
\opform{q}_{\txtbec}(f)}.\]

\begin{prop}\label{expedition0012063}
After restoring the chemical potential $\smchemicalpotential
< 0$, taking the limit $L
\to \infty$ of the local measures, and then taking the limit $\smchemicalpotential
\uparrow 0$ for the chemical potential, for any $f
\in \sphilb{D}_{0,\sminvtemperature}$,
\begin{equation}
\begin{aligned}
\sqfun{\prbexp_{\msr{\mu_{\txtbec}}}}
{\napiernum^{\imunit \opfocksegal(j_{0} f)}}
&=
\fnexp{-\oneoverfour
\opform{q}_{0}(f)}
\sqfun{\prbexp_{\msr{\mu_{\txtnonzero}}}}
{\napiernum^{\imunit \opfocksegal(j_{0} f)}}
\\ 
&=
\int_{\fldreal^{2}}
\sqfun{\prbexp_{\msr{\mu_{\txtnonzero}}}}
{\napiernum^{\imunit \rbk{\opfocksegal(j_{0} f) + \ell_{\sminvtemperature,r,\theta}(f)}}}
\opdmsr{\chi(r,\theta)}
\end{aligned}
\end{equation}
holds.
In particular, $\opform{q}_{0}$ is invariant under the action of $U_{t}$.
\end{prop}

\begin{proof}
By definition, $\opform{q}_{0}$ evaluates values at the origin in wavenumber space.
Since $U_{t}$ acts via $\napiernum^{-\sminvtemperature \omega(0)}
= 1$, the triviality of the action on $\opform{q}_{0}$ follows.

To obtain the measure estimate and limit, it suffices to evaluate the characteristic functional.
Fix $L
> 0$ and take any $f
\in \fun{\lpseq^{2}}{\setlattice_{L}^{d}}$; then$$\sqfun{\prbexp_{\msr{\mu_{\txttot,L}}}}
{\napiernum^{\imunit \opfocksegal(j_{0} f)}}
=
\fnexp{-\oneoverfour \opform{q}_{\txtnonzero,\smchemicalpotential,L}(f)}.$$Set $f_{k}
= \bkt{f}{e_{k}}$ for the above $f$.
Computing in wavenumber space,
\begin{equation}
\begin{aligned}
\opform{q}_{\txtnonzero,\smchemicalpotential,L}(f)
&=
\sum_{k \in \setlattice_{L}^{d}}
\frac{1 + \napiernum^{-\sminvtemperature (\omega(k) - \smchemicalpotential)}}
{1 - \napiernum^{-\sminvtemperature (\omega(k) - \smchemicalpotential)}}
\abs{f_{k}}^{2}
=
\opform{q}_{0,\smchemicalpotential,L}(f)
+\opform{q}_{\txtnonzero,\smchemicalpotential,L}(f)
\end{aligned}
\end{equation}
holds.
By definition, for $f
\in
\opformdomain(\opform{q}_{0})
\cap
\opformdomain(\opform{q}_{\txtnonzero})$,
under the Bose-Einstein condensation condition, taking $L
\to \infty$ followed by $\smchemicalpotential
\downarrow 0$ gives$$\opform{q}_{\txtnonzero,\smchemicalpotential,L}(f)
\to
\opform{q}_{\txtbec}(f)
=
\opform{q}_{0}(f) + \opform{q}_{\txtnonzero}(f).$$It then suffices to substitute this back into the exponent.

By Corollary \ref{expedition0012066},$$\fnexp{-\oneoverfour
\opform{q}_{0}(f)}
=
\int_{\fldreal^{2}}
\fnexp{\imunit \ell_{\sminvtemperature,r,\theta}(f)}
\opdmsr{\chi(r,\theta)}.$$Substituting,
\begin{equation}
\begin{aligned}
&\sqfun{\prbexp_{\msr{\mu_{\txtbec}}}}
{\napiernum^{\imunit \opfocksegal(j_{0} f)}}
=
\fnexp{-\oneoverfour \opform{q}_{0}(f)}
\sqfun{\prbexp_{\msr{\mu_{\txtnonzero}}}}
{\napiernum^{\imunit \opfocksegal(j_{0} f)}}
\\ 
&=
\int_{\fldreal^{2}}
\napiernum^{\imunit \ell_{\sminvtemperature,r,\theta}(f)}
\opdmsr{\chi(r,\theta)}
\cdot
\sqfun{\prbexp_{\msr{\mu_{\txtnonzero}}}}
{\napiernum^{\imunit \opfocksegal(j_{0} f)}}
\\ 
&=
\int_{\fldreal^{2}}
\sqfun{\prbexp_{\msr{\mu_{\txtnonzero}}}}
{\napiernum^{\imunit \rbk{\opfocksegal(j_{0} f) + \ell_{\sminvtemperature,r,\theta}(f)}}}
\opdmsr{\chi(r,\theta)}
\end{aligned}
\end{equation}
is obtained.
\end{proof}

\subsection{Quasi-local Construction of the van Hove Model with IR/UV Cutoffs}\label{quasi-local-construction-of-the-van-hove-model-with-iruv-cutoffs}

Using the notation of the previous section, we introduce the van Hove model for a bounded system at finite temperature by applying the perturbation theory for interacting systems of \cite{DerezinskiGerard001}, and discuss the quasi-local construction. Throughout, the chemical potential satisfies \(\smchemicalpotential
< 0\).

For any \(\opfocksegal_{L}
\in \prbqspace_{\txttot,L}\), define \(\mathsf{V}_{\kappa,\Lambda,L}(\opfocksegal_{L})
=
\int_{S_{\sminvtemperature}}
\opfocksegal_{L}(j_t \omega \mathsf{m}_{\kappa,\Lambda})
\opdmsr{t}\).

\begin{prop}
For any $L, \lambda, \Lambda$, $\smpartitionfunc_{\sminvtemperature,\kappa,\Lambda,L}
=
\int_{\prbqspace_{\txttot,L}}
\napiernum^{-\mathsf{V}_{\kappa,\Lambda,L}}
\opdmsr{\mu_{\txttot,L}}$ is strictly positive and finite.
\end{prop}

\begin{proof}
The functional $\mathsf{V}_{\kappa,\Lambda,L}$ is a linear functional in $\opfocksegal_{L}$ and is a centered Gaussian random variable under the measure $\msr{\mu_{\txttot,L}}$.
In particular, the exponential moment is finite and strictly positive.
\end{proof}

Based on this proposition, define the Gibbs measure for the bounded system as\[\opdmsr{\mu_{\txtvanhowe,\sminvtemperature,\kappa,\Lambda,\smchemicalpotential,L}}
=
\frac{1}{\smpartitionfunc_{\sminvtemperature,\kappa,\Lambda,L}}
\napiernum^{-\mathsf{V}_{\kappa,\Lambda,L}}
\opdmsr{\mu_{\txttot,L}}.\]

\begin{prop}
The measure $\msr{\mu_{\txtvanhowe,\sminvtemperature,\kappa,\Lambda,\smchemicalpotential,L}}$ defined above is a Gaussian measure, and its characteristic functional for $f
\in \faadjpresharp{{\prbqspace_{\txttot,L}}}$ is
\begin{equation}
\begin{aligned}
&\int_{\prbqspace_{\txttot,L}}
\napiernum^{\imunit \dualbkt{\opfocksegal_{L}}{j_{t} f}}
\opdmsr{\mu_{\txtvanhowe,\sminvtemperature,\kappa,\Lambda,\smchemicalpotential,L}(\opfocksegal_{L})}
\\ 
&=
\fnexp{-\oneoverfour
\opform{q}_{C_{\smchemicalpotential,L}}(j_{t} f)
-\frac{\imunit}{2}
\int_{S_{\sminvtemperature}}
\opform{q}_{C_{\smchemicalpotential},L}(j_{t} f, j_s(\omega \mathsf{m}_{\kappa,\Lambda}))
\opdmsr{s}}
\\ 
&=
\fnexp{-\oneoverfour
\opform{q}_{C_{\smchemicalpotential,L}}(j_{t} f)
-\imunit \mathsf{m}_{\kappa,\Lambda}(f)}.
\end{aligned}
\end{equation}
In particular, the mean and covariance are respectively
\begin{equation}
\begin{aligned}
\sqfun{\prbexp_{\mu_{\txtvanhowe,\sminvtemperature,\kappa,\Lambda,\smchemicalpotential,L}}}
{\opfocksegal_{L}(j_{t} f)}
&=
-\onehalf \mathsf{m}_{\kappa,\Lambda}(f),
\\ 
\sqfun{\prbcov_{\msr{\mu_{\txtvanhowe,\sminvtemperature,\kappa,\Lambda,\smchemicalpotential,L}}}}
{\opfocksegal_{L}(j_{t_1} f),\opfocksegal_{L}(j_{t_2} g)}
&=
\onehalf
\opform{q}_{C_{\smchemicalpotential,L}}(j_{t_1} f, j_{t_2} g),
\end{aligned}
\end{equation}
and the covariance coincides with that of the free field.
\end{prop}

\begin{proof}
First, compute the characteristic functional.
For brevity, set $U = \int_{S_{\sminvtemperature}} j_{s}(\omega \mathsf{m}_{\kappa,\Lambda}) \opdmsr{s}$; then by $\mathsf{V}_{\kappa,\Lambda,L}(\opfocksegal_{L}) = \dualbkt{\opfocksegal_{L}}{U}$,
$$\int_{\prbqspace_{\txttot,L}}
\napiernum^{\imunit \dualbkt{\opfocksegal_{L}}{j_{t} f}}
\opdmsr{\mu_{\txtvanhowe,\sminvtemperature,\kappa,\Lambda,\smchemicalpotential,L}(\opfocksegal_{L})}
=
\frac{\int_{\prbqspace_{\txttot,L}}
\napiernum^{\imunit \dualbkt{\opfocksegal_{L}}{j_{t} f} - \dualbkt{\opfocksegal_{L}}{U}}
\opdmsr{\mu_{\txttot,L}}}
{\int_{\prbqspace_{\txttot,L}}
\napiernum^{-\dualbkt{\opfocksegal_{L}}{U}}
\opdmsr{\mu_{\txttot,L}}}.$$Substituting $\lambda
= 1$, $F = j_{T} f$, $G = U$ in the general formula for a centered Gaussian measure$$\int
\napiernum^{\imunit \bkt{\opfocksegal}{F} + \lambda \bkt{\opfocksegal}{G}}
\opdmsr{\mu_{C}(\opfocksegal)}
=
\fnexp{-\onehalf \opform{q}_{C}(F)
+\frac{\imunit \lambda}{2} \opform{q}_{C}(F,G)
+\frac{\lambda^2}{4} \opform{q}_{C}(G)},$$and canceling with the denominator $\smpartitionfunc_{\sminvtemperature,\kappa,\Lambda,L}
=
\napiernum^{\oneoverfour \opform{q}_{C_{\smchemicalpotential},L}(U)}$:$$\int_{\prbqspace_{\txttot,L}}
\napiernum^{\imunit \dualbkt{\opfocksegal_{L}}{j_{t} f}}
\opdmsr{\mu_{\txtvanhowe,\sminvtemperature,\kappa,\Lambda,\smchemicalpotential,L}(\opfocksegal_{L})}
=
\fnexp{-\oneoverfour
\opform{q}_{C_{\smchemicalpotential,L}}(j_{t} f)
-\frac{\imunit}{2}
\int_{S_{\sminvtemperature}}
\opform{q}_{C_{\smchemicalpotential},L}(j_{t} f, j_s(\omega \mathsf{m}_{\kappa,\Lambda}))
\opdmsr{s}}.$$

It remains to simplify the mean.
Noting the equivalence with the Araki-Woods representation from \cite{DerezinskiGerard001}, the sharp-time covariance is$$\opform{q}_{C_{\smchemicalpotential},L}(j_{t} f, j_{s} g)
=
\begin{dcases}
\bkt{f}{\frac{\napiernum^{-(t-s) \omega}}{1 - \napiernum^{-\sminvtemperature \omega}} g}, & t \geq s, \\
\bkt{f}{\frac{\napiernum^{-(s-t) \omega}}{1 - \napiernum^{-\sminvtemperature \omega}} g}, & s > t.
\end{dcases}$$Expressed as a periodic kernel,$$\opform{q}_{C_{\smchemicalpotential},L}(j_{t} f, j_{s} g)
=
\bkt{f}{C_{\sminvtemperature}(t-s) g},
\quad
C_{\sminvtemperature}(\tau)
=
\frac{\napiernum^{-\abs{\tau} \omega} + \napiernum^{-(\sminvtemperature - \abs{\tau}) \omega}}
{1 - \napiernum^{-\sminvtemperature \omega}}.$$

Computing $\int_{S_{\sminvtemperature}}
C_{\sminvtemperature}(t-s)
\opdmsr{s}$:
\begin{equation}
\begin{aligned}
&\int_{S_{\sminvtemperature}}
C_{\sminvtemperature}(t-s)
\opdmsr{s}
=
\int_{-\frac{\sminvtemperature}{2}}^{\frac{\sminvtemperature}{2}}
\frac{\napiernum^{-\abs{t-s} \omega} + \napiernum^{-(\sminvtemperature - \abs{t-s}) \omega}}
{1 - \napiernum^{-\sminvtemperature \omega}}
\opdmsr{s}
\\ 
&=
\frac{2}{1 - \napiernum^{-\sminvtemperature \omega}}
\int_0^{\frac{\sminvtemperature}{2}}
\rbk{\napiernum^{-u \omega}
+\napiernum^{-(\sminvtemperature - u) \omega}}
\opdmsr{u}
\\ 
&=
\frac{2}{1 - \napiernum^{-\sminvtemperature \omega}}
\rbk{\frac{1 - \napiernum^{-\frac{\sminvtemperature}{2} \omega}}{\omega}
+\frac{\napiernum^{-\frac{\sminvtemperature}{2} \omega} - \napiernum^{-\frac{\sminvtemperature}{2} \omega}}{\omega}}
=
\frac{2}{\omega}
\end{aligned}
\end{equation}
is obtained.
Therefore$$I
=
\int_{S_{\sminvtemperature}}
\opform{q}_{C_{\smchemicalpotential},L}(j_{t} f, j_{s}(\omega \mathsf{m}_{\kappa,\Lambda}))
\opdmsr{s}
=
\bkt{f}
{\rbk{\int_{S_{\sminvtemperature}} C_{\sminvtemperature}(t-s)}
\omega \mathsf{m}_{\kappa,\Lambda}}
\opdmsr{s}
=
\mathsf{m}_{\kappa,\Lambda}(f).$$
\end{proof}

The consistency of the Gibbs measures also follows.

\begin{prop}
With the notation of the previous section, the consistency $\pushoutrbk{\pi_{L_n}^{L_{n+1}}}
\msr{\mu_{\txtvanhowe,\sminvtemperature,\kappa,\Lambda,\smchemicalpotential,L_{n+1}}}
=
\msr{\mu_{\txtvanhowe,\sminvtemperature,\kappa,\Lambda,\smchemicalpotential,L_{n}}}$ holds.
\end{prop}

\begin{thm}[Quasi-local limit]
On the projective limit $\prbqspace_{\txttot,\infty}$, there exists a probability measure $\msr{\mu_{\txtvanhowe,\sminvtemperature,\kappa,\Lambda,\smchemicalpotential}}$ satisfying the condition $\pushoutrbk{\pi_{L_n}^{\infty}}
\msr{\mu_{\txtvanhowe,\sminvtemperature,\kappa,\Lambda,\smchemicalpotential}}
=
\msr{\mu_{\txtvanhowe,\sminvtemperature,\kappa,\Lambda,\smchemicalpotential,L_{n}}}$.
\end{thm}

\begin{proof}
By the Kolmogorov extension theorem \ref{expedition0011529} for projective systems.
\end{proof}

This limiting measure has the following characterization.

\begin{prop}\label{expedition0012094}
The mean and covariance of the process $\opfocksegal_t(f)
= \opfocksegal(j_{t} f)$ under the measure $\msr{\mu_{\txtvanhowe,\sminvtemperature,\kappa,\Lambda,\smchemicalpotential}}$ are
\begin{equation}
\begin{aligned}
\sqfun{\prbexp_{\msr{\mu_{\txtvanhowe,\sminvtemperature,\kappa,\Lambda,\smchemicalpotential}}}}
{\opfocksegal(j_{t} f)}
&=
-\onehalf
\mathsf{m}_{\kappa,\Lambda}(f),
\\ 
\fun{\prbcov_{\mu_{\txtvanhowe,\sminvtemperature,\kappa,\Lambda,\smchemicalpotential}}}{\opfocksegal(j_{t_1} g_1),
\opfocksegal(j_{t_2} g_2)}
&=
\fun{\opform{q}_{C_{\smchemicalpotential}}}{j_{t_1} g_1, j_{t_2} g_2}.
\end{aligned}
\end{equation}
In particular, the characteristic function for $f
\in \faadjpresharp{{\prbqspace_{\txttot}}}$ is$$\int_{\prbqspace_{\txttot}}
\napiernum^{\imunit \dualbkt{\opfocksegal}{j_{t} f}}
\opdmsr{\mu_{\txtvanhowe,\sminvtemperature,\kappa,\Lambda,\smchemicalpotential}(\opfocksegal)}
=
\fnexp{-\oneoverfour
\opform{q}_{C_{\smchemicalpotential}}(j_{t} f)
-\frac{\imunit}{2}
\int_{S_{\sminvtemperature}}
\opform{q}_{C_{\smchemicalpotential}}(j_{t}f, j_{s}(\omega \mathsf{m}_{\kappa,\Lambda}))
\opdmsr{s}},$$and the van Hove model with IR/UV cutoffs is a Gaussian process given by adding a constant drift to the periodic Ornstein-Uhlenbeck process.
\end{prop}

\begin{proof}
Obtained as the limit of the bounded system.
\end{proof}

Let us consider the removal of the chemical potential and the emergence of the BEC component. Following the discussion in Section \ref{expedition0012093}, it suffices to handle the term \(\opform{q}_{C_{\smchemicalpotential}}(j_{t} f)\) in Proposition \ref{expedition0012094}. In particular, only the component corresponding to \(\opform{q}_0\) needs to appear; in the bounded system,\[\opform{q}_{C_{\smchemicalpotential},L}(j_{t} f)
=
\bkt{f}
{\napiernum^{-\onehalf \abs{t} \omega}
K_{\sminvtemperature,\smchemicalpotential} f}_{\sphilb{H}_{\txtreal,L}}
=
\fun{\opform{q}_{\txtnonzero,L}}
{\napiernum^{-\onehalf \abs{t} \omega} f}.\]The rest coincides with the argument for the ideal Bose gas, and the following proposition is obtained.

\begin{prop}
After taking the infinite-volume limit $L
\to \infty$ and then the limit $\smchemicalpotential
\to 0$ for the chemical potential, results analogous to the Weyl algebra and resolvent algebra discussion are obtained.
Denoting the measure that gives these results by $\msr{\mu_{\txtvanhowe,\sminvtemperature,\kappa,\Lambda}}$, the characteristic function is
\begin{equation}
\begin{aligned}
&\int_{\prbqspace_{\txttot}}
\napiernum^{\imunit \dualbkt{\opfocksegal}{j_{t} f}}
\opdmsr{\mu_{\txtvanhowe,\sminvtemperature,\kappa,\Lambda}(\opfocksegal)}
\\ 
&=
\fnexp{-\oneoverfour
\fun{\opform{q}_{0}}{\napiernum^{-\onehalf \abs{t} \omega} f}
-\oneoverfour
\fun{\opform{q}_{\txtnonzero}}{\napiernum^{-\onehalf \abs{t} \omega} f}
-\frac{\imunit}{2}
\mathsf{m}(f)}.
\end{aligned}
\end{equation}
\end{prop}

\begin{rem}
The BEC component $\opform{q}_0$ is meaningful (non-zero) only under the conditions for the ideal Bose gas discussed in \cite{AsaoArai28,YoshitsuguSekine004}.
\end{rem}

Henceforth, the chemical potential is set to \(\smchemicalpotential
= 0\) and suppressed in the notation. Furthermore, the BEC component \(\opform{q}_0\) is always written explicitly, including the case where it is \(0\).

\subsection{Removal of Infrared and Ultraviolet Cutoffs}\label{removal-of-infrared-and-ultraviolet-cutoffs-2}

The cutoff removal in this section can essentially be treated in the same way as Sections \ref{expedition0011370} and \ref{expedition0011528}, in particular Section \ref{expedition0011528}. Here we confine ourselves to describing the key points concerning convergence.

Let the space of test functions that remain meaningful after removing the infrared and ultraviolet singularities be \(\sphilb{D}_{\txtirsingular,\sminvtemperature}
=
\dom \mathsf{m} \cap \sphilb{D}_{\sminvtemperature}\), and define the index set \(I_{\sminvtemperature}
=
\setpowfin{S_{\sminvtemperature} \times \sphilb{D}_{\txtirsingular,\sminvtemperature}}\). For a finite set \(E
=
\setone{(t_1,f_1),\ldots,(t_n,f_n)}
\in I_{\sminvtemperature}\), define the finite-dimensional characteristic function of the model with cutoffs by\[\fun{\chi_{\sminvtemperature,\kappa,\Lambda}^{E}}{\nfoldvar{s}{n}}
=
\int_{\prbqspace_{\txttot}}
\fnexp{\imunit
\sum_{j=1}^n s_j
\opfocksegal(j_{t_j} f_j)}
\opdmsr{\mu_{\txtvanhowe,\sminvtemperature,\kappa,\Lambda}}.\]The van Hove model with cutoffs is a Gaussian measure that shifts the mean of the free finite-temperature Gaussian field by \(\mathsf{m}_{\kappa,\Lambda}\). Hence the finite-dimensional characteristic function can be written as\[\chi_{\sminvtemperature,\kappa,\Lambda}^{E}(s)
=
\fnexp{-\frac{\imunit}{2}
\sum_{j=1}^n s_j \mathsf{m}_{\kappa,\Lambda}(f_j)
-\oneoverfour
\sum_{j=1}^n s_j
\rbk{\fun{\opform{q}_{0}}{\napiernum^{-\onehalf \abs{t_j} \omega} f}
+\fun{\opform{q}_{\txtnonzero}}{\napiernum^{-\onehalf \abs{t_j} \omega} f}}}.\]Here the covariance \(C_{\sminvtemperature}\) is the covariance of the free finite-temperature field and is independent of the IR/UV cutoffs. The cutoff dependence appears only in the mean \(\mathsf{m}_{\kappa,\Lambda}\).

For each \(f
\in \sphilb{D}_{\txtirsingular,\sminvtemperature}\), \(\mathsf{m}_{\kappa,\Lambda}(f)
\to
\mathsf{m}(f)\) holds in the sense of weak convergence in Hilbert space. Therefore \(M_{\kappa,\Lambda,t}(f)\) appearing in the time evolution also converges to \(M_{t}(f)\). In particular, for each \(E
\in I_{\sminvtemperature}\),\[\chi_{\sminvtemperature}^E(s)
=
\lim_{\kappa\downarrow0,\Lambda\uparrow\infty}
\chi_{\sminvtemperature,\kappa,\Lambda}^{E}(s)\]exists. This convergence is pointwise convergence of characteristic functionals.

Furthermore, the family of finite-dimensional characteristic functions with cutoffs \(\fml{\chi_{\sminvtemperature,\kappa,\Lambda}^E}
{E \in I_{\sminvtemperature}}\) satisfies projective consistency. The limit is taken for each finite set \(E\) independently, and since the marginalization operation is represented by setting variables of the finite-dimensional characteristic function to \(0\), the family \(\fml{\chi_{\sminvtemperature}^E}
{E \in I_{\sminvtemperature}}\) after the limit satisfies the same projective consistency.

By the Kolmogorov-type projective limit theorem, for any \(E
=
\setone{(t_j,f_j)}_{j=1}^n
\in I_{\sminvtemperature}\), there exists a probability measure \(\mu_{\txtvanhowe,\sminvtemperature}\) satisfying\[\fun{\chi_{\sminvtemperature}^E}{\nfoldvar{s}{n}}
=
\int_{\prbqspace_{\txttot}}
\fnexp{\imunit
\sum_{j=1}^n
s_j
\opfocksegal(j_{t_j} f_j)}
\opdmsr{\mu_{\txtvanhowe,\sminvtemperature}}.\]The convergence of measures is weak convergence of probability measures arising from tightness.

This measure is called the Euclidean measure of the finite-temperature van Hove model with IR/UV cutoffs removed. The effect of the infrared and ultraviolet singularities appears not in the covariance of the measure but in the restriction of the admissible test function space to \(\sphilb{D}_{\txtirsingular,\sminvtemperature}
=
\dom \mathsf{m} \cap D_{\sminvtemperature}\).

\bibliography{myref.bib}

\end{document}